\documentclass [a4paper,12pt]{JNU-PhD}
\usepackage{algpseudocode}
\usepackage{amsthm}
\usepackage{rotating}
\usepackage{subcaption}
\usepackage{longtable}

\theoremstyle{definition}
\newtheorem{definition}{Definition}[chapter]
\newtheorem{theorem}{Theorem}[chapter]

\usepackage[chapter]{algorithm}

\begin{document}
\onehalfspacing
\title{\textcolor{black}{Design and Analysis of Pairing-Friendly Elliptic Curves for Cryptographic Primitives}}
\author{Mahender Kumar}
\advisor{Prof. Satish Chand}
\dean{}

\degree{Doctor of Philosophy}
\field{Computer Science \& Technology}
\degreeyear{2020}
\degreemonth{June}

\department{School of Computer \& Systems Sciences}
\university{Jawaharlal Nehru University}
\universitycity{New Delhi}
\universitycitypincode{110067}
\universitystate{}
\universitycountry{INDIA}
\maketitle
\copyrightpage

\addcontentsline{toc}{chapter}{Certificate}
\certificatepage

\addcontentsline{toc}{chapter}{Declaration}
\declarationpage

\justify

\addcontentsline{toc}{chapter}{Abstract}
\abstractpage

\addcontentsline{toc}{chapter}{List of contents}
\tableofcontents
\addcontentsline{toc}{chapter}{List of publications}
\listofpublications
\addcontentsline{toc}{chapter}{List of figures}
\addcontentsline{toc}{chapter}{List of tables}
\listoftables
\addcontentsline{toc}{chapter}{List of algorithms}
\listofalgorithms

\addcontentsline{toc}{chapter}{List of Symbols}
\listofsymbol

\raggedleft
\dedicationpage

\justify
\addcontentsline{toc}{chapter}{Acknowledgements}
\acknowledgments


\begin{savequote}[100mm] 
Computer science is no more about computers than astronomy
is about telescopes.
\qauthor{Edsger Dijkstra} 
\end{savequote}

\chapter{Introduction}

\begin{doublespace}
Cryptography has been used by military and defence organizations to secure communication since the early 1900s. Cryptography is increasingly used to protect online transactions, data communication, patient health information, internet access, and cloud computation. Cryptosystems are divided into two categories based on the number of cryptographic keys used: private-key cryptosystems (PrKC) or symmetric-key cryptosystems, and public-key cryptosystems (PKC) or asymmetric-key cryptosystems. In a PrKC, two parties use the same secret key for encryption and decryption. The symmetric key must be random or cryptographically generated to ensure security. Although PrKC is fast and straightforward, securely exchanging keys between two parties who have not met before is challenging. One significant cryptography breakthrough was the PKC development by Diffie and Hellman \cite{diffie1976new}. PKC addresses the key distribution problem using public and private keys. The keys are mathematically linked, and the message is encrypted with the public key and decrypted with the private key. This method is slower than PrKC but more flexible. Rivest \textit{et al.} \cite{rivest1978method} introduced RSA encryption, a PKC method based on the Diffie-Hellman model and integer factorization problem. RSA provides digital data and communication integrity, confidentiality, authenticity, and non-repudiation. The ElGamal encryption scheme, defined over the cyclic group of the elliptic curve, is another PKC method presented by ElGamal \cite{elgamal1985public} in 1985.

In 1985, Koblitz \cite{koblitz1987elliptic} and Miller \cite{miller1985use} independently discovered the Elliptic Curve Cryptosystem (ECC), a PKC’s variant, which was constructed on the group of points on an elliptic curve over a finite field. The ECC security is equivalent to solving the Discrete Logarithm Problem on the Elliptic Curve (ECDLP). In contrast, the security of the RSA-based cryptosystem is based on finding the factor of the product of two large prime integers. The factorization Problem (FP) is solved using the Number Field Sieve (NFS) Algorithm with sub-exponential complexity \footnote{In cryptography, sub-exponential-time refers to algorithms whose running time grows faster than polynomial-time and more slowly than exponential-time.}. In contrast, no known sub-exponential algorithm can solve the ECDLP if the curves are chosen suitably. Compared with RSA, the ECC generates small-size keys with the same level of security. For example, 256-bit ECC keys achieve the same security level as 3072-bit RSA keys. Hence, it is indispensable in applications requiring smaller bandwidth and memory.

The pairing-based cryptography (PBC) has recently gained much attention and is being standardized as a next-generation cryptosystem. The pairing maps the discrete logarithm in a subgroup of an elliptic curve to the discrete logarithm in a finite field. Menezes \textit{et al.} \cite{menezes1993reducing} present the Weil pairing and attack on pairing to efficiently reduce the elliptic curve-based discrete logarithm problem (ECDLP) to a discrete logarithm problem (DLP) over a finite field using Weil pairing.  This attack is known as the MOV attack, based on the initials of three inventors. Similarly, Frey and Ruck \cite{frey1999tate} use the Tate pairing to reduce the ECDLP to DLP in a finite field, known as FR attack. It is well-known that not all elliptic curves are suitable for pairing. Thus, it is challenging to find suitable elliptic curves containing a subgroup with optimal embedding degree k, such that k is big enough to secure against FR attack but small enough that the arithmetic in a finite field is efficiently computable. There have been in-depth studies of the elliptic curves suitable for pairing, known as Pairing-Friendly Elliptic Curves (PF-EC). Generally, the PF-EC has a sizeable prime order subgroup with a small embedding degree.

The ECC offers several advantages over the traditional RSA-based cryptosystem, including a smaller key size, efficient computation, and the same level of security. This has encouraged the exploration and adoption of the ECC, especially in applications that were not previously possible with the invention of the PBC. One such application is Identity-Based Encryption (IBE), widely used in real-time scenarios. Additionally, the ECC offers a solution to the key escrow problem in the identity-based cryptosystem. Our proposed work will focus on the basic building blocks of the ECC and PBC.

\section{Elliptic Curves}
Mathematicians' study of elliptic curves throughout the last century has resulted in a rich and fascinating history. These curves have been applied to solve various mathematical problems, including Fermat's last theorem. The elliptic curve's ability to construct a group structure is advantageous for implementing PKC. The difficulty of finding discrete logs in a group, due to the absence of a sub-exponential time algorithm, makes the elliptic curves a preferred choice for systems based on the multiplicative group for the finite field. In practical applications, ECC is utilized for its compact implementation and high performance. The Weierstrass equation is the fundamental concept behind ECC, as discussed below.

\subsection{Weierstrass Equations}

\theoremstyle{definition}
\begin{definition} (\textit{Weierstrass equation}). Suppose a field $K$, and its algebraic closure $\overline{K}$. A Weierstrass equation over $K$ is an equation of the form given in (\ref{eq1.1}).

 \begin{equation} \label{eq1.1}
    y^2+a_1 xy+a_3 y=x^3+a_2 x^2+a_4 x+a_6 
\end{equation}
where $a_i \in K$ and $1 \le i \le 6$.
\end{definition}

\begin{definition} (\textit{Elliptic curve}). An elliptic curve $E \subset K^2$ is a non-singular \footnote{Non-singular means that the graph has no cusps, self-intersections, or isolated points} projective closure of the smooth affine curve, given in  (\ref{eq1.2}), 

 \begin{equation} \label{eq1.2}
    y^2=x^3+ax+b 
\end{equation}
where $a,b \in K$. It is derived from the Weierstrass equation whose characteristic \footnote{The smallest number of 1s that sum to 0 is called the characteristic of the finite field, and the characteristic must be a prime number} $K$ does not equal to $2$ or $3$, $Char(K) \ne 2$ or $3$.
\end{definition}

In an elliptic curve, the projective point $O:= (0,1,0)$ is the point at infinity. Due to this, we individually deal with the elliptic curve as affine curves and viewpoints at infinity. This point plays an essential role in the implementation of modern cryptography. 

\begin{definition} (\textit{Discriminant}). To find the non-singularity of a curve defined by a Weierstrass equation, we need to compute the discriminant $\bigtriangleup$ value. The discriminant of a Weierstrass equation is defined in Equation (\ref{eq1.3}).

 \begin{equation} \label{eq1.3}
    \bigtriangleup=-16(4a^3+27b^2 ) 
\end{equation}
\end{definition}

\begin{definition} (\textit{Non-singular curve}). An elliptic curve E is considered non-singular if $\bigtriangleup \ne 0$. Let $a_i$ in (\ref{eq1.1}) are in $K \subset \overline{K}$ which state that the curve is defined over $K$ and can be written as $E/K$, where we define $E(K)$ as the set of K-rational points of $E$.

\end{definition}

\subsection{Group Structure on Elliptic Curve}
To utilize the elliptic curve in modern cryptography, the points on an elliptic curve must have an abelian group structure. Theorem 1.5 illustrates the group structure of the elliptic curve.

\begin{theorem}
Suppose an elliptic curve over $K$ is $E$, denoted as $E/K$. Let two points on elliptic curve $E/K$ are $P$ and $Q$, where $P=(x_P,y_P)$, $Q=(x_Q,y_Q)$. The points on elliptic curve $E/K$ constitute an abelian group, including identity element denoted as $O$, negation of an point $P$ is defined by $-P=(x_P,-y_P)$ and addition operation on $P$ and $Q$ such that $P \ne -Q$ defined by $P+Q=(x_{P+Q},y_{P+Q})$, where 

 \begin{equation} \label{eq1.4}
    x_{P+Q}=(\mu^2- x_P- x_Q )
\end{equation}

 \begin{equation} \label{eq1.5}
    y_{P+Q}=(\mu(x_{P+Q}- x_P )- y_P )
\end{equation}

and, $\mu$ is the slope defined in (\ref{eq1.6}),

\begin{equation} \label{eq1.6}
\mu = \left\{
  \begin{array}{lr}
    \frac{(y_Q - y_P)}{(x_Q - x_P)}modp, \hspace{1cm} if \hspace{0.5cm} P \neq Q \\
  \frac{(3x_P^2+m)}{2y_P}modp, \hspace{1.2cm} otherwise
  \end{array}\right.
\end{equation} 
\end{theorem}

From  (\ref{eq1.4})-(\ref{eq1.6}), we can say that the addition operation in $E/K$ is commutative and it is straightforward to check that $E/K$ is closure under addition and additive inverse by using the corresponding formulae into elliptic curve’s Weierstrass equation. 

\begin{definition} (\textit{Additive Inverse}). According to B´ezout’s Theorem, the line passing through a point P and O meet on E/K at a unique point, which we denote -P. The additive inverse of point P mirrors point P about the x-axis.
\end{definition}

\begin{definition} (\textit{Point Addition}). The line through P and Q intersect at the third point, denoted as $-(P+Q)$ on curve E/K and the additive inverse of point $-(P+Q)$, i.e., $(P+Q)$ is known as an addition operation on points on $P$ and $Q$ on the elliptic curve. Thus, the \textit{addition operation} is defined as given three aligned non-zero points on an elliptic curve, $P+Q-(P+Q)=O$.
\end{definition}

\begin{definition} (\textit{Point Doubling}). Consider a point $P \in E/K$, the line passing through $P$ is the tangent, and the addition of point $P$ to itself is known as the \textit{point doubling}.
\end{definition}

\begin{figure} 
\begin{subfigure}{.5\textwidth}
  \centering
  \includegraphics[width=1\linewidth]{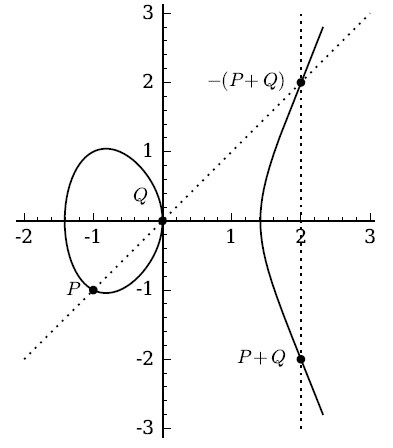}
  \caption{\\$E:y^2=x^3-2x$,\\ $P=(-1,-1),Q=(0,0)$\\
Point addition
}
  \label{fig1.1a}
\end{subfigure}%
\begin{subfigure}{.5\textwidth}
  \centering
  \includegraphics[width=.8\linewidth]{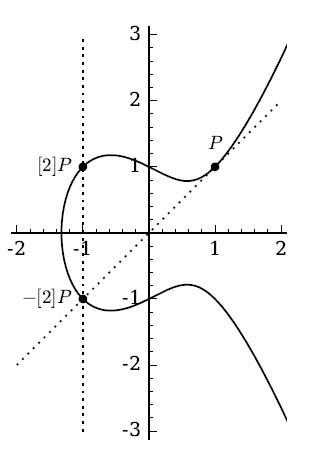}
  \caption{\\ $E:y^2=x^3-x+1$, \\$P=(1,1)$ \\
Point doubling
}
  \label{fig1.1b}
\end{subfigure}
\caption{Addition operation on elliptic curve over $\mathbb{Q}$}
\label{fig1.1}
\end{figure}

 Figure (\ref{fig1.1}) depicts the point addition and doubling on an elliptic curve. This makes it possible for the elliptic curve to be used in cryptography. 
 
\begin{definition} (\textit{Scalar multiplication on elliptic curve}). For any two points P and Q on elliptic curve $E/K$, an element $a \in K$ exists, such as $Q=aP$. Adding the point P using the group addition law is known as scalar point multiplication. 
\begin{align*}
    Q=P+P+\ldots+P \hspace{3mm}
    (a \hspace{2mm} times)
\end{align*}

\end{definition}

The smallest positive integer $q$ such that $qP=O$ is known as the order of point $P$. The points $<O,P,2P,3P,..(q-1)P>$ form a group generated by $P$, denoted as $< P>$.

\subsection{Elliptic Curve over Finite Field }
A finite field is a set of a finite number of elements, for example, the set of integer modulo prime number p, denoted as $\mathbb{Z}_p$ or $\mathbb{F}_p$ or $\mathbb{GF}(p)$ is the finite fields. The finite field $\mathbb{F}_p$ have $p$ elements from $0$ to $(p-1)$. It is noted that the value of p must be prime; otherwise, it will not be a field as it does not have the multiplicative inverse of all integers. For example, a set of integer modulo 10 is not a field because 2 has no multiplicative inverse in the field.

\begin{definition} (\textit{Elliptic curve over finite field}). The elliptic curves E over finite fields $\mathbb{F}_p$, denoted as $E(\mathbb{F}_p)$, hold the required properties to form an abelian group structure. The total number of points in a group is called the order of that group, denoted as $\#E(\mathbb{F}_p)$. 
\end{definition}

It is infeasible to count all possible values from 0 to $(p-1)$ as it needs $O(p)$ steps. For a larger prime, Schoof’s algorithm computes the order of the prime field which is run in polynomial time $Olog(p)$. Under the group addition rule, the set of points on elliptic curves forms a group where O is considered the identity element. This makes it possible for the elliptic curve to be used in cryptography. Point addition, point doubling and point inverse are the basic group operations on the elliptic curves. 

\subsection{Public and Private Keys over Elliptic Curve}
The scalar multiplication on the elliptic curves can define the public and private keys.

\begin{definition} (\textit{Public and private key on elliptic curve}). Suppose $P$ and $Q$ be the points on elliptic curve $E(\mathbb{F}_p)$, such that $Q=aP$, and $a \in \mathbb{F}_p$. In this scenario, $P$ is a group generator, $Q$ is the public key, and a is a private key. A positive integer $k$ such that $q|(p^k-1)$ is known as the embedding degree of the elliptic curve. 
\end{definition}

\subsection{Discrete Logarithm Problem on Elliptic Curves}
The security of elliptic curve cryptography is defined as the Discrete Logarithm Problem on Elliptic Curves (ECDLP).

\begin{definition} (\textit{ECDLP}). Let $\mathbb{G}_1$ be a group of points on the curve $E$. The ECDLP states that given $P,Q \in \mathbb{G}_1$, where $Q=xP$, the advantage of finding $x \in \mathbb{F}_p$ is negligible $\epsilon$, i.e.,

 \begin{equation} \label{eq1.7}
    |Pr[x \in \mathbb{F}_p | A(P,Q)]| \ge \epsilon
\end{equation}
\end{definition}

The ECDLP creates a platform for constructing the asymmetric key algorithms, for example, Elliptic Curve Digital Signature Algorithm (ECDSA), Elliptic Curve Diffie-Hellman Key Exchange (ECDHE) protocol and Elliptic Curve ElGamal-Based Encryption.

\section{Bilinear Pairing on Elliptic Curve}

\subsection{Torsion Group}

A torsion group on an elliptic curve plays a primary role in understanding the concept of PBC.

\begin{definition} (\textit{Torsion point}). Suppose an elliptic curve $E$, we define the $q$-torsion group of $E$ by $E[q]$, that is defined as

\begin{equation} \label{eq1.8}
    E[q] = \{P \in E(K)[q]|[q]P=0\}
\end{equation}
\end{definition}

The point $P$ is a point of finite order or torsion point. It can be seen that if $P,Q \in E[q]$ then $P+Q \in E[q]$ and $-P \in E[q]$. Thus, $E[q]$ is said to be a subgroup of $E$. Here, $P=E(K)[q]$ is represented as the point $P$ lies on the particular field $K$, like $\mathbb{Q}$ or $\mathbb{R}$  and $\mathbb{F}_p$.

\subsection{Divisor on Elliptic Curve}
\begin{definition} (\textit{Divisor on elliptic curve \footnote{Divisors are a way to keep track of the zeroes and poles of a function, where poles are related to projective coordinates}}). Suppose $E$ is an elliptic curve and $f(x,y)$ is a rational function of two variables $x$ and $y$. Some points exist on $E$ where the numerator of $f$ vanishes, and some points of $E$ where the denominator of $f$ vanishes. That means $f$ has zeros and poles on $E$. The divisor associated with $f$ is the formal sum, given in (\ref{eq1.9}). 

 \begin{equation} \label{eq1.9}
    D=div(f)= \sum_{P \in E}n_P[P]
\end{equation}

Where, coefficient $n_P \in \mathbb{Z}_q$, in which many of them are nonzero, so D is the finite sum. The degree of the divisor is defined as the sum of its coefficient, as given in (\ref{eq1.10}).

 \begin{equation} \label{eq1.10}
    deg(D) = deg\left(\sum_{P \in E} n_P[P] \right) = \sum_{P \in E} n_P
\end{equation}
\end{definition}

Now, we define the sum of the divisor, as given in Equation (\ref{eq1.11}).

 \begin{equation} \label{eq1.11}
    \sum(D) = \sum \left( \sum_{P \in E} n_P[P] \right) = \sum_{P \in E} n_P P
\end{equation}

Here, $n_PP$ is the additive operation on $P$ to itself $n_P$ times. 

\subsection{Miller Algorithm}

Miller's algorithm maps two points in an elliptic curve into an element of a finite field. Suppose we have two points on an elliptic curve, say P and Q, then the Miller algorithm, denoted as e yields an element $x \leftarrow e(P, Q)$ on finite fields. Recall that points addition (e.g. $R \leftarrow P+Q$) and scalar multiplication operations (e.g. $R \leftarrow kP$) on an elliptic curve are equivalent to the elements multiplication (e.g. $R \leftarrow kP$) and exponentiation operations (e.g.  $R \leftarrow p^k$) on the field, respectively.

\begin{algorithm}
	\caption{Miller Algorithm} 
	\begin{algorithmic}[1]
	    \State Set $T=P$ and $f=1$
	    \For {$iteration=(n-2),\ldots 0 $}
            \State Set $f=f^2.g_{T,T}$
	        \State Set $T=2T$
	        \If $a_1=1$
	           \State Set $f=f.g_{T,T}$
	           \State Set $T=T+P$
	       \EndIf
	   \EndFor
	   \State Return the value $f$
    \end{algorithmic} 
\end{algorithm}

\begin{definition} (\textit{Miller algorithm}). Let $a \ge 1$ is an integer, which can be written in binary expansion, as given in Equation (\ref{eq1.12}).
 
 \begin{equation} \label{eq1.12}
    a=a_02^0+a_1 2^1+ \dots +a_{n-1}2^{n-1}
\end{equation}

where, $a_i \in \{0,1\}$ and $a_{n-1} \ne 0$. The Miller algorithm gives the function $f_p$ whose divisor satisfies the Equation (\ref{eq1.13}).

 \begin{equation} \label{eq1.13}
    div(f_p )=q[P]-[qP]-(q-1)[O]
\end{equation}
\end{definition}

One of the useful properties of Miller’s algorithm is bilinearity. Suppose four points $P, Q, A$ and $B$ and two integers $a$ and $b$, where $A \leftarrow aP$ and  $B \leftarrow bQ$. The Miller algorithm computes $x \leftarrow e(P,Q)$ for $P$ and $Q$, and $y \leftarrow e(A,B)=e(aP,bQ)$ for $A$ and $B$, which will be related with  $y=x^{ab}$. In other words, the bilinearity property is defined by Equation (\ref{eq1.14}).

 \begin{equation} \label{eq1.14}
    e(aP,bQ)=e(P,Q)^{ab}=e(bP,aQ)
\end{equation}

With this fantastic property, many cryptographic protocols have been constructed, which could never be possible or too complicated, such as short signature, identity-based encryption, and attribute-based encryption.
where, $g_{P,Q}$ is defined in Equation (\ref{eq1.15}).        

\begin{equation} \label{eq1.15}
 g_{P,Q}= \left\{
  \begin{array}{lr}
    \frac{y-y_P-\mu(x-x_P)}{y+x_P+x_Q-\mu^2}, \hspace{1cm} if \hspace{0.5cm}  \mu \ne \infty \\
  x-x_P, \hspace{2cm} if \hspace{0.5cm} \mu = \infty
  \end{array} \right.             
\end{equation}

and $\mu$ is slope. 

\subsection{Weil Pairing}

\begin{definition} (\textit{Weil Pairing}). Let $P,Q \in E[q]$, and $f_P$ and $f_Q$ denote the rational functions on $E[q]$ satisfying $div(f_P)=q[P]-q[0]$ and $div(f_Q )=q[Q]-q[0]$. The Weil pairing $e_W$ of $P$ and $Q$ is given in Equation (\ref{eq1.16}).   

 \begin{equation} \label{eq1.16}
    ge_W (P,Q)= \frac{f_P (Q+S)/f_P(S)}{f_Q (P-S)/f_Q(-S)}
\end{equation}

where $S \in E[q]$ is any random point such that all elements on the right-hand side are defined and nonzero, i.e.,  $S \notin \{O,P,-Q,P-Q\}$.  
\end{definition}

The Weil pairing satisfies the following properties.
\begin{itemize}
    \item \textit{\textbf{$q^{th}$ root of unity}}:  For all $P,Q \in E[q]$, the Weil pairing is 
	
 \begin{equation} \label{eq1.17}
    e_W (P,Q)^q=1
\end{equation}
	\item \textbf{\textit{Bilinearity}}: For all $P,P_1,P_2,Q,Q_1,Q_2 \in E[q]$, 
	
\begin{equation} \label{eq1.18}
    e_W (P_1+P_2,Q)=e_W (P_1,Q) e_W (P_2,Q)
\end{equation}
 
and 
\begin{equation} \label{eq1.19}
    e_W (P,Q_1+Q_2 )=e_W (P,Q_1 ) e_W (P,Q_2)
\end{equation}
    \item \textbf{\textit{Alternating}}:  For all $P,Q \in E[q]$, the Weil pairing is 
\begin{equation} \label{eq1.20}
    e_W (P,Q)=1
\end{equation}

which implies that $e_W (P,Q)=e_W (Q,P)^{-1}$

    \item \textit{\textbf{Non-degeneracy}}: For all $P,Q \in E[q]$, If $e_W (P,Q)=1$, then $P=O$
\end{itemize}
	
\subsection{Tate Pairing }

In comparison with Weil pairing, Tate pairing is computationally more efficient. Suppose there is an elliptic curve $E$ over $\mathbb{F}q$, a prime $l$, point $P \in E(\mathbb{F}q)[l]$ and $Q \in (\mathbb{F}q)$. Pick a rational function $f_P$ on $E$ with divisor $div(f_P )=l[P]-l[O]$. The Tate pairing of $P$ and $Q$ is computed as 

\begin{equation} \label{eq1.21}
    \tau(P,Q)=\frac{f_P (Q+S)}{f_P (S)}\in \mathbb{F}q
\end{equation}

where $S \in E(\mathbb{F}q)$ is point such $f_P(Q+S)$ and $f_P(S)$ are defines and nonzero. 

\subsection{Weil Pairing over Field of Prime Power Order}

\begin{definition} (\textit{Embedding degree}).  Suppose an elliptic curve over $\mathbb{F}_p$ be $E$ and $r \ge 1$ be an integer such that $p$ is not divisible to $r$. Then, the embedding degree is defined as the smallest integer $k$ such that 

\begin{equation} \label{eq1.22}
   E(\mathbb{F}_{p^k}) \cong \mathbb{Z}_r  \times \mathbb{Z}_r
\end{equation}

\end{definition}

The embedding degree helps the Weil pairing to map ECDLP $E(\mathbb{F}_p)$ into the DLP in the field $\mathbb{F}_{p^k}$.

\begin{definition} (\textit{MOV algorithm}). The MOV algorithm reduces the ECDLP in elliptic cure $E({\mathbb{F}_p})$ to the DLP in the field $\mathbb{F}_{p^k}$. Algorithm 1.2 summarizes the MOV algorithm.
\end{definition}

\begin{algorithm}
	\caption{MOV Algorithm} 
	\begin{algorithmic}[1]
	    \State Set $N= \# E(\mathbb{F}_{p^k})$
	    \State Select any point $T \in E(\mathbb{F}_{p^k})$ but $T \notin E(\mathbb{F}_p)$
	    \State Let $T'=(N/l)T$
	    \If $T'=0$
	        \State run Step 2
	        \State $T'$ is point of order $l$, so run Step 8
	   \EndIf
	   \State Compute Weil pairing as $a=e_W (P,T') \in \mathbb{F}_p$ and $b=e_W (Q,T') \in \mathbb{F}_p$
	   \State Solve the DLP for $a, b \in  \mathbb{F}_p$, find an exponent $n$ such that $b=a^n$.
	   \State Also $Q=nP$, so the ECDLP has been solved.
    \end{algorithmic} 
\end{algorithm}

\begin{definition} (Distortion map). Suppose an elliptic curve over $\mathbb{F}_p$ be $E$ and $l \ge 3$ be a prime. Let a point $P \in E(\mathbb{F}_p)[l]$ on elliptic curve order l and $\Phi: E \rightarrow $E be mapped from $E$ to itself. Then mapping $\Phi$ is said to be a $l$-distortion map for $P$ if satisfied the two properties:
\begin{itemize}
    \item $\Phi(nP) = n\Phi(P)$ for all $n \ge 1$
    \item $e_W (P,\Phi(P))$ has $l^{th}$ root of unity, i.e., for any integer $r$, multiple of $l$, $e_W (P,\Phi(P))^r$ is 1.
\end{itemize}
\end{definition}

\begin{definition} (\textit{Modified Weil Pairing}). Suppose an elliptic curve over $\mathbb{F}_p$ be $E$ and $l \ge 3$ be a prime. Let point $P \in E(\mathbb{F}_p)[l]$ on elliptic curve order $l$ and $\Phi$ be an $l$-distortion map for $P$. The modified Weil pairing $e_W$ on $E[l]$ is defined by

\begin{equation} \label{eq1.23}
    e_W (Q,Q)=e_W (Q,\Phi(Q ))
\end{equation}
\end{definition}

Non-degeneracy is the important property of modified Weil pairing.

\subsection{Mathematical Assumptions}

We give some mathematical assumptions based on bilinear pairing on elliptic curves. 

\begin{definition} (\textit{Computational Diffie-Hellman (CDH) Problem}). The CDH problem states that given $P,Q,R \in E(\mathbb{F}_p)$ and $ \forall x,y \in \mathbb{F}_p$, where $Q=xP$  and $R=yP$, the advantage of finding $xyP$ is negligible $\epsilon$.

\begin{equation} \label{eq1.24}
    \left| \Pr \left[ xyP \in E(\mathbb{F}_p) \mid A(P,Q,R) \right] \right| \ge \epsilon
\end{equation}
\end{definition}

\begin{definition} (\textit{Decision Diffie-Hellman (DDH) problem}). The DDH problem \cite{boneh1998decision} states that given $P,Q,R,S \in E(\mathbb{F}_p)$ and  $ \forall x,y,z \in \mathbb{F}_p$, such that $Q=xP$, $R=yP$ and $S=zP$. The advantage to decide whether $z=xy$ is negligible $\epsilon$.

\begin{equation} \label{eq1.25}
    \left| \Pr \left[ z=xy \mid A(P,Q,R,S) \right] \right| \ge \epsilon
\end{equation}
\end{definition}

\begin{definition} (\textit{Bilinear Diffie-Hellman (BDH) Problem}). The BDH problem \cite{boneh2001identity} states that given $P,Q,R,S \in E(\mathbb{F}_p)$ and $ \forall x,y,z \in \mathbb{F}_p$, such that $Q=xP$, $R=yP$ and $S=zP$. The advantage of finding  $e(P,P)^{xyz} \in \mathbb{F}_p$  negligible $\epsilon$.

\begin{equation} \label{eq1.26}
    \left| \Pr \left[ e(P,P)^{xyz}\in \mathbb{F}_p \mid A(P,Q,R,S) \right] \right| \ge \epsilon
\end{equation}
\end{definition}

\begin{definition} (\textit{Gap Diffie-Hellman (GDH) Problem}). The GDH problem  \cite{choon2003identity} states that given $P,Q,R,S \in E(\mathbb{F}_p)$ and $\forall x,y,z \in \mathbb{F}_p$, such that $Q=xP$, $R=yP$ and $S=zP$. The advantage of finding  $e(P,P)^{xyz} \in \mathbb{F}_p$  negligible $\epsilon$.

\begin{equation} \label{eq1.27}
    \left| \Pr \left[ e(P,P)^{xyz}\in \mathbb{F}_p \mid A(P,Q,R,S) \right] \right| \ge \epsilon
\end{equation} 
\end{definition}

\section{Applications of PBC}

\subsection{Identity-Based Encryption}

One example of a modern pairing-based cryptosystem is the implementation of Identity-based Encryption (IBE). In 1984, Shamir proposed the concept of an identity-based cryptosystem (IBC) to address the key management problem associated with traditional PKC \cite{shamir1984identity}. IBC involves deriving a user's public key from unique identities such as email addresses, with the corresponding private key generated by a private key generator (PKG). However, until 2001, implementing a practical IBE was difficult. This changed when Boneh \textit{et al.} \cite{boneh2001identity} successfully implemented identity-based encryption using Weil pairing. The IBC framework consists of four algorithms that are defined as follows:

\begin{itemize}
    \item 	\textit{Setup}: PKG chooses a finite field $\mathbb{F}_p$ and an elliptic curve $E$ and point $P \in E(\mathbb{F}_p)[l]$ of prime order $q$ such that there exist an $l$-distortion map for $P$. Suppose a modified Weil is pairing $e_l$. Also, PKG needs to pick two cryptographic hash function $H_1:ID \rightarrow E(\mathbb{F}_p)$ and $H_2: \mathbb{F}_p \rightarrow |M|$. PKG computes its master key by selecting an integer $s$ and computing the public parameters $P_0=sP \in E(\mathbb{F}_p)$.
	\item \textit{Encryption}: Suppose Bob wish to send a message $M$ to Alice using her identity $ID_A$. Bob computes her public key as $P_A=H_1(ID_A) \in E(\mathbb{F}_p)$ and pick a random element $r \in \mathbb{Z}_q$, and computes ciphertext $C=(C_0,C_1)$ on message $M$, where 
$C_0=rP$ and $C_1=M \oplus H_2 (e_l (P_A,P_0)^r)$ 
	\item \textit{Key Extraction}: To decryption the ciphertext, Alice must have the private key corresponding to her $ID_A$. Thus, she requests PKG for her private key. Now, PKG computes the private key associated with $ID_A$ as 
$S_A=sP_A=sH_1 (ID_A) \in E(\mathbb{F}_p)$
	\item \textit{Decryption}: Using her private key $S_A$, Alice now decrypts the ciphertext and get the message. She first computes $e_l(S_A,C_0)$ which is considered as ephemeral key for decrypting ciphertext, is equal to $e_l(P_A,P_0)^r$. Alice can extract the message by evaluating it.

	\begin{align*}
	    C_1  \oplus H_2(e_l(S_A,C_0 ))= M \oplus H_2(e_l(P_A,P_0 )^r) \oplus H_2 (e_l (S_A,C_0 )) = M
	\end{align*}
	
\end{itemize}

\textit{Security of IBE is defined by two types of security}: The Weaker notion of security, which is known as semantic security (denoted as IND-ID-CPA), and the stronger notion of security (denoted as IND-ID-CCA) whose security is proven by the game playing between two entities, namely adversary Adv and challenger Ch. The complete security notion is discussed in \cite{boneh2001identity}.

\subsection{Tripartite DH One Round Key Agreement Protocol}

There is another useful pairing application that involves Joux's tripartite one-round key agreement protocol \cite{joux2000one}. This protocol enables three parties to generate a secret key in just one round, which is more efficient compared to the Diffie-Hellman protocol which requires two rounds.

Imagine three individuals named Alice, Bob, and Carl who want to exchange a shared secret key with each other in a single round of parameters. Antoine Joux has enabled this through the use of a pairing-based cryptosystem, as described below.

\begin{itemize}
    \item \textit{Setup}: First, Alice, Bob and Carl agree on an elliptic curve $E$ and point $P \in E(\mathbb{F}q)[l]$ of prime order $q$ such that there exist an $l$-distortion map for $P$. Suppose a modified Weil is pairing $e_l$. 
	\item \textit{Public parameter}: Alice picks $a$, and computes public parameter $A=aP$. Bob picks $b$ and computes the public parameter $B=bP$. Similarly, Carl picks $c$, and computes public parameter $C=cP$. Now, they publish the parameters $A, B$ and $C$.
	\item \textit{Key agreement}: To evaluate the shared secret key, Alice computes the modified pairing as defined as $K_A=e_l (B,C)^a =e_l (P,P)^{abc}$.
	
	Similarly,
	
	Bob and Carl compute the shared key as 
    
    Bob computes:                    $K_B=e_l(A,C)^b =e_l (P,P)^{abc}$
    
    Carl computes:                  $K_C=e_l (A,B)^c=e_l (P,P)^{abc}$
\end{itemize}
They now have the shared secret key as $e_l (P,P)^{abc}$.

\subsection{BLS Short Signature Scheme}

In 2003, Boneh \textit{et al.} \cite{boneh2001short} presented a short signature scheme that produces a short-size signature for the bandwidth-efficient environment. It is defined by the three algorithms: Setup, signature and verification. 
\begin{itemize}
    \item \textit{Setup}: Suppose $e_l$ be a bilinear pairing on group $(\mathbb{G}_1,\mathbb{G}_2)$ such that DHP in $\mathbb{G}_1$ is difficult. Let a cryptographic hash function $H:\{0,1\}^* \rightarrow \mathbb{G}_1$. Alice selects a random integer $a \in [1,n-1]$, which is her private key. Alice computes her public key as $A=aP$.
    \item \textit{Signature}: Using her private key $a$,  Alice computes a signature $S$ on message $m \in \{0,1\}^*$ as $S=aM$, where $M=H(m)$
	\item \textit{Verification}: Using Alice’s public key A, Bob verifies the signature $S$ by checking that the tuple $<A,P,M,S>$ is a valid decisional Diffie-Hellman tuple. Bob verifies from the Equation (\ref{eq1.28}).
	
\begin{equation} \label{eq1.28}
    e_l(P,S)=e_l (A,M)                                                 
\end{equation}
\end{itemize}

Suppose an attacker wants to forge a signature on message $m$. To do so, he needs to compute $S=aH(m)$ for given $P, A$, an instance of the DHP in $\mathbb{G}_1$. The short signature scheme is secure against existential forgery under adaptively chosen message and ID attacks (EF-ID-CMA).

Through our previous discussion, we have discovered that utilizing pairing on elliptic curves has significant advantages in creating cryptographic protocols that were previously unattainable. This has inspired us to delve deeper into the foundation of elliptic curves and pairing-based cryptography, in order to provide secure solutions for the latest advancements in technology. Moving forward, we will elaborate on our motivation and contribution to this thesis.

\section{Motivation and Contribution}

The pairing operation on an elliptic curve is crucial in numerous cryptographic applications in today's practical scenarios, such as big data analysis, cloud computing, and privacy-enhancing schemes. It is a fundamental building block for various cryptographic primitives, enabling secure and efficient operations.

One compelling application of pairing-based cryptography is Identity-Based Encryption (IBE) with pairing. This approach has gained considerable attention in wireless communication, sensor networks, e-commerce, and online transactions, where achieving user authentication with pre-interaction is crucial. By leveraging pairing, we can design cryptographic systems that enhance security and efficiency in ubiquitous technologies like wireless body area networks, the Internet of Things, and e-cash payment systems.

The architecture of an Identity-Based Cryptosystem (IBC) is another key motivation for this thesis. IBC addresses the challenges associated with certificate management in traditional public-key cryptography. However, due to issues like user slandering, key abuse, and secure key issuing problems, IBC has been limited to small network applications despite its usability features. By exploring pairing on elliptic curves, we aim to overcome these limitations and propose an IBC architecture that addresses the key escrow problem and enhances security.

In this thesis, we delve into the comprehensive study of pairing on elliptic curves and focus on selecting suitable curves that enable efficient pairing. Building upon this understanding, we propose an innovative architecture for an identity-based cryptosystem that solves the key escrow problem. Furthermore, we present novel IBE and IBS schemes based on this architecture and an identity-based blind signature scheme that enables secure signing without knowledge of the message's content. We develop e-cash payment and e-voting systems based on these schemes to showcase practical applications.

Moreover, we introduce an identity-based signcyrption scheme that eliminates key escrow concerns and proposes a data communication protocol tailored for wireless body area networks and peer-to-peer video-on-demand systems. Lastly, we present an authenticated key exchange protocol designed for resource-constrained environments. Our research in these areas aims to advance state-of-the-art pairing-based cryptography, providing improved security and usability for identity-based cryptosystems across a wide range of practical applications.

By addressing these motivations, we strive to contribute to the development of secure and efficient cryptographic solutions that can be applied to real-world scenarios, safeguarding sensitive information and facilitating trusted communication in various domains.

\section{Structure of Thesis}

The structure of the thesis is organized as follows: 

\begin{itemize}
    \item In Chapter \ref{chapter2}, we examine the significant families of PF-EC, including BN, BLS12, BLS24, KSS16, and KSS 18. However, these families are vulnerable to the recent NFS attack, necessitating an increase in the key size. We enhance Freeman \textit{et al.}'s taxonomy \cite{freeman2010taxonomy} by introducing some new families that were previously overlooked. Additionally, we evaluate the practical security of various pairing-friendly curve families to determine which families are superior to BN, KSS, and BLS in terms of the required key size. 
    \item In Chapter \ref{chapter3}, we explore a secure and efficient way of issuing keys to tackle the issues of key escrow, key issuance, and user defamation. We utilize multiple semi-trusted authorities and a single Key Generation Center (KGC), where the KGC registers the user, and the authorities provide safeguarded private key shares to the user. Additionally, our proposed structure leverages cloud computing to transfer the heavy computation workload to the cloud, with only minor operations being performed on the user and KGC. By implementing this architecture, we have created secure, escrow-free identity-based encryption against confidentiality breaches and an identity-based signature scheme impervious to forgery.
    \item In Chapter \ref{chapter4}, we introduce two secure Identity-Based Blind Signature (IDBS) schemes that are free of pairing and pairing-friendly. These schemes are resistant to forgeable attacks and possess the property of blindness. By utilizing our proposed IDBS schemes, we showcase two anonymous and secure electronic voting systems, as well as an End-to-End Verifiable Internet Voting (E2EVIV) system with batch verifiability. Our demonstration proves that these systems meet the necessary security requirements, including voter confidentiality, authenticity, anonymity, vote integrity, and resilience to bribery and coercion. 
    \item In chapter \ref{chapter5}, we propose an Identity-based Blind Signature Scheme with a Message Recovery (IDBS-MR) scheme that delegates the user's signature generation to the third-party signer without revealing the actual message except for the recipient after extracting it from the signature. It is secured against the forgeable attack under the chosen message and ID attack and achieves blindness property. It avoids expensive cryptographic operations, for example, pairing, map-to-hash function, and modular exponentiation operations and is hence suitable for the pairing-free environment. Further, we propose a privacy-preserving with integrity auditing scheme for cloud storage, whose security is defined by the proposed IDBS-MR.  
    \item In chapter \ref{chapter6}, we propose an Escrow-Resilient Identity-Based (Aggregated) Signcyrption (EF-ID[A]SC) scheme, which is secure against confidentiality and forgeability attacks. We extend the proposed IDSC to construct a secure Peer-to-Peer Video-on-Demand System (SecP2PVoD), secure against Pollution Attacks and Untrusted Service providers. Further, we implement a bandwidth-efficient secure data transmission for the wireless body area networks using the proposed IDASC scheme. 
    \item In chapter \ref{chapter7}, we discuss authenticated two-party and three-party key agreement protocols for low-constraint devices. The proposed two-party key agreement protocol is computationally efficient that avoids expensive operations. The three-party key agreement protocol allows three parties to develop a shared secret key in one round. We also present an anonymous authenticated key agreement protocol for the WBAN system, which anonymously provides an authenticated communication channel between the low-resource sensor and medical professional. These three proposed protocols achieve mutually authenticated, forward secrecy, and unlinkability.
    \item In Chapter \ref{chapter8}, we present a summary and conclusion, along with suggestions for future directions. Additionally, we expand on our thesis to address the COVID-19 pandemic by proposing a privacy-preserving data exchange scheme using Hyperledger Fabric (a permissioned blockchain) to standardize the exchange of medical data with a patient-centric approach.
\end{itemize}

\end{doublespace} \label{chapter1}
\begin{savequote}[75mm] 
Mathematics is the language with which God has written the universe.
\qauthor{Galileo galilei} 
\end{savequote}
\chapter{Recent Attacks on Pairing-Friendly Elliptic Curves and Opportunities}

\begin{doublespace}
In this chapter, we conduct a thorough analysis of various families of Pairing-Friendly Elliptic Curves (PF-EC) and their practical security, providing valuable insights into their adoption and usage in cryptographic applications. We extend Freeman's \cite{freeman2010taxonomy} taxonomy by introducing new families of PF-EC, including a complete family with variable discriminant and new sparse families of curves. Our framework comprehensively covers the construction of individual and parametric families of curves. We also estimate the practical security of several families of pairing-friendly curves, comparing them to well-known families such as BN, KSS, and BLS. Our aim is to identify families of curves that offer better security against recent attacks, such as Tower Number Field Sieve (TNFS) and Special Number Field Sieve (STNFS). We demonstrate that these attacks require an increase in key size to be effective. Furthermore, we compare the families of curves in terms of their required key size and recommend the most suitable elliptic curves for pairing. Additionally, we explore the adoption of pairing-friendly curves in international standards, cryptographic libraries, and applications.

\section{Introduction}
Tate pairing solves the discrete logarithm problem (DLP) in the divisor class group of the elliptic curve that has become the foundation of the Frey-Ruck (FR) attack \cite{frey1999tate}. The DLP can also be solved using the index calculus algorithm for a small embedding degree $k$. Menezes \textit{et al.} \cite{menezes1993reducing} have used the Weil pairing to reduce DLP in the divisor class group $E(\mathbb{F}_p)$ to the DLP in the finite field $\mathbb{F}_{p^k}$, which is termed as MOV attack. The elliptic curves with large prime-order subgroup $r$ and small embedding degree $k$, are the main ingredients for implementing pairing-based cryptographic systems. Such elliptic curves are recognized as PF-EC. The challenge with pairing is to mine those elliptic curves in the families of curves with an optimal embedding degree $k$.

The bilinear pairing maps a pair of points on the elliptic curve $E(\mathbb{F}_P)$ to the multiplicative group of the finite field $\mathbb{F}_{p^k}$. Here, $k$ is the embedding degree that embeds an instance of the DLP over an elliptic curve $E(\mathbb{F}_P)$ into an instance of DLP over the finite field elements $\mathbb{F}_{p^k}$. A Pairing-Friendly Elliptic Curve (PF-EC) must satisfy the following conditions: 
\begin{itemize}
    \item 	The subgroup of order $r$ is large enough so that Pollard’s rho method for solving DLP is hard, 
    \item The embedding degree $k$ is big enough so that the index calculus method for computing DLP in $\mathbb{F}_{p^k}$ is equally hard as in $E(\mathbb{F}_p)$, 
	\item $k$ is relatively small enough so that the arithmetic (pairing) in $\mathbb{F}_{p^k}$ is easily computable. 
\end{itemize}

Besides, the $\rho$ value, which determines the closeness of the subgroup size $r$ w.r.t. the size of field $p$, should be close to 1. Barreto \textit{et al.} \cite{barreto2005pairing} suggest that the (BN) elliptic curve with $k=12$ is optimal for a 128-bit security level an important finding in pairing-based cryptography.
Menezes \textit{et al.} \cite{menezes1993reducing} was the first to provide the supersingular curves with embedding degree $k \in \{2,3,4,6\}$. These curves are suitable for implementing pairing-based cryptosystems but have some deployment issues. Nevertheless, the possible deployment of supersingular curves is constructed on the small characteristic $K \ne 2$ or $3$. The first systematical method for making non-supersingular (ordinary) curves was described by Miyaji \textit{et al.} \cite{miyaji2001new}, which is usually known as MNT curves. The MNT curves are defined over a prime with embedding degrees $k \in \{2,3,4,6\}$. The schemes \cite{barreto2003selection,dupont2005building} introduce other ways for constructing ordinary curves with an arbitrary embedding degree and large $\rho \approx 2$.

A general construction of an elliptic curve of arbitrary embedding degree $k$ is given by Cocks \textit{et al.} \cite{cocks2001identity} such that  $\rho \approx r^2$, where $p$ is field size and $r$ is prime order of subgroup, but it responds to an inefficient implementation. Barreto at al. \cite{barreto2003selection} discuss an optimized method to select groups in ordinary curves and present the curves, usually known as BLS curves, by avoiding irrelevant factors during Tate pairing computation. Similarly, Brezing \textit{et al.} \cite{brezing2005elliptic} have extended the Cocks-Pinch method \cite{cocks2001identity} and discovered a well-known algorithm to construct a polynomial family of elliptic curves with smaller $\rho$ under certain circumstances. However, achieving $\rho \approx 1$ remains still a challenge. In 2008, Kachisa \textit{et al.} \cite{kachisa2008constructing} discussed a method to construct the Brezing-Weng curve \cite{brezing2005elliptic}, which has the least polynomials of integers in the cyclotomic field and present the families of elliptic curves with embedding degree $k \in \{16,18,36,40\}$. Tanaka \textit{et al.} \cite{tanaka2008constructing} presented more patterns of families of pairing-friendly curves with some adjustments. Drylo \cite{drylo2011constructing} extended the Brezing-Weng method to construct complete families with variable discriminant and sparse families. Barreto \textit{et al.} \cite{barreto2005pairing} presented a simple algorithm that generates the elliptic curve (BN curve) of prime order with degree $k=12$, which is considered the most effective and efficient elliptic curve in practice to date. Scott \textit{et al.} \cite{scott2006generating} presented the extension of general MNT construction to produce useful curves and show their use in real practices. Freeman \cite{freeman2006constructing} suggested that the existing methods for constructing elliptic curves of prime order with prescribed embedding degree can be defined as a general framework and utilized to construct a curve with embedding degree 10. Recently, Scott \textit{et al.} \cite{scott2018new} have presented a new family of pairing-friendly curves with embedding degree $k=54$, helpful in implementing pairing-based cryptosystems with high security. They have also tried to find the elliptic curve with $k=72$, but due to inefficiency, they did not. Freeman \textit{et al.} \cite{freeman2010taxonomy} demonstrated the selection of curve parameter size for efficient pairing-based cryptographic systems. 

The complexity to solve DLP on elliptic curve $E(\mathbb{F}_p)$ with a subgroup of order $r$ using Pollard’s rho algorithm is $O(\sqrt{r})$. Recently, ongoing progress on number field sieve (NFS) and its variants \cite{kim2016extended,kim2017extended} reduced the complexity of DLP in the extension field $\mathbb{F}_{p^k}$ for a composite embedding degree. Thus, a significant effect on selecting the elliptic curve groups over finite field extension $\mathbb{F}_{p^k}$ has been observed due to such attacks. Besides, an improvement in the special tower number field sieve (STNFS) conquered the complexity of DLP in the extension field $\mathbb{F}_{p^k}$ for the composite embedding degree. Fotiadis \textit{et al.} \cite{fotiadis2019tnfs} utilized the Brezing-Weng method to produce the families of PF-EC and extracted the list of curves with optimal embedding degrees by considering the SNTFS complexity on a 128-bit level of security. Kim \textit{et al.} \cite{kim2016extended} revisited a new kind of NFS algorithm, namely, an extended tower number field sieve (exTNFS) that has asymptotically dropped down the security level of elliptic curves for pairing. After Menezes \textit{et al.} \cite{menezes2016challenges}, Barbulescu \textit{et al.} \cite{barbulescu2019updating}, and Guillevic \textit{et al.} \cite{guillevic2021alpha} in 2019, Guillevic \textit{et al.} \cite{guillevic2020short} in 2020, re-examine the STNFS to produce an accurate finite field size for a given security level. Fotiadis \textit{et al.} \cite{fotiadis2018optimal} discussed the pairing-friendly curves with a composite embedding degree for a 128-bit security level. From the above discussions, it is concluded that the recommendations in Table 2.1 are suitable for selecting parameters of the curve with prime $k$; in the case of composite $k$, parameters need to be updated. Table 2.1 illustrates the parameters of a class of curves with embedding degree $k$ to obtain the level of security (in bits).

\begin{table}
        \centering
        \caption{Curve parameter size (in bits) with associated embedding degree for obtaining the prescribed level of security}
        \label{tbl2.1}
        \begin{tabular}{|c|c|c|c|c|}
            \hline
            \multicolumn{1}{|c|}{Security Level} & \multicolumn{1}{|c|}{Subgroup size} & \multicolumn{1}{|c|}{Extension field size} &            \multicolumn{2}{|c|}{Embedding degree k} \\
            \cline{3-5}
             (in bits) & $r$ (in bits) &  $p^k$ (in bits)  & $\rho=1$ &	$\rho=2$ \\
            \hline
            \hline
            80	& 160 &	960-1280 &	6-8 &	3-4 \\
            112	& 224 &	2200-3600 &	10-16 &	5-8 \\
            128	& 256 &	3000-5000 &	12-20 &	6-10 \\
            192 & 384 &	8000-10000 & 20-26 &	10-13\\
            256 & 512 &	14000-18000	& 28-36 &	14-18 \\
            \hline
        \end{tabular}
\end{table}

In this chapter, we discuss the following. 
\begin{itemize}
    \item 	We present a comprehensive literature study to obtain many techniques for constructing an ordinary elliptic curve with prescribed embedding degree and their classification based on the CM discriminant value. We also continue the taxonomy given by Freeman \textit{et al.} \cite{freeman2010taxonomy} and revive it with some new families such as the complete family with variable discriminant \cite{drylo2011constructing}. New sparse family \cite{fotiadis2018generating} that were not covered in the existing taxonomy \cite{drylo2011constructing, scott2018new}.
    \item We realize the generalization of various pairing-friendly curves and present a comprehensive framework for constructing the individual as well as parametric families of curves with a distinct embedding degree. After construction, we estimate the security by considering the key size of these families of curves to produce better families of curves than the previous families such as BN, KSS, and BLS.
    \item We also look at the recent attacks, such as TNFS, exTNFS, and SexTNFS \cite{kim2016extended,kim2017extended} on pairing and recommend that such attacks on pairing are required to increase the key size of most popular pairing-friendly curves such as BN, BLS and KSS. We observe that such attacks force us to update the key size for pairing. In this direction, we evaluate and compare the key size required for families of curves and consequently select the best choice of the elliptic curve from the available NFS secure family of the curve.
    \item We identify and accumulate the lists of possible curves against STNFS attacks for a 128-bit security level that have never been discussed before, except by Guillevic \textit{et al.} \cite{guillevic2020short}. Besides, we presented the adoption of pairing-friendly curves in international standards, cryptographic libraries, and applications on district security levels such as 128-bits, 192-bits, and 256-bits.
\end{itemize}

\section{Background on Elliptic Curve}
Here, we provide the basic building blocks and the framework for constructing pairing-friendly curves.

\newcommand{\embeddingDegree}{k}
\newcommand{\rhoValue}{\rho}

\subsection{Embedding Degree k and rho-value}

\theoremstyle{definition}
\begin{definition} \label{def2.1}
(Embedding Degree).
Let $E$ be an elliptic curve defined over a finite field $\mathbb{F}_{p^k}$, and let $r$ be a prime dividing $\#E(\mathbb{F}_p)$. The embedding degree of E with respect to $r$ is the smallest value of$k$such that $r$ divides $p^ {k-1}$ but does not divide $p^{i-1}$ for all $0 < i < k$. 
\end{definition}

$k$ decides the feasibility, in terms of security and efficiency, of pairing construction. The curves with large $k$ are not suitable for implementing pairing. The only feasible curves for pairing are those curves with small $k$. For example, suppose an elliptic curve $E(\mathbb{F}_p)$ over a finite field $\mathbb{F}_{p^2}$, where $log(p)=512$ bits and subgroup $log(r)=160$ bits. There exists a transformation from the elliptic curve $E(\mathbb{F}_p)$ to $\mathbb{F}_{p^2}$, where $p^2=1024$ bits. Here, exponent 2 is the embedding degree. Barreto \textit{et al.} \cite{barreto2002constructing} suggest that the optimal degree is the elliptic curve with embedding degree $k=12$. 

The security of pairing-based cryptographic protocols depends on the hardness of DLP on elliptic curve $E(\mathbb{F}_p)$ of order $r$ and finite field extension $\mathbb{F}_{p^k}$ of order $p$. Both problems are defined on their group order, so a new parameter denoted as $\rho$, evaluates the closeness of size of subgroup $r$ w.r.t. the size of field $p$, given in Equation (\ref{eq2.1}).

\begin{equation} \label{eq2.1}
    \rho = \frac{{\log(p)}}{{\log(r)}}
\end{equation}
                                                        
For $\rho \ge 1$, the elliptic curves respond to the dense coordinate of points due to large extension finite field $\mathbb{F}_{p^k}$ that affects arithmetic efficiency, as given in (\ref{eq2.2})

\begin{equation} \label{eq2.2}
    k\rho = \frac{{\log(p^k)}}{{\log(r)}}
\end{equation}

Suppose $p(x),r(x),t(x) \in \mathbb{Q}[x]$ and $(p,r,t)$ parameterizes a family of curves with embedding degree k. The $\rho$-value of $(p,r,t)$ is represented as follows: 

\begin{equation} \label{eq2.3}
    \rho(p, r, t) = \frac{{\text{deg}(p)}}{{\text{deg}(r)}}
\end{equation}

The elliptic curves with small $\rho$-values are considered robust and efficient in arithmetic computation. For instance, a curve of the 512-bit size of a subgroup with $\rho=1$ is built over an extension field size of 512-bits, while the same curve with $\rho=2$ is built over an extension field size of 1024-bits. The first group's operations are more efficient than the other. 

\textbf{Correlation between $\rho$ and k}. A pairing-friendly elliptic curve is characterized by the triplet $\langle \rho, k, d \rangle$, where $d$ is the maximum integer from the set of possible twists ${1, 2, 3, 4, 6}$, and the parameters $\rho$ and $k$ have been defined above. Suppose we define the size of the base field $\mathbb{G}_1$ as $|\mathbb{G}_1| = n \rho$ bits, the size of the extension field $\mathbb{G}_T$ as $|\mathbb{G}_T| = n \rho k$ bits, and the size of the field $\mathbb{G}_2$ as $|\mathbb{G}_2| = |\mathbb{G}_T|/d = n \rho k/d$ bits for a curve of order $n$ bits.

We can correlate the security of elliptic curves to match the security of AES variants (AES-128, AES-192, or AES-256). Usually, the group size of an elliptic curve is twice the AES-equivalent bits. For example, the group size $r$ of an elliptic curve (BN curve) with security similar to AES-128 is 256 bits. The triplet associated with such a curve and group size $r = 256$ is ${1, 12, 6}$, which gives $|\mathbb{G}_1| = 256$, $|\mathbb{G}_T| = 3072$, and $|\mathbb{G}_2| = 512$. Similarly, for the BLS12 curve with group size $n = 256$, the triplet is given as ${3/2, 12, 6}$, which yields $|\mathbb{G}_1| = 384$, $|\mathbb{G}_T| = 4608$, and $|\mathbb{G}_2| = 768$. 
	
\subsection{Pairing-Friendly Curve}

\begin{definition} \label{def2.2}
(Elliptic curve). An elliptic curve E over a field $K$, denoted as $E/K$ ($Char(K) \ne 2,3$) is the set of points in $K \times K$ of an equation $y^2=x^3+ax+b$, where $a,b \in K$ and $\Delta = -16(4a^3 + 27b^2) \ne 0$, together with some point O at infinity. 
Here,$k$can be a finite field $\mathbb{F}_p$ for a prime $p>3$ or a field extension $\mathbb{F}_{p^k}$ such that the set 

\begin{equation} \label{eq2.4}
    E(\mathbb{F}_(p^k ) )=\{(x,y) \in (\mathbb{F}_{p^k},\mathbb{F}_{p^k} )| y^2=x^3+ax+b\} \cup \{O\}
\end{equation}

is known the group of $\mathbb{F}_{p^k}$-rational points of E. The group $E(\mathbb{F}_{p^k})$ is known as the abelian group, in which O denotes the point at infinity. The number of points, also known as the order of group $E(\mathbb{F}_{p^k})$ is defined as $\#E(\mathbb{F}_{p^k })=r=p+1-t$, (Hasse’s Theorem) such that $|t| \le \sqrt{p}$, where $t$ is a trace of Frobenius. 
\end{definition}

\begin{definition} \label{def2.3}
(r-torsion point). Let $E$ be an elliptic curve defined over $\mathbb{F}_{p^k}$, and $E(\mathbb{F}_{p^k})$ denotes the group of $\mathbb{F}_{p^k}$-rational points. For some integer $r$, we suppose $E[r]$ as the group of $r$-torsion and $E(\mathbb{F}_{p^k})[r]$ the group of r-torsion points of E defined over $\mathbb{F}_{p^k}$. The set of r-torsion points in $E(\mathbb{F}_{p^k})$ is defined as 

\begin{equation} \label{eq2.5}
    E(\mathbb{F}_{p^k})[r]=\{P \in E(\mathbb{F}_{p^k})| [r]P=O\}
\end{equation}

Suppose $\mathbb{G}_1$ and $\mathbb{G}_2$ are two subgroups of $E(\mathbb{F}_{p^k})$ and $\mathbb{G}_T$ be a subgroup of $\mathbb{F}_{p^k}$, both of order $r$, there is an asymmetric pairing $e:\mathbb{G}_1 \times \mathbb{G}_2 \rightarrow \mathbb{G}_T$. The suitable combination of $p$, $k$, $E$ and $r$ constitute a strong pairing-based cryptographic protocol. 

\end{definition}

\begin{definition} \label{def2.4}
(Pairing-friendly elliptic curve). An elliptic curve over $\mathbb{F}_p$ is said to be pairing-friendly if the following conditions are satisfied:
\begin{itemize}
    \item 	there exists a prime, $r$ such that $r|E(\mathbb{F}_p)$,
    \item the value $\rho=\log(q)/\log(r)\le 2$, close to 1, while the embedding with respect to $r$ satisfies $k \le \log(r)/8$.
    \item the DLP in $\mathbb{G}_T$ is as computationally hard as DLP in $\mathbb{G}_1$ and $\mathbb{G}_2$.
    \item there is an efficient algorithm to perform the computation in $\mathbb{G}_T \subset \mathbb{F}_{p^k}$	
\end{itemize}
\end{definition}

\subsection{Parameterization of Pairing-Friendly Elliptic Eurve}

We now discuss the strategy to parameterize the traces of curves, given by Barreto \textit{et al.} \cite{barreto2005pairing}. First, we fix a polynomial $t(x)$, trace of Frobenius, and construct the polynomials $r(x)$ and $p(x)$ that are the order of elliptic curves and prime fields, respectively. More concretely, the complex multiplication (CM) algorithm can be used to construct an elliptic curve over $\mathbb{F}_{p(x_0)}$ with $r(x_0)$ points with embedding degree k, if $p(x_0)$ is prime for some $x_0$. 

\begin{theorem} \label{thm2.1}
\cite{freeman2006constructing} For an integer $k>0$, suppose $\Phi_k (x)$ is $k^{th}$ cyclotomic polynomial. Let $t(x)$ be a polynomial of trace with integer coefficient, $r(x)$ be an irreducible factor of $\Phi_k(t(x)-1)$, and $f(x)=4p(x)-t(x)^2$. For a positive square-free integer D, let there be a solution of equation $Dy^2=f(x)$ be $(x_0,y_0)$ in which $p(x_0)$ and $r(x_0)$ are prime.
For small D, there exists an efficient CM algorithm that constructs an elliptic curve E defined over $\mathbb{F}_p(x_0)$ such that $E(\mathbb{F}_p(x_0))$ has prime order $r{x_0}$ and E has embedding degree at most $k$.
\end{theorem}

\begin{proof}
The well-known solution to equation $Dy^2=f(x)$ is $(x_0,y_0)$ in which $p(x_0)$ is prime. For small $D$, the CM algorithm constructs an elliptic curve E over $\mathbb{F}_{p(x_0)}$ of order $\#E(\mathbb{F}_{p(x_0)})=r(x_0)$. The elliptic curve $E(\mathbb{F}_{p(x_0)})$ has prime order since $r(x_0)$ is prime. Lemma 1 of Barreto \textit{et al.} \cite{barreto2002constructing} proves that E with embedding degree $k$ has prime order $r(x_0)$ that divides $\Phi_k(t(x_0 )-1)$ but not $\Phi_l(t(x_0 )-1)$ for $l<k$. Because we select the polynomial $r(x)$ to divide $\Phi_k(t(x)-1)$, so it is guaranteed that $r(x_0)$ divides  $p(x_0)^{k-1}$ and the embedding degree of $E$ is atmost $k$.
\end{proof}

\subsection{Families of Pairing-Friendly Curves}
To construct a curve of prime order with embedding degree k, we have first to choose the polynomials $p(x)$, $r(x)$ and $t(x)$ that fulfil the conditions of Theorem \ref{thm2.1} on various values of $x$ until $r(x)$ and $p(x)$ are prime. 

\begin{definition} \label{def2.5}
Suppose $p(x)$, $r(x)$ and $t(x)$ are nonzero polynomials with integer coefficients. For any given positive integer $k$ and positive square-free integer $D$, we say that polynomial triplet $(p(x),r(x),t(x))$ parameterizes a family of the pairing-friendly ordinary curve if the following conditions are fulfilled.  
\begin{itemize}
    \item $p(x)=q(x)^d$ for some $d \ge 1$ and $q(x)$ is irreducible prime.
    \item $r(x)=c.r'(x)$ with $c \ge 1 \in \mathbb{Z}$ and $r'(x)$ is irreducible prime.
    \item $p(x)+1-t(x)=h(x)r(x)$ for some $h(x) \in \mathbb{Q}$. 
	\item $r(x)|(\Phi_k t(x)-1)$, where $Phi_k$ is the $k$-the cyclotomic polynomial.
	\item The CM equation $Dy^2=4q(x)-t(x)^2$ has infinite many solutions.
\end{itemize}
 
\end{definition}

In this way, we can define a family of curves and construct an elliptic curve with embedding degree$k$and CM discriminant D. In practice, it is easy to find $r(x)$ and $p(x)$ that fulfil four conditions of definition \ref{def2.5}. However, choosing these parameters is quite difficult, such that $Dy^2=4p(x)-t(x)^2$ responds to infinitely many solutions. Usually, if $f(x)$ is a square-free polynomial with a degree of at least 3, then a finite number of solutions to equation $Dy^2=f(x)$ exists. Therefore, we ensure that the parameters $(p,r,t)$ can denote a family of curves.  

\begin{definition} \label{def2.6}
Choose an integer $k>0$, and polynomials $p(x), r(x)$ and $t(x) \in Z[x]$ that satisfy the first three conditions of definition \ref{def2.5}. Let $f(x)=4p(x)-t(x)^2$. Consider $f(x)=ax^2+bx+c$, where $a,b,c \in Z$, such that $a>0$ and $b^2-4ac \ne 0$. Let a square-free integer be $D$ such that $aD$ is not a square. If there exists a solution $(x_0,y_0)$ of equation $Dy^2=f(x)$, then tuple $(t,r,p)$ denotes a family of curves with embedding degree $k$.
\end{definition}

It is clear from Theorem \ref{thm2.1} that if $f(x)$ is square-free and quadratic, we can obtain a family of curves of prescribed degree for each $D$. Conversely, if $f(x)$ is a linear function time a square, then we can still obtain a family of curves for only a single $D$. Barreto and Naerhig \cite{barreto2005pairing} utilize the same method to construct the curve with embedding degree $k=12$. 

\section{Classification of Pairing-Friendly Elliptic Curves}

The methods for constructing pairing-friendly curves are classified into two categories: individual curves and parametric families of curves. The methods for constructing individual curves considering the integers $p$, $r$ and $k$ such that there exists an elliptic curve $E$ over $\mathbb{F}_p$ with subgroup $r$ and embedding degree $k$. While the methods for constructing the family of curves assume the polynomials $p(x)$ and $r(x)$ in such a way if $p(x_0)$ is prime power for some $x_0$, there must be an elliptic curve $E$ over $\mathbb{F}_{p(x_0)}$ with subgroup $r(x_0)$ and embedding degree $k$. Figure (\ref{fig2.1}) calculates a pairing-friendly elliptic curve.

\begin{figure}
  \centering
  \includegraphics[width=1\linewidth]{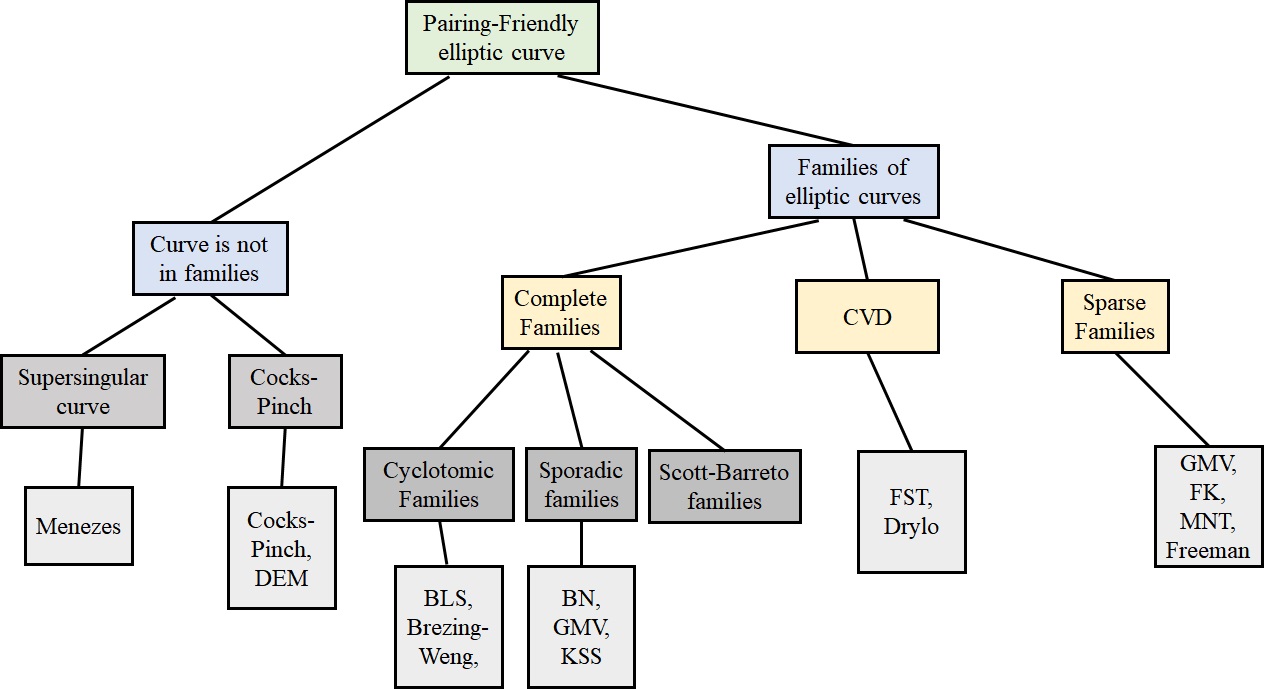}
  \caption{Revised Taxonomy of pairing-friendly elliptic curve}
\label{fig2.1}
\end{figure}

\subsection{Elliptic Curve Not in Family}

\textbf{Supersingular Curve}. An elliptic curve is called a supersingular curve if and only if t=0(mod p) and embedding degree $k\in \{1,2,3,4,6\}$; otherwise, it is an ordinary curve. The elliptic curve over a prime field $\mathbb{F}_p$ such that $p>3$, then $k=1$ or $k=2$. The supersingular elliptic curves with embedding degrees 4 and 6 are defined over characteristic two and three finite fields, respectively. Menezes \textit{et al.} \cite{menezes1993reducing} discussed the prime-order supersingular curves with embedding degree $k\in \{3,4,6\}$. These curves are considered weak due to the MOV and FR reduction attacks, so they are unsuitable for constructing a cryptographic system. Some researchers do not prefer supersingular curves since $k=2$ and 6 are too small, so they focus on the ordinary curves of prime order n, such that $r=n$. 

\textbf{Ordinary Curve}. Supersingular curves are restricted to small embedding degrees, i.e., $k=2$ for prime field and $k \le 6$ in general; so, we have to move to the ordinary curves. Cocks-Pinch \cite{cocks2001identity} method and Dupont \textit{et al.}’s \cite{dupont2005building} method have constructed the ordinary curves with a small embedding degree, however, they did not fall into the families. The main problem with ordinary curves is to determine the elliptic curve parameters $(p,r,t)$ that satisfy the conditions given in definition \ref{def2.5}. In both methods, the trace of Frobenius some integer boosts $t$ in $Z\mathbb{Z}_r$. Thus, the value of $t$ is approximately identical to the size of $r$, which implies the $\ rho$ value is around 2. In practice, such a choice of parameters is inefficient for pairing computations in well-known Tate pairing variants, such as Ate and twisted Ate pairing. Algorithm (\ref{alg2.1}) gives the Cocks-Pinch Method for constructing an ordinary curve of composite order. 

\begin{algorithm} 
	\caption{Cocks-Pinch construction of composite order} 
	Fix$k$and D, and do the following steps: 
	\begin{algorithmic}[1]
    	\State Choose a prime $r$ such that $k|(r-1)$ and $((-D))/2=1$
	    \State Evaluate $z$, a primitive $k^{th}$ root of unity in $\mathbb{Z}_r$  and $t_0=z+1$
	    \State Evaluate $t_0=(t_0-2)/\sqrt{(-D)} mod r$
	    \State Reduce $t=t_0 modr$ and $y=y_0 mod r$, let $p=(t^2+Dy^2)/4$ 
    \end{algorithmic} 
    \label{alg2.1}
\end{algorithm}

\subsection{Families of Pairing-Friendly Curves}
The Cocks-Pinch and DEM methods are inefficient that can be solved by considering the elliptic curve parameters $(p,r,t)$ as polynomials family $(p(x),r(x),t(x))\in  \mathbb{Q}[x]$.  Based on the CM  polynomial $f(x)=4q(x)-t(x)^2$ that is the right-side of the CM equation, the elliptic curves are categorised into three families: complete families with fixed discriminant \cite{barreto2005pairing,brezing2005elliptic,kachisa2008constructing,barreto2002constructing}, complete families with variable discriminant \cite{freeman2010taxonomy,drylo2011constructing,lee2009generating}, and sparse families \cite{miyaji2001new,freeman2006constructing,fotiadis2018generating}, shown in Figure (\ref{fig2.1}). Their constructions depend on the ability to find integers $x$ and $y$ that satisfy the CM equation. For some set of $(x,y)$, the equation has solutions that grow exponentially, we call such families sparse. On another side, if the equation satisfies any $x$, and we define $y$ as a polynomial in $x$, we call such families complete. Due to the small CM discriminant value, complete families, are in real-world practices. Although, such curves are easily vulnerable to many attacks on DLP. We have two more families to avoid such attacks: sparse families and complete families with variable discriminant. The sparse families of curves with embedding degree $k \in  \{3,4,6\}$ exhibit a low level of security at 80 bits. 

\subsubsection{Sparse Families of Elliptic Curve}
The primary idea of a sparse family of curves is surrounded by Miyaji \textit{et al.} method \cite{miyaji2001new}. Most of the sparse family curves are in prime order but they are restricted to embedding degree $k \le 10$. To construct a sparse family of curves, we choose the polynomials $p(x)$, $r(x)$ and $t(x)$, which satisfy the first four conditions of Definition \ref{def2.5} and for which the CM equation  $Dy^2=4p(x)-t(x)^2=4h(x)r(x)-(t(x)-2)^2$ has infinitely many solutions, where $h(x)$ is a cofactor satisfying $h(x)r(x)=p(x)-1+t(x)$. For $h(x)=1$, we can obtain a curve of prime order. Miyaji \textit{et al.} [6] are the first to construct the ordinary elliptic curves of prime order with a prescribed embedding degree. They rely on the that if $4h(x)r(x)-(t(x)-2)^2$ is a quadratic polynomial, then we can transform the equation into a generalized Pell equation. In this direction, the equation gives an infinite number of solutions in which we get a family of curves. Algorithm (\ref{alg2.2}) presents the construction of the MNT curve with embedding degree $k$.

\begin{algorithm} 
	\caption{Construction of MNT curve with embedding degree k} 
	 Input: Parameters $p(u)$, $t(u)$ and $k$.
    Output: Set of points for MNT curve
	\begin{algorithmic}[1]
	    \State Let $p(u)$ and $t(u)$ be associated with elliptic curve with embedding degree $k$. Using Hass theorem, compute $r(u)$ as $r(u)=p(u)+1-t(u)$. 
	    \State Pick an integer for u until we find $u_0$ such that $p(u_0)$ and $r(u_0)$ are both prime.
	    \State Using the CM algorithm, find the solutions to the norm equation $t^2-4p=DV^2$, such that $|D|$ is small.  Repeat this step until we have a small $|D|$.
	    \State Find solutions to norm equation $t^2-4p=DV^2$.  
	    \State If $d=-3D$ is positive and square-free, it can have infinitely many solutions.  
    \end{algorithmic} 
    \label{alg2.2}
\end{algorithm}

Freeman \textit{et al.} \cite{freeman2006constructing} exhibit such curves by considering that if $f(x)=4p(x)-t(x)^2$ is square-free, the equation defines the affine plane of genus $g=(\text{deg}(f)-1)/2$. If $f(x)$ is quadratic, then the genus is zero, which means there are either no solutions or infinitely many integral solutions. For infinitely many solutions, we can obtain a family $(t,r,p)$ using Definition \ref{def2.5}.

\subsubsection{Complete Families with Fixed Discriminant}

In order to construct the complete family of curves, we find the polynomials $t(x)$, $r(x)$, $p(x)$ that fulfil the certain divisibility conditions and the equation has infinitely many solutions $(x,y)$. The complete family construction chooses the parameters $D$, $t(x)$, $r(x)$, $p(x)$ in such a way that $4h(x)r(x)-(t(x)-2)^2$ is $D$ times the square of polynomials $y(x)$. Scott \textit{et al.} \cite{scott2006generating} and Barreto \textit{et al.} \cite{barreto2002constructing} methods are the two methods for constructing complete families. The complete family is generality to the Brezing \textit{et al.} \cite{brezing2005elliptic} curve.

\begin{algorithm} 
	\caption{Brezing-Weng Construction for a cyclotomic family of curves} 
	 Given fixed integer$k$and positive square-free integer D, perform the following steps:
	\begin{algorithmic}[1]
	    \State Picks a number field $k$ containing $\sqrt{(-D)}$ and primitive $k^{th}$ root of unity $\zeta_k$.
	    \State Compute an irreducible polynomial $r(x) \in Z[x]$ such that $\mathbb{Q}[x]/r(x) = K$
    	\State Suppose a polynomial be  $t(x) \in \mathbb{Q}[x]$  that maps to $\zeta_{k+1} \in K$.
    	\State Suppose a polynomial is  $y(x) \in \mathbb{Q}[x]$  that maps to $(\zeta_{k-1})/\sqrt{(-D)} \in K$.
    	\State Suppose a $p(x) \in \mathbb{Q}[x]$ is defined as $(t(x)^2+Dy(x)^2)/4$. 
	    \State If parameters $p(x)$ and $r(x)$ are prime then triple $(p(x),r(x),t(x))$ denotes the family of curve with embedding$k$and discriminant $\sqrt{(-D)}$.  
    \end{algorithmic} 
    \label{alg2.3}
\end{algorithm}

\textbf{Cyclotomic families of curves}. Barreto \textit{et al.} \cite{barreto2002constructing} and Brezing \textit{et al.} \cite{brezing2005elliptic} extended the Cocks-Pinch method by parameterizing $(t, r, p)$ as polynomials and reducing the $\rho$-value. Brezing \textit{et al.} discussed the fullest generality of the construction. Barreto \textit{et al.} \cite{barreto2002constructing} presented the first construction of a cyclotomic family of curves by considering the polynomial $r(x)$ defining the number field $k$ to be the $k$-th cyclotomic polynomial. It selects $\zeta_k$ in $k$ and uses the constraints that if $k$ is divisible by 3, then $\sqrt{-3} \in k$

Brezing \textit{et al.} \cite{brezing2005elliptic} discussed a more general construction by fixing $r(x)$ to be the cyclotomic polynomial $\Phi_l(t(x))$ of a multiple $l$ of the desired embedding degree $k$ and selecting distinct representatives of $\zeta_k \in \mathbb{Q}[x]/(r(x))$. In both constructions, the discriminant $D$ is usually 1 or 3, and the cyclotomic polynomial must satisfy condition 2 of Definition \ref{def2.5}.

Algorithm \ref{alg2.3} shows the Brezing-Weng construction for a cyclotomic family of curves.

\textbf{Sporadic families of curves}. To construct the elliptic curve for the cyclotomic family Brezing and Weng assume cyclotomic polynomials $r(x)$. However, the extension of cyclotomic fields, defined by the non-cyclotomic polynomials yields more effective results. To construct such an extension field to examine the cyclotomic polynomials $\Phi_l(x)$ on any polynomial $u(x)$.  For any factorization of cyclotomic polynomial $\Phi_l(u(x))$, we have some advantage for constructing a curve; otherwise, as is usually the case, it is required to evaluate the parameter $(t,r,p)$ at $u(x)$. Galbraith \textit{et al.} \cite{galbraith2007ordinary} have examined that  $\Phi_l(u(x))$ is factorized if $u$ is quadratic and $\Phi_l$ has degree 8. For $l=12$, there are two $u(x)$ that factorizes $\Phi_l(x)$. BN curve constructed by Barreto \textit{et al.} \cite{barreto2005pairing} uses one such factorization.

\textbf{Scott-Barreto families of curves}. The Scotte-Barreto family of curves assumes $k$ as an extension of the cyclotomic field, but it does not contain an element $\sqrt{-D}$. For any polynomial $t(x)$ and $r(x)$ to be irreducible to $\Phi_k(t(x)-1)$, the field $\mathbb{Q}[x]/r(x)$ evaluates to an extension of a cyclotomic field.

\subsubsection{	Complete with Variable Discriminant }
In the sparse and complete family of curves, the methods given by Brezing and Weng assume that the CM discriminant $D$ is known in advance, and based on such $D$, all curves are constructed. For instance, the mostly curves constructed by Barreto, Lynn and Scott, and Brezing and Weng demand $D=3$. On another side, there exist some families of PF-EC with variable CM discriminant $D$. For higher security, it is required to utilize the curves with sufficiently large discriminant. Thus, Freeman \textit{et al.} address the complete families with variable discriminant. 

\begin{definition} \label{def2.7}
A tuple of polynomials $(t(x),r(x),p(x))$ parameterizes a complete family of elliptic curves with embedding degree$k$and variable discriminant if it fulfils the first four conditions of Definition \ref{def2.5}, and $4p(x)-t(x)^2=xh(x)^2$ for some $h(x) \in \mathbb{Q}[x]$. Substituting $x=Dx^2$, such that $D>0$ is a square-free integer, outputs a complete family $r(Dx^2)$, $t(Dx^2)$, $p(Dx^2)$ with discriminant $D$ that satisfies the first two conditions of Definition \ref{def2.5}.
\end{definition}

Freeman \textit{et al.} \cite{freeman2010taxonomy} get the variable discriminant families from complete families $(t,r,p)$ with discriminant $D$ in such a way that there must be the polynomials $t'$, $r'$, $p'$, $y' \in \mathbb{Q}[x]$ such that  $t(x)=t'(x^2)$, $r(x)=r'(x^2)$, $p(x)=p'(x^2)$ and $4p(x)-t(x)^2=Dx^2(y'(x^2))^2$. This family $(t',r',p')$ satisfies definition \ref{def2.7} on $x \leftarrow x/D$. Drylo \textit{et al.} \cite{drylo2011constructing} illustrates a construction for a complete family with variable discriminant, as shown in Algorithm (\ref{alg2.4}). 

\begin{algorithm} 
	\caption{Drylo Construction for a complete family with variable discriminant} 
	 Given number field $k$ containing $k^{th}$ root of unity $\zeta_k$, performs the following steps:
	\begin{algorithmic}[1]
	    \State Picks $z \in K$ in such a way that $a=-z^2$ is primitive element of $K$.
	    \State Suppose a polynomial $r(x)$ is $a$, and set $K=\mathbb{Q}[x]/r(x)$.
	    \State Suppose a polynomial is  $t(x) \in \mathbb{Q}[x]$  that maps to $\zeta_{k+1} \in K$.
	    \State Suppose a polynomial is  $h(x) \in \mathbb{Q}[x]$  that maps to $(\zeta_{k-1})/\sqrt{(-x)} \in K$.
	    \State Suppose a $p(x) \in \mathbb{Q}[x]$ is defined as $(t(x)^2+Dy(x)^2)/4$. 
	    \State If parameters $p(x)$ and $r(x)$ are prime then the triple $(p(x),r(x),t(x))$ denotes the family of the curve with embedding$k$and variable discriminant.  
    \end{algorithmic} 
    \label{alg2.4}
\end{algorithm}

\section{Some Parameterized Examples of Pairing-Friendly Curves }

Here, we present the parameters of some existing PF-EC with embedding degree $k$. To select an elliptic curve for constructing pairing, we parameterize the number of points on the curve and the field size defined by the polynomials as $r(x)$ and $p(x)$, respectively. For each $x_0$, the CM algorithm constructs an elliptic curve on given prime values for polynomials $r(x_0)$ and $p(x_0)$. Some well-known elliptic curves are shown by Miyaji \textit{et al.} \cite{miyaji2001new}, Barreto \textit{et al.} \cite{barreto2002constructing}, Barreto \textit{et al.} \cite{barreto2005pairing}, Cocks \textit{et al.} \cite{cocks2001identity}, Brezing \textit{et al.} \cite{brezing2005elliptic}, Kachisa \textit{et al.} \cite{kachisa2008constructing}, and Freeman \cite{freeman2006constructing}. 

\subsection{Individual Elliptic Curve }

\textbf{Cocks-Pinch Method}. The Cocks and Pinch provided the general construction for curves of arbitrary embedding degree $k$. Algorithm 2.1 discusses the construction of the elliptic curve over $\mathbb{F}_p$ of embedding degree $k$ if $p$ is a prime integer. In this method, the size of field $\mathbb{F}_p$ is related to the subgroup of prime, $r$ such that $p=r^2$, which gives value $\rho=2$. It has been observed that the curve with small $\rho$ values efficiently performs the arithmetic operations. Thus, the Cocks-Pinch method leads to inefficient implementation, which is impractical. 

\subsection{Complete Family of Curves}

For complete family construction, Brezing and Weng presented a method which starts by selecting a fixed embedding degree and square-free CM discriminant $D$. The complete family with small discriminant $D$ is fascinating for implementation. However, such families are vulnerable to various attacks on DLP. 
Brezing-Weng method. Brezing and Weng \cite{brezing2005elliptic} extended the generalization of the Cock and Pinch \cite{cocks2001identity} method to construct the families of PF-EC with small $\rho$ values. This method has a similar form, but instead of using polynomial $r(x)$, it uses $r$. This method gives the PF-EC having $\rho$ values less than two and closer to 1. However, finding a suitable polynomial $r(x)$ is challenging. For $k<1000$, and $k$ is odd, the parameters  $(t(x),r(x),p(x))$ for complete family of a pairing-friendly elliptic curve with embedding degree$k$and discriminant 1 are as follows:

\begin{align*}
r(x)&=\Phi_{4k}(x)\\
t(x)&=-x^2+1\\
p(x)&=1/4(x^{2k+4}+2x^{2k+2}+x^2k+x^4-2x^2+1)
\end{align*}

This family has $\rho$ values $\frac{(k+2)}{\phi(k)}$

\textbf{Barretto-Lynn-Scott Curve with embedding degree $k \in \{12,48\}$}. In 2002, Barretto, Lynn, and Scott \cite{barreto2002constructing} presented the elliptic curves, known as BLS curves, suitable for constructing optimal ate pairing. BLS curves have a small discriminant parameter that yields a simple equation from which we can derive PF-EC of embedding degree $k=0(mod6)$ with a maximal twist of $d=6$, and have a relatively small $\rho$ values, as defined as $\rho=\frac{(2+k/3)}{\phi(k)}$.  

The BLS curve is a unique elliptic curve over a finite field $\mathbb{F}_p$ defined by  $E:y^2=x^3+b$. Similar to the BN curve, it has a twist of order six but does not have a prime order. While its order is divisible by a large parameterized prime, $r$ and the pairing is defined on r-torsion points. On distinct embedding degrees, the BLS curves vary. Here, we discuss the BLS12 and BLS48 curves families with embedding degrees 12 and 48, respectively. The following equations define the parameters p and $r$ to construct the BLS curves. Table 2.2 describes the parameters p and $r$ to create the BLS curves with embedding degrees 12 and 48.

\begin{table}
        \centering
        \caption{Complete family of the curve with embedding degrees 12 and 48, using BLS method}
        \label{tbl2.2}
        \begin{tabular}{|c|c|c|}
            \hline
            \textbf{BLS12:} & $p= \frac{((u-1)^2(u^4-u^2+1))}{3+u}$ &     $r=u^4-u^2+1$ \\
            \hline
            \textbf{BLS48:} & $p= \frac{((u-1)^2 (u^16-u^8+1))}{3+u}$     &     $r=u^16-u^8+1$ \\
           
            \hline
        \end{tabular}
\end{table}

Where $u$ is a randomly chosen integer.

\textbf{Barreto and Naerhig curve with embedding degree $k=12$}. In 2005, Barrato and Naehrig \cite{barreto2005pairing} presented a PF-EC, generally known as the BN curve, which is considered a breakthrough in cryptography. An Ate pairing construct over BN curves is considered the optimal pairing. An elliptic curve $E$ over finite field $\mathbb{F}_p$, where $p \ge 5$, such that $p$ and curve $E$ order $r$ are prime numbers that are defined by 

\begin{align*}
p&=36u^4+36u^3+24u^2+6u+1\\
r&=36u^4+36u^3+18u^2+6u+1
\end{align*}

Where $u$ is a randomly chosen integer, BN curves are suitable for pairing with a strong level of security and achieve efficient pairing-based cryptographic primitives. For example, a pairing over a 256-bit BN curve is considered the secure method for signing a message with a small signature size. 

\textbf{KSS curve with embedding degree $k \in \{8,16,18,32,36,40\}$}. Kachisa \textit{et al.} \cite{kachisa2008constructing} extended the idea of Brezing \textit{et al.} \cite{brezing2005elliptic} to discuss a new method to construct the PFEC.  This method uses a minimal polynomial of elements other than cyclotomic polynomial $\Phi_l(x)$ to define the cyclotomic field. This method constructs new families of PF-EC with embedding degree $k \in \{8,16,18,32,36,40\}$. 

\begin{sidewaystable}
        \centering
        \caption{Complete families of curve with embedding degree $k \in \{8,16,18,32,36,40\}$ by KSS method}
        \label{tbl2.3}
        \begin{tabular}{|c|c|}
            \hline
            k=8, D=3, & $p(u)= \frac{1}{3}(u^{10}+u^9+u^8-u^6+2u^5-u^4+u^2-32u+1)$ \\
            $\rho=5/4$ & $r(u)=u^8-u^4+1$ \hspace{10mm} $t(u)=u^5-u+1$ \\
            \hline
            \hline
            k=8, D=1, & $p(u)= \frac{1}{180}(u^6+2u^5-3u^4+8u^3-15u^2-8u+125)$ \\
            $\rho=3/2$ & $r(u)=u^4-8u^2+25$ \hspace{10mm} $t(u)=1/15(2u^3-11u+15)$ \\
            \hline
             \hline
            k=16, D=1, & $p(u)= \frac{1}{980} (u^{10}+2u^9+5u^8+48u^6+152u^5+240u^4+625u^2+2398u+3125)$ \\
            $rho=5/4$ & $r(u)=u^8+48u^4+625$ \hspace{10mm} $t(u)= \frac{1}{35}(2u^5+41u+35)$ \\
            \hline
             \hline
            k=18, D=3, & $p(u)= \frac{1}{21} (u^{10}+5u^7+7u^6+37u^5+188u^4+259u^3+343u^2+1763u+2401)$ \\
            $\rho=4/3$ & $r(u)=u^6+37u^3+343$ \hspace{10mm} $t(u)= \frac{1}{7}(u^4+16u+7)$ \\
            \hline
            \hline
            k=32, D=1, & $p(u)= \frac{1}{2970292} ((u^{18}-6u^{17}+13u^{16}+57120u^{10}+344632u^9+742560u^8+815730721u^2-4948305594u+1060449373))$ \\
            $\rho=9/8$ & $r(u)=u^{16}+57120u^8+815730721$ \hspace{10mm} $t(u)= \frac{1}{3107}(-2u^9+56403u+3107)$ \\
            \hline
             \hline
            k=32, D=3, & $p(u)= \frac{1}{28749}(u^{14}+46u^{13}+7u^{12}+683u^8-2510u^7+4781u^6+117649u^2-386569u+823543)$ \\
            $\rho=7/6$ & $r(u)=u^{12}+683u^6+117649$ \hspace{10mm} $t(u)= \frac{1}{259}(2u^7+757u+259)$ \\
            \hline
             \hline
            k=40, D=1, & $p(u)= \frac{1}{1123} (u^{22}-2u^{21}+5u^{20}+6232u^{12}+10568u^{11}+31160u^{10}+9765625u^2-13398638u+48828125)$ \\
            $\rho=11/8$ & $r(u)=u^{16}+8u^{14}+39u^{12}+112u^{10}-79u^8+2800u^6+24375u^4+125000u^2+390625$ \hspace{10mm} $t(u)= \frac{1}{1185}(2u^{11}+6469u+1185)$ \\
            \hline
        \end{tabular}
\end{sidewaystable}

\textbf{Scott-Guillevic with embedding degree k=54}. The elliptic curve with embedding degree 48 suggested by Barreto \textit{et al.} \cite{barreto2002constructing}  achieved the AES-256 level of security in practice. The curves with embedding degrees up to 50 are only considered by the Freeman \textit{et al.} \cite{freeman2010taxonomy}. Recently, Scott \textit{et al.} \cite{scott2018new} noticed new curves using the KSS method, but the curve with embedding degree 54 is not the kind of KSS curve. It is observed that such a family of curves was occasional and unrelated to any existing family. However, it exhibits a certain pattern, which shows some possibility that it might be a member of an undiscovered family of families. Thus, the authors announced the invention of new PF-EC with embedding degree 54, strengthening the difficulty of the discrete logarithm problems since it is used for pairing-based cryptography. There exists an element $-\zeta_{54}-\zeta_{54}^{10} \in Q(\zeta_{54})$ that using the KSS method yields the following parameters 

\begin{align*}
p(u)=310u^{20}+310u^{19}+39u^{18}+36u^{11}+36u^{10}+35u^9+3u^2+3u+1 \\
r(u)=39u^{18}+35u^9+1 \\
t(u)=35u^{10}+1 \\
c(u)=3u^2+3u+1
\end{align*}

where $c$ is cofactor. 

The total number of points on the curve is $\#E=cr$. Since $4p-t^2=3f^2$ for some polynomial $f$, it can be observed that the CM discriminant is 3. Therefore, like BN and BLS curves, it has twists of degree 6 that provide essential optimizations. 

\subsection{Complete Families with Variable Discriminant}
Recall that the family of curves with enormous discriminant values are less susceptible to various attacks on DLP. In this respect, complete families with variable discriminants use the CM discriminant in polynomial representation. In this family, the CM polynomial has the form $f(x)=g(x)y(x)^2$ and $\text{deg}(g)=1$. Such families can be constructed using the Brezing-Weng method by substituting the square-free $D$ with a linear term $g(x)$ such that $\sqrt{-g(x)}\in \mathbb{Q}[x]/r(x)$. These families exhibit preferable CM discriminant but with limited choices. To evaluate suitable parameters for variable CM discriminant, we look for $x_0=Dy^2 \in \mathbb{Z}$ such that $r(x_0)$ and $p(x_0)$ are prime.

Robert Drylo method for CVD. Drylo constructed the complete families with variable discriminants by generalizing the Brezing-Weng method. This method enables us to find families with few embedding degrees, which improves on $\rho$ -values with parameter $\rho=\frac{max\{2degt,1+2degh\}}{degr}$, which is equal to $\rho= \frac{(2degr-1)}{degr}$. 

\begin{sidewaystable}
        \centering
        \caption{Complete families of curve with variable discriminate with embedding degree $k \in \{8,9,15,28,30\}$ by KSS method}
        \label{tbl2.4}
        \begin{tabular}{|c|c|}
            \hline
            k=8, D=1,11, & $p(u)= \frac{1}{576}(4u^7-39u^6+170u^5-311u^4+52u^3+716u^2-384u+196)$ \\
            $\rho=7/4$ & $r(u)=u^4-4u^3+8u^2+8u+4$ \hspace{10mm} $t(u)= \frac{1}{12} (-u^3+5u^2-16u+14)$ \\
            \hline
            \hline
            k=9, D=1, & $p(u)= \frac{1}{4}(59049u^{10}+6561u^9+8748u^8+2916u^7+972u^6+1296u^5+108u^4+36u^3+12u^2+u+1)$ \\
            $\rho=5/3$ & $r(u)=729u^5+27u^3+1$ \hspace{10mm} $t(u)=243u^5+1$ \\
            \hline
             \hline
            k=15, D=1, & $p(u)= \frac{1}{4} (531441u^{13}-236196u^{11}+39366u^{10}+39366u^9-8748u^8-729u^7+486u^6-243u^5+135u^4+18u^3+18u^2+u+1)$ \\
            $\rho=13/8$ & $r(u)=6561u^8-2187u^7+243u^5-81u^4+27u^3-3u+1$ \hspace{10mm} $t(u)= 9u^2+1$ \\
            \hline
            \hline
            k=28, D=3,1, & $p(u)= \frac{1}{4} (2624144u^{18}+65536u^{17}-32768u^{15}+16384u^{14}+12288u^{13}-3072u^{11}+2816u^9-192u^7+48u^5+16u^4-8u^3+u+1)$ \\
            $\rho=3/2$ & $r(u)=4096u^{12}-1024u^{10}+256u^8-64u^6+16u^4-4u^2+1$ \hspace{10mm} $t(u)=512u^9+1$ \\
            \hline
             \hline
            k=30, D=1, & $p(u)= \frac{1}{4} (244140625u^{13}+195312500u^{12}+78125000u^{11}+19531250u^{10}+2353750u^9-140625u^9-43750u^6-6875u^5-125u^4+150u^3$\\
            $\rho=13/8$ & $-50u^2+9u+1)$ \\
             & $r(u)=390625u^8+78125u^7-3125u^5-625u^4-125u^2+5u+1$ \hspace{10mm} $t(u)=-25u^2+1$ \\
            \hline
        \end{tabular}
\end{sidewaystable}

\subsection{Sparse Families}
Identical to complete families with variable discriminant, sparse families have a large discriminant, so they avoid any attack on DLP. In this family, the CM polynomial has the form $f(x)=g(x)y(x)^2$ and $\text{deg}(g)=2$ and non-square. Several sparse family curves have prime order but are restricted to embedding degree $k\le 10$. Such families of curves fulfill conditions 1-4 of the definitions on given polynomials $p(x)$, $r(x)$, and $t(x)$, and help in responding to infinitely many solutions in the CM equation $Dy^2=4p(x)-t(x)^2=4h(x)r(x)-(t(x)-2)^2$, where $h(x)$ is a cofactor. It transforms the CM polynomial $4h(x)r(x)-(t(x)-2)^2$ into the generalized Pell equation.

\textbf{MNT curves with embedding degree $k \in \{3,4,6\}$}. In 2002, Miyaji, Nakabayashi, and Takano \cite{miyaji2001new} described a straightforward approach for constructing ordinary curves $E(\mathbb{F}_p)$ of prime order $n=r$ with embedding degree $k \in \{3,4,6\}$. Consider a prime $p$ and $E(\mathbb{F}_p)$ as an elliptic curve such that $r=\#E(\mathbb{F}_p)$ is prime. Let $t=p+1-r$ be a trace. Then, the elliptic curve has to embed degree$k$and if only $u \in Z$ exists, such that $p(u)$ and $t(u)$ are given in Table 2.5. 

\begin{table}
        \centering
        \caption{Parameterization for MNT curves with embedded degree $k \in \{3,4,6\}$}
        \label{tbl2.4a}
        \begin{tabular}{|c|c|c|}
            \hline
           $k$&	p(u) &	t(u) \\
            \hline
            \hline
            3 &	$12u^2-1$ &	$-1 \pm 6u$ \\
            4 &	$u^2+u-1$	& -u or $u+1$ \\
            6 &	$4u^2+1$	& $1 \pm 2u$ \\
            \hline
        \end{tabular}
\end{table}

For embedding degree, $k \in \{3,4,6\}$, the CM equation $Dy^2=4p(x)-t(x)^2$ determines an elliptic curve of genus-zero, with the right-hand side, are quadratic in $x$. With linear variables changes, this CM equation can be transformed into the generalized Pell equation $x^2-SDy^2=M$. The method finds the solution to the generalized Pell equation with a minimal positive integer solution. Algorithm 1 summarizes to find the resolution of the MNT curve.

\textbf{Scott and Barreto method}. Scott and Barreto \cite{scott2006generating} extended the MNT method by aiding a small constant cofactor $h$. They fix small integers h and d and replace $r$ with $\Phi_k (t-1)/d$ and $t$ with $x+1$ in (\ref{eq2.6}) that gives 
\begin{equation} \label{eq2.6}
    Dy^2=(4h\Phi_k (x))/d-(x-1)^2     
\end{equation}

It can be observed that the right-hand side of the equation (\ref{eq2.6}) is quadratic in $x$ for embedding degree $k \in \{3,4,6\}$ as similar to the MNT curve. So we write (\ref{eq2.6}) into the generalized Pell equation. The MNT method is applied to explore the curves with embedding degree $k \in \{3,4,6\}$ for almost prime order.

\textbf{GMV method}: Galbraith \textit{et al.} \cite{galbraith2007ordinary} also extended the MNT method for constructing an elliptic curve with an embedding degree $k \in \{3,4,6\}$ by allowing cofactor $h \in \{2,3,4,5\}$. They have followed the MNT method but substitute hr with $\#E(\mathbb{F}_p)$. In the case of prime order, every parameterization of $t(u)$ is linear, and every parameterization of $p(u)$ is quadratic. Thus final CM equations $Dy^2=4p(x)-t(x)^2$ are quadratic, which enables for transforming into a generalized Pell equation.

\textbf{Freeman curve of embedding degree $k=10$}. Freeman \cite{freeman2006constructing} addressed the open problem posed by Boneh \textit{et al.} \cite{boneh2001short} by constructing an elliptic curve with embedding degree 10. He chooses $n(x)$ and $t(x)$ in such a way that the high degree value of $t(x)^2$ is cancelled out with those of $4r(x)$ in quadratic equation $f(x)=4r(x)-(t(x)-2)^2$. He found that this is possible only with the curve of embedding degree 10. 

To construct a curve with embedding degree $k = 10$, Freeman chooses the following parameters: 

\begin{align*}
       r(u)=25u^4+25u^3+15u^2+5u+1\\
p(u)=25u^4+25u^3+25u^2+10u+3 
\end{align*}

We have 
\begin{align*}
    t(u)=10u^2+5u+3
\end{align*}

such that 
\begin{align*}
    r(u)|\Phi_10 (p(u))
\end{align*}

and,
\begin{align*}
    t(u)^2+4p(u)=-(15u^2+10u+3)
\end{align*}

If the norm equation $u^2-15Dv^2=-20$ has a solution with $u \equiv 5 (mod 15)$, then $\{p,r,t\}$ denotes the family of curves with embedding degree 10. 

\textbf{Robert Drylo method for sparse}. Drylo \cite{drylo2011constructing} also constructed the variable discriminant families using a sparse family. He generalized the Brezing-Weng method to construct such families, but the method is less efficient than the method for constructing variable discriminants using complete families. 

\begin{table}
        \centering
        \caption{Sparse families with variable discriminant and embedding degree $k \in \{8,10,12\}$ by Drylo method}
        \label{tbl2.6}
        \begin{tabular}{|c|c|}
            \hline
            k=8, D=1,11, & $p(u)=25u^4+25u^3+15u^2+5u+1$ \\
            $\rho=1$ & $r(u)=25u^4+25u^3+25u^2+10u+4$ \hspace{10mm} $t(u)= 10u^2+5u+3$ \\
            \hline
            \hline
            k=8, D=1, & $p(u)= \frac{1}{576}(u^6-6u^5+7u^4-36u^3+135u^2+186u-63)$ \\
            $\rho=3/2$ & $r(u)=u^4-2u^2+9$ \hspace{10mm} $t(u)= \frac{1}{12}(-u^3+3u^2+5u+9)$ \\
            \hline
             \hline
            k=12, D=1, & $p(u)= \frac{1}{900} (u^6-8u^5+18u^4-56u^3+202u^2+258u-423)$ \\
            $\rho=3/2$ & $r(u)=u^4-2u^3-3u^2+4u+13$ \hspace{10mm} $t(u)= \frac{1}{12} (-u^3+4u^2+5u+6)$ \\
            \hline
        \end{tabular}
\end{table}

\textbf{Fotiadis and Konstantinou}. Fotiadis \textit{et al.} \cite{fotiadis2018generating} extended the Drylo’s method \cite{drylo2011constructing} for constructing the first sparse families of curves and implementing different embedding degrees $k \in \{5,8,10,12\}$. Fotiadis \textit{et al.} \cite{fotiadis2018generating} suggested that the families $\rho(p,r,t)=2$ provide the right balance between the size of prime order $r$ and the size of the extension field. It also presented many numerical examples of cryptographic parameters after considering the recent TNFS attack method to reduce the complexity of DLP in field extension of composite degree. This method constructed the sparse families of curves by adopting the Lee-Park’s \cite{lee2009generating} method to produce an irreducible polynomial $r(x) \in \mathbb{Q}[x]$ and trace of polynomial $t(x) \in \mathbb{Q}[x]$ in such a way that $r(x)|\Phi_k(t(x)-1)$ on some fix embedding degree $k$. Then, it determines a non-square quadratic polynomial $g(x)$, such that the CM polynomial is $f(x)=g(x)y(x)^2$, with $p(x) \in \mathbb{Q}[x]$ using the Drylo’s method \cite{drylo2011constructing}. The rest polynomials $p(x)$ and $y(x)$ are straightforward. Table 2.7 provides the parameters of cyclotomic sparse families of curves.

\begin{table}
        \centering
        \caption{Cyclotomic sparse families of curves with an embedding degree and $\rho$-value}
        \label{tbl2.7}
        \begin{tabular}{|c|c|}
            \hline
            $k=5 \Phi(k)=4$, & $y(x)=-(2x^2+2x+1)$ \\
            $\rho=3/2$ & $g(x)=3x^2-2x+3$ \hspace{10mm} $t(x)= x+1 $ \\
            \hline
            \hline
            $k=8 \Phi(k)=4$, & $y(x)=-\frac{1}{17}(3x^2-x+3)$ \\
            $\rho=3/2$ & $g(x)=7x^2-26x+7$ \hspace{10mm} $t(x)= -x^3+1 $ \\
            \hline
            \hline
            $k=10 \Phi(k)=4$, & $y(x)= \frac{1}{11}(x^2+3x+1)$ \\
            $\rho=3/2$ & $g(x)=3x^2+10x+3$ \hspace{10mm} $t(x)= x^3+1 $ \\
            \hline
            \hline
            $k=10 \Phi(k)=4$, & $y(x)= \frac{1}{71}(7x^2-x+7)$ \\
            $\rho=3/2$ & $g(x)=15x^2+50x+15$ \hspace{10mm} $t(x)= x^3+1 $ \\
            \hline
            \hline
            $k=7 \Phi(k)=6$, & $y(x)= \frac{1}{55}(38x^4-23x^3+50x^2-23x+38)$ \\
            $\rho=5/3$ & $g(x)=208x^2+375x+208$ \hspace{10mm} $t(x)= x^5+1 $ \\
            \hline
            \hline
            $k=9 \Phi(k)=6$, & $y(x)= -\frac{1}{109}(x^4-18x^3-4x^2-18x+1)$ \\
            $\rho=5/3$ & $g(x)=8x^2+35x+8$ \hspace{10mm} $t(x)= x^5+1 $ \\
            \hline
            \hline
            $k=14 \Phi(k)=6$, & $y(x)= -(2x^4-5x^3+6x^2-5x+2)$ \\
            $\rho=5/3$ & $g(x)=4x^2+5x+4$ \hspace{10mm} $t(x)= x^5+1 $ \\
            \hline
            \hline
            $k=18 \Phi(k)=6$, & $y(x)= -\frac{1}{19}(3x^4-2x^3-8x^2-2x+3)$ \\
            $\rho=5/3$ & $g(x)=4x^2+9x+4$ \hspace{10mm} $t(x)= x^5+1 $ \\
            \hline
             \hline
            $k=30 \Phi(k)=8$, & $y(x)= \frac{1}{9755}(433x^6-293x^5-149x^4+637x^3-149x^2-293x+433)$ \\
            $\rho=7/4$ & $g(x)=155x^2+350x+155$ \hspace{10mm} $t(x)= x^7+1 $ \\
            \hline
            \hline
            $k=10 \Phi(k)=4$, & $y(x)= -\frac{1}{19}(8x^3-8x^2+1)$ \\
            $\rho=2$ & $g(x)=15x^2+50x+15$ \hspace{10mm} $t(x)= x+1 $ \\
            \hline
            \hline
            $k=14 \Phi(k)=6$, & $y(x)= 3x^5-4x^4+3x^3-2x+2$ \\
            $\rho=2$ & $g(x)=4x^2+5x+4$ \hspace{10mm} $t(x)= -x^2+1 $ \\
            \hline
             \hline
            $k=18 \Phi(k)=6$, & $y(x)= -\frac{1}{19} (7x^5-x^4-6x^2-6x+10)$ \\
            $\rho=2$ & $g(x)=4x^2+9x+4$ \hspace{10mm} $t(x)= x+1 $ \\
            \hline
             \hline
            $k=18 \Phi(k)=6$, & $y(x)= \frac{1}{37} (26x^5-14x^4-12x^2-12x^2-12x+29)$ \\
            $\rho=2$ & $g(x)=19x^2+30x+19$ \hspace{10mm} $t(x)= x+1 $ \\
            \hline
             \hline
            $k=15$, $\Phi(k)=15$, & $y(x)= \frac{1}{93} (20x^7-8x^6-22x^5+20x^4+14x^3\pm 6x^2+7x-15)$ \\
            $\rho=2$ & $g(x)=3x^2-18x+3$ \hspace{10mm} $t(x)= x^2+1 $ \\
            \hline
             \hline
            $k=20 \Phi(k)=15$, & $y(x)= -\frac{1}{505} (20x^7+23x^6-43x^5-4x^4+24x^3+68x^2-88x+20)$ \\
            $\rho=2$ & $g(x)=40x^2-55$ \hspace{10mm} $t(x)= x+1 $ \\
            \hline
        \end{tabular}
\end{table}

\section{Recent Attacks on PF-EC: TNFS}

\subsection{Security of Pairing-Friendly Curves}
The security of a pairing-friendly curve is based on solving the DLP on an elliptic curve and DLP on a finite field $\mathbb{F}_p$. The General Number Field Sieve (GNFS) is one of the algorithms to compute the discrete logs over a finite field $\mathbb{F}_p$ in sub-exponential time while Pollard’s rho algorithm breaks ECDLP in time $O(\sqrt{p})$, which is equivalent to $O(r)$, where $r$ is the subgroup order. An n-bit AES symmetric key achieves an equivalent security level as an elliptic curve of order 2n prime. Although, it is suggested that a good choice for security level should be larger than 128-bit which could be achieved by the curves with prime factor order larger than 256-bits.

There are other difficult problems over which the security of PF-EC is also evaluated. Such hard problems include computational Diffie-Hellman problem (CDHP) \cite{joux2003separating}, Decision Diffie-Hellman problem (DDHP) \cite{boneh1998decision} and gap DDHP \cite{choon2003identity}. It is believed that all such problems are reduced to the hardness of DLP and assumed to be easier than DLP. To break the pairing-based cryptography, an attacker could solve these problems, although no such attacks have yet been identified, so we address the hardness of DLP.

The next factor that estimates the pairing-friendly curve's security level is the computation overhead of an algorithm that solves the DLP. Some such well-known algorithms are index calculus and the Pollard rho algorithm. In 2001, the number field sieve algorithm was utilized to make index calculus more efficient. Thus, the choice of selecting the embedding degree$k$and the characteristic of the extension field determines the security of DLP in the finite extension field. To represent the complexity of DLP, we recall the L-notation given by

\begin{equation} \label{eq2.7}
    L_N[l:c] = \exp\big[\big(c+o(1)\big) \big(\ln N\big)^l \big(\ln\ln N\big)^{(t-l)}\big]
\end{equation}

where $l \in [0,1]$ is real constant and $c>0$, and $N=p^k$. Usually, the NFS attacks finite field extension with complexity $L_N[1/3,1.923]$. This complexity still works for finite extensions of prime degrees. Now, there has been recent progress on the TNFS method, such as extension TNFS (exTNFS) \cite{kim2016extended} and special exTNFS (SexTNFS) approaches that reduce the complexity of DLP on finite field extensions for composite$k$and special form of p, to $L_N[1/3,1.526]$.

\subsection{Recent Attacks on PF-EC}
Here, we discuss the impact of recent attacks on pairing-friendly curves. A pairing map points on the elliptic curve to values in some "extension field." Security depends on both the EC-DLP for the elliptic curve and also DLP on the extension field. Thus, the pairings can be attacked in two points: elliptic curve or finite extension field.

A p-bit BN curve has a prime field of p bits and an embedding degree $k$. Thus, the extension field is $modp^k$ bits. For instance, 256-bit BN curves have a prime field of 256-bit and embedding degree 12, so the embedding degree is $p^{12}$, i.e., 3072 bits. It has been claimed that 256-bit DLP on the elliptic curve and 3072-bit DLP on the finite field would have 128 bits of security. It is suggested that BN curves are not secure enough. In 2016, after Kim \textit{et al.}’s paper \cite{kim2017extended}, the security in the extension field case changed, making it challenging to provide a concrete estimation. Kim \textit{et al.} \cite{kim2017extended} proposed a new kind of NFS algorithm, the extended tower number field sieve (exTNFS) reduces the complexity of solving DLP in a finite field. The exTNFS asymptotically dropped down the security level of elliptic curves for pairing.

In this respect, it is suggested that DLP for 3072-bit $prime^{12}$  is easy to attack than DLP for a 3072-bit prime. Recently, Barbulescu \textit{et al.} \cite{barbulescu2019updating} estimated that the BN curves which had previously achieved 128 bits of security dropped to around 100 bits. Menezes \textit{et al.} \cite{menezes2016challenges} indicate that the BN curve used 383-bits length of p after applying exTNFS for achieving 128-bit of security, and that of BLS12 curves as 384-bits. Kiyomura \textit{et al.} \cite{kiyomura2017secure} indicates that a bit length of $p^k$ for BLS48 curves is 27,410 bits after applying exTNFS with 572 bit of $p$.

\section{Pairing-Friendly Curves Resilient to TNFS Attacks}
We recommend the pairing-friendly elliptic curve secured against the recent TNFS and its variant attacks. The new progress in TNFS work has a significant effect on constructing the pairing-friendly elliptic curve with a composite degree $k$. One immediate response is that the finite extension field must be larger than the previously chosen field extension and, thus, $\rho=1$ will not be a better option for composite k. For instance, BN curves with embedding degree 12 were an ideal choice that generates 256-bit prime and 3072-bit extension field and $\rho=1$. Before TNFS attacks, such parameters achieved a 128-bit security level. Accidentally, after the progress of the TNFS method, the extension field of such size holds a 110-bit security level. One should opt for the size of the extension field around 4608 bits in order to reach 128-bit security. Recall, BN curves have $\rho=1$ this will respond to the size of a subgroup of order $r=384$, that will imbalance the security of $\mathbb{G}_1$, $\mathbb{G}_2$ and the security of target group $\mathbb{G}_T$.

Recently, Fotiadis \textit{et al.} \cite{fotiadis2019tnfs} acknowledged the impact of variants of TNFS and revised the conditions for constructing polynomial families $(p(x),r(x),t(x))$ of curves. They presented the optimal families that achieve a balance of security level in a tuple of pairing group $(\mathbb{G}_1,\mathbb{G}_2,\mathbb{G}_T)$ and give the pairing-friendly parameters for composite embedding degrees that are secured against the TNFS attacks. Besides, they suggested the use of those polynomial families of prime embedding degrees that were not considered previously due to the enormous $\rho$ value.

\begin{sidewaystable}
        \centering
        \caption{Recommended  family of curves at 128-bit security level \cite{guillevic2020short}}
        \label{tbl2.8}
        \begin{tabular}{|c|c|c|c|c|c|c|c|}
            \hline
            Curves &	k	& D &	r bits &	p bits &	$p^k$ bits &	Seed u &	Security level \\
            \hline
            \hline
            Cock-Pinch &	6 &	3 &	256 &	672 &	12255 &	$2^{128}-2^{124}-2^{59}$	& 128 \\
            Cock-Pinch &	8 &	1 &	256	& 544 &	13799 &	$2^{64}-2^{54}+2^{37}+2^{32}-4$ &	131 \\
            Cyclo FM &	10 &	15 & 256 &	446 &	12255 &	$2^{32}-2^{26}-2^{17}+2^{10}-1$, $a=-3$ &	133 \\
            Cyclo FM &	11	& 3 &	258 &	333 &	11477 &	$-2^{13}+2^{10}-2^8-2^5-2^3-2$, $b=13$	& 131 \\
            Cyclo FM &	11	& 11 &	256 &	412 &	12255 &	$-2^{56}+2^{21}+2^{19}-2^{11}-2^9-1$, $a=2$ &	145 \\
            BN & 12 &	3	& 446 &	446 &	13799 &	$2^{110}+2^{36}+1$, $b=257$ &	132 \\
            Cyclo BLS &	12 &	3 &	229 &	446 &	12255 &	$-2^{74}-2^{73}-2^{63}-2^{57}-2^{50}-1$, $b=1$ &	132 \\
            FK &	12 &	3 &	296 &	446 &	11477 &	$-2^{72}-2^{71}-2^{36}$, $b=-2$	& 136 \\
            Cyclo &	13 &	3 &	267 &	310 &	12255 & $2^{11}+2^8-2^6-2^4$, $b=-17$ &	140 \\
            Cyclo &	14 &	3 &	256 &	340 &	13799 &	$2^{21}+2^{19}+2^{10}-2^6$, $b=-4$ &	148 \\
            KSS16 &	16 &	1 &	257 &	330 &	12255 &	$-2^{34}+2^{27}-2^{23}+2^{20}-2^{11}+1$, $a=1$	& 140 \\
            KSS16 &	16 &	1 &	256 &	330 &	11477 &	$2^{34}-2^{30}+2^{26}+2^{23}+2^{14}-2^5+1$, $a=1$ &	140\\
            \hline
        \end{tabular}
\end{sidewaystable}

\textbf{Composite embedding degree}. Fotiadis \textit{et al.} \cite{fotiadis2019tnfs} construct complete and CVD families of curves for different composite embedding degrees that have generated the suitable curve parameters secured against the TNFs attacks. Compared to the previously discussed results, these families have a large $\rho$ value to expand the size of the extension field. Subsequently, it enhances the complexity of the DLP in target group $\mathbb{G}_T$. The variant of TNFS attacks, SexTNFS achieved the asymptotic complexity of DLP in the extension field of the composite degree to $L_N [1/3,1.526]$ where $N=p^k$.

\begin{table}
        \centering
        \caption{Recommended  family of curves at the 192-bit security level \cite{guillevic2020short}}
        \label{tbl2.9}
        \begin{tabular}{|c|c|c|c|c|c|c|c|}
            \hline
            Curves &	k &	r bits &	p bits &	$p^k$ bits &	Seed u &	Security level \\
            \hline
            \hline
            BN	& 12 &	1024 &	1022 &	12255 &	$-2^{254}+2^{33}+2^6$ &	191 \\
            BLS12 &	12 &	768 &	1150 &	13799 &	$-2^{192}+2^{188}-2^{115}-2^{110}-2^{44}-1$	& 193 \\
            KSS16 &	16 &	605 &	766 &	12255 &	$2^{78}-2^{76}-2^{28}+2^{14}+2^7+1$ &	194 \\
            KSS18 &	18 &	474 &	638 &	11477 &	$2^{80}+2^{77}+2^{76}-2^{61}-2^{53}-2^{14}$ &	193\\
            BLS24 &	24 &	409 &	509 &	12202 &	$-2^{51}\neg -2^{28}+2^{11}-1$ & 	193\\
            
            \hline
        \end{tabular}
\end{table}

\textbf{Prime embedding degree}. Fotiadis \textit{et al.} \cite{fotiadis2019tnfs} also suggested those inefficient complete and CVD families of curves for prime embedding degrees that did not previously notice due to large $\rho$ value. The objective of choosing curves with large $\rho$ value is to produce a balanced security level in three $\mathbb{G}_1$, $\mathbb{G}_2$ and $\mathbb{G}_T$. However, the variant of TNFS attacks does not affect the prime degree extension fields, where the complexity of the DLP is evaluated to $L_N[1/3,1.923]$, where $N=p^k$.

\textbf{128-bit security}. For 128-bit security, the finite field size required for BLS12 and BN curves is about $12 \times 448 = 5376$. To ensure the security of the curve, the size of $r$ should be at least 256 bits. At the 128-bit security level, we set the limit as $3072 \le 256\rho k \le 5376$ to narrow down the choices for families of curves. We obtain $k \le 21$ and $k \ge 6$ for $\rho=1$ and $\rho=2$, respectively. Table 2.8 summarizes the recommended parameters for secure PF-EC with embedding degree $k={6,8,10,11,12,13,14,16}$ at the 128-bit security level.

\textbf{192-bit security}. Guillevic \textit{et al.} \cite{guillevic2020short} assumed the limit $7168 \le 384\rho k \le 14336$ for a 192-bit level of security. We can obtain $k \le 37$ and $k \ge 10$ for $\rho=1$ and $\rho=2$, respectively. Fotiadis \textit{et al.} \cite{fotiadis2019tnfs} curves with $\rho=2$ satisfy the boundary for $10 \le k \le 18$. There cannot find any cyclotomic family of curves with embedding degree $k=32$. Table 2.8 shows the recommended parameters for families of curves: BN, BLS12, BLS24, KSS16, and KSS18 at a 192-bit level of security.

\section{Adoption of Pairing-Friendly Curves}

Here, we discuss the adoption in international standards, libraries, and applications, followed by their classification by considering the general security level such as 128-bits, 192-bits, and 256-bits. Some of these adoptions considered the TNFS secured pairing-friendly curves. After considering the variant of TNFS attacks, we recommend the first two of three aspects to select the pairing-friendly curves: security, usage, and efficiency.
\subsection{International Standards}

The ISO/IEC is a technical board of the International Organization for Standardization (ISO) and the International Electro-technical Commission (IEC) that deals with the maintenance, development, and standards promotion in information technology \cite{gorbenko2017examination}. The 15946 series of ISO/IEC reviews the elliptic curve-based cryptographic primitives. The 15946-5 series of ISO/IEC includes some samples of BN curves of field size 160-bit, 192-bit, 224-bit, 256-bit, 384-bit, and 512-bit, Freeman curves of field size 224-bit and 256-bit and MNT curves of field size 160-bit and 256-bit. Such examples do not consider the impact of exTNFS. The trading card game (TCG) standard selects the BN256 and a BN curve of field size 638-bit. Fast ID Online (FIDO) alliance is a group of technology-agnostic security specifications for secure authentication. The World Wide Web Community (W3C) is an international community that articulates web standards that guarantee that the functionality and framework of websites are related in each web browser. Both the FIDO alliance and W3C adopt BN256 and BN512 \cite{verheul2016practical}.

\subsection{Cryptographic Libraries}
Many cryptographic libraries support the elliptic curve and pairing operations. PBC is one of the standard cryptographic libraries for pairing-based cryptography proposed by Ben Lynn that supports MNT curves, BN curves, Freeman curves, and supersingular curves \cite{lynn2010pairing}. Another library for pairing-based cryptography is the mcl that supports BN-254 curves, BN-SNARK1, BN-382, BN-462, and BLS12-381 \cite{abdulla2019hitc}. Tsukuba Elliptic Curve and Pairing Library (TEPLA) support BN254N and BN254B \footnote{A. Kanaoka, “Tepla elliptic curve and pairing library,” online, 2018. [Online]. Available: https://github.com/TEPLA/tepla-library}. Intel announces a cryptographic library named Intel Integrated Performance Primitives (Intel-IPP), which supports BN-256I.

\begin{table}
        \centering
        \caption{Adoption of pairing-friendly elliptic curve}
        \label{tbl2.10}
        \begin{tabular}{|c|c|c|c|c|c|}
            \hline
            \multicolumn{1}{|c|}{Category} & \multicolumn{1}{|c|}{Name} & \multicolumn{4}{|c|}{Security level} \\            
            \cline{3-6}
             & & 100-110 bit &	128 bit &	192 bit &	256 bit \\
            \hline
            \hline
            Standards &	ISO/IEC 15946-5 &	BN256 &	BN384 & & \\		
            \cline{2-6}
            & TCG &	BN256 & & & \\			
            \cline{2-6}
	            &FIDO/W3C	& BN256		& & & \\			
	       \hline
	       \hline
            Libraries &	PBC &	BN & & & \\		
                        \cline{2-6}
	                    & mcl	& BN254  & 	BN381\_1 (*) & & \\		
	                    \cline{2-6}
                        & & BN\_SNARK1 & BN462  & & \\
                        
                        &  & &   BLS12-381    & & \\
                        \cline{2-6}
		            & RELIC &	BN254 &  BLS12-381 & & \\
		             & & BN256	& BLS12-455	& & \\	
		             \cline{2-6}
		            & TEPLA &	BN254 &&& \\
		            \cline{2-6}
                    &	AMCL &	BN254 & BLS12-381 (*) & & BLS48\\
                    && BN256 & BLS12-383 (*) &&\\
                    &&& BLS12-461	&&\\
                    \cline{2-6}
                    & Intel IPP &	BN256 &&& \\
                    \cline{2-6}
                    & Kyushu Univ. &&&		&		BLS48\\
                    \cline{2-6}
                    & MIRACL &	BN254 &	BLS12		&&\\
            \hline
            \hline
            Applications &	Zcash &	BN128 & BLS12-381 && \\
            && BNSNARK			&&&\\
            \cline{2-6}
	        &Ethereum &	BN254 &	BN382 (*) && \\
	        &&& BLS12-381 (*)		&& \\
	        \cline{2-6}
	        & Chia Network & &		BLS12-381 (*)	&& \\
	    \hline
            
        \end{tabular}
\end{table}

RELIC \footnote{Aranha, “RELIC is an Efficient LIbrary for Cryptography,” online, 2017. [Online]. Available: https://github.com/relic-toolkit/relic} is an efficient cryptographic library that has various kinds of PF-EC including six BN curves (BN-158, BN-254R, BN-256R, BN-382, BN-446, and BN-638), six curves (BLS12-381, BLS12-446, BLS12-445, BLS12-638, BLS24-477, and BLS48-575), cock-pinch curves of embedding degree 8 with field size 544-bit, and KSS curve of embedding degree 18 with field size 511-bit. Besides, the Multi-precision Integer and Rational Arithmetic Cryptographic Library (MIRACL) \footnote{M. Integer and C. Rational Arithmetic, “C++ Library (MIRACL).” 2013} is a C software library that is extensively suggested by developers as the gold standard open-source SDK for elliptic curve cryptography (ECC) that supports BN-462 and BLS48-581.

\subsection{Applications}
In their library libsnark \cite{gabizon2019security}, Zcash implemented a BN128 curve. It is a cryptocurrency that uses elliptic curve cryptography to achieve privacy for its users. They proposed a new parameter of BLS12 as BLS12-381 after utilizing exTNFS and presented their implementations. Ethereum is a public, blockchain-based distributed computing platform that is built specifically for creating smart contracts. Ethereum 2.0 adopts BLS12-381, BN-254 and BN-382 curves \cite{bowe2017scalable}. The implementation of the Chia network integrated the RELIC toolkit. Based on the security level of PF-EC curves given in \cite{barbulescu2019updating, fotiadis2018optimal, guillevic2020cocks,mbiang2020computing}, we summarise the adoption status of pairing-friendly curves in Table 2.10.

\section{Summary}

This chapter presents a comprehensive overview of various methods for constructing elliptic curves suitable for pairing, categorizing them into individual and families of elliptic curves. The taxonomy introduced by Freeman \textit{et al.} \cite{freeman2010taxonomy} is extended to include new families such as complete with variable discriminant \cite{drylo2011constructing} and new sparse family \cite{fotiadis2018generating}. A comprehensive framework for constructing individual and parameterized families of curves is also presented. The functional security of various families of pairing-friendly curves is evaluated in terms of required key size, aiming to identify better families of curves than BN, KSS, and BLS. Recent attacks, such as TNFS on pairing, highlight the need to enhance the key size for pairing. As such, families of curves are analyzed in terms of their key size, and an appropriate choice of an elliptic curve is selected. The adoption of pairing-friendly curves in international standards, cryptographic libraries, and applications at different security levels is also demonstrated. This chapter provides a comprehensive overview of the state-of-the-art in constructing PF-EC and highlights the importance of selecting suitable elliptic curves for efficient and secure cryptographic applications.

\end{doublespace} \label{chapter2}
\begin{savequote}[75mm] 
"I never dreamed about success, I worked for it.
\qauthor{Estée Lauder- American founder of the cosmetic company} 
\end{savequote}
\chapter{Secure Key Issuing for ID-Based Encryption and Signature Schemes}

\begin{doublespace}

The key-escrow problem is a major challenge when using identity-based cryptosystems (IBC). This arises when the private key generator (PKG) generates private keys for users. To trust the PKG, it must be a reliable entity. However, if the PKG is involved in malicious actions, it can intentionally distribute private keys to unauthorized identities. This results in two significant difficulties: the PKG possessing the sender's private key can forge signatures undetected, and the user can accuse the PKG of misusing their private key. Furthermore, IBC faces the problem of secure private key distribution, making it unsuitable for large-network applications due to its key escrow, key abuse, and user slandering issues. Our chapter proposes an efficient and secure key issuing scheme that resolves these issues. We also use this scheme to implement escrow-free identity-based encryption (IBE) that is secure against confidentiality and an escrow-free identity-based signature (IBS) scheme that is forgeable and secured. Lastly, we discuss some existing solutions to these problems, but our proposed solution stands out as the most effective and reliable.

There have been discussed several solutions that address the key-escrow problem. Boneh \textit{et al.} \cite{boneh2001identity} exercise multiple PKGs to reduce the trust on a single PKG, where the master key is distributed among several PKGs, and each PKG computes the private key share. Paterson \textit{et al.} \cite{paterson2002cryptography} and Chen \textit{et al.} \cite{chen2002applications} employ multiple trusted authorities (TAs) that authenticate the user and provide private key share using their secret keys to him. However, these schemes \cite{boneh2001identity,paterson2002cryptography,chen2002applications} interestingly explicated the key escrow problem, but it needs multiple user authentications for each TA/PKG. Thus, the cost of injecting additional infrastructure and communication leads to inefficiency. Lee \textit{et al.} \cite{lee2004secure} introduce the multiple independent key privacy authorities (KPAs) along with the single KGC that sequentially protects the user’s private key so that PKG could not regain it. For mitigating the trust in a single TA, a Certificate-less public key cryptosystem (CL-PKC) scheme \cite{al2003certificateless} enables the TA to authenticate a user and provide a partial private key to the user. The partial private key and user-selected secret value form the user’s private key. Since many signature schemes based on the CL-PKC approach have been discussed \cite{karati2018pairing,karati2018provably}  but they lose the identity-based property because the recipient needs the user public key for verifying the signature.

Recently, Chen \textit{et al.} proposed five different schemes that address the key escrow problem \cite{chen2015removing,chen2015t,chen2015escrowcomp,chen2015escrowcloud,chen2016escrow}. In \cite{chen2015removing}, Chen \textit{et al.} employ multiple Key Privacy Authorities (KPAs) to limit the power of PKG for a Hierarchical Identity-Based Encryption (HIBE) scheme. In \cite{chen2015t}, Chen \textit{et al.} extend \cite{chen2015removing} to achieve the goal of mitigating the trust in PKG and resolving the secure key issuing problem using the blinding technique. In \cite{chen2015escrowcomp}, Chen \textit{et al.} use user-chosen secret value to tackle the key escrow problem, key abusing problem, and user slandering problem. The scheme is IND-ID-CCA secured using the Dual System Encryption methodology. In \cite{chen2015escrowcloud}, Chen \textit{et al.} address the key escrow problem in the HIBS scheme for cloud storage protection, in which they employ a third party to identify the misbehaviour of malicious PKG. In \cite{chen2016escrow}, Chen \textit{et al.} modify the signing algorithm and employ an arbitral party to keep the user public parameter for signature verification purposes to fix key escrow, user slandering, and key abusing problems for the HIBS scheme. 

From the above discussions, we observe that the escrow-free model using various authorities requires significant computation overhead. Inspired by this idea, we utilize cloud technology and propose an efficient and secure key issuing protocol (ESKI) in which most of the computation is offloaded to the cloud server, i.e., cloud privacy centres (CPCs). Now, we introduce our proposed solution in order to address the key-escrow problem.

\section{System Framework}

In this section, we introduce the framework of an efficient and secure key issuing (ESKI) scheme that will be the basis for constructing the proposed escrow-free IBE and IBS schemes. Further, we discuss the adversarial model and security goals for IBE and IBS schemes.

\subsection{System Model}

Here, we first give the system model of the proposed efficient and secure key issuing (ESKI) scheme, followed by IBE and IBS schemes based on the proposed ESKI scheme. The proposed ESKI scheme includes three entities: user, key generation centre (KGC) and multiple Cloud privacy centres (CPCs). They have the following accountabilities.
\begin{itemize}
    \item \textit{Key generation centre}: The KGC is a semi-trusted authority whose primary accountability is to authenticate a user and provide a blinded partial private key to him without knowing the original partial private key. 
    \item \textit{Cloud privacy centres}: The CPCs are the crux of our proposed ESKI model. Under the assumption that at least one CPC is honest, we employ n number of CPCs in our proposed model. The primary duties of CPCs are the following: they protect the user’s private key using their secret keys and send them to the user without knowing them, and perform each operation on the cloud. The number of CPCs will depend on the nature of the application. For instance, in an election mechanism, the KGC behaves as the election commission that monitors and controls the election processes, and multiple CPCs act as the observers that are forwarded by the nominated candidates to bring out unlawful activities to the notice of the election commission. In such an environment, the number of CPCs depends on the number of the nominated candidate in the election process. 
    \item \textit{User}: The user requests the KGC for his private key using a randomly chosen blinding value. The user obtains a private key share from CPCs, and combines them and extracts them to form a private key using his randomly chosen blind value.
\end{itemize}

It is noted that the user authentication is done once by the KGC. We use the ECC-based blinding technique to secure the communication medium between two entities. In this way, we design an efficient and secure key issuing mechanism to address the key escrow problem. 

Second, we deploy our proposed ESKI scheme to Boneh-Franklin IBE \cite{boneh2001identity} to design an Efficient and Secure Key Issuing Identity-Based Encryption (ESKI-IBE) scheme. The ESKI-IBE scheme utilizes the ESKI approach to generate the user’s private key, in a sequential manner, which eliminates the key escrow problem, secure key issuing problem, and user-slandering problem. 

Third, we present an escrow-free ID-based signature (EF-IBS) scheme based on the proposed ESKI approach. In the proposed EF-IBS scheme, the user generates his private key using the ESKI scheme in a parallel manner.

\subsection{Adversarial Model}

In order to gain any useful information, an active adversary requires an original private key whose production depends on the partial private key and collusion of KGC and CPCs. We must assume that an adversary does not compromise with at least one-out-of-n CPC. For our model, we give more capabilities to the adversary. First, an adversary can attack on random identity $ID_{ch} \in \{0,1\}^*$, such that $ID_{ch} \ne ID$. Second, during a chosen ciphertext attack on $ID_{ch}$, an adversary can obtain the partial private key from the KGC on $ID_{ch}$. Third, an adversary can obtain the private key from the CPCs on $ID_{ch}$. In this way, we model an adversary that possesses a set of partial private keys and private keys against some identities and is capable of adaptive attacking other identity IDs of its choice. Based on the adversary’s capability, we essentially define two different types of adversaries: 

\textit{Type I Adversary}: An active adversary that can access KGC’s master key or behave as a malicious KGC is defined as a Type-I adversary, denoted as $Adv_I$. The $Adv_I$’s intention is to retrieve some useful information about plaintext from the ciphertext using random identity $ID_{ch}$, where $ID$ is an identity for which a private key has been computed. In our model, $Adv_I$ can extract the partial key-issuing key and make decryption queries but cannot extract the private key extraction query. Such an adversary wants to obtain the user’s private key colluded with other entities. 

\textit{Type II Adversary}: An active adversary that can access CPC’s secret key or behave as a malicious CPC is defined as a Type-II adversary, denoted as $Adv_{II}$. The $Adv_{II}$’s intention is to retrieve some useful information about plaintext from the ciphertext using random identity $ID_{ch}$, where ID is an identity for which a private key has been computed. In our model, $Adv_{II}$ can extract the private key and make decryption queries but cannot extract the partial-key issuing query. Such an adversary wants to obtain the user’s private key colluded with other CPCs and KGC. 
	
\subsection{Security Goals for ID-based Encryption}

For our security model, we strengthened the security model of \cite{boneh2001identity} that gives adversaries the capabilities defined in the previous section. We give a brief overview of the indistinguishability of ciphertext against chosen ciphertext security (IND-ID-CCA) in terms of game playing between the adversary $Adv_k$, where $k \in \{I, II\}$ and challenger, say $Ch$. The game defined is as:

    \textit{Setup}: On given security parameter $k$, challenger $Ch$ runs the this algorithm and outputs public parameters $pp$, master key and secret keys. $Ch$ sends $pp$ to the adversary $Adv_k$. 
    
    \textit{Phase1}: $Adv_k$ runs the following queries:
    \begin{itemize}
        \item \textit{Partial Key-issuing query}: Given $ID_i$ and output of $q_1$ queries on $ID_i$, $Ch$ runs this queries, outputs $D_i$  and gives it to $Adv_k$. 
    	\item \textit{Private Key Extraction query}: Given $ID_i$ and $D_{IDi}$, $Ch$ runs this query to compute the private key $d_{IDi}$ and responds it to the $Adv_k$.
	    \item \textit{Decryption Query}. Given identity and ciphertext pair $<ID_i,C_i>$, $Ch$ runs decryption queries on $C_i$ which gives the plaintext $M$ if $C_i$ is in the ciphertext space, null otherwise. $Ch$ then pass $M$ to the $Adv_k$.
    \end{itemize}
	\textit{Challenger phase}: $Adv_k$ gives the two equal size message $M_0,M_1$ and challenge identity $ID_{ch} \ne ID_i$. Then, $Ch$ computes $C=Enc(M_c,ID_{ch},pp)$, where $c \in \{0,1\}$  is a random bit. 
	
    \textit{Phase 2}: $Adv_k$ again runs, $q_{pp}$,$q_p$, and $q_d$ queries, as similar to Phase 1.  
    
    \textit{Guess}: At the end, $Adv_k$ pick value $c' \in \{0,1\}$.  If $c = c'$, $Adv_k$ wins the game with advantage 
    
     \begin{equation} \label{eq3.1}
    Advantage(Adv_K ):= |Pr[c=c']-1/2| \ge \epsilon 
\end{equation}
where $a_i \in K$ and $1 \le i \le 6$.

\begin{definition} (\textit{IND-ID-CCA}). We define that our ESKI-IBE is indistinguishability secure against adaptive identity chosen ciphertext attack (IND-ID-CCA) if no polynomially bounded adversary of $Adv_k$ has a non-negligible advantage $\epsilon$ in the above game.
\end{definition}

\subsection{Security Goals for ID-based Signature}
According to [11], we consider three security attacks for our proposed Escrow-Free Identity-Based Signature (EF-IBS) scheme: 1) existential forgery attack against adaptive chosen message and identity (EF-ID-CMA); 2) existential user slander against adaptive chosen message and identity (EUS-ID-CMA), and 3) existential key abusing against adaptive chosen message and identity (EKA-ID-CMA). We ensure that the proposed EF-IBS is secured against such attacks.

\textit{EF-ID-CMA}: A signature scheme is secured against the existential forgery attack if the forger $F$ can forge a signature for at least one message. Following game playing between the forger $F$ and challenger$Ch$ shows that the proposed EF-IBS scheme is existential forgery secured against adaptive chosen message and identity attack (EF-ID-CMA). 

\textit{Setup}: Similar to IND-ID-CCA model, $Ch$ gives $pp$ to $F$. 

\textit{Oracles}. $F$ can run the following oracles: 

\begin{itemize}
    \item \textit{Private key extraction oracles}: For given identity $ID$, $Ch$ runs KeyGen algorithm and responds the private key $d_{ID}$ associated with $ID$ to $F$. $C$ also provides the system public parameter $Y$ to $F$. However, the secret keys of CPCs are kept secret.
    \item \textit{Signing oracles}: For a given input $\{ID, m\}$, $Ch$ runs the KeyGen algorithm if it does not have the private key corresponding to $ID$. Then, $Ch$ signs a message $m$ with key $d_{ID}$ and system public parameter $Y$ by running the signing algorithm.
\end{itemize}
\textit{Challenge}. Finally, $F$ responses the challenge with given triple $<ID^*,m^*,\sigma^*>$ with following restrictions
\begin{itemize}
    \item $F$ has not been queried the private key on $ID^*$ before.
    \item $F$ has not been queried the signature on $<ID^*,m^*>$  before.
\end{itemize}

If $\sigma^*$ is successfully verified on $ID^*$ then $F$ wins the EF-ID-CMA game. 

\begin{definition} (\textit{EF-ID-CMA}). In the above EF-ID-CMA game, if no probabilistic polynomial time (PPT) forger $F$ has a non-negligible advantage against the challenger $Ch$, we say that our proposed EF-IBS scheme is secure. 
\end{definition}

\textit{EKA-ID-CMA}: Following game playing between the attacker $Adv$ and challenger$Ch$ shows that the proposed EF-IBS scheme is resilient to existential key abuse against adaptive chosen massage and identity attack (EKA-ID-CMA). 

\textit{Setup}: Similar to IND-ID-CCA model, $Ch$ gives $pp$ to $Adv$. 
\textit{Oracles}. $Adv$ can run the private key extraction and signing oracles.
\begin{itemize}
    \item \textit{Private key extraction oracles}: Identical to private key extraction oracle of EF-ID-CMA game.  
	\item \textit{Signing oracles}: Identical to private key extraction oracle of EF-ID-CMA game. 
\end{itemize}

\textit{Challenge}. Finally, $Adv$ responses the challenge with given triplet $<ID^*,m^*,\sigma^*>$ with following restrictions
\begin{itemize}
    \item $Adv$ has queried the private key on $ID^*$ to obtain $<d^*_ID,Y^*>$
    \item $Adv$ has queried the signature on $<ID^*,m^*>$  for at least one $m^* \ne m$.
\end{itemize}
	
If $\sigma^*$ is successfully verified by verification algorithm on input $<ID^*,m^*,\sigma^*,Y^*>$ then $Adv$ wins the EKA-ID-CMA game. 

\begin{definition}(\textit{EKA-ID-CMA}). In the above EKA-ID-CMA game, if no PPT adversary $Adv$ has a non-negligible advantage against $Ch$, we say that our proposed EF-IBS scheme is secure. 
\end{definition}

\textit{\textbf{EUS-ID-CMA}}: Following game playing between the adversary $Adv$ and challenger $Ch$ shows that our proposed scheme is resilient to existential user slander under adaptive chosen message and identity.  

\textit{Setup}: Similar to IND-ID-CCA model, $Ch$ gives $pp$ to $Adv$.

\textit{Oracles}. $Adv$ runs the following oracle. 
\begin{itemize}
    \item \textit{Private key extraction oracles}: For given identity $ID$, $Ch$ runs Keygen algorithm and responds the private key $d_{ID}$ associated with $ID$ to $Adv$. $Ch$ also provides the system public parameter $Y$ and secret keys of CPCs to $Adv$. 
    \item \textit{Signing oracles}: Identical to private key extraction oracle of EF-ID-CMA game.  
\end{itemize}

\textit{Challenge}. At the end of this game, $Adv$ responses the challenge of given triple $<ID^*,m^*,\sigma^*>$ having following restrictions
\begin{itemize}
    \item $Adv$ has queried the private key on $ID^*$ before to obtain $<d_{ID}^*,Y^*>$.
    \item $Adv$ has not been queried the signature on $<ID^*,m^*>$ previsously.
\end{itemize}

Adv wins the EUS-ID-CMA game if and only if $\sigma^*$ is successfully verified on $ID^*$ with input $<ID^*,m^*,\sigma^*,Y^*>$.

\begin{definition} (\textit{EUS-ID-CMA}). Our proposed EF-IBS scheme is EUS-ID-CMA secure if no PPT adversary $Adv$ has a non-negligible advantage over the challenger $Ch$ in the above game.
\end{definition}

\section{Proposed Efficient and Secure Key Issuing ID-based Encryption)}

Here, we implement the IBE system based on the idea of the proposed ESKI scheme. It includes four algorithms: Setup, Key Extract, Encryption, and Decryption.

\begin{enumerate}
    \item \textit{\textbf{Setup}}: Suppose a security parameter k. Select a prime number $q$ of $k$-bits, two groups, $\mathbb{G}_1$ and $\mathbb{G}_2$ each of order $q$, and $P$ be the generator of group $\mathbb{G}_1$. Let a map function be  $e: \mathbb{G}_1  \times \mathbb{G}_1 \rightarrow \mathbb{G}_2$, and four hash function as $H_1:\{0,1\}^* \rightarrow \mathbb{G}_1, H_2:\mathbb{G}_2 \rightarrow \{0,1\}^n$, $H_3:\{0,1\}^n  \times \{0,1\}^n  \rightarrow \mathbb{Z}_q^*$, and $H_4: \{0,1\}^n \rightarrow \{0,1\}^n$. KGC selects an element $s_0 \in \mathbb{Z}_q^*$ known as a master secret key, computes public key $P_0=s_0P$, and sends $P_0$ to $CPC^1$.
    
    On given parameters $pp$, $CPC^i$ selects its secret key $s_i \in \mathbb{Z}_q^*$, and compute its public key  $P_i=s_iP$ and sets systems  public key as $Y_i=s_i Y_{(i-1)}$. $CPC^i$ sends $Y_i$ to $CPC^{i+1}$.  $CPC^i$  can verify correctness of $Y_{i-1}$ using $e(Y_{i-1},P_i) \overset{?}{=}e(Y_i,P)$, where  $Y_0=P_0$. Now, $CPC^n$ sets the system public key as $Y= Y_n=s_n s_{n-1}s_0P$ and sends it to the KGC. It is noted the system's public key is derived sequentially using each secret key of CPC. Now, KGC publishes the parameter $pp=<q,e,P,P_0,\mathbb{G}_1,\mathbb{G}_2,H_1,H_2,H_3,H_4,Y,P_1,..P_n>$ and the parameter $s_0$ is kept secret.
    \item \textit{\textbf{KeyGeneration}}: On given identity $ID$, the user’ private key $S_{ID}$ is computed as follows.
    \begin{itemize}
        \item KeyIssuing. User picks an element $x \in \mathbb{Z}_q^*$, computes  $X=xP$,  $D_{ID}=xQ_{ID}$, where   $Q_{ID}=H_1 (ID)$, $ID$ is his identity, and  sends $<X,ID,D_{ID}>$ to the KGC. KGC validates the user identity along with other receiving parameters, using $e(Q_{ID},X) \overset{?}{=}e(D_{ID}, P)$, computes the blinded partial private key $D_0=s_0D_{ID}$ and $X_0=s_0X$, and sends $<D_0, X_0>$ to the user. The user validates the parameter $<D_0,X_0>$ using $e(D_{ID},X_0 ) \overset{?}{=}e(D_0,X)$.
        \item \textit{KeySecuring}. On receiving parameter $<D_0,X_0>$, user sends $<ID,D_{i-1},X_{i-1}>$ to $CPC^i$, ($i=1,2,..,n$)  and asks it for the key protection. $CPC^i$ checks the correctness of parameters using $e(Q_{ID},X_{i-1}) \overset{?}{=}e(D_{i-1},P)$ and compute $D_i=s_iD_{i-1}$ and  $X_i=s_i X_{i-1}$, and send $D_i$ to the user. The user verifies the correctness of the blinded partial private key using $e(D_i, P) \overset{?}{=} e(D_{i-1}, P_i )$. Finally, user gets $D_n=s_n D_{n-1}$ from the last $CPC^n$.
        \item \textit{KeyExtraction}. User retrieves private key  $d_{ID}=x^{-1} D_n= s_0  s_1.... Q_{ID}$ that can be verified using  $e(d_{ID},P) \overset{?}{=} e(Q_ID,Y)$.
    \end{itemize}
    \item \textit{\textbf{Encryption}}:  On given message $M \in \{0,1\}^n$,  $pp$ and recipient identity $ID$, sender chooses an element  $z \in\{0,1\}^n$, sets $r= H_3(z,M)$, $g= e(Q_{ID},Y)$, where  $Q_{ID}=H_1{ID}$, and outputs the ciphertext $C=<U,V,W>=<rP,z   H_2 (g^r ),   M   H_4 (z)>$ . 
    \item \textit{\textbf{Decryption}}: On given ciphertext $C=<U,V,W>$ and  private key $S_{ID}$, receiver sets $g'= e(U,d_{ID})$ and $z'=  V \oplus  H_2 (g')$. Now, recipient extracts message  $M=W \oplus H_4 (z')$ and sets $r'= H_3 (z',M)$. The recipient accepts the  message $M$ is valid if $U=r'P$ otherwise $M$ is invalid.
\end{enumerate}
This completes the full description of our ESKI-IBE scheme.
	
\section{Proposed Escrow-free ID-based Signature Scheme }	

Here, we construct the proposed EF-IBS scheme based on the idea of the proposed ESKI scheme. The proposed EF-IBS scheme involves of four randomized PPT algorithms: Setup, keyGeneration, Signing and Verification. They are defined as follows. 

\begin{enumerate}
    \item \textit{\textbf{Setup}}: The KGC considers the security parameter  $k$ and two groupas $\mathbb{G}_1$ (additive group) and $\mathbb{G}_2$ (multiplicative group); both have order $q$, where $q$ is the length of $k$-bit. Suppose $P$ be a generator of $\mathbb{G}_1$, a bilinear function is $e: \mathbb{G}_1 \times \mathbb{G}_1 \rightarrow \mathbb{G}_2$ that maps the $\mathbb{G}_1$ to $\mathbb{G}_2$ and, $H_1:\{0,1\}^* \rightarrow \mathbb{G}_1$, $H_2:\{0,1\}^* \rightarrow \mathbb{Z}_q^*$  and $H_3:\mathbb{G}_1 \rightarrow \mathbb{Z}_q^*$ are three cryptographic hash functions. 
    
    The KGC chooses an element $s_0 \in \mathbb{Z}_q^*$ (master key), evaluate  $P_0=s_0P$ (Public key),  $g=e(P,P)$ and simultaneously sends $P_0$ to each CPCs.  
    
    For given $P_0$, $CPC^i$ chooses an element $s_i \in \mathbb{Z}_q^*$ (their secret keys), evaluate $P_i=s_iP$ (their public keys) and $Y_i=s_iP_0$ (system public key shares), where $1 \le i \le n$.  Now, $CPC^i$ sends $<P_i,Y_i>$ to the KGC and keeps $s_i$ secret. 
    
    KGC verifies the received parameters $<P_i,Y_i>$ using equation $Y_i=s_0 P_i$, and computes the system public key as $Y= \sum_{i=1}^n Y_i =s_0  (s_1+s_2+...+s_n)P$. Now, KGC publishes the parameter  $pp=<q,e,g,P,P_0,\mathbb{G}_1,\mathbb{G}_2,H_1,H_2,H_3,P_1,..P_n,Y>$ and keeps safe its master key $s_0$. 
    \item \textit{\textbf{KeyGeneration}}: It involves three sub-algorithms: KeyIssuing, KeySecuring and KeyExtraction, define as follows:
    \begin{itemize}
        \item \textit{KeyIssuing}. The user requests to the KGC for the partial private key on his identity $ID$. The KGC authenticates the user identity $ID$ and computes the partial private key as $P_{ID}=s_0Q_{ID}$, where $Q_{ID}=H_1(ID)$, and sends $P_ID$ to the user. 
        \item \textit{KeySecuring}. The user simultaneously asks $CPC^i$ for protecting the requested partial private key $P_{ID}$. The $CPC^i$ computes the private key share $X_i=s_iP_{ID}$, using its secret key $s_i$ and sends them to the user. The user can verify the received parameters using  $e(X_i,P) \overset{?}{=} e(P_{ID},P_i)$.
        \item \textit{KeyExtraction}. On given private key share $X_i$, the user combines and computes his original private key as $d_{ID}= \sum_{i=1}^nX_i = s_0 (s_1+s_2+...+s_n)Q_{ID}$.
    \end{itemize}
    Note that $d_{ID}$ is the original private key against the identity $ID$, which can be verified using $e(d_{ID},P) \overset{?}{=} e(Q_{ID},Y)$.
    \item \textit{\textbf{Signing}}. The sender chooses an element $r \in \mathbb{Z}_p^*$, and sign a given message $m$ to compute the signature $\sigma=<R,S>$  using his private key $d_{ID}$, where $R=rP$, and $S=r^{-1}(H_2(m)P+H_3 (R)d_{ID})$.
    \item \textit{\textbf{Verification}}. On given message-signature pair $<m, S, R>$, sender’s $ID$, and param, the recipient verifies the signature-message pair using (\ref{eq3.2}).
\vspace{-5mm}    
    \begin{equation} \label{eq3.2}
    e(S,R) \overset{?}{=} g^{H_2(m)} e(Q_ID,Y)^{H_3 (R)} 
    \end{equation}
\end{enumerate}

This completes the full description of our EF-IBS scheme.

\section{Security Analysis}

This section gives the security of proposed ESKI-IBE and EF-IBS schemes. 

\subsection{Security Analysis of ESKI-IBE}

\begin{theorem} \label{Thm3.1}
The proposed ESKI-IBE is correct.
\end{theorem}
\begin{proof}
Proof. The correctness of our proposed ESKI-IBE scheme is defined as follows.
\vspace{-5mm}    
\begin{align*} 
W \oplus  H_4 (z') &= W \oplus  H_4 (V  \oplus H_2 (g')) \\
&= W  \oplus H_4 (V \oplus  H_2 ( e(U,d_{ID} ))) \\
&= M  \oplus H_4 (z) \oplus H_4 (V  \oplus H_2 ( e(U,d_{ID} ))) \\
&= M \oplus  H_4 (z) \oplus H_4 ( z \oplus  H_2 (e(s_0  s_1...  .s_n Q_{ID},rP))  \oplus  H_2 ( e(U,d_{ID} ))) \\
&= M \oplus  H_4 (z) \oplus  H_4 ( z  \oplus H_2 (e(d_{ID},U))  \oplus  H_2 ( e(U,d_{ID} ))) \\
&= M \oplus  H_4 (z) \oplus H_4 ( z ) = M
\end{align*}

Thus, the proposed ESKI-IBE scheme is corrected.
\end{proof}

\begin{theorem} \label{thm 3.2}
Assume $H_1$, $H_2$, $H_3$, and $H_4$ are four random oracle models, and the BDH problem is infeasible to solve. Our proposed ESKI-IBE scheme is semantically secure against IND-ID-CCA attack. Assume there exists an adaptively chosen message $M_i$ and chosen identity $ID_i$. Adversary $Adv(q_{pp}, q_p, q_d, q_1, q_2, q_3, q_4, t, \epsilon)$ runs at most Partial-Key-Extract queries $q_{pp}$, private-key-Extract queries $q_p$, decryption queries $q_d$, $H_1$ queries $q_1$, $H_2$ queries $q_2$, $H_3$ queries $q_3$, and $H_4$ queries $q_4$, respectively. Adversary $Adv_k$ can break our proposed ESKI-IBE scheme in polynomial time $t$ and success probability $\epsilon$. Then, there exists a simulator $B$ that helps an adversary $Adv_k$ against ESKI-IBE and can solve the BDH problem with probability $\epsilon'$.

\begin{equation} \label{eq3.4}
    \epsilon' \ge \frac{1}{(q_3+q_4)} \Bigg( \big( \frac{\epsilon(k)}{e(1+nq_{pp} q_p+q_d)q_2}+1 \big) \big(1-\frac{2}{q}\big)^{q_d}-1 \Bigg)    
\end{equation}

and                                             $t'=t+t_s+q_{pp}+q_p+q_d+q_1+q_2+q_3+q_4$\\
where $t_s$ is algorithm B’s simulation time. 
\end{theorem}

\begin{proof}
We proceed with the security reduction through the followings steps:

\textbf{Reduction 1}. Construct IND-CCA adversary $B$ against $BasicPubhy$ of [1] with the help of IND-ID-CCA adversary $Adv$ against ESKI-IBE. 

Given IND-ID-CCA adversary $Adv$ against ESKI-IBE, we first construct IND-CCA adversary B against BasicPubhy. B obtains public key $K_{Pub}=<P,P_{Pub},D_{ID},H_2,H_3,H_4>$ from $Adv$ during game against $BasicPub^{hy}$.

\textit{\textbf{Setup}}: $B$ helps $Adv_k$ by giving the public parameter $pp= <P,P_Pub,H_1,H_2,H_3,H_4>$. $H_1$ queries are simulated in the same as in Theorem 3.2. 

\textit{\textbf{Phase1}}: At any time during attacks , the $Adv_k$ runs following queries.
\begin{itemize}
    \item \textit{Partial key Extraction queries}. After executing $H_1$ query, $B$ gets the $D_i \in \mathbb{G}_1$ such that $H_1(ID_i)= D_i$. Assume tuples $<ID_i,D_i,b_i,c_i>$ against $ID_i$ in $H_1^{List}$ list, $B$ makes a list $PP^{List}$ of tuples $<ID_j,D_j,S_j,b_j,c_j>$ and do the following: 
    \begin{itemize}
        \item Aborts, if $c_i=1$. That means B’s attack on BasicPub failed.
        \item If $c_i=0$, we know $H_1(ID_i)= b_iP$. Compute $S_i= b_iP_{Pub}  \in \mathbb{G}_1$. Note that $S_i=s_0D_i$   is associated with a partial private key corresponding to $ID_i$. 
        \item $B$ then adds an entry $<ID_i,D_i,S_i,b_i,c_i>$ in $H_1^{List}$ list and gives $S_i$ to $Adv_k$.
    \end{itemize}
	
	\item \textit{Private Key Extraction queries}. After executing partial private key extraction queries, $B$ gets $S_i \in \mathbb{G}_1$ such that $S_i= b_iP_{Pub}$. Assume tuples $<ID_i,D_i,b_i,c_i>$ against $ID_i$ in $H_1^{List}$ list and tuples $<ID_i,D_i,S_i,b_i,c_i>$ against $S_i$ in $PP^P{List}$ list, $B$ do the following:
	\begin{itemize}
	    \item Aborts if $c_i=1$. That means $B$’s attack on BasicPub is failed. 
	    \item If $c_i=0$, we know $S_i= b_iP_Pub$. Computes $S_{IDi}= b_iY \in \mathbb{G}_1$. Note that $S_{IDi}= s_0s_1..s_nD_i \in  \mathbb{G}_1$ is associated with private key correspond to $ID_i$ and outputs to $Adv_k$. 
	\end{itemize}
 
	\item \textit{Decryption Query}. Suppose $Adv_k$ runs an decryption on given input $<ID_i,C_i>$, where $C_i= <U_i,V_i,W_i>$. Now simulator $B$ outputs as follows: 
	\begin{itemize}
	    \item To obtain $D_i \in \mathbb{G}_1$ such that $H_1 (ID_i )= D_i$, $B$ executes the above $H_1$ queries. Assume tuple $<ID_i,D_i,b_i,f_i>$ be one entry in the $H_1^{List}$.
	    \item To obtaining the $D_i \in \mathbb{G}_1$ such that $S_i= s_0 D_i$, execute the above partial private key extraction queries. Assume tuple $<ID_i,D_i,b_i,f_i>$ be one entry in the $H_1^{List}$.
	    \item If $f_i=0$, $B$ runs the private key extraction queries to produce the private key against $ID_i$.
	    \item If $f_i=1$, then $D_i= b_iQ_{ID}$. Recall that $U_i \in \mathbb{G}_1$. $B$ sets $C'_i = <b_i U_i,V_i,W_i>$. Let $S_i= s_0D_i$ be the FullESKI-IBE partial private key and private keys as $S_{IDi}= s_0 s_1..s_n D_i$ against the $ID_i$. The BasicPubhy decryption of $C'_i$ using $S_{ID}$ is equivalent to the decryption of $C_i$ using $S_{IDi}$ in FullESKI-IBE. Thus, we examine this observation as 
	    
	    $e(b_iU_i,S_{ID}) = e(b_i U_i,s_0 s_1..s_n Q_{ID})
        = e(U_i,b_is_0 s_1..s_n Q_{ID})
        = e(U_i,s_0s_1..s_n D_i )= e(U_i,S_ik )$
	\end{itemize}
	Then, gives the decryption query to $Ch$ and $Ch$ response back to $Adv_k$. 
\end{itemize}
	
\textit{\textbf{Challenger phase}}: Now, $Adv_k$ gives the two equal size message $M_0,M_1 \in \{0,1\}$ and challenge identity $ID_{ch} \ne ID$. Then, $Ch$ computes Ciphertext $C=Enc(M_c,ID_{ch},pp)$ on some chosen random bit $c \in  \{0,1\}$.  $B$ does as follows: 
\begin{itemize}
    \item Suppose any message $M_0,M_1 \in  \{0,1\}$ be challenged to the $Ch$ which it respond with a ciphertext $C= <U,V,W>$ where $C$ is the encrypted form of $M_c$ on bit $c in \{0, 1\}$.
    \item $B$ then runs the $H_1$ query to get $D \in \mathbb{G}_1$  such that $H_1(ID_{ch} )=D$ on tuple $<ID_{ch},D,b,f>$. $B$ aborts if $f = 0$ and attack become failed.
    \item For $f = 1$, $D=bQ_{ID}$ and we know $C= <U,V,W>$. For $b^{-1}$ is multiplicative inverse of $b$, $B$ set $C'= <b^{-1} U,V,W>$ and respond to $Adv_k$ with challenged $C'$. 
\end{itemize}
	
\textit{\textbf{Phase 2}}: Similar to phase 1, $Adv_k$ adaptively issue more query where no private key extraction is allowed on $ID_ch$. 

\textit{\textbf{Guess}}: At the end, $Adv_k$ gives the guess $c'\in \{0,1\}$. The game is in favour of the adversary if $c = c'$. 
If $B$ does not abort, then $Adv_k(A) = \left| Pr[c' = c] - \frac{1}{2} \right| \geq \epsilon$. The proof is given in [1].

The algorithm could abort for 3 reasons (Events):
\begin{itemize}
    \item $E_1$: During phase 1 or phase 2, an invalid partial private key and private key extraction queries run by $Adv$. 
    \item $E_2$: On invalid challenge identity $ID^{ch}$ challenged by $Adv_k$. 
	\item $E_3$: During phase 2, an invalid decryption query is run by adversary $Adv_k$.
\end{itemize}
 
For $nq_p$ private key queries in each $q_{pp}$ partial private key queries and $q_d$ decryption queries, we have a total $nq_{pp}q_p$ queries. We know the probability that B does not abort the game is:

\begin{center}
    $Pr[\neg E_1 \wedge \neg E_2 \wedge \neg E_3] \geq a^{nq_{pp} q_p + q_d}(1-\alpha)$
\end{center}

This value is maximized at $\alpha = \left(1 - \frac{1}{(1+q_{pp}q_p^n)}\right)$, which means $Pr[\text{not\_abort}] \geq \frac{1}{e(1+q_{pp}q_p^n+q_d)}$, where $e$ is the natural logarithm. Therefore, B's advantage against $BasicPub^{hy}$ is $\text{Advantage}(B) \geq \frac{1}{e(1+q_{pp}q_p^n+q_d)}$.

\textbf{Reduction 2}. We obtain PKE system BasicPubhy from BasicPub system using FO-transformation.  

From theorem 4.5 of BF-IBE \cite{boneh2001identity}, Let IND-CCA $Adv_k$ have advantage $\epsilon(k)$ while attacking BasicPubhy. Suppose $Adv_k$ has at most $q_d$ decryption queries and makes at most $q_3$ and $q_4$ queries to the hash functions $H_3$ and $H_4$ respectively. Then there is an IND-CPA simulator B against BasicPub with advantage $\epsilon_1$, where 

\begin{center}
    $ \epsilon_1 \ge \bigg( \frac{1}{(q_3+q_4 )}\bigg) \bigg( (\epsilon(k)+1) (1-\frac{2}{q})^{q_d}-1 \bigg)$
\end{center}

\textbf{Reduction 3}. As an output of FO-transformation, there exists an adversary $Adv_k$ against BasicPub, which can be used to construct a simulator B to solve the BDH problem by IND-ID-CPA of \cite{boneh2001identity}.

From the above 3 reductions, we bound an IND-ID-CPA adversary $Adv_k$ on BasicESKI-IBE with advantage $\epsilon(k)$ gives a BDH algorithm with advantage $\epsilon'$ as required, where,

\begin{center}
    $ \epsilon' \ge \bigg( \frac{1}{(q_3+q_4)} \bigg) \bigg( \frac{\epsilon(k)}{(e(1+q_{pp} q_p^n+q_d)q_2)+1)} \bigg) \bigg((1-\frac{2}{q})^{q_d}-1\bigg)$
\end{center}

\end{proof}

\begin{theorem} \label{thm3.3}
 Our proposed ESKI-IBE scheme is secured against an impersonation attack.
\end{theorem}

\begin{proof}
Proof. In the proposed keyIssuing algorithm, the user gives parameters $<ID, X, D_{ID}>$ such that $D_{ID}=xH_1(ID)$ for secret information x, and requests to the KGC for his partial private key. Suppose there is an adversary $Adv$ with identity $ID_A$, who wishes to impersonate the user by modifying the public parameters $<ID,X, D_{ID}>$ to $<ID_A, D_{ID}^*, X^*>$, where $D_{ID}^*=x^*H_1(ID)$, and $X^*= x^*P$ such that  $x^*\in \mathbb{Z}_q$, and passes it to KGC. The KGC first the authenticates the $ID_A$, then computes $Q_{IDA}=H_1(ID_A)$ to validate the received parameter $<ID_A, D_{ID}^*, X^*>$ using (\ref{eq3.5}).
\vspace{-5mm}
\begin{equation} \label{eq3.5}
    e(Q_{IDA},X^*) \overset{?}{=}  e(D_{ID}^*,P)
\end{equation}

The correctness of (\ref{eq3.5}) is verified as:
\vspace{-5mm}
\begin{align*}
    e(Q_{IDA},X^* )&= e(Q_{IDA},x^* P)
    =  e(x^* Q_{IDA},P)\\
    &=  e(D_{IDA},P) 
    \ne  e(D_{ID}^*,P)  
\end{align*}
The above inequalities demonstrate that the KGC detects if any malicious entity impersonates the user in the ESKI scheme. Therefore, our proposed scheme avoids the impersonation attack. 
\end{proof}

\begin{theorem} \label{thm3.4}
The proposed ESKI-IBE scheme is secured against insider attacks. 
\end{theorem}

\begin{proof}
Suppose a malicious CPC wishes to issue a signature on random parameters, instead of the received parameters without disclosing any hint to other CPC. Any malicious CPC and group of malicious CPC can control the user's private key. The proposed ESKI-IBE scheme is said to be secured against insider attacks if the user detects any curious CPC generates the signature in its choice parameter.
In our proposed scheme, the user receives the parameter $<X_{i-1},D_{i-1}>$ from the $CPC^{i-1}$, validates it and requests to $CPC^i$  for partial private key share by sending $<ID,X_{i-1},D_{i-1}>$. Suppose $CPC^i$ is an adversary who generates the signature on his choice without giving any hint to other CPC. $CPC^i$ chooses random $r \in \mathbb{Z}_q$, computes $D_i^*=rs_1 H_1 (ID)$, and $X_i^*=rs_1 P$, and sends  $<D_i^*, X_i^*>$ to user. The user verifies the correctness of parameters using (3.6).
\vspace{-5mm}
\begin{equation} \label{eq3.6}
    e(D_{i-1},P_i)?=  e(D_i^*,P)
\end{equation}

The correctness of  (\ref{eq3.6}) is verified as:
\vspace{-10mm}
\begin{align*}
    e(D_{i-1},P_i ) &= e(D_[i-1],s_iP)
    =  e(s_iD_{i-1},P)  \\
    &=   e( s_i s_{i-1}.. s_1 s_0 xQ(ID),P)
    = e(D_i,P) \ne e(D_i^*,P)
\end{align*}

The above inequalities show that the user detects if any curious CPC generates a false signature in its choice. This proves that our proposed ESKI systems are secure against insider attacks.
\end{proof}

\begin{theorem} \label{thm3.5}
The proposed ESKI scheme avoids the CPC’s incompetency. 
\end{theorem}

\begin{proof}
Incompetency refers to the inability of CPC to check the validity of the received parameter without knowing any information about the received parameters. The proposed ESKI-IBE scheme is said to be secured against the incompetency attack if the CPC can check the validity of receiver parameters and sign it without knowing any information about the received parameters. 

In the key-securing phase of our proposed system, the user sends their identity ID, signed blinding factor $X_{i-1}'$, and blinded partial private key $D_{i-1}$ to $CPC^i$ and sequentially requests them for key protection. The $CPC^i$ validates the parameters $\langle X_{i-1}, D_{i-1} \rangle$ using (\ref{eq3.7}).

\begin{equation} \label{eq3.7}
     e(Q_{ID},X_{i-1}) \overset{?}{=}e(D_{i-1},P)
\end{equation}

Suppose an adversary $Adv$ modified $<X_{i-1},D_{i-1}>$  as $<D_{i-1}^*=rD_{i-1}, X_{i-1}^*=rX_{i-1}>$ where $r$ is random variable. $CPC^i$ checks these parameters using (3.8).

\begin{equation} \label{eq3.8}
    e(Q_{ID},X_{i-1}^*) \overset{?}{=}e(D_{i-1}^*,P)
\end{equation}

On successful validation, $CPC^i$ set $<D_i^*=s_i D_{i-1}^*, X_i^*=s_i X_{i-1}^*>$ and outputs it to the user who verifies the correctness of parameters $<D_i^*, X_i^*>$ using Equation (\ref{eq3.9a}).
\vspace{-5mm}

\begin{equation} \label{eq3.9b}
    e(D_{i-1}, P_i) \overset{?}{=}  e(D_i^*,P)
\end{equation}
 
The correctness of (\ref{eq3.9b}) is verified as:

\vspace{-10mm}

\begin{align*}
    e(D_{i-1},P_i ) &= e(D_{i-1}, s_i P) =  e( s_i D_{i-1},P) \\
    &=   e( s_i s_{i-1}..s_1s_0 xQ(ID),P)  = e(D_i,P)     \ne e(D_i^*,P)
\end{align*}

These inequalities show that the user can easily detect adversarial misbehaviour during communication between the user and $CPC^i$. Therefore, our proposed ESKI Systems are secure against CPC’s incompetency attack. 
\end{proof}

\subsection{Security Analysis of EF-IBS scheme}

\begin{theorem} \label{thm3.6}
The proposed EF-IBS is consistent.
\end{theorem}

\begin{proof}
The Equation (\ref{eq3.2}) ensures the validity of the signature-message pair. Since $R=rP$ and $S=r^{-1} (H_2(m),H_3 (R)d_{ID} )$ are given, the consistency of Equation (\ref{eq3.2}) is verified as follows. From LHS, 

\vspace{-15mm}

\begin{align*}
    e(S,R)&=e(r^{-1}(H_2 (m)P+H_3(R)d_{ID}),rP)\\
&=e((H_2 (m)P+H_3(R)d_{ID}),P) \\
&=e(H_2 (m)P,P)e(H_3(R)d_{ID},P) \\
&=e(H_2 (m)P,P)e(H_3(R)s_0  (s_1+s_2+...+s_n) Q_{ID},P) \\
&=e(H_2 (m)P,P)e(H_3(R)Q_{ID},Y) \\
&=e(P,P)^{H_2(m)}e(Q_{ID},Y)^{H_3 (R))}\\
&=g^{H_2(m)}e(Q_{ID},Y)^{H_3 (R)}
\end{align*}

This completes the correctness of Equation (\ref{eq3.2}).
\end{proof}

\begin{theorem} \label{thm3.7}
(\textbf{EF-ID-CMA}). Suppose $H_1, H_2$ and $H_3$ are three random oracle models and forger $F$ wants to forge a signature on message $m$. Suppose forger $F$ for an adaptive chosen message and identity attacks (EF-ID-CMA) to our EL-IBS scheme with advantage $k^{-n}$ and running time $t$, there exists an algorithm B that helps $F$. Thus, $F(t,q_1,q_2,q_3,q_p,q_S,k^{-n})$ have the following advantage to breaks the proposed scheme.
\vspace{-5mm}    
\begin{equation} \label{eq3.9a}
    |Pr[F(t,q_1,q_2,q_3,q_p,q_S,k^{-n})]| \ge k^{-n} (1-q_1/2^k  )^{q_p} (1-(q_2+q_3)/2^k)^{q_s}
\end{equation}
and,  $Time t'=t+t_B+q_1+q_2+q_3+q_P+q_S$

Where, $q_1,q_2,q_3,q_p$ and $q_S$ are the most queries to $H_1, H_2, H_3$ hash oracles, private key extraction and signature oracle respectively asked by $F$ and $t_B$ is running time for algorithm $B$.
\end{theorem}
 
\begin{proof}
Suppose a forger $F$ wants to forge any signature and let there must be a simulator $B$, which helps $F$. We have to design an algorithm $B$ to solve the CDHP. 

\textbf{Setup}: Consider three hash functions $H_1:\{0,1\}^* \times \mathbb{G}_1 \rightarrow \mathbb{Z}_q^*, H_2:\{0,1\}^* \rightarrow \mathbb{Z}_q^*$  and $H_3:\mathbb{G}_1 \rightarrow \mathbb{Z}_q^*$, generator $P$  of $\mathbb{G}_1$  and bilinear function  $e: \mathbb{G}_1  \times \mathbb{G}_1 \rightarrow \mathbb{G}_2$, $B$ chooses an element $a \in \mathbb{Z}_q^*$, and set system public key as $Y=aP$ and $g=e(P,P)$. $B$ now set parameter $pp=<q,e,g,P,P_0,\mathbb{G}_1,\mathbb{G}_2,H_1,H_2,H_3,Y>$ and gives $pp$ to $F$.

\textbf{Oracles}: Forger $F$ performs the following oracles.
\begin{itemize}
    \item \textit{$H_1$ oracle}: $B$ makes a list $H_1^{List}$, which is initially empty having tuple $<ID_i,h_1i,b_i,x_i>$. When $F$ queries on identity $ID_i$, $B$ responds $F$ in the following way.
    \begin{itemize}
        \item $B$ gives $h_1i$ to the $F$, if $ID_i$ found in the $H_1^{List}$ of  tuple $<ID_i,h_1i,b_i,x_i>$.
        \item Otherwise, $B$ chooses a bit $x \in \{0,1\}$, where $Pr[x=0] = \alpha$. Pick an element $b \in \mathbb{Z}_q^*$, set $H_1(ID)=h_1=bQ_{ID}$ for $x=1$, else, set $H_1(ID)=h_1=bP$. Adds tuple $<ID_i,h_{1i},b,x>$ to list $H_1^{List}$, and gives $h_{1i}$ to $F$.
    \end{itemize}
	
Note, $H_1 (ID)$ gives no information to $F$ until they query the $H_1$ oracle on $ID$ because $H_1$ is the random oracle. 

    \item \textit{$H_2$ oracle}. For given message $m$, $B$ chooses a random element $h_2 \in \mathbb{Z}_q^*$, set $H_2(m)=h_2$ and returns it to $F$.
    \item \textit{$H_3$ oracle}. For given message $R$, $B$ chooses a random element $h_3 \in \mathbb{Z}_q^*$, set $H_3(R)=h_3$ and returns it to the $F$.
    \item \textit{Private key extraction oracle}. $F$ submit an identity $ID$. $B$ maintains a list $H_1^{List}$ having tuple $<ID_i,d_{IDi},b_i>$.  $B$ responds private key $d_{ID}$ to {F}, if queried key contained in the list. Otherwise, $B$ chooses a random element $b \in \mathbb{Z}_q^*$, set $d_{ID}=bY$, adds $d_{ID}$ in the list and gives $d_{ID}$ to the $F$.  
    \item \textit{Signing oracle}. For a given input $<ID,m>$, $B$ queries the private key from the EF-IBS scheme  if it does not have a private key $d_{ID}$ corresponding to $ID$. $B$ runs the signing algorithm of EF-IBS to get the signature  $\sigma=<S,R>$. $B$ invokes a hash query such that $h_1=H_1 (ID)$, $h_2=H_2(m)$ and $h_3=H_3(R)$. It then computes $\sigma'=<S',R'>$, where $S'=rS$ and returns it to $F$.
\end{itemize}
 
\textbf{Challenge}: At the end, $F$ responds a tuple $<ID^*,m^*,S^*,R^*>$ against identity $ID^*$. Since $F$ is allowed to win the EF-ID-CMA game with a non-negligible advantage against the EF-IBS scheme as assumed, there is $e(S^*,R^* ) \overset{?}{=}g^{h_2} e(h_1,Y)^{h_3}$, where $h_1=H_1(ID^*)$, $h_2=H_2(m^*)$ and $h_3=H_3(R^*)$.

To break the EF-IBS scheme, B submits the challenge tuple $\sigma'=<S',R'>$, where $S'=rS^*$ and $R'=r^{-1}R^*$, $h_2=H_2 (m)$ and $h'_3=H_3(R' )$. We can get 
\vspace{-10mm}

\begin{align*}
    e(S',R')=e(rS^*,r^{-1}R^*)=e((h_2 P+h'_3abP),R^*)
\end{align*}

\textbf{Analysis}. The probability that B does not abort the game is defined by four events.
\begin{itemize}
    \item $E_1$: The extract oracles fail if $H_1$ oracle outputs the inconsistent outputs with probability at most $(q_1/2)^k$. The simulation is complete $q_p$ time, which happens with probability at least $(1-q_1/2^k)^{q_p}$. 
	\item $E_2$: The signature oracles fail if $H_2$ and $H_3$ oracles output the inconsistent outputs with probability at most $(q_2+q_3)/2^k$. The simulation is complete $q_s$ time, which happens with probability at least $(1-(q_2+q_3)/2^k  )^{q_s}$. 
	\item $E_3$: When $F$ runs the sign oracle, $B$ does not stop. 
	\item $E_4$: The forger $B$ computes the valid message-signature forgery.
\end{itemize}
	
From the above four events, we obtain the probability that $F$ can break the scheme is 
\vspace{-10mm}

\begin{align*}
    Pr[\neg E_1 \land \neg E_2 \land E_4 | E_3] \ge k^{-n} \left(1 - \frac{q_1}{2^k}\right)^{q_p} \left(1 - \frac{q_2 + q_3}{2^k}\right)^{q_s}
\end{align*}

\end{proof} 

\begin{theorem} \label{thm3.8}
(\textbf{EKA-ID-CMA}). Suppose $H_1, H_2$ and $H_3$ are three random oracle models and adversary A wishes to defame the KGC that his private key is used by the KGC and let an algorithm $B$, which helps $A$. Suppose $A$ executes private key extract oracles, signature oracles, $H_1$ hash oracles, $H_2$ hash oracles, and runs at most $t$ times with advantage at most $k^{-n}$. Under the assumption of ROM and intractable to solve the CDH problem, our proposed EF-IBS Scheme is resilient to existential key abuse against the adaptive chosen message and identity attacks (EUS-ID-CMA).

\end{theorem}

\begin{proof}
Suppose a forger $F$ wants to abuse a user by forging a signature and let there must be a simulator B, which helps $F$. We have to design an algorithm $B$ to solve the CDH problem.  

\textbf{Setup}: $B$ considers $P$ as a generator of group $\mathbb{G}_1$, three hash functions $H_1:\{0,1\}^* \times \mathbb{G}_1 \rightarrow \mathbb{Z}_q^*$, $H_2:\{0,1\}^* \rightarrow \mathbb{Z}_q^*$  and $H_3:\mathbb{G}_1 \rightarrow \mathbb{Z}_q^*$, a bilinear map function $e: \mathbb{G}_1  \times \mathbb{G}_1 \rightarrow \mathbb{G}_2$. $B$ chooses $a \in \mathbb{Z}_q^*$, and compute $Y=aP$, and public parameter  $pp=<q,e,g,P,P_0,\mathbb{G}_1,\mathbb{G}_2,H_1,H_2,H_3,Y>$. $B$ outputs $pp$ to $F$.

\textbf{Oracles}: Forger $F$ asks the queries, where $B$ responds the following oracles.
\begin{itemize}
    \item \textit{$H_1$ oracle}: $B$ prepares a list $H_1^{List}$ having tuple $<ID_i,h_1i,b_i,x_i>$, which is initially empty. $F$ queries on Identity $ID$, $B$ responds $F$ in the following way.
    \begin{itemize}
        \item $B$ gives $h_{1i}$  to $F$, if $ID$ found in the $H_1^{List}$ of  tuple $<ID_i,h_1i,b_i,x_i>$.
        \item Otherwise, $B$ chooses a bit $x \in \{0,1\}$, where $Pr[x=0]=\alpha$. Pick a random element $b \in \mathbb{Z}_q^*$, set $H_1(ID)=h_1=bQ_{ID}$ for $x=1$, else, set $H_1(ID)=h_1=bP$. Add the tuple $<ID_i,h_{1i},b,x>$ to the list $H_1^{List}$, and give $h_{1i}$ to $F$.

    \end{itemize}
	Note, $H_1(ID)$ gives no information to $F$ until they query the $H_1$ oracle on $ID$ because $H_1$ is the random oracle. 
	\item \textit{$H_2$ oracle}. For given message $m$, $B$ chooses a random element $h_2 \in \mathbb{Z}_q^*$, set $H_2 (m)=h_2$ and returns to the $F$. 
	\item \textit{$H_3$ oracle}. For a given message $R$, $B$ chooses a random element $r\in \mathbb{Z}_q^*$, set $h_3=H_3(R)=rP$ and returns to $F$.
	\item \textit{Private key extraction oracles}. $F$ submit an identity $ID$. $B$ maintains a list $H_1^{List}$ having tuple $<ID_i,d_{IDi},b_i>$.  $B$ responds private key $d_{ID}$ to $F$, if queried key contained in the list. Otherwise, $B$ chooses a random element $b \in \mathbb{Z}_q^*$, set $d_{ID}=bY$, adds $d_{ID}$ in the list and gives $d_{ID}$ to the $F$. 
	\item \textit{Signing oracle}. For signing oracle on $<ID,m>$, if B does not have private key $d_{ID}$ corresponding to $ID$, it computes the private key by running the private key extraction oracle. Using private key $d_{ID}$ and $H_2$ and $H_3$ responds, $B$ outputs the signing query as follows: 
	\begin{itemize}
	    \item If $ID$ is queried previously, $B$ abort the simulation. 
	    \item Otherwise, $B$ execute the private key oracle using $ID_i$ corresponding to $H_1^{List}$ and corresponding private key $d_{ID}$, queries the $H_2$ oracle to get $h_2$, queries the $H_3$ oracle to get $h_3$, computes the signature  $R=rP,S=r^{-1}(h_2i P+h_3i d_{ID})$. $B$ sends $<S,R>$ to $A$.
	\end{itemize}
\end{itemize}

\textbf{Challenge}: At the end, $A$ responds to the challenge $<ID^*,m^*,S^*,R^*>$ against identity $ID^*$ such that the private key against $ID^*$ that has been queried in private key extraction oracle previously and there must be a signing oracle that outputs $<ID^*,m,S,R>$. It seems that $B$ does not abort the game with probability $(1-1/q_1 )(1-1/(q_1 (q_2+q_3)))$. Also, the probability that $m^*=m$ is $1/(q_1 (q_2+q_3))$. So, there is 

\vspace{-15mm}

\begin{align*}
    e(S^*,R^*)&=g^(h_2)e(bP,aP)^(h_3)\\
&=e(P,P)^(h_2 ) e(bP,aP)^(h_3 ) \\
&=e(P,P)^(h_2 ) e(bP,aP)^(h_3 )\\
&=e(h_2 P,P)e(h_3 abP,P)\\
&=e(h_2 P+h_3 abP,P)
\end{align*}

Thus, $B$ obtains $r^* S^*=h_2 P+h_3 abP$, i.e., $abP=h_3^{-1}(r^* S^*-h_2 P)$. Therefore, $B$ can break the CDH problem.

\end{proof}
 
\begin{theorem} \label{thm3.9}
(\textbf{EUS-ID-CMA}). Suppose $H_1$, $H_2$ and $H_3$ are three random oracle models and adversary $A$ who wish to defame the KGC that his private key is used by the KGC and let an algorithm $B$, which helps $A$. Suppose A executes at most $q_p$ private key extract oracles, $q_s$ signature oracles, $q_1$ $H_1$ hash oracles, $q_2$ $H_2$ hash oracles, runs at most $t$ times with advantage at most $k^{-n}$. Under the assumption of ROM and intractable to solve the CDH problem, our proposed EF-IBS Scheme is existential user slandering secured against the adaptive chosen message and identity attacks (EUS-ID-CMA).
\end{theorem}

\begin{proof}
Suppose a malicious user $A$ who wishes to defame the KGC that his private key is used by him and lets an algorithm $B$, which helps $A$. We have to design an algorithm $B$ to solve the CDHP. 
\textbf{Setup}: Consider three hash functions $H_1:\{0,1\}^* \times \mathbb{G}_1 \rightarrow \mathbb{Z}_q^*$, $H_2:{0,1}^* \rightarrow \mathbb{Z}_q^*$, and $H_3:\mathbb{G}_1 \rightarrow \mathbb{Z}_q^*$, generator $P$ of $\mathbb{G}_1$, and bilinear function $e: \mathbb{G}_1 \times \mathbb{G}_1 \rightarrow \mathbb{G}_2$. $B$ chooses an element $a \in \mathbb{Z}_q^*$ and sets the system public key as $Y=aP$ and $g=e(P,P)$. $B$ now sets the parameter $pp=\langle q,e,g,P,P_0,\mathbb{G}_1,\mathbb{G}_2,H_1,H_2,H_3,Y \rangle$ and gives $pp$ to $F$.

\textbf{Oracles}: Forger $F$ performs the following oracles.
\begin{itemize}
    \item \textit{$H_1$ oracle}: $B$ makes a list $H_1^{List}$, which is initially empty having tuple $<ID_i,h_1i,b_i,x_i>$. When $F$ queries on identity $ID_i$, $B$ responds $F$ in the following way.
    \begin{itemize}
        \item 	$B$ gives $h_{1i}$ to the $F$, if $ID_i$ found in the $H_1^{List}$ of  tuple $<ID_i,h_1i,b_i,x_i>$.
        \item Otherwise, $B$ chooses a bit $x \in \{0,1\}$, where $Pr[x=0] =\alpha$. Pick an element $b \in \mathbb{Z}_q^*$, set $H_1(ID)=h_1=bQ_{ID}$ for $x=1$, else, set $H_1(ID)=h_1=bP$. Add the tuple $<ID_i,h_{1i},b,x>$ to the list $H_1^{List}$, and give $h_{1i}$ to $F$.
    \end{itemize}

Note, $H_1 (ID)$ gives no information to $F$ until they query the $H_1$ oracle on $ID$ because $H_1$ is the random oracle. 
    \item \textit{$H_2$ oracle}. For given message $m_i$, $B$ chooses a random element $h_2 \in \mathbb{Z}_q^*$, set $H_2(m)=h_2$ and returns it to the $F$.
    \item \textit{$H_3$ oracle}. For given message $R$, $B$ chooses a random element $h_3 \in \mathbb{Z}_q^*$, set $H_3(R)=h_3$ and returns it to the $F$.
    \item \textit{Private key extraction oracle}. After successful execution of $H_1$ oracle, $B$ obtains $h_{1i}$ against $ID_i$. Using tuple $<ID_i,h_1i,b_i,x_i>$ in $H_1^{List}$, $B$ performs the following steps. 
    \begin{itemize}
        \item If $x_i =1$, abort the game. 
        \item Otherwise, choose a random element $b \in \mathbb{Z}_q^*$, set $d_{ID}=bY$ and gives $d_{ID}$ to the $F$.
    \end{itemize}
	 
	\item \textit{Signing oracle}. For a given input, $<ID,m>$, $B$ executes the signature oracle to generate the signature $<S,R>$ using the output of the $h_2$ and $h_3$ oracles as follows.
	\begin{itemize}
	    \item If $ID$ is queried previously, $B$ abort the simulation. 
	    \item Otherwise, $B$ execute the private key oracle using $ID$ corresponding to $H_1^{List}$ and corresponding private key $bY$ and picks $r \in \mathbb{Z}_q^*$ and set the signature $<S,R>$, where $R=rP$ and $S=r^{-1}(h_2P+h_3abP)$
	\end{itemize}
\end{itemize}
	
\textbf{Challenge}: At the end, $A$ responds to the challenge $<ID^*,m^*,S^*,R^*>$ against identity $ID^*$ where, the private key against $ID^*$ that has been queried in private key extraction oracle previously. So, $B$ generates the private key $d_i$. It has been observed that A can break our EF-IBS against the EUS-ID-CMA if $d_i=d_{ID}$. 

From the forking lemma in [13], B can replay the same tuple but with a distinct $H_2$ choice. It then produces two valid signature $<S^*,R^*>$ and $<S^{'*},R^{'*}>$ on message $m^*$ corresponding to hash function $H_2$ and $H'_2$ having distinct values $h_2\ne h'_2$ respectively. Since, $R^*=r^*P$, $S^*=r^{*-1}(H_2(m^* )P+H_3 (R^*)abP)$ and $R^{'*}=r^{'*} P, S^{'*}=(r^{'*-1})(H'_2 (m^* )P+H_3 (R^{'*})abP)$.

\vspace{-10mm}

\begin{align*}
    \frac{e(S^*,R^*)}{e(S'^{*},R'^{*})}
    &= \frac{g^{h_2}e(bP,aP)^{h_3}}{g^{h'_2}e(bP,aP)^{h'_3}}
    =g^{h_2-h'_2}e(bP,aP)^{h_3-h'_3}\\
    &=g^{h_2-h'_2}e(bP,aP)^{h_3-h'_3}
    =e(P,P)^{h_2-h'_2}e(bP,aP)^{h_3-h'_3}\\
    &=e(((h_2-h'_2)+(h_3-h'_3)ab)P,P)\\
    &= e(r^* S^* (r^{'*}S^{'*})^{-1},P)
\end{align*}

\vspace{-10mm}

\begin{align*}
    r^*S^*(r'^* S'^*)^{-1} &= ((h_2-h'_2) + (h_3-h'_3)ab)P \\
    (h_3-h'_3)abP &= r^*S^*(r^{'*}S^{'*})^{-1} - (h_2-h'_2)P \\
    abP &= (h_3-h'_3)^{-1}(r^* S^*(r^{'*} S^{'*})^{-1} - (h_2-h'_2))P
\end{align*}

Therefore B can solve the CDH problem.  
\end{proof}

\section{Performance Analysis}

This section analyzes the performance of the proposed ESKI-IBE and EF-IBS schemes, in terms of computation overhead (in msec), communication overhead (in Bytes) and storage cost (in Bytes). 
	
\subsection{Experiment Simulation}
In order to validate the proposed schemes, we simulate the experiments on the \textit{Intel(R) Core(TM) i7-2600K CPU @ 3.4 GHz}, and \textit{8 GB of RAM}, using \textit{gcc 4.6}. We consider the Tate pairing constructs over the Type-A curve of the PBC library \cite{lynn2010pairing} in order to achieve the 1024-bit RSA level of security. The Type-A super-singular elliptic curve $E/\mathbb{F}_P: y^2=x^3+x$ built on two prime $p$ and $q$, such that $|p|=512-bit$, $q=2^{159}+2^{17}+1$ is Solaris prime ($|q|=160-bit$) satisfying $p+1=12qr$, where $\mathbb{G}_1=\mathbb{G}_2$, with embedding degree is 2. We take the following cryptographic operations for analysis of the computational cost of the proposed EF-IBS system: bilinear pairing operations, scalar multiplication in ECC, addition in ECC, modular multiplication, modular inversion, exponentiation and map-to-point hash function. Table \ref{tbl3.1} shows the computation cost of required cryptographic operations used in the entire thesis. We follow the methodologies given in \cite{cao2010pairing,debiao2011id,barreto2003selection,chung2007id} in order to compute the computation cost of required cryptographic operations.

\begin{table}[h!]
\centering
\begin{tabular}{|c| c| c| c|} 
 \hline
 Operations &	Notations &	Cost (in $T_M$) & Cost (in ms) \\ [0.5ex] 
 \hline\hline
 Modular multiplication &	$T_M$ &	01.00 $T_M$ 	& $\approx 0.23$\\
Modular inversion &	$T_I$ &	11.60 $T_M$ &	$\approx (11.60*0.23)$ $\approx 2.67$\\
Points addition &	$T_A$ &	00.12 $T_M$ &	$\approx (0.12*0.23)$ $\approx  0.03$\\
Scalar point multiplication &	$T_{SM}$	& 29.00 $T_M$	& $\approx (29*0.23)$ $\approx  6.67$\\
Modular Exponentiation &	$T_E$ &	240.00 $T_M$	& $\approx (240*0.23)$ $\approx  55.2$\\
Map-To-Point hash operation &	$T_H$	 & 29.00 $T_M$ &	$\approx  6.67$\\
Bilinear pairing &	$T_P$ &	87.00 $T_M$	& $\approx  20.01$\\ 
 \hline
\end{tabular}
\caption{Computational cost of required cryptographic operations (in ms)}
\label{tbl3.1}
\end{table}

\subsection{Performance Evaluation of ESKI}
\textit{Computation cost}. Now, we compare the computation cost of setup and key generation phases of our proposed ESKI model with the corresponding phases of other Escrow-free Schemes in the literature \cite{karati2018pairing,karati2018provably,chen2015removing,chen2015t,chen2015escrowcomp,chen2016escrow,li2017efficient,zhang2012efficient,zhang2013id}. Using notations and computation cost of operations given in Table \ref{tbl3.1}, we summarize the computation cost of KGC-Setup, CPC-Setup, Key issuing, Key securing and Key extraction phase of our ESKI model for $n=0,1,10$ and $100$ number of CPCs, given in Table \ref{tbl3.2}. In Figure \ref{fig3.1}, it can be seen that the computation cost of CPC-Setup and key issuing phase increases linearly, that is $2n(T_P+T_S )=(51.36*n) ms$ as the value of $n$ grows. On the other side, the computation cost of the Setup, Key issuing, and key retrieving phases are $13.34 ms, 106.72 ms$, and $46.7 ms$, respectively, and remain constant for any number of CPCs. It is noted that our scheme acts as a general IBE scheme which takes 166.8 ms in setup and key generation for $n=0$.

\subsection{Performance Evaluation of ESKI-IBE scheme}
This section estimates the performance of our proposed ESKI-IBE scheme, in terms of computation cost and communication cost. For comparison of our scheme with other related schemes, we assume that HIBE/HIBS scheme in the related schemes is 1 level, which is equivalent to a pure IBE/IBS scheme. As our proposed ESKI scheme employs multiple authorities, we compare our scheme with Chen \textit{et al.} \cite{chen2015removing}, Chen \textit{et al.} \cite{chen2015t} and Li \textit{et al.} scheme \cite{li2017efficient}. 

\begin{table}
        \centering
        \caption{Computation time (in ms) of KGC-Setup, CPC-Setup, Key-Issuing, Key-Securing and Key-Extraction phase of our ESKI scheme for n = \{0, 1, 10 and 100\}, the number of CPC participated in the system.}
        \label{tbl3.2}
        \begin{tabular}{|c|c|c|c|c|c|}
            \hline
            \multicolumn{1}{|c|}{Phases} & \multicolumn{1}{|c|}{\#operations} &
            \multicolumn{4}{|c|}{Computation cost (in ms)} \\
            \cline{3-6}
              &    & n=0 & n=1 & n=10 & n=100 \\
            \hline
            \hline
            KGC-Setup &	$1T_{SM}+1T_H$	& 13.34	& 13.34	& 13.34	& 13.34 \\
            CPC-Setup &	$2n(T_P+T_{SM})$	&  0	&  53.36	& 533.6	& 5,336\\
            Key Issuing	& $4(T_P+T_{SM})$	&  106.72 &	106.72	& 106.72	& 106.72\\
            Key Securing	& $2n(T_P+T_{SM})$ & 0	&53.36	& 533.6	& 5,336\\
            Key Extracting	& $2T_P+1T_{SM}$	& 46.7	& 46.7	& 46.7	& 46.7\\
            Total &	$(4n+6)(T_P+T_{SM})+1T_H$	& 166.8	& 273.5	& 1,234	& 10,838\\
            \hline
        \end{tabular}
    \end{table}

\textit{Computation cost}. Table \ref{tbl3.3} summarizes the computation cost of Chen \textit{et al.} \cite{chen2015removing}, Chen \textit{et al.} \cite{chen2015t} and Li \textit{et al.} scheme \cite{li2017efficient}, our scheme-I and our scheme-II. It can be observed from Figure \ref{fig3.2} that our scheme (when we do not consider the CPC cost) takes constant time (46.69 ms) while the overhead of other scheme increase linearly. Similarly, Figure \ref{fig3.3} shows the comparison of our scheme (when we considered CPC cost) with others where it can be seen that our scheme performed better as compared to Chen \textit{et al.} \cite{chen2015removing} and Chen \textit{et al.} \cite{chen2015t}. 

\begin{figure}
  \centering
  \includegraphics[width=0.8\linewidth]{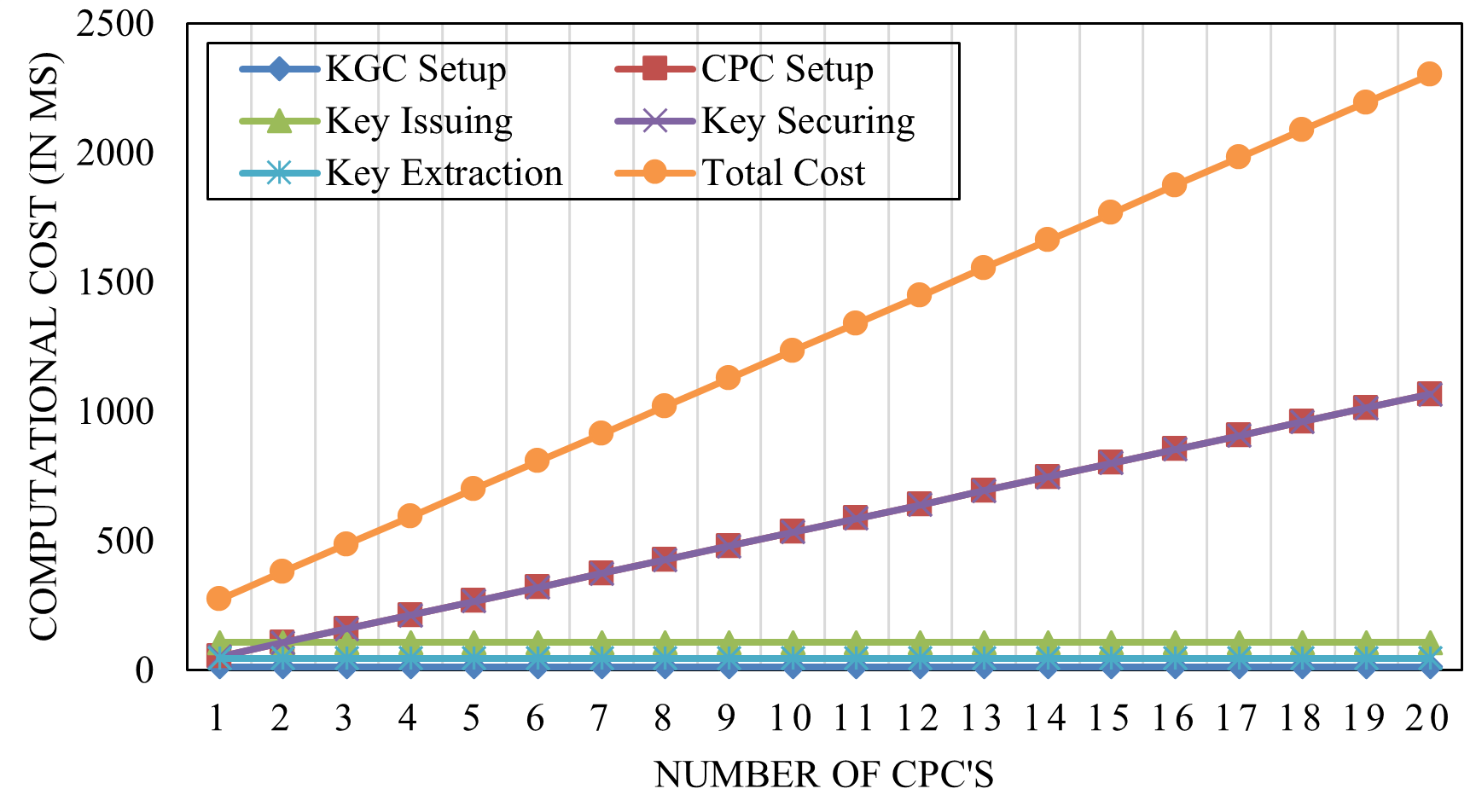}
  \caption{Computation overhead (in ms) of the different algorithms of proposed ESKI-IBE scheme.}
\label{fig3.1}
\end{figure}

\begin{figure}
  \centering
  \includegraphics[width=0.8\linewidth]{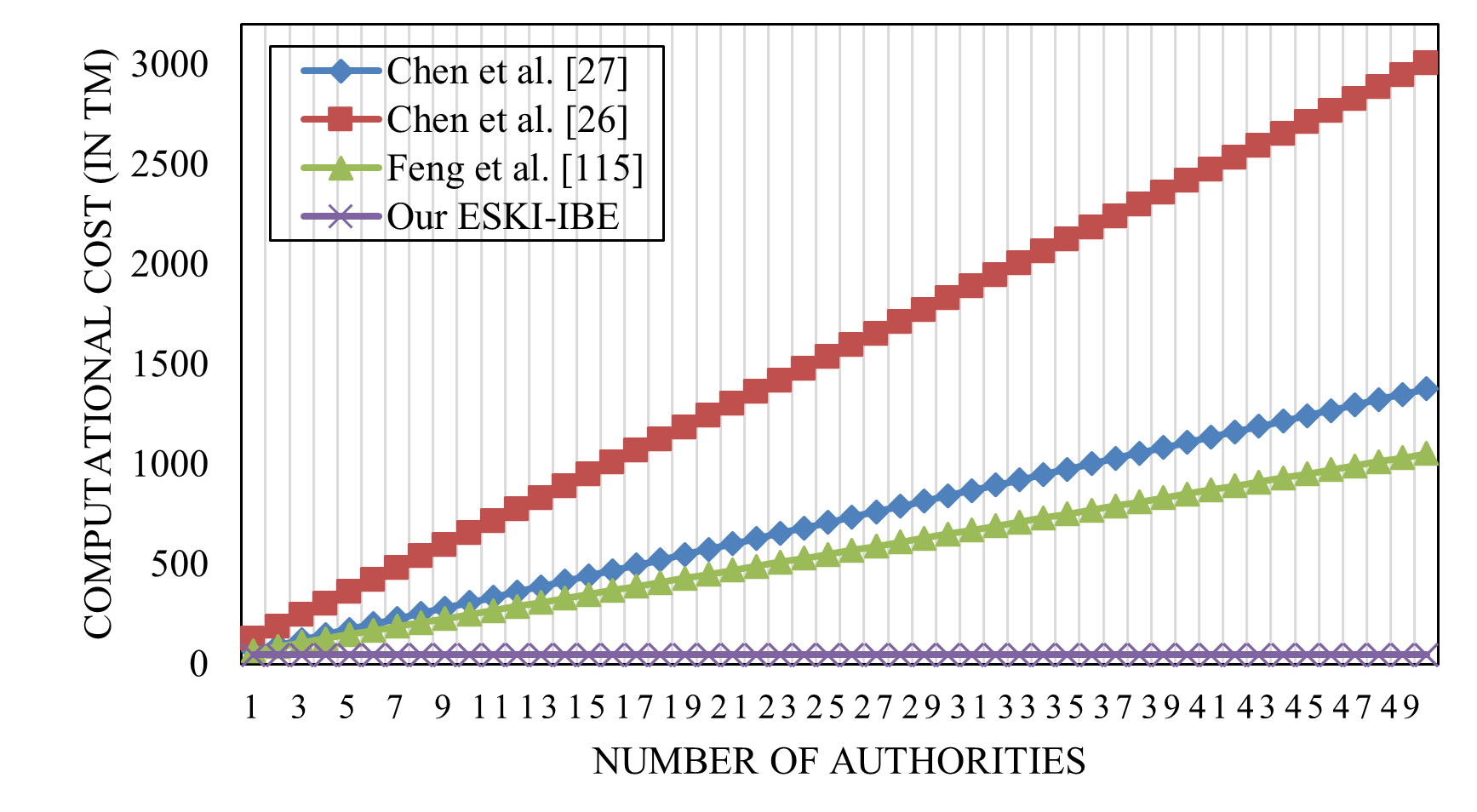}
  \caption{Setup and Key generation comparison of proposed ESKI-IBE scheme with others in terms of computation Cost (Without outsource Computing)}
\label{fig3.2}
\end{figure}

\begin{table}
        \centering
        \caption{Computation comparison of our proposed scheme with related schemes.}
        \label{tbl3.3}
        \begin{tabular}{|c|c|c|c|}
            \hline
             Scheme &	Set up (A)	& Key Extraction (B)&	Total  Cost (A+B) \\
            \hline
            \hline
            Chen \textit{et al.} \cite{chen2015removing} & $(n+1)T_P$ &	$(3+n)T_{SM}+(2+n)T_A$	& 26.71n+40.1 \\
           Chen \textit{et al.} \cite{chen2015t} & $(n+1)T_P$ &	$(4n+7)(T_{SM}+T_A)$	& 58.7n+66.9 \\
            Feng Li. \cite{li2017efficient} & $(2n+1)T_{SM}+nT_A+1T_H$ & 	$(2+n)T_{SM}+nT_A+1T_P$	& $20.1n+46.7$ \\
            Our W/o CPC &	$1T_{SM}+1T_H$ &	$5T_{SM}$ 	& 46.49\\
            Our CPC	& $(2n+1)T_{SM}+1T_H$ &	$(2n+5)T_{SM}$ 	& 26.68n+46.49\\
            \hline
            
        \end{tabular}
    \end{table}

\begin{figure}
  \centering
  \includegraphics[width=0.8\linewidth]{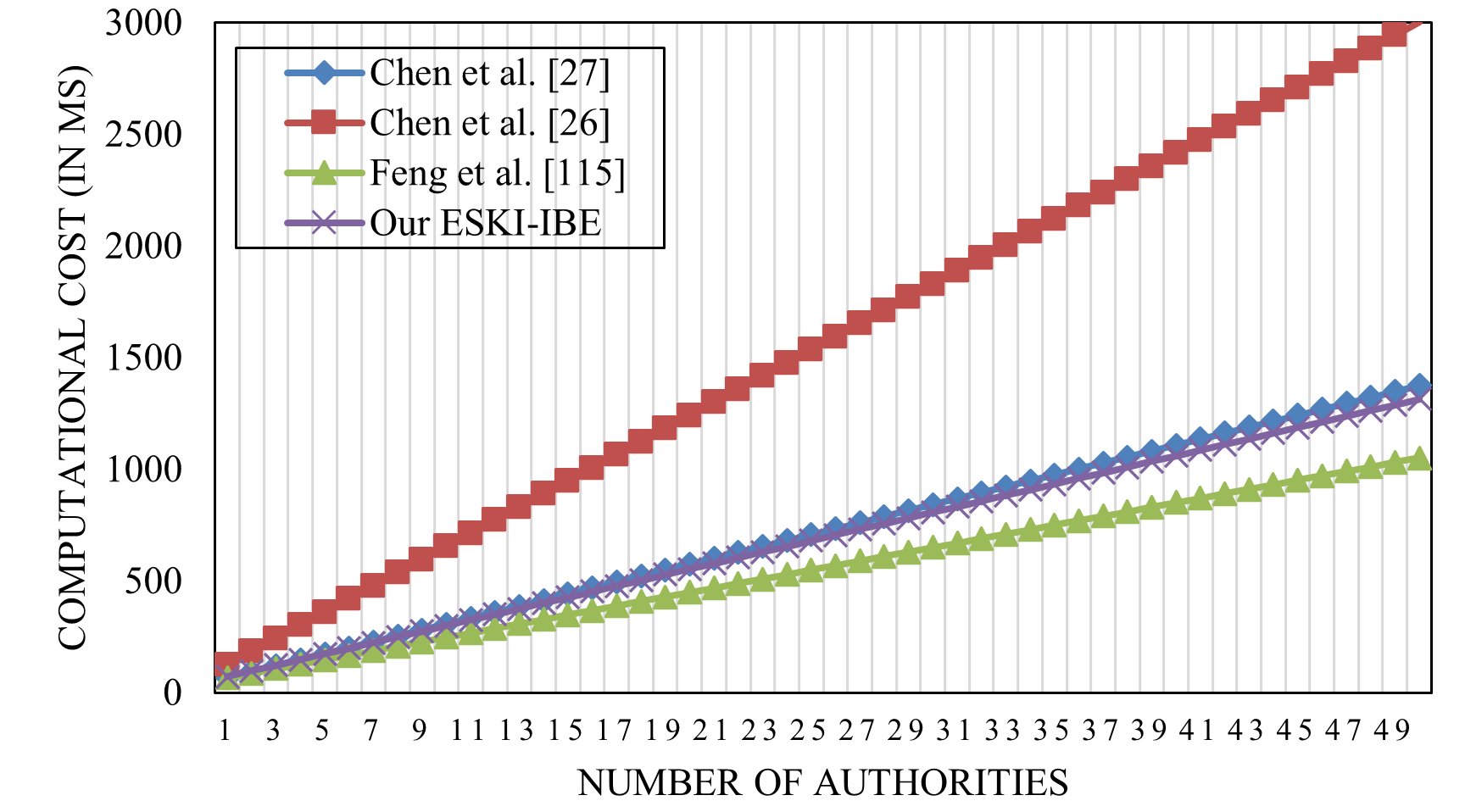}
  \caption{Setup and Key generation comparison of proposed ESKI-IBE scheme with others in terms of computation Cost (With outsource Computing)}
\label{fig3.3}
\end{figure}

\begin{figure}
  \centering
  \includegraphics[width=0.8\linewidth]{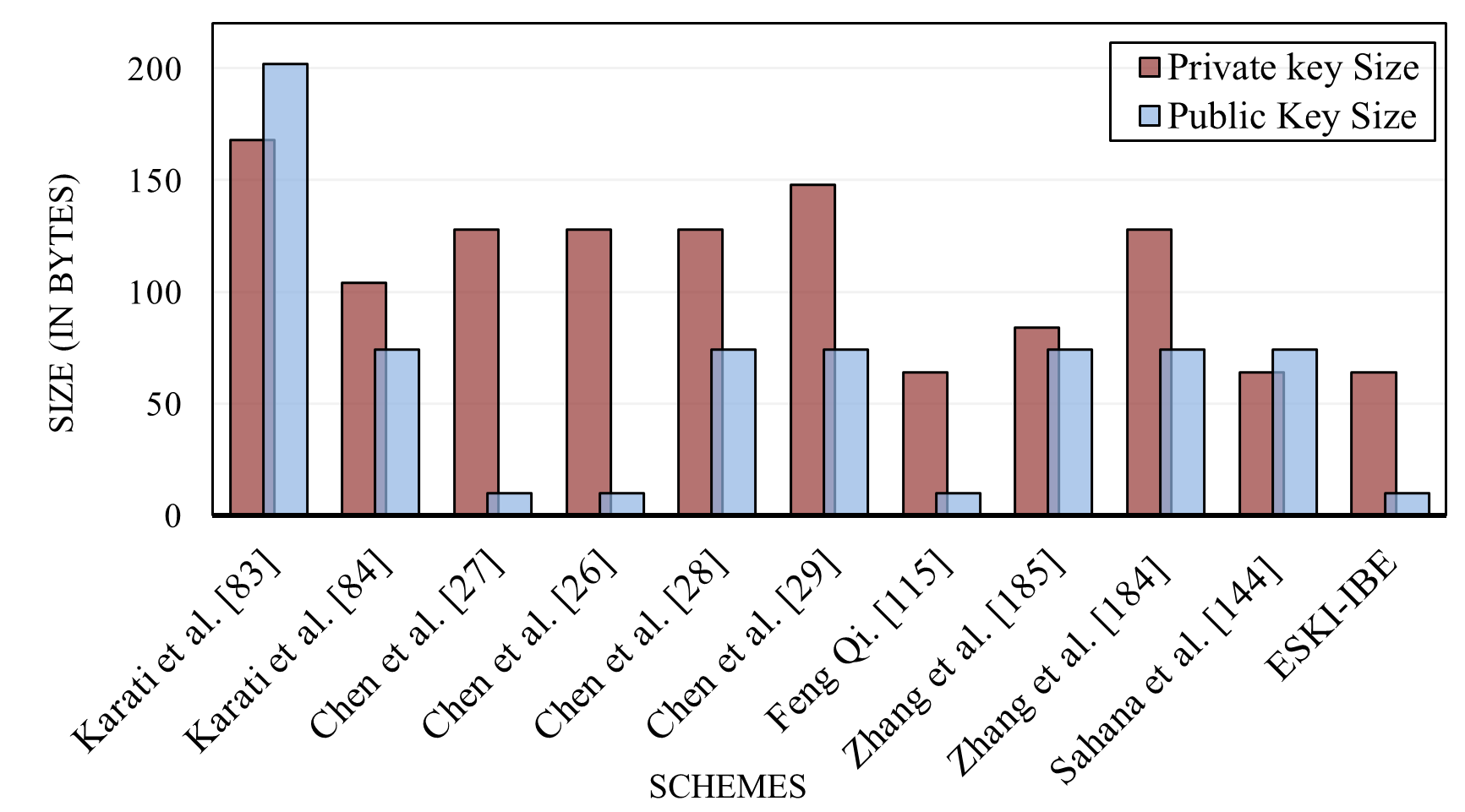}
  \caption{Communication Cost (in Bytes) of proposed ESKI-IBE proposed scheme}
\label{fig3.4}
\end{figure}

\textit{Communication cost}. In order to compare the communication cost of our proposed scheme with other schemes, we consider the size of $|\mathbb{G}_1|= |\mathbb{G}_3 |$ = 512-bit (64 Bytes), $|\mathbb{Z}_q |$ = 160-bit (20 Bytes),  and $|ID|$ = 80-bit (10 Bytes), where $\mathbb{G}_3$ is a group used in Chen \textit{et al.} \cite{chen2015removing} and Chen \textit{et al.} \cite{chen2015t}. Table \ref{tbl3.4} summarizes the comparison of the communication cost of our proposed scheme with other Escrow-free Schemes in the literature [\cite{karati2018pairing,karati2018provably,chen2015removing,chen2015t,chen2015escrowcomp,chen2016escrow,li2017efficient,zhang2012efficient,zhang2013id,sahana2019provable}. It can be observed from Figure \ref{fig3.4} that the proposed ESKI model have the private key and public key size of 64 Bytes and 10 Bytes respectively. 

\begin{table}
        \centering
        \caption{Bandwidth and security comparison of our proposed scheme with related schemes}
        \label{tbl3.4}
        \begin{tabular}{|c|c|c|c|c|}
            \hline
             Schemes &	Public key size &	Private key size &	Crypt. Tool &	Assumption \\
            \hline
            \hline
Karati \textit{et al.} \cite{karati2018pairing} &  168 Bytes & 202 Bytes &	ECC &	ECDLP\\
Karati \textit{et al.} \cite{karati2018provably} & 104 Bytes &	 74 Bytes &	Bilinear Pairing & 	EBSDH$^@$, BSDH$^\#$\\
Chen \textit{et al.} \cite{chen2015removing} & 128 Bytes & 10 Bytes &	COBG &	DSE$^\&$ \\
Chen \textit{et al.} \cite{chen2015t} & 128 Bytes & 10 Bytes &	COBG$^\&\&$ &	DSE\\
Chen \textit{et al.} \cite{chen2015escrowcomp} & 128 Bytes & 74 Bytes &	COBG &	DSE\\
Chen \textit{et al.} \cite{chen2016escrow} &148 Bytes	& 74 Bytes &	Bilinear Pairing & 	CDH$^{\#\#}$\\
Feng Qi. \cite{li2017efficient} &  64 Bytes &	 10 Bytes &	Bilinear pairing & 	BDH$^{@@}$\\
Zhang \textit{et al.} \cite{zhang2012efficient} & 84 Bytes &	74 Bytes &	Bilinear pairing &	CDH\\
Zhang \textit{et al.} \cite{zhang2013id} & 128 Bytes &	 74 Bytes &	Bilinear pairing &	BDH \\
Sahana \textit{et al.} \cite{sahana2019provable} & 64 Bytes &	 74 Bytes &	Bilinear pairing &	CDH\\
Our scheme	& 64 Bytes & 10 Bytes &	Bilinear pairing & 	BDH\\
            \hline
        \end{tabular}
        {\\$^@$Extended Bilinear Strong Diffie-Hellman, $^\#$Bilinear Strong Diffie-Hellman, $^{\#\#}$  computational Diffie-Hellman problem, $^{@@}$ Bilinear Diffie-Hellman problem,  $^\&$Dual system encryption, $^\&\&$Composite Order Bilinear Groups}
    \end{table}

\begin{table}
        \centering
        \caption{General comparison of our proposed scheme with related schemes.}
        \label{tbl3.5}
        \begin{tabular}{|c|c|c|c|c|}
            \hline
             Scheme	& RKEP $^@$ &	RSKIP $^\#$	& RUSP $^*$	& FID $^{@@}$ \\
            \hline
            \hline
            Karati \textit{et al.} \cite{karati2018pairing} & $\checkmark$&	$\checkmark$& 	$\times$& 	$\times$ \\
            Karati et  al. \cite{karati2018provably} & $\checkmark$&	$\checkmark$&	$\times$&	$\times$\\
            Chen \textit{et al.} \cite{chen2015removing} & $\checkmark$&	$\times$&	$\times$&	$\checkmark$\\
            Chen \textit{et al.} \cite{chen2015t} &$\checkmark$&	$\checkmark$	&$\times$	&$\checkmark$ \\
            Chen \textit{et al.} \cite{chen2015escrowcomp} &$\checkmark$	&$\times$	&$\checkmark$	&$\times$ \\
            Chen \textit{et al.} \cite{chen2016escrow} & $\checkmark$ &	$\times$&	$\checkmark$&	$\times$\\
            Feng Qi. \cite{li2017efficient} &$\checkmark$&	$\times$&	$\times$&	$\checkmark$\\
            Zhang \textit{et al.} \cite{zhang2012efficient} &$\checkmark$	&$\checkmark$&	$\times$&	$\times$\\
            Zhang \textit{et al.} \cite{zhang2013id} &$\checkmark$&	$\times$&	$\times$&	$\times$\\
            Sahana \textit{et al.} \cite{sahana2019provable} & $\checkmark$&	$\times$&	$\times$&	$\times$\\
            Our scheme	&$\checkmark$	&$\checkmark$	&$\checkmark$	&$\checkmark$ \\
            \hline
        \end{tabular}
        \footnotesize{\\$^@$ Resilient to Key escrow problem, $^\#$  Resilient to secure key issuing problem, $^*$ Resilient to user slandering problem, $^{@@}$ Achieve full identity based advantage}\\
\end{table}

\textit{General comparison}. Table \ref{tbl3.5} summarizes the general comparison of our scheme with related schemes in which we observed that our proposed ESKI scheme is resilient to key escrow problem (RKEP), resilient to secure key issuing problem (RSKI), resilient to user slandering problem (RUSP) and achieve full identity-based feature (FID).

\subsection{Performance Analysis of EF-IBS Scheme }
This section estimates the performance of our proposed EF-IBS scheme, in terms of computation cost and communication cost. 

\begin{table}
        \centering
        \caption{Computation cost (in ms) comparison of proposed IBS with related schemes.}
        \label{tbl3.6}
        \begin{tabular}{|c|c|c|c|}
            \hline
             Schemes &Signing & 	Verification & 	Total\\
            \hline
            \hline
            \cite{karati2018provably} & $2T_{SM}+1T_A+1T_E$ & 	$1T_p+3T_E$	& $1T_p+2T_{SM}+1T_A+4T_E$ \\
            \cite{chen2015escrowcomp} & $5T_{SM}+2T_A$ & $5T_p+2T_{SM}+1T_A$ &	$5T_p+7T_{SM}+3T_A$ \\
            \cite{chen2015escrowcloud} & $4T_{SM}+1T_A $ &	$4T_p+2T_{SM}$ &	$4T_p+6T_{SM}+1T_A$ \\
            \cite{sahana2019provable} & $2T_{SM}+1T_A$ &	$2T_p+1T_{SM}+1T_A$ &	$2T_p+3T_{SM}+2T_A$ \\
            \cite{shim2013eibas} & $4T_{SM}+2T_A+1T_E$ & $4T_p+4T_{SM}+1T_E$ &	$4T_p+8T_SM+2T_A+2T_E$ \\
            \cite{tseng2019top} & $2T_SM+7T_E$ &	$7T_p+4T_{SM}+1T_E$ &	$7T_p+6T_{SM}+8T_E$ \\
            \cite{yang2019top} & $3T_{SM}+3T_E$ &	$3T_p+2T_{SM}$ & $3T_p+7T_{SM}+3T_E$ \\
            Our &	$3T_{SM}+1T_A$ & $2T_p+1T_E$ & $2T_p+3T_{SM}+1T_E+1T_A$ \\
            \hline
        \end{tabular}
    \end{table}

\begin{table}
        \centering
        \caption{Communication cost and security comparison of proposed IBS with related schemes}
        \label{tbl3.7}
        \begin{tabular}{|c|c|c|c|c|c|c|}
            \hline
            Schemes & Private  &Public  & 	Signature  &	Cryptographic  &	Crypt.  &	Security \\
             & key size &  key size & 	size &	 Primitive &	 Tool & Assump.\\
            \hline
            \hline
            \cite{karati2018provably} & $1|\mathbb{G}_1 |+2|\mathbb{Z}_q |$ &	$1|\mathbb{G}_1 |+1|ID|$ &	$3|\mathbb{G}_1 |$	& CLS &	BP $^{\#\#}$ 	& EBSDH$^a$,  \\
            & & & & & & BSDH$^b$ \\
            \cite{chen2015escrowcomp} & $3|\mathbb{G}_1 |$ &	$1|\mathbb{G}_1 |+1|ID|$ &	$4|\mathbb{G}_1 |$	& HIBS &	COBG$^{@@}$ &	DSE $^{**}$\\
            \cite{chen2015escrowcloud} & $2|\mathbb{G}_1 |+1|\mathbb{Z}_q |$ &	$1|\mathbb{G}_1 |+1|ID|$ &	$3|\mathbb{G}_1 |$ &	HIBS &	BP &	CDH $^@$ \\
            \cite{sahana2019provable} & $2|\mathbb{G}_1 |$	& $1|\mathbb{G}_1 |+1|ID|$ &	$2|\mathbb{G}_1 |$ &	Partial IBS &	BP &	CDH\\
            \cite{shim2013eibas} & $2(|\mathbb{G}_1 |+|\mathbb{Z}_q |)$	& $2|\mathbb{G}_1 |+1|ID|$ & 	$3|\mathbb{G}_1 |$ &	CLS &	BP &	CDH\\
            \cite{tseng2019top} & $2|\mathbb{G}_1 |$ &	$2|\mathbb{G}_1 |+1|ID|$ &	$5|\mathbb{G}_1 |$	& CLS &	BP &	GCDH $^\#$ \\
            \cite{yang2019top} & $2|\mathbb{G}_1 |$ &	$1|\mathbb{G}_1 |+1|ID|$ &	$3|\mathbb{G}_1 |$ &	CLS &	BP &	CDH \\
            Our & 	$1|\mathbb{G}_1 |$ &	$1|ID|$	& $2|\mathbb{G}_1 |$ &	IBS & BP & 	BDH $^*$ \\
            \hline
        \end{tabular}
        \footnotesize{$^a$ Extended Bilinear Strong Diffie-Hellman, $^b$ Bilinear Strong Diffie-Hellman, $^@$ computational Diffie-Hellman problem, $^\#$ generalized computational Diffie-Hellman problem, $^*$ Bilinear Diffie-Hellman problem, $^{@@}$ Composite order bilinear group, $^{\#\#}$ Bilinear Pairing and $^{**}$ dual system encryption}\\
       
    \end{table}

\begin{figure}
  \centering
  \includegraphics[width=0.8\linewidth]{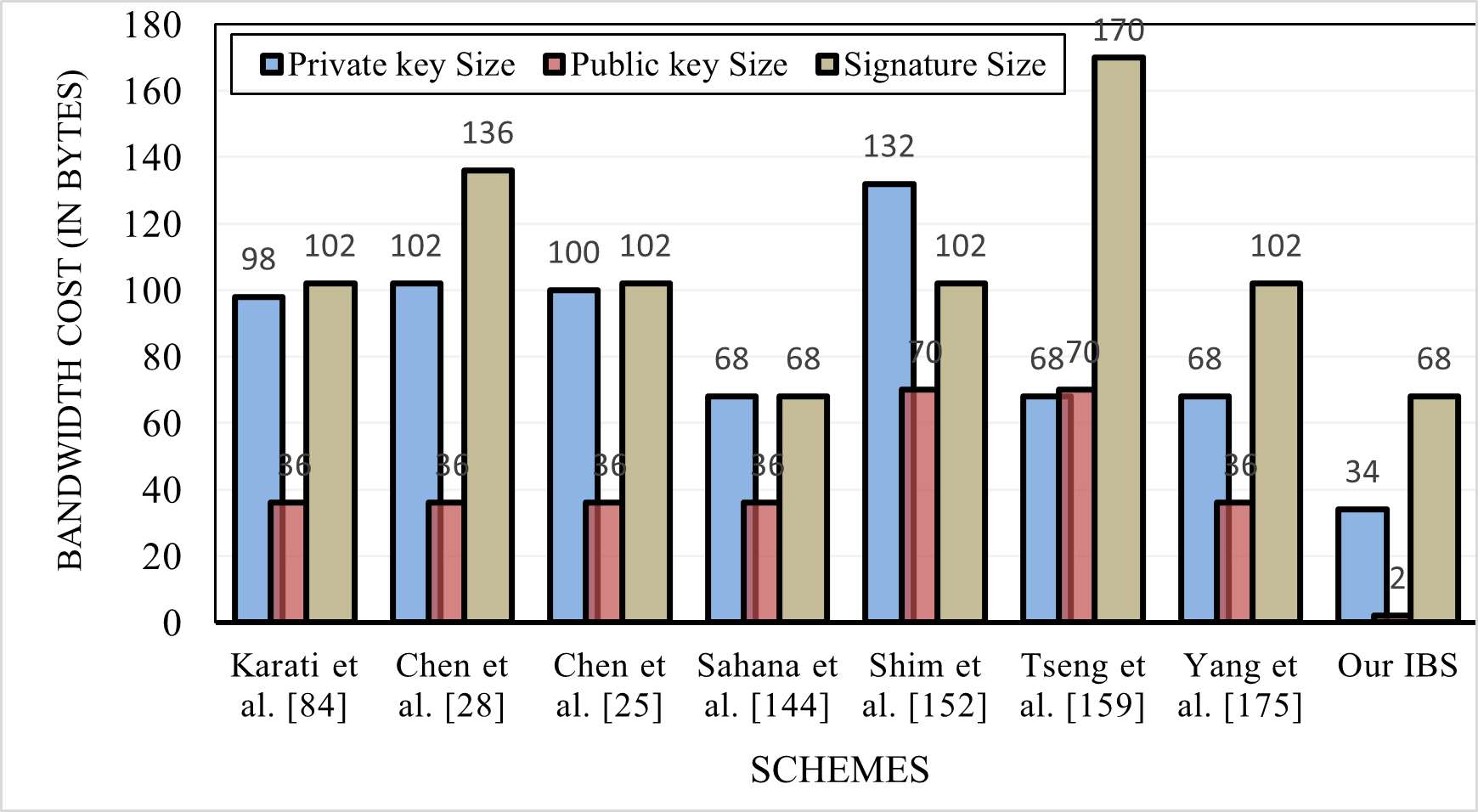}
  \caption{Signing, verification and total computational cost (in ms) comparison of proposed EF-IBS scheme with other schemes}
\label{fig3.5}
\end{figure}

\begin{figure}
  \centering
  \includegraphics[width=0.8\linewidth]{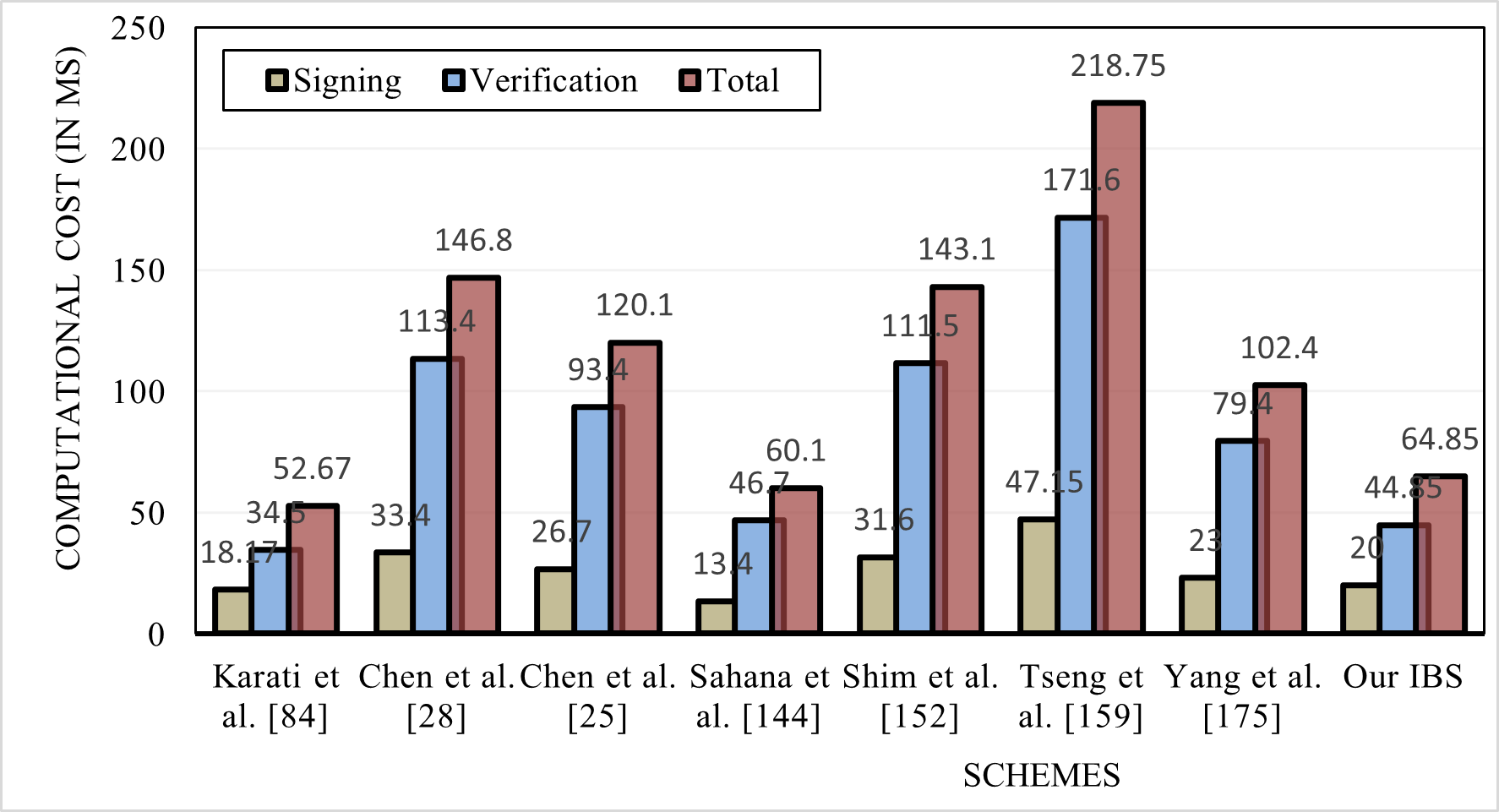}
  \caption{Private key, public key and signature size comparison of proposed EF-IBS scheme with related schemes}
\label{fig3.6}
\end{figure}

\textit{Computation cost}. In Table \ref{tbl3.6}, we compare the signing and verifying cost of our proposed EF-IBS scheme with related schemes such as Karati \textit{et al.} \cite{karati2018provably}, Chen \textit{et al.} \cite{chen2015escrowcomp}, Chen \textit{et al.} \cite{chen2015escrowcloud}, Sahana \textit{et al.} \cite{sahana2019provable}, Shim \textit{et al.} \cite{shim2013eibas}, Tseng \textit{et al.} \cite{tseng2019top} and Yang \textit{et al.} \cite{yang2019top} IBS schemes. For signing a message, the proposed scheme takes $3T_{SM}+1T_A=20.04 ms$, while Karati \textit{et al.} \cite{karati2018provably}, Chen \textit{et al.} \cite{chen2015escrowcomp}, Chen \textit{et al.} \cite{chen2015escrowcloud}, Sahana \textit{et al.} \cite{sahana2019provable}, Shim \textit{et al.} \cite{shim2013eibas}, Tseng \textit{et al.} \cite{tseng2019top} and Yang \textit{et al.} \cite{yang2019top} IBS schemes take $18.17 ms, 33.4 ms, 26.7 ms, 13.4 ms, 31.6 ms, 47.15 ms$ and $23 ms$ respectively. For verifying a message-signature pair, the proposed scheme takes $2T_{BP}+1T_E=44.85 ms$, while Karati \textit{et al.} \cite{karati2018provably}, Chen \textit{et al.} \cite{chen2015escrowcomp}, Chen \textit{et al.} \cite{chen2015escrowcloud}, Sahana \textit{et al.} \cite{sahana2019provable}, Shim \textit{et al.} \cite{shim2013eibas}, Tseng \textit{et al.} \cite{tseng2019top} and Yang \textit{et al.} \cite{yang2019top} IBS schemes take $34.5 ms$, $113.4 ms$, $93.4 ms$, $46.7 ms$, $111.5 ms$, $171.6 ms$ and $79.4 ms$ respectively. Thus, the proposed system takes 64.85 ms for total computational cost, while Karati \textit{et al.} \cite{karati2018provably}, Chen \textit{et al.} \cite{chen2015escrowcomp}, Chen \textit{et al.} \cite{chen2015escrowcloud}, Sahana \textit{et al.} \cite{sahana2019provable}, Shim \textit{et al.} \cite{shim2013eibas}, Tseng \textit{et al.} \cite{tseng2019top} and Yang \textit{et al.} \cite{yang2019top} IBS schemes take $54.67 ms$, $146.8 ms$, $120.1 ms$, $60.1 ms$, $143.1 ms$, $218.75 ms$, and $102.4 ms$ respectively. Figure \ref{fig3.5} compares our proposed scheme with other related schemes.

\textit{Storage cost}. Table \ref{tbl3.7} compares the communication cost and security assumption of our scheme with Karati \textit{et al.} \cite{karati2018provably}, Chen \textit{et al.} \cite{chen2015escrowcomp}, Chen \textit{et al.} \cite{chen2015escrowcloud}, Sahana \textit{et al.} \cite{sahana2019provable}, Shim \textit{et al.} \cite{shim2013eibas}, Tseng \textit{et al.} \cite{tseng2019top} and Yang \textit{et al.} \cite{yang2019top} IBS schemes. For communication cost evaluation, we assume $|ID| = 2 bytes$. For the super-singular curve over the binary field  $\mathbb{F}_{2^{271}}$ with the order of $\mathbb{G}_1$ is 252-bit prime and using compression technique [27], we consider $|\mathbb{G}_1 |=34 bytes$ and $|\mathbb{Z}_q | = 32 bytes$. We assume that the HIBE/HIBS scheme in \cite{karati2018provably}, \cite{chen2015escrowcomp}, \cite{chen2015escrowcloud} [7], [10], [11] are $1$-level HIBE/HIBS scheme, which is equivalent to standard IBE/IBS scheme. The proposed scheme requires to transmit $2|\mathbb{G}_1 |$= 68 Bytes while Karati \textit{et al.} \cite{karati2018provably}, Chen \textit{et al.} \cite{chen2015escrowcomp}, Chen \textit{et al.} \cite{chen2015escrowcloud}, Sahana \textit{et al.} \cite{sahana2019provable}, Shim \textit{et al.} \cite{shim2013eibas}, Tseng \textit{et al.} \cite{tseng2019top} and Yang \textit{et al.} \cite{yang2019top} IBS schemes take $102 Bytes$, $136 Bytes$, $102 Bytes$, $68 Bytes$, $102 Bytes$, $170 Bytes$ and $102 Bytes$ respectively. Also, the proposed scheme has the least public key ($2 Bytes$) and private key ($34 Bytes$) size in comparison with the related schemes, shown in Figure (\ref{fig3.6}).

\textit{Batch Verification}. Here, we discuss that the proposed EF-IBS scheme verifies multiple message-signature pairs simultaneously in the least communication and computation cost compared to the situation where multiple message-signature pairs are verified one by one. On given message-signature pairs $<m_k,S_k,R_k>$, sender’s $ID$, and $pp$, the recipient verifies the signature-message pairs using  (\ref{eq3.10}).

\begin{equation} \label{eq3.10}
    \prod_{k=0}^ne(S_k,R_k ) \overset{?}{=}g^{\sum _(k=0)^nH_3(m_k)}e(Q_ID,Y)^{
    sum_(k=0)^nH_3 (R_k)}
\end{equation}

 The Equation (\ref{eq3.10}) ensures the validity of multiple signature-message pairs. The consistency of the above Equation (3.7) is verified as follows. From LHS,
\vspace{-10mm} 
 
\begin{align*}
    \prod_{k=0}^ne(S_k,R_k) &=\prod_{k=0}^ne(r_k^{-1}(H_2(m_k)P+H_3(R_k)d_ID),r_kP)\\ 
&=\prod_{k=0}^ne((H_2(m_k)P+H_3(R_k)d_{ID}),P) \\
&=\prod_{k=0}^ne(H_2(m_k)P,P)  \prod_{k=0}^ne(H_3 (R_k)Q_{ID},Y) \\
&=\prod_{k=0}^ne(P,P)^{H_2(m_k)}  \prod_{k=0}^ne(Q_{ID},Y)^(H_2(R_k))\\ 
&=\prod_{k=0}^ng^{H_2(m_k)}  \prod_{k=0}^n(Q_{ID},Y)^{H_2(R_k)}\\
&=g^{\sum_{k=0}^nH_3(m_k)}e(Q_{ID},Y)^{\sum_{k=0}^nH_3(R_k)}
\end{align*}

This completes the correctness of (\ref{eq3.10}). For verifying k message-signature pairs simultaneously, the batch verification needs $(k+1) T_BP+kT_E$ while individual verification needs $2kT_BP+2kT_E$ of computation cost.

\textit{Comparison with secret-sharing based solution}. In the secret-sharing solution \cite{boneh2001identity}, the master key is distributed among multiple PKGs, and the user requests each PKG for its private key share using the secret-sharing scheme. Each PKG authenticates the user and computes the private key shares using their master key shares and issues them to the user. The user then combines the shares and extracts their private key. While in our scheme, the user’s private key is derived from the KGC, further protected by multiple authorities under their secret keys. In the proposed IBS, the KGC authenticates the user and issues a partial private key to him, and multiple authorities protect the private key. Furthermore, private key distribution among the entities is secured under the user’s secret key. Therefore, our proposed IBS scheme is secured against the malicious-but-passive PKG that overcomes the inefficiency due to multiple-user authentication.

\textit{Comparison with CL-PKC and CB-PKC}. In CL-PLC \cite{al2003certificateless}, the user derived his private key, which consists of a partial private key derived from the trusted certificate authority and the user’s secret key to sign a message. For verifying the message-signature pair, the recipient needs the signer identity ID and the public key corresponding to the signer’s secret key. Similarly, in CB-PKC \cite{gentry2003certificate}, the user’s private key involves the secret key of the certificate authority and his secret key. Both these two cryptosystems were partially built over the framework of IBC directly without introducing more infrastructures. However, due to the introduction of the user’s secret key, they do not adequately preserve the identity-based property. On comparing CL-PKC \cite{al2003certificateless} and CB-PKC \cite{gentry2003certificate}, the user’s private key in our proposed approach is derived from the KGC and simultaneously protected by the multiple CPCs. Assuming at least one CPC is honest, our approach addresses the key escrow, key abusing, and user slandering problem.   
	
\section{Summary}
This section introduces a solution that can enhance trust in Private Key Generators (PKG) within Identity-Based Cryptography (IBC). It presents a secure and efficient key issuing scheme that addresses key escrow, secure key issuing, and user slandering issues. The proposed scheme is then utilized to implement an escrow-free Identity-Based Encryption (IBE) scheme that ensures confidentiality, and an escrow-free Identity-Based Signature (IBS) scheme that is secure against forgery. The proposed schemes are benchmarked against existing solutions in computation and communication costs, and the results demonstrate that the proposed schemes require the least computation and communication cost, making them ideal for practical applications. This approach provides an alternative to traditional IBC, which relies on a trusted third party, the PKG. The use of an escrow-free key issuing scheme provides a secure and efficient way to generate keys without relying on a trusted PKG.

\end{doublespace} \label{chapter3}
\begin{savequote}[75mm] 
Difficulties in your life do not come to destroy you but to help you realise your hidden potential and power. Let difficulties know that you too are difficult.
\qauthor{Avul Pacir Zainulabidin Abdul Kalam-Atomic Scientist and President of India} 
\end{savequote}

\chapter{Identity-Based Blind Signature Scheme for End-to-End Verifiable Internet-Voting System}
\justify
\begin{doublespace}

The Electronic Voting System (EVS) \cite{vassil2016diffusion} is undoubtedly more efficient, participatory, and fair than the traditional postal voting system. However, inconsistencies have been observed in the EVS system, including ballot privacy, voter anonymity, transparency, and auditing. To address these shortcomings, experts have utilized several cryptographic tools, such as homomorphic cryptosystems \cite{cervero2014efficient, chillotti2016homomorphic, yang2018secure}, public key encryption \cite{kusters2012clash}, and zero-knowledge proof \cite{kusters2012clash}. An early secure EVS system using a mix-net approach was developed by David Chaum \cite{chaum1981untraceable}. Blind signature schemes \cite{fujioka1992practical, mu1998anonymous, lopez2014pairing}, proxy servers \cite{chen2014secure}, and secret sharing mechanisms \cite{porkodi2011multi} have also been used to implement the EVS system. In this chapter, we introduce two secure \textit{Identity-Based Blind Signature} (IDBS) schemes: \textit{IDBS-I} and \textit{IDBS-II}. These schemes have been designed to prevent forgeable attacks and ensure blindness. We have implemented an \textit{End-to-End verifiable internet voting (E2E-VIV)} system with batch verification based on the proposed IDBS-II scheme. Additionally, we have demonstrated that the \textit{E2E-VIV} system meets the required security goals.

There have been discussed several practical E2E-VIV systems. The first system based on the principle of E2E-VIV is the Rijnland Internet-based election system (RIES) \cite{hubbers2005ries} that enables a voter to access the public algorithm and parameters to confirm his participation in the election. This verification procedure is supposed to be secure as compared to the traditional postal system, but it is far weaker than desired individual verification. Peter Ryan invented \textit{Pret a Voter} system \cite{ryan2009pret} that uses a list of shuffled ordering candidates and a space of vote marking and a traditional cryptographic method. It enables a voter to verify his vote based on the shuffled ordering candidate. Identical to \textit{Pret a voter}’s election experience, \textit{Punchscan} system discussed by Popoveniuc \textit{et al.} \cite{popoveniuc2006introduction} allows voting on the paper ballots and the security is achieved by the use of an optical scanner that reads the privacy-preserving area. The voter is permitted to verify his receipt with the published record.

The \textit{Scantegrity} system \cite{chaum2008scantegrity, chaum2009scantegrity} generates a random number and distributes it to the election officials using a secret sharing scheme. It computes a three-letter code using the random number for each printed ballot and generates a table for post-election verification. During the election, a voter marks a choice with invisible ink on the paper to obtain the three-letter code. He can record the code and ballot ID to check if his vote was tabulated. In \textit{Remotegrity} system \cite{zagorski2013remotegrity}, the voter obtains a lottery-style code and authentication card with serial numbers. It uses both the serial number and code against the candidate’s selection and authentication code to cast a vote. \textit{Adida} presents a web-based voting system, known as Helios, that can achieve E2E verifiability and privacy \cite{adida2008helios}. This method has a unique design mechanism in the voting device to help the voter in verification during vote casting. A voter would cast the encoded vote only if he is convinced that the voting device is not cheating him. Corteir \textit{et al.} \cite{cortier2013attacking} point out that the repetition of the voter’s choice in the Helios will leak any information about his vote. In \cite{chondros2016d, zacharias2016demos}, the E-voting system achieves E2E verifiability in the standard model and is secure under the well-known decisional Diffie-Hellman problem. Joaquim \textit{et al.} \cite{joaquim2013eviv} present an E2E-VIV system that provides the advantage of mobility and vote privacy to the voters. In [22], each voter has a voter security token that encrypts the vote to transmit his encoded candidate choice to the officials.

Most of the above-discussed E2E-VIV systems are paper ballot based. Recently, a team of experts has suggested moving the election process online as a voter is particularly interested in the online system \cite{joaquim2013eviv}. In 1992, Fujiokta \textit{et al.} \cite{fujioka1992practical} gave the first secret E-voting mechanism based on the blind signature scheme, which ensures privacy and fairness. Since then, several blind signatures-based E-voting systems have been discussed \cite{lopez2014pairing, yang2018secure, chen2014secure, porkodi2011multi, kusters2012clash}. However, they have lacks security and transparency. The E2E verifiability is one of the tools to promise the integrity and transparency of the system [23]. Such internet voting systems are based on traditional public key cryptosystem which requires a significant computational cost for certificate management and public keys. In this chapter, we proposed two identity-based blind signature schemes and implement an End-to-End Verifiable Internet Voting system. We consider the smartphone as a voting device to collect the voter’s biometric features and authenticate him against his biometric features and identity, for example, national identification number, employee ID, and social security number. The proposed E2E-IVS system provides mobility to a voter and enables him to cast his vote secretly on a public computer/device without revealing any information about the vote. Now, we discuss the proposed E2E-IVS system.

\section{System Model}

In this section, we give the brief architecture of the E2E-IVS system and its security requirements. We also define the identity-based blind signature scheme and its security model.

\subsection{Architecture Design }

Here, we discuss the architecture of the proposed E2E-VIV system, which involves five entities: key generation centre (KGC), voters, candidates, voting device, and election commission authority (ECA). 
\begin{itemize}
    \item \textit{Key Generation Center}. The KGC initializes the system, authenticates each entity (voters, candidates, and ECA) and issues the private keys to them. Additionally, the KGC maintains a record of registered voters and registered candidates. 
    \item \textit{Voter}.  A voter is a person in any organization or a citizen of the country having a unique ID and registered with KGC to obtain his private key. 
    \item \textit{Election Commission Authority}. The ECA, in the proposed system, is an election-conducting server that performs the following three processes. (1) It authenticates the voter and distributes a secret digital blank ballot to him without knowing his ID. It maintains the record of the valid voters for whom blank ballots are already issued. (2) It receives and verifies the vote during the vote-casting phase and issues a witness to him. (3) It filters the duplicate ballots, counts them, and announces the winner.
    \item \textit{Voting Device}. A voting device could be a smartphone or computer that contains a fingerprint and iris sensor, which is either implanted or attached to the voting device.
    \item \textit{Candidate}. A person nominated by any group or political party who can participate in the election process. He must be first registered with KGC and obtain a private key. 
\end{itemize}
	
The proposed E2E-IVS system consists of the following six algorithms.
\begin{itemize}
    \item Initialization. The KGC initializes the system and computes its master key $s_0$ and public parameter $pp$. The KGC keeps $s_0$ secret and publishes $pp$. 
	\item \textit{Registration}. Before the election process, each participating entity is required to register with KGC. The KGC authenticates the entities against their IDs and biometric information and issues the private keys to them over an insecure channel. Further, KGC prepares two lists: $List^{RV}$ and $List^{RC}$ which contain the biometric details of all the registered voters and registered candidates, respectively, and passes them to the ECA after registration. 
	\item \textit{Authentication and Ballot distribution}. In PKI-based cryptography, the user’s private key is associated with his ID. Thus, we cannot use the voter’s private key for signature purposes as it may lose the voter’s anonymity. The pseudonym public and private keys preserve voter anonymity. In order to preserve anonymity, the system enables the voter to generate the pseudonym public and pseudonym private keys. During the blank distribution phase, a voter blinds and signs the pseudonym public key using his pseudonym private key and asks the ECA for a blind blank ballot against his pseudonym public key. To authenticate a voter, the ECA maintains two lists ($List^{RV}$  and $List^{VV}$) for the registered voters and voters for whom the ballots have previously been distributed. The ECA first checks if the voter is registered (i.e., found in list $List^{RV}$), and the blank ballot is not issued to him previously (i.e., not found in list $List^{VV}$), only then ECA issues a fresh pseudonym certificate, known as a blank ballot, against the pseudonym public key to the voter and adds an entry in the list $List^{VV}$; otherwise, it raises an error. Then, the voter extracts the blank digital ballot using his secret values.
	\item \textit{Vote casting}. An IVS system has many inconsistencies, e.g., network congestion due to voter traffic during the election period, election server failure, and people with no internet facility during the election period. The proposed system allows a voter for early voting to overcome these difficulties. A voter may cast a vote before the election date using a timestamp T that ensures the freshness of the vote. During the election and/or pre-election, a voter chooses a candidate from the list $List^{RC}$, computes the electronic ballot $Bt$ using BLS short signature scheme and passes it to ECA. The ECA authenticates the voter and validates the ballot B, and issues an electronic receipt if both conditions are satisfied; otherwise, it rejects the vote and adds $Bt$ in $List^B$ (initially empty). The ECA includes a unique random value in each receipt that makes it fresh and unique. 
	\item \textit{Vote counting}. When the voting day is over, the ECA stops to receive additional ballots. It makes sure that there are no invalid or duplicate electronic ballots. To identify the duplicate votes in the list $List^B$, ECA manages two lists: $List^V$ to store the details of valid voters, and $List^{IV}$ to store the details of invalid voters. The ECA separates the invalid votes from the list $List^B$ by identifying two ballots having the same signature. Finally, the ECA publishes both lists. 
    \item \textit{Auditing}. After the election process is over, the voter verifies if his vote corresponding to his receipt is in the ballot list $List^V$ or $List^{IV}$. Also, anyone can verify if all ballots are counted correctly. 
\end{itemize}

\subsection{Security Requirements for E2E-IVS System}
This section discusses the security requirements for the internet-based voting system.
\begin{itemize}
    \item 	\textit{Voter’s anonymity}. It must ensure the vote anonymous against the voter’s ID who is able to vote. 
    \item \textit{Vote integrity}. The system must be secure against an adversarial attack on an individual voter’s devices and on its architecture. It must also ensure the vote's privacy. 
    \item \textit{Voter’s eligibility}. The system must permit only eligible/authentic voters to vote.
    \item \textit{Uniqueness}. It must ensure that a voter gets only one ballot, and can cast only one vote.
    \item \textit{Resistant to bribery and coercion}. In any election process, the collusion between a voter and a candidate may give two main inconsistencies: i) the candidate may give some bribe to the voter (bribery), or ii) the candidate may threaten or force a voter (coercion). During as well as after the vote counting process, a voter can freely cast his vote under coercion and bribery situations and cannot prove whom he supported.
    \item \textit{Ballot stuffing}. An illegal practice of voting, i.e., casting multiple ballots during a vote should not be possible, whereby one ballot is allowed.
    \item \textit{Individual verifiability}. Any voter can verify that his ballot is correctly involved in the ballot list. 
    \item \textit{Universal verifiability}. Anyone can verify that all ballots in the list have been correctly tallied.
    \item \textit{Batch verifiability}. In order to optimize the cost of verification, the system must verify the massive ballots simultaneously. 
\end{itemize}

\subsection{Identity-based Blind Signature Scheme}
It involves three entities: Signer, User and Private Key Generator (PKG) that perform four randomized probabilistic polynomial-time (PPT) algorithms: Setup, Key Extract, Blind Signature, and Verify. 
\begin{itemize}
    \item \textit{Setup}: The PKG computes the master key s and a public parameter $pp$. It keeps s secret and publishes$pp$.
    \item \textit{Key Extract}: The PKG computes the private key $d_{IDS}$ against the signer’s $ID_S$ using and param and sends $d_{IDS}$ to the signer. 
    \item \textit{Blind Signature}: Without knowing the content of message $M$, the signer signs on it as follows. 
    \begin{itemize}
        \item \textit{Blinding}: The user blinds m using its secret key and requests to the signer for signature on the blinded $M$.
        \item \textit{Signing}: Given blinded $M$, the signer computes a blind signature using its private key $d_S$ and sends it to user.
        \item \textit{Unblinding}: The user retrieves the original signature $\sigma$. 
    \end{itemize}

	Verifying: On given tuple $<M,\sigma>$ and signer’s identity $ID_S$, the user verifies the signature.
\end{itemize}
	 
\subsection{Security Threat}
The proposed IDBS scheme is considered to be secure if it achieves the blindness property and under the chosen message and ID attack, it is existentially unforgeable (EUF-ID-CMA). 

\textbf{\textit{EUF-ID-CMA}}: It can be defined by the following game playing between the Forger $F$ that acts as a malicious user and challenger $Ch$ that acts as the signer.

The IDBS scheme is EUF-ID-CMA secured if forger $F$ has a negligible advantage in the following game. 

\textbf{Setup}: The $Ch$ computes its master key and public parameters, where it keeps the master key and responds public parameter to $F$. 

\textbf{Oracle}: $F$ performs the following oracles. 
\begin{itemize}
    \item \textit{Extract oracle}: On given $ID$, $Ch$ runs this oracle to compute the private key $d_S$ and sends to $F$.
    \item \textit{Blind signature oracle}: For given $ID$ and message $M$, $F$ blindly requests a signature on m from $Ch$. The $Ch$ runs the oracle to generate the signature and gives to $F$ 
\end{itemize}

\textbf{Forgery}: At the end, the forger $F$ responds with a Message-Signature pair $<M^*,\sigma^*>$ against $ID^*$. The forger $F$ will win the game if it fulfils the following conditions. 
\begin{itemize}
    \item $<M^*,\sigma^*>$ is the valid message-signature pair against $ID^*$.
    \item The Blindsig oracle has not been queried on $<M^*,\sigma^*>$. 
    \item The extract oracle has not been queried on $ID^*$.
\end{itemize}

\textbf{Unlinkability}. The blindness security notion can be understood using the following game playing between the adversary $Adv$ that acts as a malicious signer and Challenger Ch that acts as the honest user. The user have two distinct message $<M_b,M_{1-b}>$ engaged in the proposed IDBS scheme and obtains the signatures $<\sigma_b,\sigma_{1-b}>$, where $b\in \{0,1\}$.

The IDBS scheme is unlinkable if forger $F$ has a negligible advantage in the following game.

\textbf{Setup}: The Ch computes its master key and public parameters, where it keeps the master key and responds public parameter to the $F$. 
Oracle: $F$ performs the following oracles. 
\begin{itemize}
    \item \textit{Extract oracle}: For given $ID$, $Ch$ runs this oracle to get the private key $d_S$ and sends it to $F$.
	\item \textit{Blind signature oracle}: For given $ID$ and message $M$, $F$ blindly requests a signature on $M$ from $Ch$. The $Ch$ runs the oracle to generate the signature $\sigma$ and gives to $F$.
\end{itemize}

\textbf{Challenge}: $Ch$ pick a bit $b \in \{0,1\}$ at random and ask $Adv$ for signature on $M_b$ and $M_{1-b}$. Finally, $Ch$ generates the signature $<\sigma_b,\sigma_{1-b}>$ on message $<M_b,M_{1-b}>$ and gives $<\sigma_b,\sigma_{1-b}>$ to $F$.

\textbf{Response}: For the given tuple $\langle M_0, M_1, \sigma_b, \sigma_{1-b} \rangle$, $F$ predicts a bit $b' \in {0,1}$ and wins the game if $b = b'$ holds with advantage $|\Pr[b=b']| \geq \frac{1}{2} + k^{-n}$.

\section{Proposed Identity-based Blind Signature schemes}
In this section, we proposed two identity-based blind signature schemes: IDBS-I and IDBS-II. 

\subsection{ID-based Blind Signature-I Scheme}
The proposed IDBS-I scheme consists of four algorithms run between the user and the signer defined as follows:
\begin{itemize}
    \item \textit{Setup}: Suppose $P$ be the generator of group $\mathbb{G}_1$ of prime order $q$. Bilinear map $e:\mathbb{G}_1 \times \mathbb{G}_1 \rightarrow \mathbb{G}_2$. Let three pre-image resistant cryptographic hash functions are $H_1:\{0,1\}^* \rightarrow \mathbb{G}_1$, $H_2:\{0,1\}^* \rightarrow \mathbb{Z}_q$, $H_3:\{0,1\}^* \times \mathbb{G}_1 \rightarrow \mathbb{Z}_q$, and $H_4: \mathbb{G}_2 \times \mathbb{Z}_q \rightarrow \{0,1\}^*$. Let the private key of the signer, and the user is denoted as $d_S$ and $d_U$, respectively. Let $t$ indicate the timestamp. PKG selects randomly $s \in \mathbb{Z}_q$ and computes public key $P_{Pub}=sP$.  PKG publishes $pp=<\mathbb{G}_1,q,e,P,P_{Pub},H_1,H_2,H_3,H_4>$, and keeps secret key $s$ secretly.
    \item \textit{Extract}: PKG computes $d_S=sQ_S$, and $d_U=sd_U$, where $Q_S=H_1(ID_S)$ and $Q_U=H_1(ID_U)$, and sends $d_S$ and $d_U$ to the signer and the user, respectively. 
    \item \textit{BlindSig}: Suppose a user wants to sign on message M, four sub-algorithms (Commitment, Authenticating \& Blinding, Signing, Unblinding) must run between the signer and the user as follows:
    \begin{itemize}
        \item \textit{Commitment}: On random chosen integer $s \in \mathbb{Z}_q$, the signer computes $k=e(d_S,\\ rH_2(t)Q_U)$ and $R=rH_2(t)Q_S$, and passes $R$ to the user. 
        \item \textit{Authenticating \& Blinding}: Using his private key, the user computes $K=e(d_U,R)$. If $k \ne K$, the user picks a random number $a \in \mathbb{Z}_q$ as a blinding factor, computes $A=a^{-1}R$, $h=H_3 (M,A)$, $b_M=ah$ and $X=H_4 (b_M,K)$ sends $<b_M,X>$ to the signer. 
        \item \textit{Signing}: The signer computes $X'=H_4 (b_M,k)$ and check if $X'=X$ holds. For valid justification, the signer produces a signature with his private key as $S=(rH_2 (t)+b_M)d_S$ and sends it back to the user.
        \item \textit{Unblinding}: The user unblinds the blinded signature $S$ with blinding factor $a$ as $S'=a^{-1}S$, and publishes signature $<S',A,M>$  on the message $m$. 
    \end{itemize}
    \item \textit{Verify}: On given $<S',A,M>$, verify that $(S',P,A+H_3 (M,A) Q_S,P_{Pub})$ is a valid Gap-Diffie-Hellman tuple, i.e., if DDHP is easy to solve while CDHP is hard, the signature is valid. Using the pairing mapping function, the verifier accepts the signature if and only if
    \begin{equation} \label{eq4.1}
        e(S',P)=e(A+H_3 (M,A) Q_S,P_{Pub})
    \end{equation}
\end{itemize}

\subsection{ID-based Blind Signature-II Scheme}
Our proposed IDBS-II scheme consists of four PPT algorithms: setup, extract, blind signature, and verify. These are defined as follows. 
\begin{itemize}
    \item \textit{Setup}: Given a security parameter $k$, the PKG assumes an additive group $\mathbb{G}_1$ of order $q$, where $q$ is large prime number of $k$-bit and its generator be $P$. PKG defines two hash function $H_1:\{0,1\}^* \times \mathbb{G}_1 \rightarrow \mathbb{Z}_q$ and $H_2:\{0,1\}^* \times \mathbb{G}_1 \rightarrow \mathbb{Z}_q$. The PKG chooses a random element $s \in \mathbb{Z}_q$ (its master key) and computes the public key $P_{Pub}=sP$.  The PKG publishes the public parameter $pp=<\mathbb{G}_1,q,P,P_0,H_1,H_2>$, and keeps $s$ secret. 
    \item \textit{Extract}: For given signer’s identity $ID_S$, $pp$ and its master key $s$, PKG chooses a random number $a\in \mathbb{Z}_q$, computes the signer’s private key $<A=aP,d_{IDS}=a+sQ_S>$, where $Q_S=H_1 (A,ID_S)$ and gives $<A,d_S>$ to the signer. 
    \begin{itemize}
        \item \textit{Blind signature}: The signer and user perform the following steps to obtain a signature on the message. 
        \item \textit{Commitment}: The signer selects two random elements $n_1,n_2 \in \mathbb{Z}_q$, computes $Q_1= n_1P$  and $Q_2= n_2 P$, and sends them to the user. 
        \item \textit{Blinding}: For given received parameters $<Q_1,Q_2>$ and message $m$, the user selects six elements $g,h, i,j,k,l \in \mathbb{Z}q$, such that $\gcd(i,j)=1$ and $ki+lj=\gcd(i,j)$. For the selection of elements $k$ and $l$, we use the Extended Euclidean algorithm. The user then computes the parameters $b{M1}$ and $b_{M2}$ using equations (\ref{eq4.2})-(\ref{eq4.7}) and asks the signer for a signature on $<b_{M1},b_{M2}>$.

        \vspace{-5mm}
        \begin{equation} \label{eq4.2}
            R_1=gQ_1+iP  \hspace{3mm} and \hspace{3mm}  r_1=absc(R_1)modq
        \end{equation}
        \vspace{-5mm}
        \begin{equation} \label{eq4.3}
            R_2=hQ_2+jP  \hspace{3mm} and \hspace{3mm}  r_2=absc(R_2 )modq
        \end{equation}
        \vspace{-5mm}
        \begin{equation}\label{eq4.4}
            r= r_1 r_2 modq
        \end{equation}
        \vspace{-5mm}
        \begin{equation} \label{eq4.5}
            R=R_1+R_2+r(A+Q_{ID}P_0)
        \end{equation}
        \vspace{-5mm}
        \begin{equation} \label{eq4.6}
            b_{M1}=kg^{-1} i(H_2(M,R)-r)modq
        \end{equation}
        \vspace{-5mm}
        \begin{equation} \label{eq4.7}
            b_{M2}=lh^{-1} j(H_2 (M,R)-r)modq                     
        \end{equation}
        \vspace{-5mm}
        \item \textit{Signature}: On received parameters $<b_{M1},b_{M2}>$ and its private key $d_{IDS}$, signer computes $S'_1=(d_{IDS} b_{M1}- n_1)modq$ and $S'_2=(d_{IDS} b_{M2}- n_2)modq$, and sends $<S'_1,S'_2>$ to the user.
        \item \textit{Unblinding}: On received parameters $<S'_1,S'_2>$ using his secret values, i.e., $<g,h,\\ i,j>$, user computes the original signature $\sigma=<S,R>$, where $S=(S_1+S_2)$, $S_1=(S'_1 g-i )$ and $S_2=(S'_2 h-j)$. The user sends the signature pair $<S,R,A>$ to the verifier. 
    \end{itemize}
    \item \textit{Verify}: For given signature pair $<A, S, R>$, the user verifies the valid signature if and only if it satisfied the Eequation $H_2(M,R)(A+H_1 (A,ID_S )P_{Pub})= SP+R$.
\end{itemize}

\section{Security Analysis}
In this section, we provide the security proof of the proposed IDBS-I and IDBS-II schemes. 

\begin{theorem} \label{thm4.1}
(\textbf{Consistency}). The proposed IDBS-I and IDBS-II schemes are consistency.
\end{theorem}

\begin{proof}
The correctness of the proposed IDBS-I scheme is verified as follows. 

Since, $S'=a^{-1}S$, we have 

\vspace{-10mm}

\begin{align*}
    e(S',P)&=e(a^{-1}S,P)\\
&=e(a^{-1}(rH_2 (t)+b_M)d_{IDS},P) \\
&=e(a^{-1} (rH_2 (t)+aH_3 (m,A))Q_{IDS} ,P_{Pub}) \\
&=e(a^{-1} R+H_3 (m,A)Q_{IDS}  ,P_{Pub} ) \\
&=e(A+H_3(m,A)Q_{IDS}  ,P_{Pub})=RHS
\end{align*}

This completes the correctness of the proposed IDBS-I scheme. Now, we prove the consistency of the proposed IDBS-II scheme. 

We have $S=(S_1+S_2)$,  then 
\vspace{-10mm}

\begin{align*}
    SP+R&=(S_1+S_2)P+R
=(S'_1g-i+S'_2h-j)P+R \\
&=((d_{IDS} b_{M1}- n_1 )g-i+(d_{IDS} b_{M2}- n_2 )h-j)P+R \\
&=(d_{IDS} ki(H_2 (M,R)-r)- n_1 g-i+d_{IDS} lj(H_2 (M,R)-r)- n_2 h-j)P+R \\
&=d_{IDS} H_2 (M,R)(ki+lj)P-(ki+lj)rd_{IDS} P-n_1 gP-n_2 hP-iP-jP+R \\
&=d_{IDS} H_2 (M,R)P-(rd_{IDS} P+n_1 gP+n_2 hP+iP+jP)+R \\
&=(a+sQ_{ID} ) H_2 (M,R)P-(r(a+sQ_{ID})P+gQ_1+hQ_2+iP+jP)+R \\
&=(A+Q_{ID} P_{Pub} ) H_2 (M,R)-R+R \\
&=(A+Q_{ID} P_{Pub})H_2 (M,R)
\end{align*}

This proves the consistency of the proposed IDBS-II scheme.
\end{proof}

\begin{theorem} \label{thm4.2}
(\textbf{Un-forgeability}). Suppose $H_1$ and $H_2$ are two random oracle models and a forger $F$ wants to forge a signature on message $M$. Suppose forger $F$ executes at most $q_E$ extract oracles, $q_B$ Blind signature oracles, $q_1$ $H_1$ hash oracles, $q_2$ $H_2$ hash oracles, runs at most $t$ times with advantage at most $k^{-n}$. Under the assumption of ROM and intractable to solve the ECDLP, our proposed IDBS Scheme is existential and Un-forgeable secured against the adaptive chosen message and identity attacks. Forger $F(t,q_1,q_2,q_E,q_B,k^{-n} )$ have the following advantage to break the proposed IDBS-II scheme.
\begin{equation} \label{eq4.8}
    |Pr[F(t,q_1,q_2,q_E,q_B,k^{-n})]| \ge \epsilon(1-q_1/k)^{q_2+q_E}
\end{equation}
   
\end{theorem}

\begin{proof}
Consider a forger $F$ wish to forge any signature in the proposed ID-BS scheme and let there exist an algorithm $B$ which helps $F$. We design an algorithm $B$ that helps $F$ to solve the ECDLP.

\textbf{Setup}: $B$ considers two cryptographic hash functions $H_1:\{0,1\}^* \times \mathbb{G}_1 \rightarrow \mathbb{Z}_q$ and $H_2:\{0,1\}^* \rightarrow \mathbb{Z}_q$ and is accountable to simulate these oracles. $B$ picks $a \in \mathbb{Z}_q$, set $P_0=aP$ and gives public parameter $pp=<\mathbb{G}_1,q,P,P_0,H_1,H_2>$ to $F$.

\textbf{Oracles}: Forger $F$ can perform the following oracles.
\begin{itemize}
    \item \textit{$H_1$  oracle}: $B$ prepares an empty list $H_1^{List}$ having tuple $<R_{Xi},ID_i,H_1 (ID_i,R_{Xi}),*>$. When $F$ queries to $H_1^{List}$ on $<ID_i,R_{Xi}>$, $B$ responds $F$ in the following way.
    \begin{itemize}
        \item $B$ gives $H_1(ID_i,R_{Xi})$ to $F$ and adds the tuple $<R_{Xi},ID_i,H_1(ID_i,R_{Xi}),*>$ to list $H_1^{List}$, if $ID_i=ID^*$.
        \item Otherwise, $B$ chooses randomly $m_i \in \mathbb{Z}_q$ and gives $H_1 (ID_i,R_{Xi})=-m_i$ to $F$ and adds tuple $<R_{Xi},ID_i,H_1 (ID_i,R_{Xi}),m_i>$ to list $H_1^{List}$.
        \item $B$ gives $H_1(ID_i,R_{Xi})$ to $F$, if $<ID_i,R_{Xi}>$ found in the $H_1^{List}$ in the tuple of $<R_i,ID_i,\\ H_1 (ID_i,R_{Xi}),m_i>$ or $<R_{Xi},ID_i,H_1(ID_i,R_{Xi} ),*>$.
    \end{itemize}
	
Note, $H_1(ID,R_X)$ gives no information to $F$ until he queries the $H_1$ oracle on $ID$ as $H_1$ is the random oracle.

    \item \textit{$H_2$ oracle}: On given parameters $R$ and $h$, $B$ runs  $H_2$ oracle and gives the output to the $F$. Suppose $z_1$, $z_2$ and $z_3$ are outputs when $H_2$ oracle executes three times. 
    \item \textit{Key Extract oracle}: $B$ simulate the extract oracles. It pick $n_i \in \mathbb{Z}_q$ and set $R_{Xi}= m_iP_0+n_iP$ in such that $d_{Xi}=m_i$ and $H_1(ID,A_i)=a_i$, and add $<R_{Xi},ID_i,H_1 (ID_i,R_{Xi}),d_{Xi}>$ in list.  We set these parameters such that they satisfy the Equation (\ref{eq4.9}).
    
    \begin{equation} \label{eq4.9}
        d_{Xi} P=R_{Xi}+H_1 (ID_i,R_{Xi})P_0
    \end{equation}
    \item  \textit{Ballot issuing oracle}:  $F$  queries the blind ballot issuing algorithm to get a blind signature on message $M_i$ with identity $ID_i$. Let $F$ give the blinded ballot $<b'_{R1},b'_{R2}>$ to $B$. Then, $B$ responds to the following oracles.
    \begin{itemize}
        \item 	If $ID_i \ne ID^*$, using $ID_i$ corresponding to $H_1^{List}$, $B$ executes the extraction oracles and signs the ballot using the corresponding private key.
        \item If $ID_i=ID^*$, $B$ executes $H_2$ oracles on ballot $R_i$ and picks up the tuple $<R_{Xi},ID_i,\\ H_1 (ID_i,R_{Xi}),d_{Xi}>$ from $H_1^{List}$ to sign ballot $R_i$ and outputs the corresponding signature $<\sigma_i^*=(s_i^*,R^*,h^*),R_X^*>$ to $F$.
    \end{itemize}

\end{itemize}

\textbf{Forgery}. $F$ responds a signature $<
\sigma_i^*=(s_i^*,R^*,h^* ), R_X^*>$ pair against identity $ID^*$. In order to forge a signature, suppose the forger $F$ creates three distinct signatures $(\sigma_1^*,\sigma_2^*,\sigma_3^*)$ on same message $M$, where $<\sigma_1^*=(s_1^*,R^*,h_V^*),R_X^*>$,    $<\sigma_2^*=(s_2^*,R^*,h_V^*),R_X^*>$ and $<
\sigma_3^*=(s_3^*,R^*,h_V^* ),R_X^*>$.
We consider $s_0,r_X$  and $u$ as the discrete logarithm of $P_0,R_X$  and $R$, respectively, i.e., $P_0=s_0 P$, $R_X=r_X P$ and $R=uP$. From $sP=H_2(R,h_V )(R_X+H_1 (ID_X,R_X )P_0 )-R$, we get $s_i^*=(z_i (r_X+s_0 H_1 (ID,R_X ))-u )modq$ for $1 \le i \le 3$. The parameters $s_0,r_X$  and $u$ from these equations are unknown to $F$. Thus, $F$ solves three linear equations to obtain these values, which is equivalent to solving the ECDL problem. 

\textbf{Analysis}. The probability that $B$ does not abort the game is denoted as $Pr[\neg E_1 \wedge \neg E_2]$.

\begin{itemize}
    \item $E_1$: The extract oracle fails if $H_1$ oracle outputs the inconsistent outputs with probability at most $q_1/k$. The simulation is completed in  $q_E$ time, which happens with probability at least $(1-q_1/k)^{q_E}$. 
    \item $E_2$: The execution of $H_2$ oracle fails if $H_2$ oracle gives the inconsistent outputs with probability at most $q_1/k$. The simulation is completed in $q_2$ time, which happens with probability at least $(1-q_1/k)^{q_2}$.
\end{itemize}

From these two events, we obtain the probability that Adv can break the scheme with

\centering
$Pr[\neg E_1 \wedge \neg E_2] \ge \epsilon (1-\frac{q_1}{k})^{q_2+q_E}$

\end{proof}

\begin{theorem} \label{thm4.3}
(Blindness). The proposed IDBS-II scheme achieves the blindness property.
\end{theorem}

\begin{proof}
Suppose an adversary $Adv$ that plays the role of ECA and challenger $Ch$ that plays a role of an honest voter runs the blind signature algorithm. Suppose Adv obtains the parameters $<b_{R1},b_{R2},s'_1,s'_2>$ by executing the ballot issuing oracle. Let the corresponding signature be $<R,h_V,s>$. There exists a tuple of values $<a,b,c,d,e,f>$ that links $<b_{R1},b_{R2},s'_1,s'_2>$ to $<R,h_V,s>$.

From $b_{R1}=ea^{-1}c(H_2 (R,h_V )-r)$ and $b_{R2}=fb^{-1}d(H_2 (R,h_V )-r)$, we get $a=eb_{R1}^{-1} c(H_2 (R,h)-r)$ and $b=fb_{R2}^{-1} d(H_2 (R,h_V )-r)$, respectively. Similarly, from $s_1=(s'_1a-c )$ and $s_2=(s'_2b-d )$, we get $c=(s'_1 a-s_1)$ and $d=(s'_2b-s_2)$, respectively. By substituting these values in $R=aA_1+bA_2+(c+d)P+r(R_E+Q_E P_0)$, it can be noted that $Adv$ must know the value of $r$ to compute $R$. The value $r$ depends on $<a,b,c,d>$ whose production is equivalent to solving the ECDL problem. Thus, the proposed E2E-VIV system provides voter anonymity.
\end{proof}

\section{Construction of End-to-End Verifiable Internet Voting System}
Here, we construct the E2E-IVS system using the IDBS-II scheme, which consists of six algorithms. In Algorithm (\ref{alg4.1}), the KGC initializes the system, and computes its master key $s_0$ and public parameter param. It keeps $s_0$ secret and publishes $pp$. 

\begin{algorithm} 
	\caption{System initialization } 
	\begin{algorithmic}[1]
	    \State Given a security parameter $k$, choose three groups $\mathbb{G}_1$,$\mathbb{G}_2$ and $\mathbb{G}_T$ of order $q$ ($k$-bit), a generator $P$ of $\mathbb{G}_1$, a generator $Q$ of $\mathbb{G}_2$, and bilinear map $e:\mathbb{G}_1 \times \mathbb{G}_2 \rightarrow \mathbb{G}_T$.
	    \State Choose five hash functions, $H_1:\{0,1\}^* \times \mathbb{G}_1 \rightarrow \mathbb{Z}_q$, $H_2:\{0,1\}^* \times \mathbb{G}_1 \rightarrow  \mathbb{Z}_q$, $H_3:\{0,1\}^* \rightarrow \mathbb{Z}_q$, $H_4:\{0,1\}^* \rightarrow \mathbb{G}_2$ and $H_5:\mathbb{Z}_q  \times \mathbb{G}_1 \times \mathbb{Z}_q  \times \mathbb{G}_1  \times \mathbb{G}_2 \times \{0,1\}^* \times \mathbb{G}_1 \rightarrow \mathbb{Z}_q$.
	    \State Choose an element $s_0 \in \mathbb{Z}_q$ (its master key) and compute the public key $P_{Pub}=sP$.
	    \State Publishes the public parameter  $pp=<q,e,P,P_{Pub},\mathbb{G}_1,\mathbb{G}_2,\mathbb{G}_T,H_1,H_2,H_3,H_4,H_5>$, and keeps $s_0$ secret.
    \end{algorithmic} 
    \label{alg4.1}
\end{algorithm}

\begin{algorithm} 
	\caption{Registration and key generation} 
	\begin{algorithmic}[1]
	    \State KGC prepares two lists: $List^{RV}$ and $List^{RC}$.
	    \State On given identity $ID_E$. The KGC picks a random element $r_E \in \mathbb{Z}_q$, and computes the private key $<d_E, R_E>$, where $R_E=r_EP$, $Q_E=H_1(ID_E ||R_E)$ and $d_E=r_E+sQ_E$, and sends it to the to ECA over a secure channel.
	    \State Candidate scans his biometric data and computes $h_{3C}=H_3(binfo_C ||ID_C)$, using his  identity $ID_C$.
	    \State 	Given $h_{3C}$  and $ID_C$, candidate asks KGC for its private key. The KGC picks a random element $r_C \in \mathbb{Z}_q$, and computes a private key $<d_C,R_C>$, where $R_C=r_CP$, $Q_C=H_1(ID_C||R_C)$ and $d_C=r_C+sQ_C$, and sends it to the candidate and adds an entry $h_{3C}$ to $List^{RC}$.
	    \State Similarly to step 4, the voter obtains his private key $<d_V, R_V>$ from KGC, where KGC maintains list $List^{RV}$ for registered voters.
	    \State	KGC sends lists ($List^{RV}$ and $List^{RC}$) to ECA and candidate and $List^{RC}$ to voter.
    \end{algorithmic} 
    \label{alg4.2}
\end{algorithm}

In Algorithm (\ref{alg4.2}), KGC registers every entity (voter, ECA and candidate) against their $IDs$ and biometric information $h_3$. It issues a private key $(d_x, R_x)$ to the entity over a secure channel and adds an entry of entity’s $h_3$ in the respective list: $List^{RV}$ or $List^{RC}$. The KGC outputs both lists to ECA and $List^{RC}$ to voter. The private key $(d_x, R_x)$ can be verified using Equation (\ref{eq4.10}), where x denotes the entity.
\vspace{-10mm}

\begin{equation} \label{eq4.10}
    d_xP=H_1(ID_x ||R_x ) P_{Pub}+R_x
\end{equation}

\begin{algorithm} 
	\caption{Authentication and Ballot distribution} 
	\begin{algorithmic}[1]
	    \State ECA chooses two secret integers $a_1$, $a_2 \in \mathbb{Z}_q$, computes $<A_1=a_1 P,A_2=a_1 P,R_E>$ and gives them to voter.
    	\State \textit{Ballot request}: Voter performs the following steps:
    	\State \hspace{5mm}	Select six elements $a,b,c,d,e,f \in \mathbb{Z}_q$ such that $gcd(c,d)=1$ and $ec+fd=gcd(c,d)$.
    	\State \hspace{5mm}Compute a pseudonym public key $R$ as follows. 
    	\State \hspace{10mm}$R_1=aA_1+cP$  and   $r_1=absc(R_1 )mod q$
    	\State \hspace{10mm} $R_2=bA_2+dP$ and   $r_2=absc(R_2 )  mod q$
    	\State \hspace{10mm} $r= r_1 r_2 mod q $
    	\State \hspace{10mm} $R= R_1 \mathbb{Z}_q+R_2+r(R_E+Q_E P_0)$
    	\State \hspace{10mm} $h_{3V}=H_3(binfo_V,ID_V)$
    	\State \hspace{10mm} $h_V=ah_3V+c$
    	\State \hspace{10mm} $b_{R1}=ea^{-1}c(H_2 (R,h_V )-r)mod q$
    	\State \hspace{10mm} $b_{R2}=fb^{-1} d(H_2 (R,h_V )-r)mod q$
    	\State \hspace{10mm} $U_1=(ab_{R1}+bb_{R2})P$
    	\State \hspace{10mm} $U_2=rP$
    	\State \hspace{5mm}	Send  $<b_{R1},b_{R2},R,h_V,h_3V,U_1,U_2>$ to ECA for requesting a blank ballot. 
    	\State \textit{Authentication}: ECA authenticates the voter using three conditions. 
    	\State \hspace{10mm} $U_1==(H_2 (R,h_V )P-U_2)$
    	\State \hspace{10mm} $h_{3V} \in List^{RV}$
    	\State \hspace{10mm} $h_{3V} \notin List^{VV}$
	    \State \textit{Ballot issuing}: If the above conditions are satisfied, ECA performs the following steps. 
	    \State \hspace{5mm} Compute $s'_1=(d_E b_{R1}-a_1 )mod q$ and $s'_2=(d_E b_{R2}- a_2)mod q$
    	\State \hspace{5mm} Issue a blind ballot $<s'_1,s'_2>$ to voter and add $h_{3V}$ to $List^{VV}$.
    	\State \textit{Unblinding}: Voter unblinds the blind ballot and computes the blank ballot $BB$ as follow.
    	\State \hspace{5mm}	Compute $s_1=(s'_1a-c )mod q$, $s_2=(s'_2b-d)  modq$ and $s_{12}=(s_1+s_2)mod q$.
    	\State The ballot is $BB = <s_{12},R,h_V>$.
    \end{algorithmic} 
    \label{alg4.3}
\end{algorithm}

\begin{algorithm} 
	\caption{Vote Casting } 
	\begin{algorithmic}[1]
	    \State The voter chooses $C_{name} \in List^{RC}$, timestamp $T$ and preforms the following steps.
	    \State Compute $V_1=aH_4(C_{name}||T)$ and $V_2=a(R_E+Q_E P_0)$.
	    \State Set ballot $B=<BB,V_1,V_2,C_{name},T>$ and sends it to ECA.
	    \State ECA validates the vote ballot $B$ using (\ref{eq4.11}) and (\ref{eq4.12}).
	    \State{
	        \begin{equation} \label{eq4.11}
	            d_E H_2 (R,h_V )P \overset{?}{=} s_12 P+R
	        \end{equation}
	        \begin{equation} \label{eq4.12}
	            e(d_E P,V_1 ) \overset{?}{=}e(V_2,H_4 (C_name||T))
	        \end{equation}
	    }
                                           
	    \If     (\ref{eq4.11}) and (\ref{eq4.12}) are valid,
	    \State ECA picks a random element $n_c \in \mathbb{Z}_q$, and sets receipt as $P_{n_c}=n_c P$, $Rcpt=H_5(B||P_{n_c})$ and $s_{Rcpt}=n_c+d_ERcpt$ and adds $<B,P_{n_c},Rcpt,s_{Rcpt}>$ in $List^B$.
	    \State Sends receipt  $<P_{n_c},Rcpt,s_{Rcpt}>$ to  voter
	    \EndIf
    \end{algorithmic} 
    \label{alg4.4}
\end{algorithm}

In Algorithm (\ref{alg4.3}), the voter requests for the digital ballot in which ECA authenticates the voter by checking his biometric parameters and makes a list $List^{VV}$ for valid voters to whom a ballot is issued. If the valid voter has not received any ballot previously, the ECA issues a fresh blind blank ballot to him without knowing any information about his identity. Finally, the voter extracts the blank ballot. Algorithm (\ref{alg4.4}) defines the vote-casting process in which a voter picks a candidate that he supports, computes the signature on his choice of support using the short signature scheme \cite{boneh2001short} and attaches it with the blank ballot to produce the electronic ballot B. The voter sends the electronic ballot to ECA which verifies the ballot and vote and acknowledges the receipt to the voter. ECA prepares a list $List^B$ where it stores the ballot’s information $<B,P_{n_c},S_{Rcpt},Rcpt>$. 

\begin{table}
        \centering
        \caption{Filtrating duplicate ballots and separating them into valid and invalid ballot lists.}
        \label{tbl4.1}
        \begin{tabular}{|c|c|c|c|c|c|c|}
            \hline
             $R$ &	$s_{12}$&	$h_V$ & C$^a$ &	$V_1$	& $V_2$	& Actions \\
            \hline
            \hline
            
            0$^@$&	0&	0&	0&	$\times$$^\#$&	$\times$&	add one ballot in $List^{VB}$ and reject other \\
            0&	0&	0&	1$^{@@}$&	$\times$&	$\times$&	add one ballot in $List^{VB}$ and other in $List^{IB}$\\
            $\times$&	1&	$\times$&	$\times$&	$\times$&	$\times$&	both ballot are add in $List^{IB}$ \\
            1&	$\times$&	$\times$&	$\times$&	$\times$&	$\times$&	both ballot are add in $List^{IB}$ \\
            $\times$&	$\times$&	1&	$\times$&	$\times$&	$\times$&	both ballot are add in $List^{IB}$ \\
            1&	1&	1&	$\times$&	$\times$&	$\times$&	both ballot are add in $List^{VB}$ \\
            \hline
    \end{tabular}
    \footnotesize{\\$^@$if the corresponding parameter in two ballots has the same values,  $^{@@}$if the corresponding parameter in two ballots has the different value, $^a$candidate vote and $^\#$denotes don’t care.}
\end{table}

\begin{algorithm} 
	
	\begin{algorithmic}[1]
	    \State ECA prepares two list: $List^{VB}$ and $List^{IB}$
	    \State Pick two tuples
	    \State 
	    {
	   \hspace{10mm} $B_i=<s_{12i},R_i,h_{Vi},V_{1i},V_{2i},C_{name_i}>$ and\\
	   \hspace{10mm} $B_j=<s_{12j},R_j,h_{Vj},V_{1j},V_{2j},C_{name_j}> \in List^B$
	    }
	    \State Using Table \ref{tbl4.1}, ECA filters the two ballots into two lists $List^{VB}$ and $List^{IB}$. 
	    \State Publishes $List^{VB}$ and $List^{IB}$. 
    \end{algorithmic} 
    \caption{Vote Counting} 
    \label{alg4.5}
\end{algorithm}

Algorithm (\ref{alg4.5}) counts the valid ballot from the pre-maintained list $List^B$. Suppose two ballots in list $List^B$ are $B_i=< R_i,s_{12i},h_{Vi},C_{name_i},V_{1i},V_{2i}>$  and $B_j=<R_j,s_{12j},h_{Vj},C_{name_j},V_{1j},V_{2j}>$. The ECA creates lists $List^{VB}$, which contains the valid ballot details, and $List^{IB}$, which contains the invalid ballot details. The ECA filters the duplicate ballot and separates it into $List^{VB}$ or $List^{IB}$. The ECA performs $B_i \oplus B_j$, where $\oplus$ denotes the exclusive-or operation that outputs the set of bits for each ballot’s parameters. The resultant bits match with the tuple of Table \ref{tbl4.1} that decide which corresponding action has to be performed with the two ballots. In the end, the ECA publishes both lists: $List^{VB}$ and $List^{IB}$. Algorithm (\ref{alg4.6}) enables a voter to check if his vote is present in the ballot list. If a voter finds that his receipt is present in $List^{VB}$ and, using Equation (\ref{eq4.13}) he checks the ballot that he intended. Also, anyone can find that a receipt is recorded in $List^{VB}$  and can check if the vote is recorded as the voter intended using Equation (\ref{eq4.14}). Further, anyone can count ballots from the published list $List^{VB}$ and can verify that all ballots are correctly counted. 

\begin{algorithm} 
	\caption{Auditing} 
	\begin{algorithmic}[1]
	    \State	On given $<s_Rcpt, P_{n_c}>$, voter finds his ballot in lists $List^{VB}$ and $List^{IB}$
        \State Voter checks whether his ballot is correctly recorded as he intended, using Equation (\ref{eq4.13}).
        \State{
        \begin{equation} \label{eq4.13}
            Q_V s_{Rcpt}P=Q_EH_5(B,P_{n_c})(Q_E^{-1}Q_V R_E+s_V P-R_V )+Q_V P_{n_c}
        \end{equation}
        }
        \State Any entity, say $X$, can check using Equation (\ref{eq4.14}) whether the ballot is correctly recorded as the voter intended without knowing the actual vote.
        \State{ 
        \begin{equation} \label{eq4.14}
            Q_X s_{Rcpt}P=Q_EH_5(B,P_{n_c})(Q_E^{-1}Q_X R_E+s_X P-R_X)+Q_XP_{n_c}
        \end{equation}
        }
    \end{algorithmic} 
    \label{alg4.6}
\end{algorithm}

\section{Security Analysis of E2E-VIV System}
Here, we carry out the security analysis of our proposed E2E-VIV systems.
\begin{itemize}
    \item \textit{Voter anonymity}. During and after the election process, the voter’s identity must be hidden from other entities in the system. In the ballot distribution process, Theorem 4.2 ensures the voter’s anonymity. For requesting a ballot and casting a vote, a voter generates the pseudonym public-private key pair ($R,<a,b,c,d,e,f>$) instead of using his private key ($d_V, R_V$). In the vote casting phase, the ECA verifies the blank ballot BB that corresponds to the pseudonym public key $R$ without knowing the information about the voter’s real identity. In order to protect the voter’s biometric information, we use the pre-image collision-resistant hash function $H_3$ that encrypts the biometric information. So, the use of pseudonym public and private key pair, pre-image resistant hash function $H_3$ and Theorem 4.2 prove that our system preserves the voter’s anonymity. 
    \item \textit{Resistant to Bribery and Coercion}. After the election, in our proposed system, the ECA issues a receipt <$P_{n_c}, Rcpt,s_{Rcpt}$> to each voter, which is only the witness to the voter to prove that he has cast his vote to the intended one. In our proposed system, the ECA adds a random number $n_c$ to each receipt to ensure the freshness of the ballot. Any curious voter can show his choice of vote to others only if he guesses the correct value of $n_c$ which is equivalent to computing the ECDLP. Thus, our proposed scheme is resistant to bribery and coercion attacks. 
    \item \textit{Voter’s eligibility/Authentication}. The use of a functional digital signature authenticates a voter against his $ID_V$ and biometric details $h_{3V}$ and stores $h_{3V}$ in list $List^{RV}$ during the registration process. In blank ballot distribution, the ECA maintains a list $List^{VV}$ for those valid voters who have been already issued a ballot. From a given blank ballot $BB=<b_{R1},b_{R2},R,h_V,h_{3V},\\ U_1,U_2>$, the ECA authenticates the voter only if $U_1==(H_2(R,h_V)P-U_2)$, and $h_{3V}$ is in $List^{RV}$, not in $List^{VV}$. The consistency of the first condition is proved as follows. From LHS,  
    \begin{align*}
        U_1&=(ab_{R1}+bb_{R2})P
            =(ec(H_2 (R,h_V)-r)+fd(H_2 (R,h_V )-r))P\\
&=(H_2 (R,h_V )-r)(ec+fd)P
=H_2 (R,h_V )P-rP
=H_2 (R,h_V )P-U_2
    \end{align*}
    The above equalities show the consistency of the first conditions. 
    \item \textit{Vote Integrity}. During the ballot-issuing process, the ECA issues a blank ballot to a voter whose integrity has been proved in Theorem 4.2. During the vote-casting phase, the integrity of the vote is ensured by the short signature scheme whose security is equivalent to solving the GDH problem \cite{boneh2001short}. Thus, the proposed system achieves vote integrity during and after the election process.
    \item \textit{Uniqueness and Fairness (ballot stuffing)}. In the proposed system, a voter assists a pair of pseudonym public and pseudonym private keys for requesting a blank ballot from the ECA. Before issuing a ballot to a voter, the ECA prepares a list of valid voters to whom the ballots have already been issued. In order to authenticate a voter against his voter’s pseudonym public key $R$ and biometric information $h_{3V}$, ECA checks if $U_1==(H_2 (R,h_{3V})P-U_2)$ holds and $h_{3V}$ is in lists $List^{RV}$, not in $List^{VV}$. Now, the ECA issues a fresh blank ballot $BB$ to the voter and adds $h_{3V}$ in $List^{VV}$. Thus, a voter obtains only one blank ballot and can cast only one vote. Therefore, the proposed system is resistant to ballot stuffing.
    \textit{Verifiability}. After the election, the ECA publishes the list of invalid ballots $List^{IB}$ and valid ballot list $List^{VB}$. From given $<P_{n_c},Rcpt,S_{Rcpt}>$, an eligible voter can check if $<P_{n_c},Rcpt,S_{Rcpt}>$ belongs to $List^{IB}$ or $List^{VB}$. The voter ensures that his vote is counted in the election process if the parameters $<P_{n_c}, Rcpt, S_{Rcpt}>$ satisfy Equation (\ref{eq4.12}). The consistency of Equation (\ref{eq4.13}) is similar to the consistency of the IDBS-II scheme, which is proved in Theorem 1. 
\end{itemize}

This proves the consistency of Equation (\ref{eq4.12}). Similarly, anyone in the system can also check participation of receipt using Equation (\ref{eq4.13}). Since a voter is allowed to check the ballot’s status, it can bring the coercion problem where the voter can check if his vote was correctly counted or not because the voter can show the value of the vote to anyone. In order to overcome this inconsistency, the proposed system publishes two lists: invalid ballot and valid ballot lists, so that he can check to which list his ballot belongs. In order to count the valid voter, anyone can verify that only valid votes were counted. Thus, the ECA proves that all receipts of the votes have been counted corresponding to the valid ballots that were previously verified and stored. Therefore, the proposed system achieves individual verifiability and universal verifiability. Since the proposed system uses the fingerprint and iris sensor-based smartphone as a voting device, the voter can cast a vote at any place. The use of smartphones with sensor technologies makes the E2E-VIV system convenient and user-friendly for non-technical users. The performance efficiency shows that the proposed system is efficient as compared to other systems/schemes. Also, the system provides an early voting facility for those voters who are unable to access the network during the election due to any reason.

\section{Performance Analysis}
This section provides a comparison of our E2E-VIV system with the related electronic voting systems and examines their performance. 

\subsection{Experiment Simulation}

\begin{table} 
        \centering
        \caption{Number of machine cycles required by cryptographic operations.}
        \label{tbl4.2}
        \begin{tabular}{|c|c|c|}
            \hline
             Notation &	Operations	& \# Cycles (in Kilo) \\
            \hline
            \hline
\#Pr    &RSA public key &	766 \\
\#Sk	&RSA private key	& 51,805 \\
\#Exp	&DSA exponentiation &	17,334 \\
\#Sm	&ECC scalar point multiplication &	300 \\
\#Pm	&Pairing multiplication 	& 450 \\
\#MTP	&Map-To-Point hash function	& 310 \\
\#Bp	&Bilinear pairing	& 2,100 \\
        \hline
    \end{tabular}
\end{table}

Suppose the Weil pairing is defined over the Type-F curve (BN curve) of PBC library \cite{lynn2010pairing} with 256-bit of one group and 512 bit of another group, and embedding degree is 12, whose security level is identical to 3072-bit RSA. For 128 bit of AES security of BN curve $E/\mathbb{F}_P:y^2=x^3+b$, where $b \ne 0$, gives the group size $|\mathbb{Z}_q |$ = 256 bit, the field size of $\mathbb{G}_1$, i.e., $|\mathbb{G}_1 |$  = 160 bit, the field size of $\mathbb{G}_2$, i.e., $|\mathbb{G}_2 |$  = 512 bit, and the field size of $\mathbb{G}_T$, i.e., $|\mathbb{G}_T|$ = 3072 bit. To compare our scheme with the non-ECC based schemes \cite{chen2014secure, chung2009approach, li2009verifiable,wu2014electronic}, we consider two large prime numbers of 3072 bit, i.e., $|p_1|=|p_2|$=3072 bit, such that their product is hard to factorize and consider a large prime p of 1024 bit and a prime factor q of 160 bits, i.e., $|p|$ = 1024 bit, and $|q|$ =160 bit, such that $(p-1)|q$. Further, we assume that $|C_{name}|=|ID|=|T|$ = 80 bit.

\subsection{Implementation and Benchmark}

Here, we compare the performance of our proposed internet-voting system with related schemes \cite{lopez2014pairing}, \cite{yang2018secure}, \cite{chen2014secure}, \cite{chung2009approach}, \cite{li2009verifiable}, and \cite{wu2014electronic}, in terms of machine cycles obtained by simulating experiment on \textit{Intel(R) Core(TM) i7-2600K CPU @ 3.4 GHz, and 8 GB of RAM}, using \textit{gcc 4.6}. We consider the same methodology as discussed in \cite{yang2018secure} to estimate the machine cycles consumed by operations (for example, private and public operations based on RSA, operations on elliptic curve and pairing). Table \ref{tbl4.1} summarizes the notations and number of machine cycles consumed by the required cryptographic operations.

Table 4.2 shows the performance of our system with the existing systems \cite{lopez2014pairing}, \cite{yang2018secure}, \cite{chen2014secure}, \cite{chung2009approach}, \cite{li2009verifiable}, and \cite{wu2014electronic}. The proposed system and  system \cite{lopez2014pairing} use elliptic curve cryptography, and their security is based on ECDL and GDH problems. The systems \cite{yang2018secure}, \cite{chen2014secure}, \cite{chung2009approach}, \cite{li2009verifiable}, and \cite{wu2014electronic} are based on the traditional RSA public key cryptosystem and their security is based on the discrete logarithm problem (DLP) and integer factorization problems (IFP). It can be seen from Table \ref{tbl4.2} that the ECC-based operations (scalar multiplication and addition) are efficient than the RSA based operations.
\begin{table} 
        \centering
        \caption{Efficiency comparison of our scheme with related schemes}
        \label{tbl4.3}
        \begin{tabular}{|c|c|c|c|c|c|}
            \hline
            Scheme & \#Cycles &	Blank Ballot  & Vote Ballot & Cryptographic  &	Security  \\
             & & Size (in Bytes) &  size (in Bytes) &  primitives &	 Assumption \\
            \hline
            \hline
            \cite{lopez2014pairing} & 20,940,500 &	128B &	202B &	Short signature &	GDHP \\
             & & & & Blind signature &	DDHP \\
            \cite{yang2018secure} & 173,340,000 & 1152B & 2098B & HE$^*$ & DDHP \\
            & & & & ZKP$^\#$ &	DDHP \\
            \cite{chen2014secure} & 364,935,431 & 1152B &	2304B &	Digital signature &	IFP$^a$ \\
			 & & & & Blind signature &	IFP \\
            \cite{chung2009approach} & 211,053,653 & 768B &	1930B &	Digital signature	& IFP \\
            & & & & Blind signature &	IFP \\
            \cite{li2009verifiable} &77,745,803	& 1152B &	1152B &	Digital signature &	IFP \\
            & & & & 			Blind signature  & 	IFP \\
            \cite{wu2014electronic} & 160,015,290 &	188B &	296B &	Digital signature &	IFP \\
            & & & & 				Mix Net &	DLP$^@$ \\
            Our & 	9,620,000 &	96B &	212B &	Short signature &	GDHP \\
            & & & & 				IDBS &	ECDLP \\

        \hline
    \end{tabular}
    \footnotesize{$^a$integer factorization problem, $^@$discrete logarithm problem, $^\#$Zero Knowledge proof, $^*$Homomorphic Encryption}
\end{table}

Table 4.2 shows that our proposed system needs 16*\#Sm + 2*\#Bp + 2*\#MTP = 9,620  Kcycles, whereas the schemes \cite{lopez2014pairing}, \cite{yang2018secure}, \cite{chen2014secure}, \cite{chung2009approach}, \cite{li2009verifiable}, and \cite{wu2014electronic} need 20,940, 173,340, 364,935, 211,053, 477,746, and 160,015 Kcycles, respectively. The computation cost of our system is much better, i.e., it takes 46\%, 5\%, 2\%, 4\%, 2\% and 6\% of the computation costs of the schemes  \cite{lopez2014pairing}, \cite{yang2018secure}, \cite{chen2014secure}, \cite{chung2009approach}, \cite{li2009verifiable}, and \cite{wu2014electronic}, respectively, shown in Figure (\ref{fig4.1}). For a blank ballot, our system needs 3*32 = 96 bytes, whereas the schemes \cite{lopez2014pairing}, \cite{yang2018secure}, \cite{chen2014secure}, \cite{chung2009approach}, \cite{li2009verifiable}, and \cite{wu2014electronic} need 128, 1152, 1152, 768, 1152, and 188 bytes, respectively.  For a voting ballot, our system needs 4*32+1*64 + (80)/8 +(80)/8 =212 bytes, whereas the schemes \cite{lopez2014pairing}, \cite{yang2018secure}, \cite{chen2014secure}, \cite{chung2009approach}, \cite{li2009verifiable}, and \cite{wu2014electronic} need 202, 2098, 2304, 1930, 1152, and 296  bytes, respectively. From Figure (\ref{fig4.2}), we observe that the ballot bandwidth in our scheme is much better, i.e., it takes 75\%, 8\%, 8\%, 12\%, 8\% and 51\% of ballot bandwidth size of the schemes \cite{lopez2014pairing}, \cite{yang2018secure}, \cite{chen2014secure}, \cite{chung2009approach}, \cite{li2009verifiable}, and \cite{wu2014electronic} respectively, and the vote-ballot bandwidth in our scheme is 105\%, 10\%, 9\%, 10\%, 18\% and 68\% of the vote-ballot bandwidth of  schemes \cite{lopez2014pairing}, \cite{yang2018secure}, \cite{chen2014secure}, \cite{chung2009approach}, \cite{li2009verifiable}, and \cite{wu2014electronic}, respectively.

\begin{figure}
  \centering
  \includegraphics[width=0.8\linewidth]{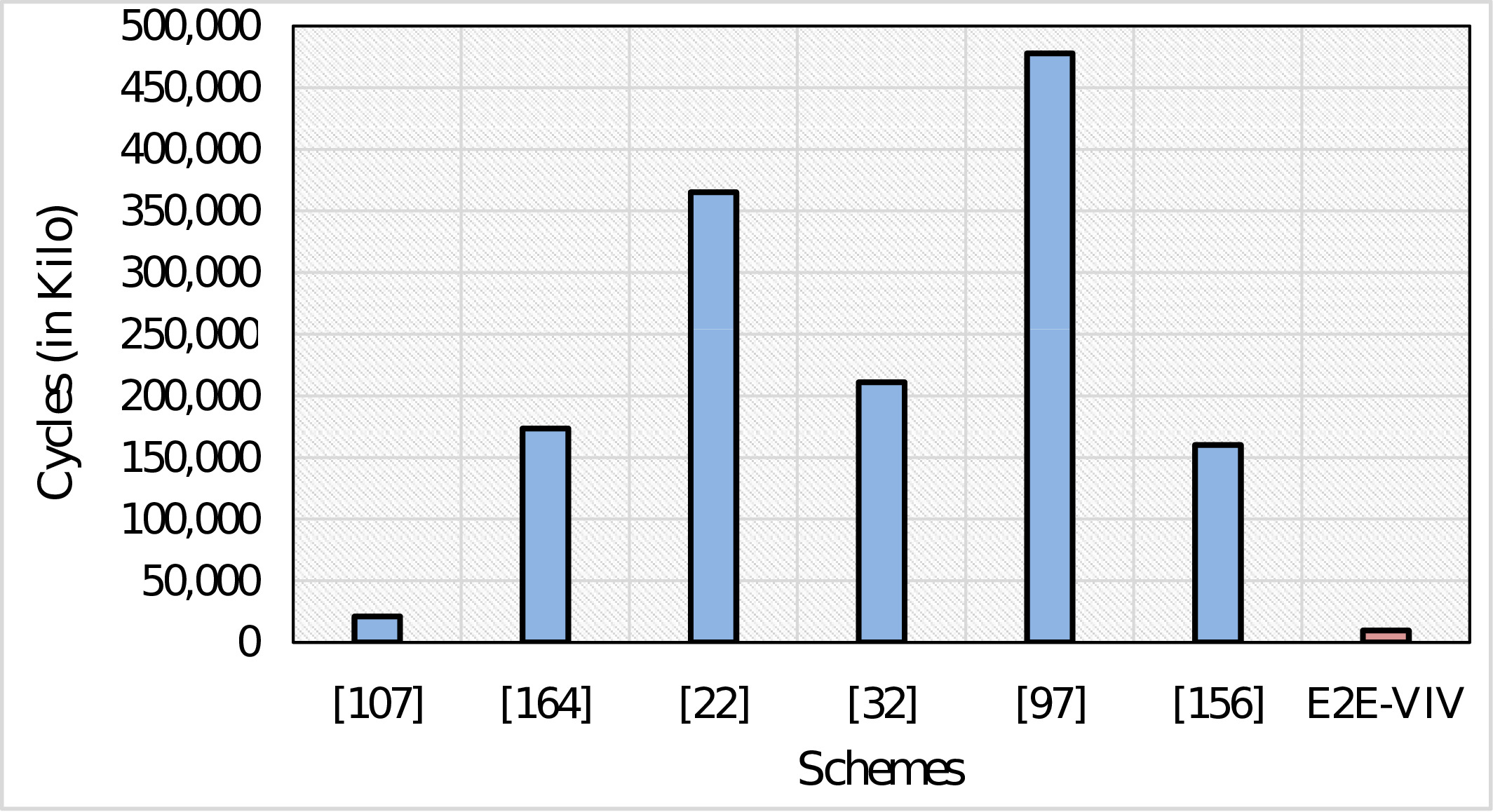}
  \caption{Computation costs of proposed E2E-VIV and other systems}
\label{fig4.1}
\end{figure}

\begin{figure}
  \centering
  \includegraphics[width=0.8\linewidth]{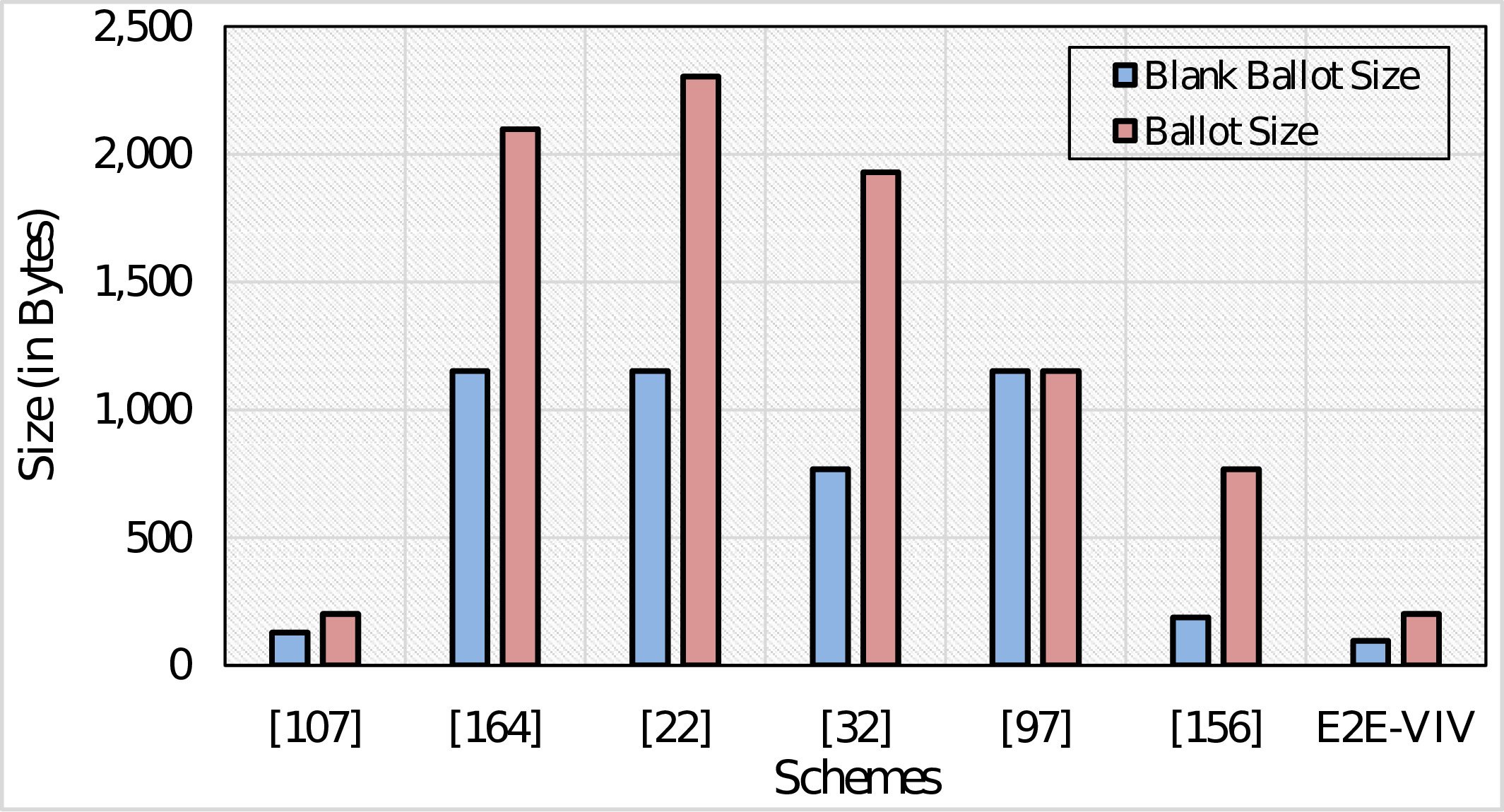}
  \caption{Ballot and vote-ballot size of proposed E2E-VIV and other systems}
\label{fig4.2}
\end{figure}

Now, we compared our E2E-VIV system with related voting systems, such as \cite{lopez2014pairing}, \cite{yang2018secure}, \cite{chen2014secure}, \cite{li2019multi} \cite{chung2009approach}, \cite{li2009verifiable}, and \cite{wu2014electronic}, in terms of following parameters, where $C_1$: voter’s anonymity, $C_2$: resilient to coercion and bribery, $C_3$: Eligibility, $C_4$: Integrity, $C_5$: Fairness, $C_6$: Mobility, $C_7$: Efficiency, $C_8$: Uniqueness, $C_9$: Accuracy, $C_{10}$: Convenience, $C_{11}$: Verifiability, $C_{12}$: early voting, $C_{13}$: Batch verifiability, $C_{14}$: E2E-VIV, as summarized in Table \ref{tbl4.3}. As evident from this table, our E2E-VIV system provides all security features, but no existing scheme satisfies all of them. Thus, the E2E-VIV system is suitable for implementing a voting system in an internet environment with E2E verifiable feature, shown in Table in \ref{tbl4.4}.

\begin{table}
        \centering
        \caption{Security features of our System and existing Schemes}
        \label{tbl4.4}
        \begin{tabular}{|c|c|c|c|c|c|c|c|c|c|c|c|c|c|c|}
            \hline
            Scheme & $C_1$ & $C_2$& $C_3$& $C_4$& $C_5$& $C_6$& $C_7$& $C_8$& $C_9$& $C_{10}$& $C_{11}$& $C_{12}$& $C_{13}$& $C_{14}$ \\
            \hline
            \hline
\cite{lopez2014pairing} & $\checkmark$ & $\checkmark$& $\checkmark$& $\checkmark$& $\checkmark$ & $\times$ & Avg & $\checkmark$& $\checkmark$ & Avg & $\checkmark$& $\times$ & $\times$& $\times$ \\
\cite{yang2018secure} & $\checkmark$ & $\checkmark$& $\checkmark$& $\checkmark$& -- & $\checkmark$ & Avg & $\checkmark$ & $\checkmark$& Avg & $\checkmark$& $\times$ & $\times$& $\checkmark$ \\
\cite{chen2014secure} & $\checkmark$ & $\checkmark$& $\checkmark$& $\checkmark$ & $\checkmark$& $\checkmark$ & Avg & $\checkmark$& $\checkmark$ & high & $\checkmark$& $\times$ & $\times$& $\times$ \\
\cite{li2019multi} & $\checkmark$ & -- & $\checkmark$& $\checkmark$& -- & $\checkmark$ & High & $\checkmark$ & $\checkmark$& High & $\checkmark$& $\times$ & $\times$& $\times$ \\
\cite{chung2009approach} & $\checkmark$ & $\checkmark$& $\checkmark$& $\checkmark$ & $\checkmark$& $\checkmark$& Low & $\checkmark$& $\checkmark$ & High & $\checkmark$& $\times$ & $\times$& $\times$ \\
\cite{li2009verifiable} & $\times$ & $\checkmark$ & $\checkmark$& $\checkmark$& $\checkmark$ & $\checkmark$ & Low & -- & $\checkmark$ & Avg & $\checkmark$& $\times$ & $\times$& $\times$ \\
\cite{wu2014electronic} & $\checkmark$ & $\checkmark$& $\checkmark$& $\checkmark$ & $\checkmark$& $\checkmark$& Avg & $\checkmark$ & $\checkmark$& Avg & $\checkmark$& $\times$ & $\times$& $\times$ \\
Our  & $\checkmark$ & $\checkmark$& $\checkmark$ & $\checkmark$& $\checkmark$ & $\checkmark$& High & $\checkmark$ & $\checkmark$& High & $\checkmark$ & $\checkmark$& $\checkmark$& $\checkmark$ \\

        \hline
    \end{tabular}
\end{table}

\subsection{Batch Verification}
This section illustrates the performance of systems, in terms of computation cost on verifying the blank ballot and vote ballot simultaneously. For convenience, we suppose 20 ballots as a batch. Figure (\ref{fig4.3}) compares the machine cycles (in KCycle) consumed for verifying the batch of 20 valid blank ballots simultaneously, 20 valid/invalid blank ballots simultaneously and 20 valid/invalid blank ballots individually. Similarly, Figure (\ref{fig4.4}) compares the machine cycles (in KCycle) consumed for verifying the batch of 20 valid vote ballots simultaneously, 20 valid/invalid vote ballots simultaneously and 20 valid/invalid vote ballots individually.
Suppose ECA obtains  $BB_i= <s_i,R_i,h_Vi>$ from $i^{th}$ voter, $1 \le i \le n$. The ECA verifies n valid blank ballots using Equation (\ref{eq4.15}). 

\vspace{-10mm}
\begin{equation} \label{eq4.15}
    \sum_{i=0}^nH_2(R_i,h_{Vi})d_EP= \sum_{i=0}^ns_iP+\sum_{i=0}^nR_i
\end{equation}

Similarly, the ECA verifies multiple ballots using Equation (\ref{eq4.16}).
\vspace{-5mm}

\begin{equation} \label{eq4.16}
    e(\sum_{i=0}^nV_{1i} ,d_EP) = e(\sum_{i=0}^nH_4(C_{name_i}),\sum_{i=0}^nV_{2i})
\end{equation}

The consistency of  (\ref{eq4.15}) is verified as follows. From RHS of (\ref{eq4.16}), 

\vspace{-15mm}

\begin{align*}
    \sum_{i=0}^n(s_i)P+\sum_{i=0}^nR_i &= \sum_{i=0}^n(s_1i+s_2i)P+\sum_{i=0}^nR_i \\
&=\sum_{i=0}^n(s'_1i a_i-c_i+s'_2i b_i-d_i)P+\sum_{i=0}^nR_i \\
&=\sum_{i=0}^n(d_E H_2 (R_i,h_{Vi})(e_i c_i+f_i d_i)-(e_i c_i+f_i d_i)r_id_E- a_1a_i-c_i- a_2b_i-d_i)P+\sum_{i=0}^nR_i \\
&=\sum_{i=0}^nd_E H_2 (R_i,h_{Vi})P-\sum_{i=0}^n(r_i (R_E+Q_E P_0)+ R_{1i}+R_{2i})+\sum_{i=0}^nR_i \\
&=\sum_{i=0}^nd_E H_2 (R_i,h_{Vi})P-\sum_{i=0}^n(R_i)+\sum_{i=0}^nR_i \\
&=\sum_{i=0}^n(H_2 (R_i,h_{Vi} ))d_E P\\
\end{align*}

\vspace{-15mm}

This proves the consistency of (\ref{eq4.15}). Similarly, the consistency of (\ref{eq4.16}) is verified as follows. From LHS of Equation (\ref{eq4.16}), we have

\vspace{-15mm}

\begin{align*}
e(\sum_{i=0}^nV_{i1},d_E P)&=\sum_{i=0}^na_iH_4(C_{name_i},d_E P)\\
&= e(\sum_{i=0}^nH_4(C_{name_i}),\sum_{i=0}^na_i(R_E+Q_EP_0 ) ) \\
&= e(\sum_{i=0}^nH_4(C_{name_i}),\sum_{i=0}^nV_{2i} )
\end{align*}
\vspace{-15mm}

This proves the consistency of (\ref{eq4.16}).

\begin{figure}
  \centering
  \includegraphics[width=0.8\linewidth]{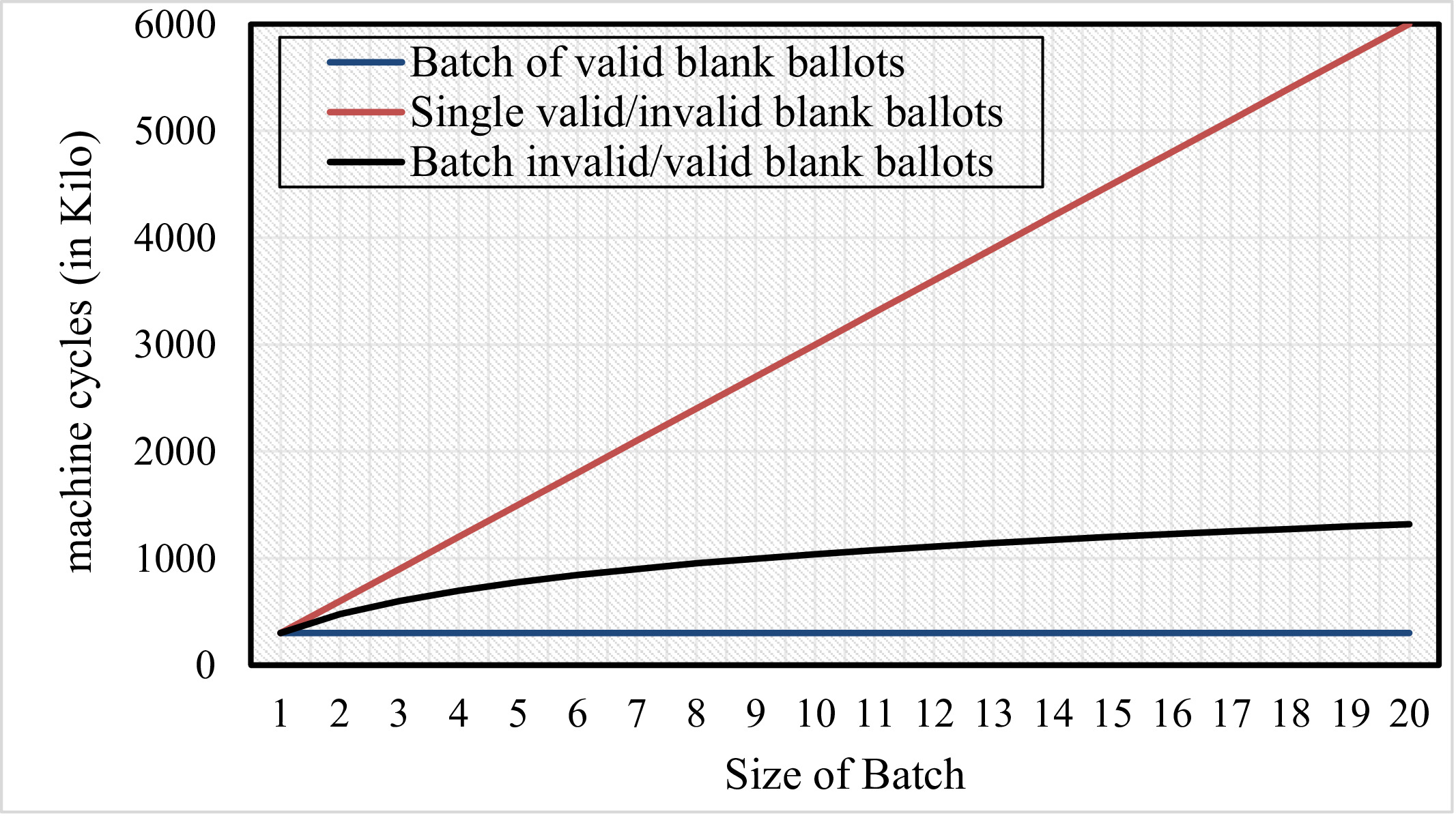}
  \caption{Computational cost of verifying batch of blank ballots of different size}
\label{fig4.3}
\end{figure}

\begin{figure}
  \centering
  \includegraphics[width=0.8\linewidth]{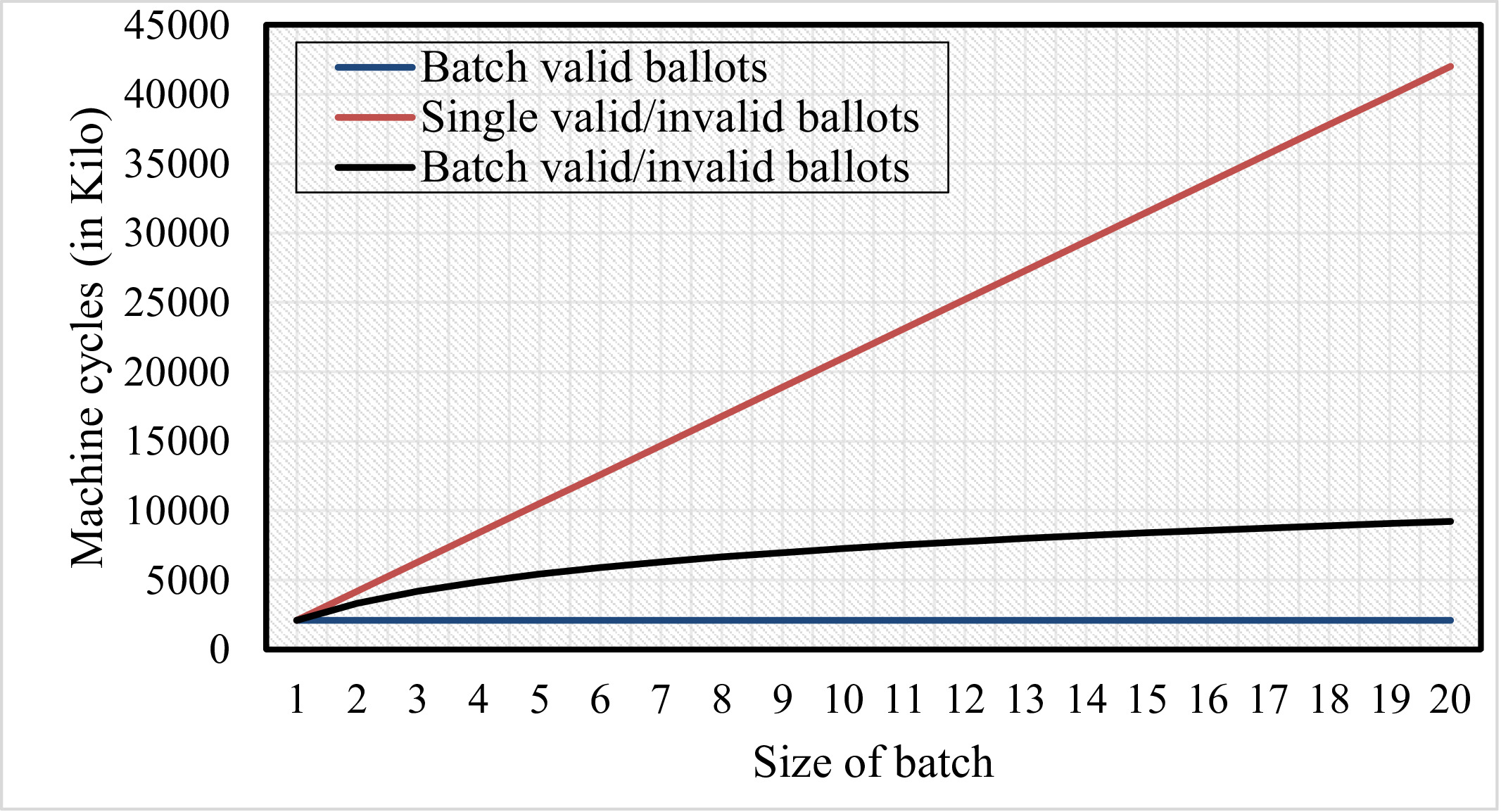}
  \caption{Computational cost of verifying batch of ballots of different size}
\label{fig4.4}
\end{figure}

\textit{Analysis}: The Equations (\ref{eq4.15}) and (\ref{eq4.16}) verify the multiple blank ballots and ballots, respectively, if all blank ballots and ballots are valid, i.e., they individually verify Equations (\ref{eq4.15}) and \ref{eq4.16}), respectively. If any ballot is invalid or modified, the Equations (\ref{eq4.15}) and (\ref{eq4.16}) will not be satisfied. For the batch of valid and invalid ballots, the proposed system uses the divide and conquer mythology to optimize the computation cost. The proposed system divides $n$ ballots into two $n/2$ ballots, and checks Equations (\ref{eq4.15}) and (\ref{eq4.16}). If the ballots satisfy these equations, stop the process; otherwise, the process is repeated until the size of the ballot becomes one. Thus, the proposed system has $Olog(n)$ time complexity for the batch of $n$ invalid and valid ballots, the complexities for n valid ballots and single ballot are $O(1)$ and $O(n)$, respectively, as shown in Table \ref{tbl4.5}.

\begin{table}
        \centering
        \caption{Batch verification analysis of blank ballots and ballots}
        \label{tbl4.5}
        \begin{tabular}{|c|c|c|c|c|}
          \hline
         \multicolumn{1}{|c|}{Verification} & \multicolumn{2}{|c|}{Blank ballots} & \multicolumn{2}{|c|}{Blank ballots} \\
            \cline{2-5}
            &Computation  &	Complexity &	Computation	& Complexity \\
            & cost &	 &	 cost	&  \\
            \hline
            \hline
            Single &	2n*\#Sm &	$O(n)$ &	2n*\#Bp &	$O(n)$ \\
            n (valid) &	2*\#Sm	& $O(1)$	& 2*\#Bp &	$O(1)$ \\
            n (invalid \&valid) &	$Log(n)$ & $2*\#Sm$	$OLog(n)$ &	$OLog(n)*2*\#Bp$	& $OLog(n)$\\
            \hline
    \end{tabular}
\end{table}

\section{Summary}
This chapter presents two Identity-Based Blind Signature (IDBS) schemes: the pairing-free IDBS-I and the pairing-friendly IDBS-II. Both proposed systems are secure against the existential forgery attack under chosen message and $ID$. An End-to-End Verifiable Internet Voting (E2E-VIV) system is also designed in this chapter. The E2E-VIV system authenticates each voter using a unique identifier issued by the appropriate authority and their biometric information. The proposed IDBS-II scheme is used to issue a blank ballot to the voter, and the BLS short signature scheme is used to protect the vote from any modification. The proposed IDBS schemes are compared with related existing schemes and have the least machine cycles. The security analysis shows that the proposed schemes are secure against existential forgery attacks under chosen message and $ID$. The E2E-VIV system provides a secure and verifiable voting process. The use of unique identifiers and biometric information ensures that each vote is cast by an eligible voter, and the IDBS-II scheme ensures that each ballot is issued anonymously. The BLS short signature scheme protects the vote from any modification.

\end{doublespace} \label{chapter4}
\begin{savequote}[75mm] 
Reality is created by the mind, we can change our reality by changing our minds.
\qauthor{Plato- Greek philosopher, pedagogue and mathematician} 
\end{savequote}

\chapter{ID-Based Blind Signature with Message Recovery Scheme for Privacy-Preserving in Cloud Storage}
\justify
\begin{doublespace}
Data outsourcing, which allows users and organizations to exploit external distributed servers such as the cloud, has become a notable breakthrough in big data technology. In cloud-enabled systems, a cloud service provider (CSP) manages, maintains, and makes data available to data owners over the Internet. This enables users to store and process their data remotely, providing them with significant benefits such as cost reduction, scalability, and flexibility. However, data outsourcing also presents significant privacy and security challenges. Once the data owner shifts the data to the untrusted cloud, they lose direct control over the data, and the data becomes susceptible to alteration by the CSP or any outsider attacker.

One of the primary concerns of data owners is the privacy of their stored data. The CSP may have access to the data, and they may not have adequate security measures in place to protect the data from unauthorized access. This could lead to data breaches, data theft, and unauthorized access, which can have severe consequences for the data owner. Moreover, the CSP or any outsider attacker could also alter the data, leading to various security breaches such as data tampering, data loss, and data corruption. These attacks could be challenging to detect, and the data owner may not be aware of the alterations until it is too late. Therefore, data owners need to take appropriate measures to ensure the privacy and security of their data when outsourcing it to the cloud. They can use various security mechanisms such as encryption, access control, and secure communication protocols to protect their data from unauthorized access and alteration. Additionally, regular security audits and risk assessments can help identify potential vulnerabilities and threats, enabling data owners to take proactive measures to mitigate the risks.

In this chapter, we propose an identity-based blind signature with a message recovery (IDBS-MR) scheme, which is secured against existential forgery attacks under the adaptive chosen message and ID attacks (EF-ID-CMA) and ECDL problems. Further, we present a \textit{privacy-preserving data outsourcing mechanism} that audits the integrity of stored data for resource-limited devices in cloud computing using the proposed IDBS-MR scheme. In the proposed privacy-preserving scheme, the data owner verifies the blinded signature received from the TPA and the TPA audits data integrity stored on the cloud. The performance analysis shows that our scheme is efficient as compared to related existing schemes.

In the last decades, there have been discussed various primitives to achieve integrity auditing for outsourced data on the cloud—we categorise such primitives into three distinct verification approaches. The first approach is to authenticate the data structure; generally, Merkel hash tree (MHT), to achieve secure search results \cite{bertino2004selective}. The basic idea is to produce an index for data files based on MHT, and the integrity of results can be achieved by re-generating the signature on the root of MHT. Unfortunately, it is expensive in terms of computation and communication. Besides, they do not ensure the consistency of the results. Another approach is the probabilistic integrity auditing approach that allows the user to insert some bogus block in the cloud \cite{wang2010privacy}. The disadvantage of such methods as they do not support common database operations, e.g., projection. The last approach is to use of signature-verification method \cite{pang2009scalable, mykletun2006authentication} that saves computation costs during integrity auditing.

In 2004, Deswarte \textit{et al.} \cite{deswarte2003remote} discussed the concept of integrity auditing and gave the first integrity auditing scheme based on RSA. However, efficiency was an issue that is addressed by Filho \textit{et al.} \cite{gazzoni2006demonstrating}, followed by they gave an improved integrity auditing scheme. Yamamoto \textit{et al.} \cite{yamamoto2007fast} utilized homomorphic hash function and batch processing in order to audit multiple data blocks. Ateniese \textit{et al.} \cite{ateniese2007provable}, in 2007, presented a provable data possession (PDP) scheme that enables a user to verify whether the server possesses data. In the same year, Juels \textit{et al.} \cite{juels2007pors} extended the PDP scheme and presented a proof of retrievability (POR) scheme that ensures the user checks for stored data possession and retrievability in the remote server. Although both schemes \cite{ateniese2007provable, juels2007pors} have high computation costs on the user side, they are not suitable for implementing secure cloud storage. For cloud storage, Wang \textit{et al.} \cite{wang2010privacy} discussed an integrity auditing scheme by employing a third-party auditor (TPA), in which the auditing process is delegated to the TPA.

In the schemes mentioned earlier, the user generates the signature for the data, and TPA audits the data integrity, which can be optimized for the user with a resource-limited device. Recently, Liu \textit{et al.} \cite{liu2018privacy} addressed this issue and presented a privacy-preserving integrity auditing for lightweight devices using a blind signature scheme. We found that such schemes have expensive operations, such as bilinear pairing on the TPA side. Besides, they need certificate maintenance costs since they are based on traditional public key infrastructure (PKI). Using the IBC technique, Zhang \textit{et al.} \cite{zhang2002id} were the first to propose the IDBS schemes. Later, many IDBS schemes have been presented, such as \cite{gao2012round, kumar2017new, islam2015design, galindo2006generic, he2011efficient, kumar2017secure, dong2014efficient, islam2016provably, mao2006linkability, tian2009security, kumar2017untraceable}. In 2007, Han \textit{et al.} \cite{han2005pairing} presented an identity-based blind signature with a message recovery scheme. However, this scheme \cite{han2005pairing} is implemented using pairing on elliptic curves. Since then, very few IDBS-MR schemes have been offered, given in the literature \cite{han2005pairing,elkamchouchi2008new, james2017identity, james2018pairing, diao2013new}. Recently, Verma \textit{et al.} \cite{verma2017efficient} present a new IDBS-MR scheme using pairing whose security is based on the assumption of ROM and the solving k-CAA problem. 

From the above-discussed schemes, it has been noticed that the existing IDBS-MR schemes are designed on pairing on elliptic curves. Hence, due to the high computation cost of the pairing operation, they could not be suitable for the cloud-storage data. Here, we present a new identity-based blind signature scheme with message recovery without bilinear pairing. We first discuss the system architecture followed by its implementation.

\section{System Model}

Here, we discuss the system architecture of privacy-preserving with integrity auditing for cloud storage. 

\subsection{Network Architecture}
The network architecture for privacy-preserving data outsourcing with integrity consists of four entities: private key generator (PKG), third-party auditor (TPA), user, and cloud, as shown in Figure (\ref{fig5.1}).

\begin{figure}
  \centering
  \includegraphics[width=1\linewidth]{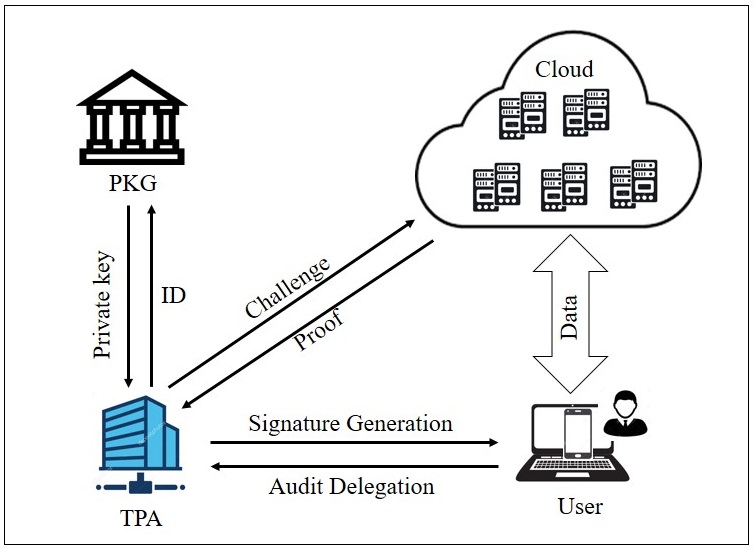}
  \caption{Architecture of proposed Privacy-preserving data outsourcing with integrity auditing}
\label{fig5.1}
\end{figure}

\begin{enumerate}
    \item \textit{Third party Auditor}: TPA is a semi-trusted third-party auditor, which can generate a signature on blinded data blocks requested by the user and can audit the stored cloud data on behalf of the user.
    \item \textit{Private key generator}: PKG generates a private key corresponding to the TPA’s identity and issues it to TPA. \item \textit{Cloud}: Several distributed servers combine to form an untrusted cloud, which is managed and maintained by CSP. CSP may reject the request made by the user or provide an incomplete query to the user in order to save cloud computation and storage. Moreover, the cloud generates storage proof against the audit challenge challenged by the TPA.
    \item \textit{User}: The users are the data owners having resource-constrained devices, for example, smartphones and tablets. In comparing with TPA and CSP, the user has lower storage and computation power. The user wants to use cloud services in a pay-per-use manner.
\end{enumerate}
	
\subsection{Design Goals}

Here, we give the design goals that must achieve the following security requirements in order to enable the privacy-preserving with integrity auditing for cloud data storage:
\begin{itemize}
    \item \textit{Signature delegation and verification}. The user delegates the signing process to TPA since the user has a resource-limited device.
    \item \textit{Integrity auditing}. It allows TPA to validate the consistency of stored data on the cloud on-demand without knowing any information about the original message. 
    \item \textit{Privacy-preserving}. It ensures that no entity can drive users’ data stored on the cloud from the information collected during integrity auditing. 
    \item \textit{Batch verification}. Since the data file is divided into many blocks, which are signed by TPA, it is costly on the user side to verify each signature individually. It enables the user to verify multiple signatures simultaneously. 
    \item \textit{Lightweight}. It ensures that the signature verification on the user side and integrity auditing on the TPA side should be performed with minimum communication as well as computation cost.
\end{itemize}

\subsection{Identity-based Blind Signature with Message Recovery} 
This section gives a brief introduction to identity-based blind signature with message recovery that helps in signature delegation without leaking any content to TSP.

\begin{definition}
(\textit{Identity-Based Blind Signature with data recovery scheme}). The IDBS-MR scheme involves four entities: signer, user, verifier and private key generator (PKG). In the proposed scheme, the signer acts as a TPA, the user is a data owner and PKG generates the private key for TPA. It consists of four randomized probabilistic polynomial-time (PPT) algorithms: setup, extract, blind signature, and verification, as discussed below.
\begin{itemize}
    \item Setup: Using security parameters, the PKG computes the master key and public parameters. PKG keeps the master key and publishes public parameters.
    \item \textit{Extract}: For a given signer’s identity, the PKG computes the private key using its master key and public parameter.
    \item \textit{Blind signature}: The signer and user perform the following steps to sign on a given message. 
    \begin{itemize}
        \item \textit{Commitment}: For the secret number, the signer computes public parameters, passes them to the user and keeps the secret number to him. 
        \item \textit{Blinding}: On a given public parameter and message M, the user blinds the message using random secret numbers. The user then requests the signer for the signature on the blinded message.
        \item \textit{Signature}: For each blinded message, the signer computes the blind signature using his private key and outputs it to the user.
        \item \textit{Unblinding}: The user retrieves the blinded signature using his secret key and outputs the original signature.
    \end{itemize}
    \item \textit{Verification}: On a given blind signature, Verifier recovers the message and then verifies the signature.
\end{itemize}

\end{definition}

\subsection{Security Threat}

To discuss the security of our proposed IDBS-MR schemes, we go through the definition of Zhang \textit{et al.}’s identity-based blind signature scheme [10] and Tso \textit{et al.}’s identity-based signature with message recovery scheme [39] where they discussed the unlinkability and unforgeability. Thus, the proposed IDBS-MR scheme is considered to be secure if it is secured against existential forgeable attack under the chosen message and ID attack (EF-ID-CMA) and achieves the blindness property.

\begin{definition}
Definition 5.2. (EF-ID-CMA). We discuss the unforgeability of our proposed IDBS-MR scheme through the following game playing between a forger $F$ that acts as a malicious user and challenger $Ch$ that acts as the honest signer under adaptive chosen message and identity attack in the random oracle model (ROM).

\textbf{Setup}: The challenger $Ch$ runs the setup algorithm and computes the master key and public parameters. The $Ch$ responds public parameter to $F$. 

\textbf{Oracle}: $F$ performs the following oracles. 
	\textit{Extract oracle}: For given $ID$, forger $F$ requests to run the extract algorithm. The $Ch$ runs this oracle to compute the private key corresponding to an identity $ID_i$, where $1 \le i \le q_k$ and sends it to $F$. Besides, $Ch$ saves the record in list $L_{ext}$, which is initially empty. \textit{Blind signature oracle}: For a message $M_i\in \{0,1\}^{l2}$ of its choice in an adaptive manner, forger $F$ asks blind signature oracle to obtain the blind signature $\sigma$. The $Ch$ executes the Blind signature oracle and responds to the result to $F$ and saves it in the list $L_{BS}$, which is initially empty.
	
\textbf{Forgery}: At the end, the forger $F$ responds a signature $\sigma^*$ on given message $M^*$ with signer’s identity $ID^*$. The forger $F$ will win the game if it fulfils the following conditions.
\begin{itemize}
    \item $\sigma^*$ is the valid signature against $M^*$ and $ID^*$.
    \item The blind signature oracle has not been queried on $<M^*>$. 
	\item The extract oracle has not been queried on $ID^*$.
\end{itemize}
	
Under chosen message and identity attacks, the proposed IDBS-MR scheme is said to be existentially unforgeable, if any forger $F$ has a negligible probability to succeed in the above game. 
\end{definition}

\begin{definition}
\textit{\textbf{Blindness}}. The blindness security notion can be defined by the adversary $Adv$ that acts as a malicious signer and is engaged with two users $U_0$ and $U_1$ in the following game:

\textit{Setup}: It computes the master key and public parameters and responds to public parameters to $Adv$.
For given $ID$, forger $Adv$ requests to run the extract algorithm. This oracle gives the private key corresponding to an identity $ID_i$, where $1 \le i \le q_k$ and sends it to $Adv$. 
For selective $b \in \{0,1\}$, the $U_0$ and $U_1$ get two distinct message $M_b$ and $M_{1-b}$, respectively. 

$U_0$ and $U_1$ compute $\sigma_b$ (signature on $M_b$) and $\sigma_{1-b}$ (signature on $M_{1-b}$), respectively, and give it to $Adv$. 

At the end, $Adv$ predicts a bit $b'\in \{0,1\}$ and wins the game if $b = b'$ holds with advantage $|Pr[b=b']|\ge 1/2+k^{-n}$. The proposed IDBS-MR scheme is blind if $Adv$ wins the above game with a negligible advantage.
\end{definition}

\section{Proposed ID-Based Blind Signature Scheme with Message Recovery}

This section discusses the implementation of the IDBS-MR scheme, which consists of four PPT algorithms: Setup, Key Generation, Blind signature, and Verification, as defined below. 

\begin{figure}
  \centering
  \includegraphics[width=1\linewidth]{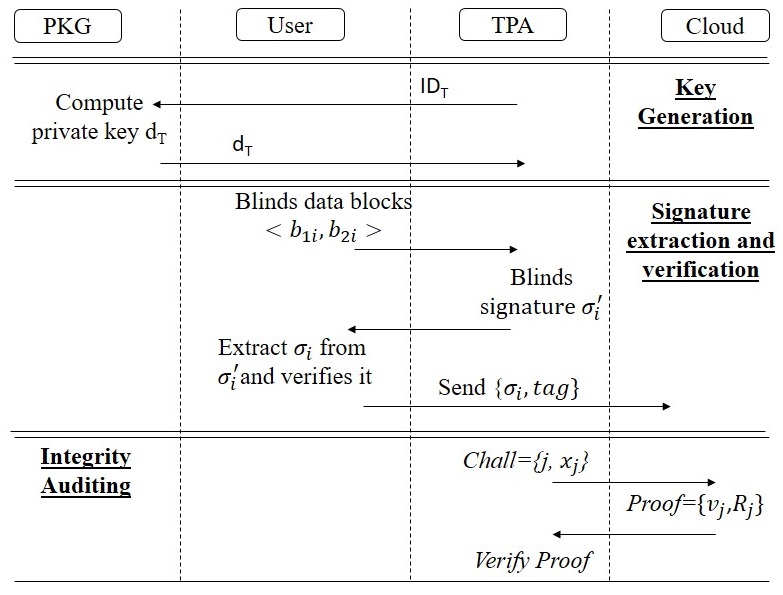}
  \caption{Flow process of proposed privacy-preserving with integrity auditing }
\label{fig5.2}
\end{figure}

\begin{enumerate}
    \item \textit{Setup}: Given a security parameter $k$, the PKG assumes an additive group $\mathbb{G}_1$ of order $q$, where $q$ is large prime number of $k$-bit and $P$ be its generator. Suppose three hash functions $H_1:\{0,1\}^* \times \mathbb{G}_1 \rightarrow \mathbb{Z}_q$, $H_2:\mathbb{G}_1 \rightarrow \mathbb{Z}_q$ and $H_3:\mathbb{G}_1 \rightarrow \{0,1\}^{|q|}$ and two functions $F_1:\{0,1\}^{l_2}  \rightarrow \{0,1\}^{l_1}$ and $F_2: \{0,1\}^{l_1}  \rightarrow \{0,1\}^{l2}$, where $l1$ and $l2$ are two positive integers such that $l_1+l_2=|q|$. Suppose $absc(P)$ gives the $x$-coordinate of point $P$. The PKG chooses a random element $s_0 \in \mathbb{Z}_q$ (its master key) and computes the public key $P_0 = s_0P$. The PKG publishes the public parameter $pp=<\mathbb{G}_1,q,P,P_0,H_1,H_2,F_1,F_2,l_1, l_2,k>$, and keeps $s_0$ secret. 
    \item \textit{Key Generation}: For given signer’s identity $ID_S$, $pp$ and its master key $s_0$, PKG chooses a random number $a\in \mathbb{Z}_q$, computes the signer’s private key $d_{IDS}=a+s_0Q_{IDs}$, where $Q_{IDs}=H_1(A,ID_S)$ and $A=aP$ and gives $<A,d_{IDS}>$ to the signer. 
    \item \textit{Blind signature}: The signer and user perform the following steps to obtain a blind signature on the message.
    \begin{itemize}
        \item \textit{Commitment}: The signer selects two random elements $n_1,n_2 \in \mathbb{Z}_q$, computes $Q_1= n_1 P$  and $Q_2= n_2 P$, and sends them with its public parameter $A$ to the user. 
        \item \textit{Blinding}: For the given received parameters $<Q_1,Q_2,A>$ and message $m \in \{0,1\}^{l2}$, the user selects six elements $g,h,i,j,k,l \in \mathbb{Z}_q$, such that $gcd(i,j)=1$ and $ki+lj=gcd(i,j)$. We use the extended Euclidean algorithm to select elements $k$ and $l$. The user then computes the parameters $b_{M1}$ and $b_{M2}$ given in (\ref{eq5.1})-(\ref{eq5.7}) and asks the signer for a signature on $<b_{M1},b_{M2}>$.
 
        \vspace{-5mm}
        \begin{equation} \label{eq5.1}
            R_1=gQ_1+iP  \hspace{5mm} and \hspace{5mm}   r_1=absc(R_1 )
        \end{equation}
        
        \vspace{-12mm}
        \begin{equation} \label{eq5.2}
            R_2=hQ_2+jP \hspace{5mm} and \hspace{5mm}  r_2=absc(R_2 )
        \end{equation}
        
        \vspace{-12mm}
        \begin{equation} \label{eq5.3}
            r= r_1 r_2
        \end{equation}
        
        \vspace{-12mm}
                \begin{equation} \label{eq5.4}
            u=F_1 (m)||F_2 (F_1 (m)) \oplus m 
        \end{equation} 
        
        \vspace{-12mm}
                \begin{equation} \label{eq5.5}
            R=R_1+R_2+ru(A+Q_{ID} P_0)
        \end{equation}
        
        \vspace{-12mm}
        \begin{equation} \label{eq5.6}
              b_{M1}=kg^{-1} i(H_2(R)-ru)mod q
        \end{equation}
        
        \vspace{-12mm}
        \begin{equation} \label{eq5.7}
            b_{M2}=lh^{-1}j(H_2(R)-ru)mod q
        \end{equation}
        
        \item \textit{Signature}: On received parameters $<b_{M1},b_{M2}>$ and its private key $d_{IDS}$, signer computes the blind signature $<s'_1,s'_2>$, where $s'_1=(d_{IDS} b_M1- n_1 )$ and $s'_2=(d_{IDS} b_{M2}- n_2 )$, and sends $<s'_1,s'_2>$ to the user. 
        \item \textit{Unblinding}: On received parameters $<s'_1,s'_2>$, user sets $s_1=(s'_1g-i )$ and $s_2=(s'_2h-j)$ and computes the original signature $<v,R>$, where $s=(s_1+s_2)$ and  $v=u \oplus H_3 (sP+R)$. The user sends the signature $<v,R,A>$ to the verifier. 
    \end{itemize}
    \item \textit{Verification}: For given signature pair $<v,R,A>$, the user computes u and recovers M', given in Equations (\ref{eq5.8})-(\ref{eq5.9}). 
    \vspace{-5mm}
    \begin{equation} \label{eq5.8}
        u=v\oplus H_3 (H_2 (R)(A+H_1 (A,ID_S ) P_0 )) 
    \end{equation}
    
    \vspace{-12mm}
    \begin{equation} \label{eq5.9}
        m'=F_2 ({_{l_1}}^l|u| ) \oplus|u|_{l_2}^R 
    \end{equation}
\end{enumerate}
	
The user accepts the signature and message $m'$ if and only if ${_{l_1}}^l|u|=F_1 (m')$. This completes the implementation of the proposed IDBS-MR scheme.

\section{Data Outsourcing with Integrity Auditing on Cloud}

Here, we discuss the implementation of privacy-preserving with integrity auditing on the cloud using the proposed IDBS-MR scheme. 

Since the user (data owner) has the limited-resources in terms of computation power and storage space, the signature generation is delegated to TPA in order to reduce the computation cost on the user side. The user blinds the data blocks and requests TPA to sign on his blinded data. TPA generates a blind signature on the blinded data blocks without knowing the original content and responds to it to the user. The correctness of the signature determines the user’s judgment of data integrity in the cloud, which is essential for CSP and the user. Hence, the user receives the corresponding blind signature and verifies its correctness and outsources it to the cloud. Since CSP may respond to an incomplete or invalid output against the user query in order to save its computation overhead, the user needs to check the integrity of the returned output. Therefore, the user requests TPA for integrity auditing of given blocks stored in the cloud. TPA generates an auditing challenge for given blocks and sends it to the cloud for auditing proof generation. Then, the cloud computes the auditing proof and sends it to TPA for integrity auditing. Figure (\ref{fig5.2}). shows the complete process of the proposed scheme.
 
The detail of the proposed scheme is given below. 

\begin{enumerate}
    \item 	\textbf{Systems Setup}
    
    Given a security parameter $k$, the PKG assumes an additive group $\mathbb{G}_1$ of order $q$, where $q$ is a large prime number of $k$-bit and $P$ be its generator. Suppose three hash functions $H_1:\{0,1\}^* \times \mathbb{G}_1 \rightarrow \mathbb{Z}_q$, $H_2:\mathbb{G}_1 \rightarrow \mathbb{Z}_q$, $H_3:\mathbb{G}_1 \rightarrow \{0,1\}^{|q|}$   and $H_4:\{0,1\}^* \rightarrow \mathbb{G}_1$ and two functions $F_1: \{0,1\}^{l_2}  \rightarrow \{0,1\}^{l_1}$ and $F_2:\{0,1\}^{l_1} \rightarrow \{0,1\}^{l_2}$, where $l_1$ and $l_2$ are two positive integers such that $l_1+l_2=|q|$. Suppose $absc(P)$ gives the $x$-coordinate of point $P$. The PKG chooses a random element $s_0 \in \mathbb{Z}_q$ (its master key) and computes the public key $P_0=s_0 P$. The PKG publishes the public parameter $pp=<\mathbb{G}_1,q,P,P_0,H_1,H_2,F_1,F_2,l_1, l_2,k>$, and keeps $s_0$ secret. 
    Suppose the data file of user is denoted as $DF={d_i}$, where $1 \le i \le n$ with corresponding data index are denoted as $I=\{I_1,I_2,..,I_n \}$.
    \item \textbf{Key Generation}
    Given $ID_T$ (TPA’s identity), $PP$ and master secret key $s_0$, PKG generates private key $d_T=a_T+s_0 Q_T$, where $Q_T=H_1 (A_T,ID_T)$ and $A_T=a_TP$ on random chosen integer $a_T \in \mathbb{Z}_q$,  and passes $<A_T,d_T>$ to the TPA. Similarly, PKG computes private key $<A_U,d_U>$ and gives it to the user with identity $ID_U$, where, $d_U=a_U+s_0Q_U$, $Q_U=H_1 (A_U,ID_U)$ and $A_U=a_UP$ on random chosen integer $a_U \in \mathbb{Z}_q$.  
    \item \textbf{Blinded signature delegation} 
    The following phases executed between TPA and user produce a blind signature on data block $d_i \in \{0,1\}^{l2}$, shown in Figure (\ref{fig5.2}). 
    The TPA chooses elements $n_1,n_2 \in \mathbb{Z}_q$, computes $Q_1= n_1 P$  and $Q_2= n_2 P$, and sends $<Q_1,Q_2,A>$ to the user. 
    The user receives the parameters $<Q_1,Q_2,A>$ and picks integers $g,h,i,j,k,l \in \mathbb{Z}_q$, such that $gcd(i,j)=1$ and $ki+lj=gcd(i,j)$. Given data block $d_i$ and file identity $F_{ID}$, the user first hides it and generates the parameters $b_{1i}$ and $b_{2i}$ given in Eq. (5.11)-(5.17) and requests TPA to sign on it.
    \vspace{-5mm}
    \begin{equation} \label{eq5.10}
        R_1=gQ_1+iP  \hspace{5mm} and  \hspace{5mm} r_1=absc(R_1 )
    \end{equation}
    
    \vspace{-12mm}
    \begin{equation} \label{eq5.11}
        R_2=hQ_2+jP \hspace{5mm} and  \hspace{5mm} r_2=absc(R_2 )
    \end{equation}

    \vspace{-12mm}
    \begin{equation} \label{eq5.12}
        r= r_1 r_2 
    \end{equation}
                         
    \vspace{-12mm}                        
    \begin{equation} \label{eq5.13}
        u_i=F_1 (d_i)||F_2 (F_1 (d_i )) \oplus d_i 
    \end{equation}                                         
    \vspace{-12mm}
    \begin{equation} \label{eq5.14}
        R_i=R_1+R_2+ru_i (A+Q_{ID}P_0) 
    \end{equation}
    
    \vspace{-12mm}
    \begin{equation} \label{eq5.15}
          b_{1i}=kg^{-1}i(H_2 (R_i ||F_{ID} ||I_i )-ru_i )
    \end{equation}
    
    \vspace{-12mm}
    \begin{equation} \label{eq5.16}
        b_{2i}=lh^{-1}j(H_2 (R_i ||F_{ID} ||I_i )-ru_i )
    \end{equation} 

    On received parameters $<b_{1i},b_{2i}>$ and its private key $d_{T}$, signer computes the blind signature $\sigma'_i=<s'_{1i},s'_{2i}>$, where $s'_{1i}=(d_T b_{1i}- n_1)$, $s'_{2i}=(d_T b_{2i}- n_2)$ and sends $\sigma'_i$ to the user. 
    
    \item \textbf{Original signature extraction and verification}
    
	On given blinded signature $\sigma'_i$ from the TPA, user extracts the original signature $\sigma_i=<R_i,u_i,s_i>$ using his secret values, i.e., $<g,h,i,j>$ (given in Equations (\ref{eq5.17})-(\ref{eq5.19})). 
	\vspace{-5mm}
	\begin{equation} \label{eq5.17}
	    s_{1i}=(s'_{1i} g-i )
	\end{equation}
	
    \vspace{-12mm}                      
    \begin{equation} \label{eq5.18}
        s_{2i}=(s'_{2i} h-j)
    \end{equation}                                          
    \vspace{-12mm}
    \begin{equation} \label{eq5.19}
        s_i=(s_{1i}+s_{2i})
    \end{equation}
    The user verifies the correctness of the signature using Equation (\ref{eq5.20}).
    \vspace{-5mm}
    \begin{equation} \label{eq5.20}
        s_i P+R_i=H_2 (R_i ||F_{ID} ||I_i)(A+H_1 (A,ID_T) P_0)
    \end{equation}

    If (\ref{eq5.20}) holds, the user outsourced the signature $<\sigma_i,tag>$ on the cloud over a secure channel. 
    
    \item \textbf{Stored data integrity auditing}
    
    The user delegates the TPA and requests him to audit the stored data integrity in the cloud. The user generates the tag for data file $F$ with its identity be $F_{ID}$ as $tag=F_{ID}||Sig(F_{ID})||A_U$, where $Sig(F_{ID} )=d_UH_4(F_{ID})$, and it to the TPA. TPA receives the data file tag and checks its correctness using $e(Sig(F_{ID}),P)=e(H_4 (F_{ID}),A_U+H_1 (A_U,ID_U)P_0) $. The correctness of the equation is verified as follows. 
    
    \begin{align*}
        e(Sig(F_{ID}),P) &=e(d_U H_4 (F_{ID}),P) \\
            &=e(H_4 (F_{ID}),(a_U+s_0 Q_U)P) \\
            &=e(H_4 (F_{ID}),A_U+H_1(A_U,ID_U)P_0 )
    \end{align*}
    If the tag is correct, TPA extracts the data file identifier name $F_{ID}$ from the tag. On given auditing request from the user, TPA picks a random subset $S={s_j} \subset I$ with integer $x_j$ associated with corresponding $s_j$, where $1 \le j \le m$, such that $m<n$. TPA generates a challenge $Chall={j,x_j}$ and sends it to the cloud for auditing proof. The challenge $Chall$ determines the location of the blocks needed to be verified. On receiving $Chall$ from TPA, the cloud computes the auditing proof using Equation (\ref{eq5.21}).

    \begin{equation} \label{eq5.21}
        v_j=u_j \oplus x_jH_3(s_j P+R_j) 
    \end{equation}
    \vspace{-5mm}
     
     Then, the cloud gives the proof $<v_j, R_j>$ to the TPA for integrity auditing of stored data block $d_j$. In order to check the correctness of stored data block $d_j$, TPA checks whether Equations (\ref{eq5.22})-(\ref{eq5.23}) hold. 
     \vspace{-5mm}
     \begin{equation} \label{eq5.22}
         u'_j=v_j \oplus x_j H_3 (H_2 (R_j |(|F_{ID}|)| I_j )(A+H_1 (A,ID_S )P_0 ))
     \end{equation}
     
     \vspace{-12mm}
     \begin{equation} \label{eq5.23}
         d'_j=F_2(({_{l_1}^l}|u'_j|) \oplus |u'_j|_{l_2}^R 
     \end{equation}
     \vspace{-5mm}

    The TPA accepts the challenge associated with the data block ${d'_j}$ if and only if  ${_{l_1}^l}|u'_j|=F_1(d'_j)$.
    
    \item \textbf{Batch signature verification}
    
    Let the user request the TPA to sign p data blocks, where $1 \le p \le n$. Users can verify the correctness of p data blocks in a single equation, given in Equation (\ref{eq5.24}).
    \begin{equation} \label{eq5.24}
        \sum_{i=1}^p(s_i)P+\sum_{i=1}^pR_i =\sum_{i=1}^p(H_2(R_i))(A+H_1 (A,ID_S)P_0)
    \end{equation}

\end{enumerate}

\section{Security Analysis}

Here, we provide the security proof of the proposed IDBS-MR scheme secured against the EF-ID-CMA and blindness. 

\begin{theorem} \label{thm5.1}
\textbf{(Consistency)}. The proposed IDBS-MR scheme is consistent.
\end{theorem}

\begin{proof}
The consistency of the proposed IDBS-MR scheme is verified as follows.

we have $s=(s_1+s_2)$ then,
\vspace{-12mm}

\begin{align*}
    sP+R &=(s_1+s_2)P+R \\
    &=(s'_1g-i+s'_2h-j)P+R \\
    &=((d_{IDS}b_{M1}- n_1)g-i+(d_{IDS}b_{M2}- n_2)h-j)P+R \\
    &=((d_{IDS}kg^{-1}i(H_2(R)-ru)-n_1)g-i+(d_{IDS}lh^{-1}j(H_2(R)-ru)- n_2 )h-j)P+R \\
    & =(d_{IDS}ki(H_2(R)-ru)- n_1g-i+d_{IDS}lj(H_2(R)-ru)- n_2 h-j)P+R\\
    & =d_{IDS}H_2(R)(ki+lj)P-(ki+lj)rud_{IDS} P-n_1 gP-n_2 hP-iP-jP+R \\
    &=d_{IDS}H_2(R)P-(rud_{IDS}P+n_1 gP+n_2 hP+iP+jP)+R \\
    &=(a+s_0Q_{ID}) H_2(R)P-(ru(a+s_0 Q_{ID})P+gQ_1+hQ_2+iP+jP)+R \\
    &=(a+s_0 Q_{ID})uH_2(R)P-(ru(A+Q_{ID}P_0)+R_1+R_2)+R \\
    & =(A+Q_{ID}P_0 )H_2(R)-R+R\\
    & =(A+Q_{ID}P_0)H_2(R)
\end{align*}
\vspace{-5mm}

This proves the consistency of the IDBS-MR scheme. Thus,  $sP+R=(A+Q_{ID} P_0)H_1 (R)$, and  $u=v \oplus H_3(sP+R)$. Now, $u=F_1(m)||F_2(F_1 (m)) \oplus m$ and hence ${_{l_1}^l}|u|=F_1 (m)$, $|u|_{l_2}^R=m \oplus F_2 (F_1 (m))$ and check if $(_l1^l)|u|=F_1 (m_0)$. If this equality holds, the verifier accepts the parameter $<A,v,R>$ as a correct signature on message $m$.
\end{proof}

\begin{theorem}
\textbf{(Un-forgeability)}. Suppose $H_1$ and $H_2$ are two random oracles models and a forger $F$ wants to forge a signature on message $M$. Suppose forger $F$ executes at most $q_E$ extract oracles, $q_B$ blind signature oracles, $q_1$  $H_1$ hash oracles, $q_2$  $H_2$ hash oracles, $q_3$ $H_3$ hash oracles runs at most $t$ times with advantage at most $k^{-n}$. Under the assumption of ROM and intractable to solve the ECDLP, our proposed IDBS-MR Scheme is existentially unforgeable under adaptive chosen message and identity attacks. Forger $F(t,q_1,q_2  q_E,q_B,k^{-n})$ have the following advantage to break the proposed IDBS-MR scheme.
\vspace{-5mm}

\begin{equation} \label{eq5.25}
    |Pr[ F(t,q_1,q_2,q_E,q_B,k^{-n})]| \ge(1-q_1/k)^{q_2+q_E} 
\end{equation}
\end{theorem}

\begin{proof}

Consider a forger $F$ wants to forge any signature in our proposed IDBS-MR scheme and let there exist an algorithm B which helps $F$. We design an algorithm B that helps $F$ to solve the ECDLP.

\textbf{Setup}: B supposes two hash function $H_1:\{0,1\}^* \times \mathbb{G}_1 \rightarrow \mathbb{Z}_q$ and $H_2:\mathbb{G}_1 \rightarrow \{0,1\}^{|q|}$. Let two function are $F_1:\{0,1\}^{l_1} \rightarrow \{0,1\}^{l_2}$ and $F_2:\{0,1\}^{l_2} \rightarrow \{0,1\}^{l_1}$, where $l_1$ and $l_2$ are two positive integers such that $l_1+l_2=|q|$. $B$ is accountable for simulating these oracles. $B$ picks $a \in \mathbb{Z}_q$ and set $P_0=aP$ and gives public parameter $pp=<\mathbb{G}_1,q,P,P_0,H_1,H_2,F_1,F_2,l_1, l_2,k>$ to $F$.

\textbf{Oracles}: Forger $F$ can perform the following oracles.
\begin{itemize}
    \item \textit{$H_1$ oracle}: $B$ prepares an empty list $H_1^{List}$ having tuple $<A_i,ID_i,H_1 (ID_i,A_i ),x_i>$. When $F$ queries to $H_1^{List}$ on Identity $ID_i$, $B$ responds $F$ in the following way.
    \begin{itemize}
        \item $B$ gives $H_1 (ID_i,A_i)$ to $F$, if $ID_i$ found in the $H_1^{List}$ in the tuple of $<A_i,ID_i,H_1 (ID_i,A_i ),x_i>$ or $<A_i,ID_i,H_1 (ID_i,A_i ),*>$.
        \item $B$ sets $H_1(ID_i,A_i )=h$ and gives to $F$ and adds the tuple $<A_i,ID_i,H_1(ID_i ),*>$ to list $H_1^{List}$, if $ID_i=ID^*$.
        \item Otherwise, $B$ chooses randomly $b_i\in \mathbb{Z}_q$ and gives $H_1 (ID_i,A_i)=-b_i$ to $F$ and adds tuple $(ID_i,H_1(ID_i,A_i),x_i)$ to list $H_1^{List}$.
    \end{itemize}
    Note, $H_1(ID,A)$ gives no information to $F$ until he queries the $H_1$ oracle on $ID$ because $H_1$ is the random oracle. 
    \item \textit{$H_2$ oracle}: $B$ provides the $R_i \in \mathbb{G}_1$ on applying the queries $M_i$ to $H_2(R_i)$ and gives to $F$. $B$ maintains a list $H_2^{List}$ for storing the responses to $F$. $B$ search an entry in the list, and if found, gives it to the $F$; otherwise, select $y_i \in \mathbb{Z}_q$ and sends  $y_i=H_2(R_i)$ to $F$ and adds an entry in list $H_2^{List}$.
    \item \textit{$H_3$ oracle}: $B$ maintains a list  $H_3^{List}$, initially empty for storing responses to $F$. $B$ search an entry in list  $H_3^{List}$, if found response to the $F$. otherwise, picks $z_i \in \mathbb{Z}_q$ and sends  $z_i=H_2(s_i P+R_i)$ to $F$ and adds an entry in list $H_3^{List}$.
    \item \textit{$F_1$ and $F_2$ oracles}: $B$ makes two lists $L_{F1}$ and $L_{F2}$ for responding to $F$ against  the requested message. $B$ check list $L_{F1}$, if found then search in $L_{F2}$. Otherwise, $B$ picks $e_{1i} \in \{0,1\}^{l_1}$ and $e_{2_i} \in \{0,1\}^{l_2}$ and set $F_1 (m_i )=e_{1i}$, $F_2(e_{1i})=e_{2i}$. $B$ responds to $F$ and stores them in $L_{F1}$ and $L_{F2}$.
    \item \textit{Extract oracle}: $B$ simulate the extract oracles on given identity $ID_i$ to obtain a private key. $B$ maintains a list $L_{Ext}$, which is initially empty and search an entry in $L_{Ext}$, If found, respond it to the $F$. otherwise do the following:
    \begin{itemize}
        \item For $ID_i=ID$, $B$ abort the process.  
        \item For $ID_i \ne ID$, $B$ abort the process. $B$ pick $c_i \in \mathbb{Z}_q$ and set $A_i=b_iP_0+c_i P$, such that $d_{IDi}=c_i$. We set these parameters such that they satisfy Eq. (6.27). B responds the secret key as $(d_{IDi},A_i)$ and store the tuple $<A_i,ID_i,H_1(ID_i,A_i),d_{IDi})$ in list $L_{Ext}$.
        \begin{equation} \label{eq5.27}
            d_{IDi} P=A_i+H_1 (ID_i,A_i)P_0
        \end{equation}
    \end{itemize}
    \item  \textit{Blind signature oracle}: $F$ queries the Blind Signature algorithm to get a signature on message $M_i$ with identity $ID_i$. Let $F$ gives the blinded message $<b'_M1,b'_M2>$ to $B$. Then, $B$ responds to the following oracles.
    \begin{itemize}
        \item For $ID_i=ID$, $B$ abort the process.
        \item  For $ID_i \ne ID$, $B$ runs the  $F_1$, $F_2$, $H_2$  and $H_3$ oracles and respond to the $F$. $B$ makes a list $L_{BS}$ to store the responses coming from blind signature and responds to the $F$. $B$ first check in list $L_{BS}$ if found then reply to the $F$ otherwise, set $v_i=u \oplus z_i$ and sends it to $F$ and make an entry in list $L_{BS}$. 
    \end{itemize}

\end{itemize}

\textbf{Forgery}: $F$ responses the signature $\sigma_i^*=<v_i^*,R_i^*,A_i^*>$  against message $M^*$ and $ID^*$. In order to forge a signature, suppose forger $F$ creates three distinct signature $<\sigma_A^*,\sigma_B^*,\sigma_C^*>$ on same message $M$, where $\sigma_A^*=<v_A^*,R^*,A^*>$, $\sigma_B^*= <v_B^*,R^*,A^*>$ and $\sigma_C^*=<v_C^*,R^*,A^*>$.

We consider parameter $s_0$,$a$ and $z$ are the discret logarithm of parameter $P_0$, $A$ and $R$ respectively, i.e., $P_0=s_0 P$, $A=aP$ and $R=zP$. From $sP=H_2 (M,R)(A+H_1 (A,ID_S )P_0 )-R$, we can get 

\vspace{-12mm}
\begin{align*}
    s_A^*&=d_A (a+s_0 H_1 (ID,A))-z \\
    s_B^*&=d_B (a+s_0 H_1 (ID,A))-z \\
    s_C^*&=d_C (a+s_0 H_1 (ID,A))-z 
\end{align*}
\vspace{-12mm}

The parameter $s_0$,$a$ and $z$ from the above equations are unknown to $F$. Thus, $F$ solves these values from the linear equations which are equivalent to solving the ECDL problem.

\textbf{Analysis}. The probability that B does not abort the game is defined by two events.
\centering
$Pr[\neg E_1 \wedge \neg E_2]$

\begin{itemize}
    \item $E_1$: The extract oracle fails if $H_1$ oracle gives the inconsistent outputs with probability at most $q_1/k$. The simulation is completed $q_E$ times which happens with probability at least $(1-q_1/k)^{q_E}$.
    \item $E_2$: The execution of $H_2$ oracle fails if $H_2$ oracle gives the inconsistent outputs with probability at most $q_1/k$. The simulation is completed $q_2$ times which happens with probability at least $(1-q)_1/k)^{q_2}$.
\end{itemize}

From the above two events, we obtain the probability that $Adv$ can break the scheme. 

\centering
$Pr[\neg E_1 \wedge \neg E_2] \ge \epsilon \left(1 - \frac{q_1}{k}\right)^{q_2} \left(1 - \frac{q_1}{k}\right)^{q_E}$

\end{proof}

\begin{theorem}
\textbf{(Blindness)}. The proposed IDBSMR scheme achieves the blindness property.
\end{theorem}

\begin{proof}
Suppose an adversary $Adv$ that plays the role of signer and challenger $Ch$ that play a role of an honest user runs the blind signature algorithm. Suppose $Adv$ obtains the parameters $<b_{M1},b_{M2},s'_1,s'_2>$ by executing the blind signature oracles, identical to Theorem 5.2. Let the corresponding signature be $<R,v,A>$. There exists a tuple of random values $<g,h,i,j,k,l>$ that links the $<b_{M1},b_{M2},s'_1,s'_2>$ to $<R,v,A>$. From the proposed IDBS-MR scheme, we have the following (\ref{eq5.27})-(\ref{eq5.33}).

\vspace{-12mm}
\begin{equation} \label{eq5.27a}
    b_{M1}=kg^{-1}i(H_2(R)-ur)
\end{equation}

\vspace{-12mm}
\begin{equation} \label{eq5.28}
    b_{M2}=lh^{-1}j(H_2(R)-ur)mod q
\end{equation}

\vspace{-12mm}
\begin{equation} \label{eq5.29}
    s'_1=(d_{IDS} b_{M1}- n_1 )mod q
\end{equation}

\vspace{-12mm}
\begin{equation} \label{eq5.30}
    s'_2=(d_{IDS}b_{M2}- n_2)mod q
\end{equation}

\vspace{-12mm}
\begin{equation} \label{eq5.31}
    s_1=(s'_1 g-i )mod q
\end{equation}
\vspace{-12mm}
\begin{equation} \label{eq5.32}
    s_2=(s'_2 h-j)mod q
\end{equation}

\vspace{-12mm}
  \begin{equation} \label{eq5.33}
      R=(g+h)Q_1+(i+j)P+ru(A+Q_{ID}P_0) 
  \end{equation}

From (\ref{eq5.28})-(\ref{eq5.33}), we can get $g=kb_{M1}^{-1}i(H_2 (R)-ur)$ and $h=lb_{M2}^{-1}j(H_2 (R)-ur)$ individually. Similarly, from Equations (\ref{eq5.32}) and (\ref{eq5.33}), we can get $i=(s'_1 g-s_1)$ and $j=(s'_2 h-s_2)$ individually. If we substitute these values in (\ref{eq5.33}), it seems that we must know the value of r to satisfy the (\ref{eq5.33}). However, the value of r depends on the elements $<g,h,i,j>$ whose estimation is equivalent to solving the ECDLP problem. Thus, it is very difficult for any adversary $Adv$ to compute the values $<g,h,i,j,k,l>$ from the (\ref{eq5.28})-(\ref{eq5.33}).
\end{proof}

\section{Performance Analysis }

Here, we evaluate the performance of the proposed scheme in terms of computation and communication overheads. 

\subsection{Implementation and Experiment Simulation}

The implementation of all algorithms is run on  \textit{Intel(R) Core(TM) i7-2600K CPU @ 3.4 GHz}, and \textit{8 GB of RAM}, under  \textit{gcc 4.6} compiler  and we call PBC library \cite{lynn2010pairing} API in order to construct the elliptic curves. We utilize the Type-A super-singular elliptic curve $E/\mathbb{F}_p:y^2=x^2+x$ built on two prime $p$ and $q$, such that $|p|$ = 512 bit, $q=2^{159}+2^{17}+1$ is Solaris prime ($|q|$ = 160 bit) satisfying $p+1=12pq$, and embedding degree is $2$, which is identical to the 1024-bit RSA security level. We consider super-singular curve over the binary field $\mathbb{F}_{2^{271}}$ with the order of $\mathbb{G}_1$ is 252 bit prime and $\mathbb{G}_2$ is 1024 bit. Using compression technique \cite{shim2013eibas}, we consider $|\mathbb{G}_1|$ = 34 bytes, $|\mathbb{G}_2|$ = 128 bytes and $|\mathbb{Z}_q|$ = 32 bytes. Besides, we suppose the block size $|d|$ = 20 bytes. We neglect the XOR operation and find the left $l$-bit and right $l$-bit of a string since they have negligible computation costs. Table \ref{tbl3.1} of chapter \ref{chapter3} illustrates the notations and computation cost (in ms) of different cryptographic operations. 

\subsection{Computation Cost }

Now, we will provide the implementation results by considering the following aspect. We consider the block size 10 KB and the file consist of 10,000 blocks to test the performance on cloud and TSP by taking the challenged blocks from 100 to 1000 with an increment of 100 for each step. We evaluate the computation cost of the signature delegation algorithm, signature recovery and verification algorithm, integrity auditing algorithm and batch signature verification algorithm. We observe in Figure (\ref{fig5.3}) that the computation cost of signature delegation and signature verification depends on the side of data blocks. It grows linearly with respect to the data blocks d. Figure (\ref{fig5.4}) illustrates the batch signature verification cost in which the cost grows very slowly with the increment of data blocks. Figure (\ref{fig5.5}) shows the cost of integrity auditing on the cloud and TPA side that grows linearly with the size of data blocks.

\begin{figure}
  \centering
  \includegraphics[width=0.8\linewidth]{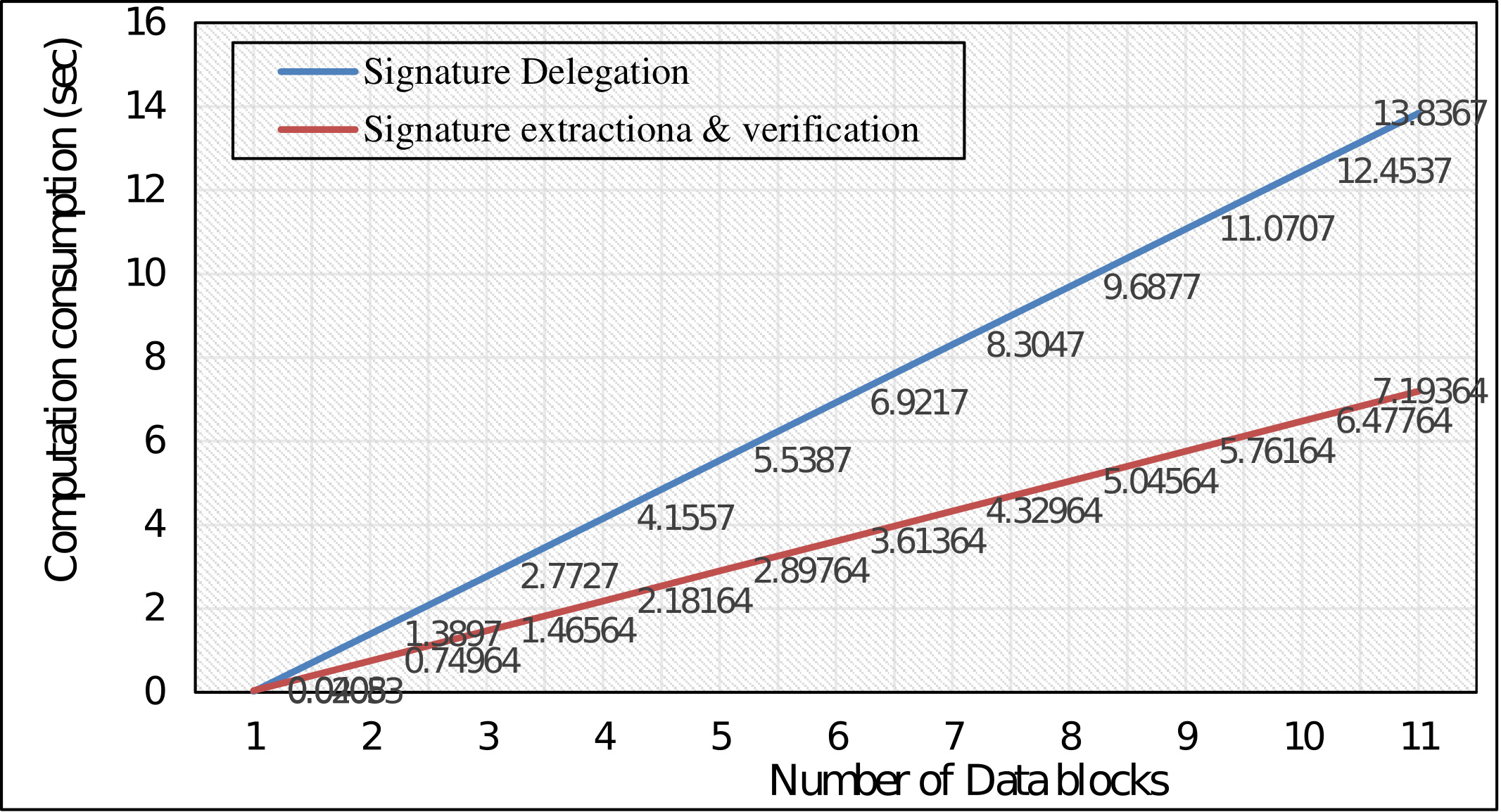}
  \caption{Computation consumption in signature delegation, and recovery and verification}
\label{fig5.3}
\end{figure}

\begin{figure}
  \centering
  \includegraphics[width=0.8\linewidth]{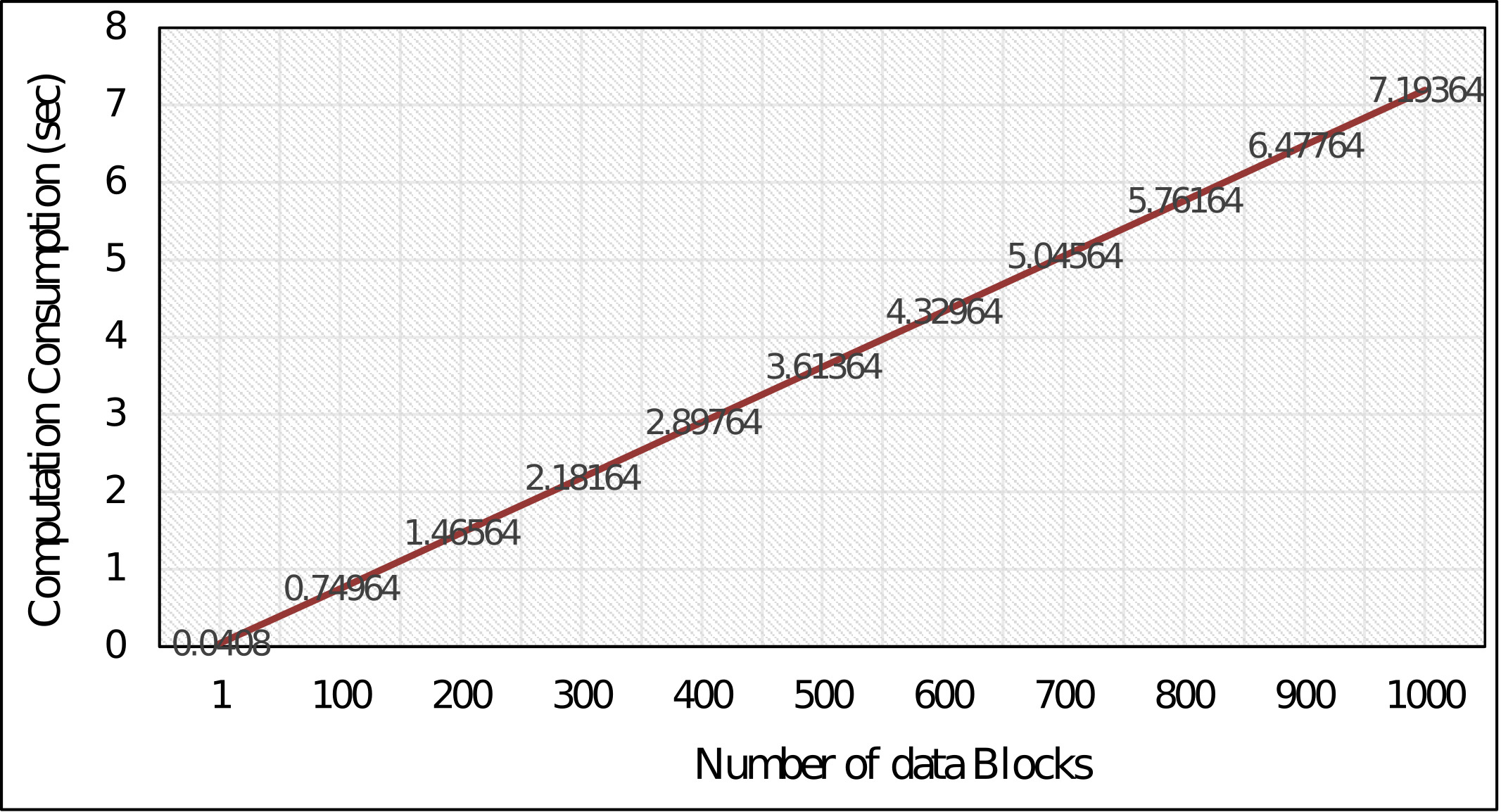}
  \caption{Computation consumption in batch signature recovery and verification}
\label{fig5.4}
\end{figure}

\begin{table}
        \centering
        \caption{Computational cost comparison of proposed IDBS-MR scheme with existing schemes}
        \label{tbl5.1}
        \begin{tabular}{|c|c|c|c|c|}
            \hline
            \multicolumn{1}{|c|}{Schemes} & \multicolumn{3}{|c|}{Computation cost (in ms)} & \multicolumn{1}{|c|}{Signature Size } \\
            \cline{2-4}
                &  BlindSig & Verify     & Total    & (in Bytes) \\
            \hline
            \hline
            \cite{gao2012round} &  126.73 &	80.04 &	206.77 &	122 \\
           \cite{kumar2017secure} &  80.5 & 46.77 &	147.2 &	88 \\
            \cite{dong2014efficient} & 40.02 &	26.68 &	66.70 &	154 \\
            \cite{han2005pairing} & 82.77   &	64.86 &	147.63 &	256 \\
            \cite{islam2015design} &   53.6 &	46.7 &	100.3 &	156 \\
            \cite{elkamchouchi2008new}  & 62.76 & 44.85 &	107.61 &	148 \\
            \cite{james2017identity} & 62.7 & 46.7 &	109.4 &	168 \\
            \cite{verma2017efficient} &  40.08 &	51.5 &	91.58 &	66\\
            Our  & 27.35 &	 6.9 &	34.25 &	98 \\
            \hline
        \end{tabular}
\end{table}

\subsection{Comparative Performance}

Here, we compare our proposed IDBS-MR scheme with the existing schemes \cite{gao2012round}, \cite{kumar2017secure}, \cite{dong2014efficient}, \cite{han2005pairing}, \cite{islam2015design}, \cite{elkamchouchi2008new}, \cite{james2017identity} and \cite{verma2017efficient} in terms of the computation, bandwidth cost and security.

\begin{figure}
  \centering
  \includegraphics[width=0.8\linewidth]{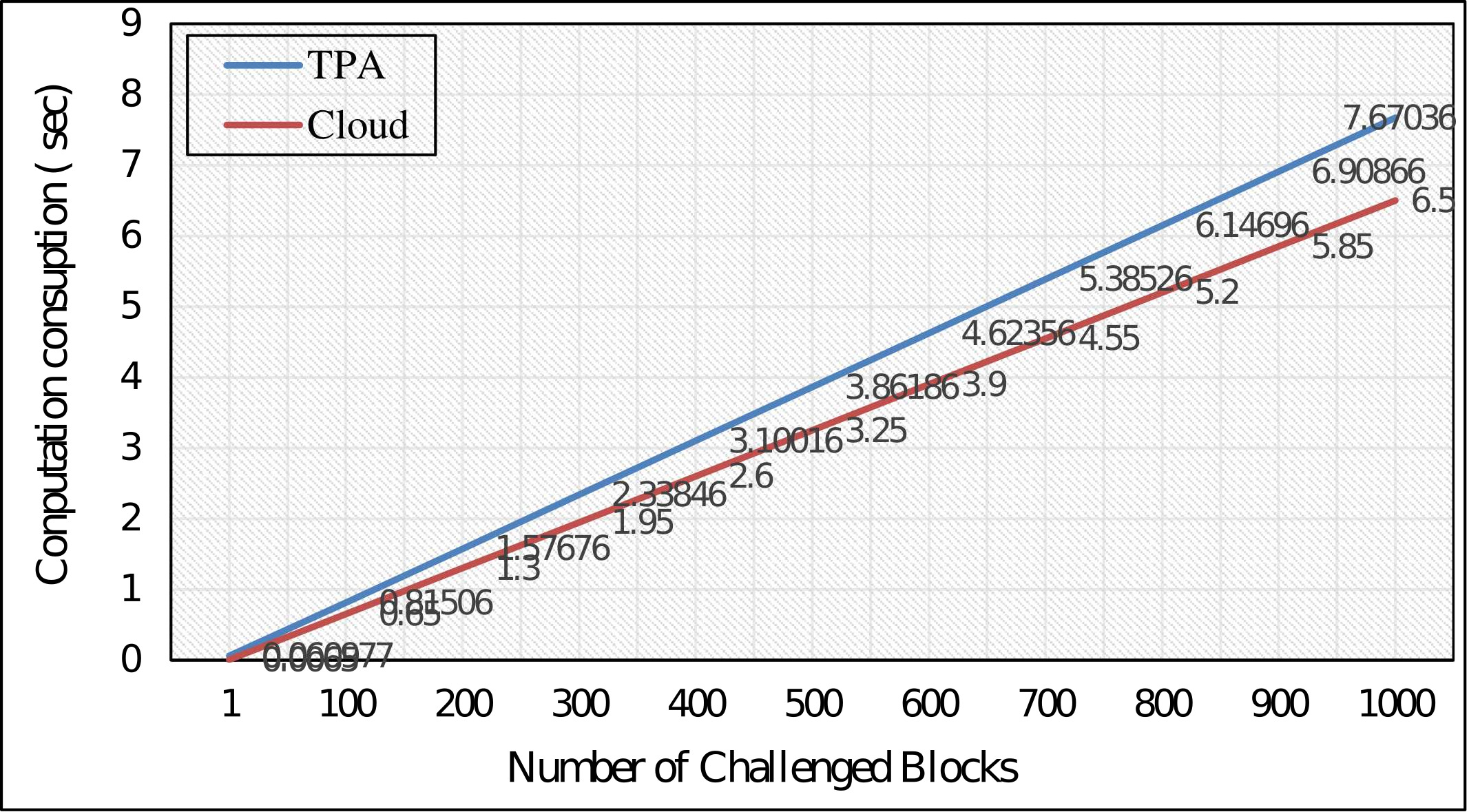}
  \caption{Computation consumption in integrity auditing on Cloud and TPA}
\label{fig5.5}
\end{figure}

\begin{itemize}
    \item 	\textit{\textbf{Computation cost}}. Here, we evaluate the computation of our proposed scheme with related schemes \cite{gao2012round}, \cite{kumar2017secure}, \cite{dong2014efficient}, \cite{han2005pairing}, \cite{islam2015design}, \cite{elkamchouchi2008new}, \cite{james2017identity} and \cite{verma2017efficient}, as summarized in Table \ref{tbl5.1}. In the blind signature algorithm, our proposed scheme consumes  $4T_{SM}+2T_A+7T_M$ = 4*6.67+2*0.23+7*0.03=27.35 ms, while the schemes \cite{gao2012round}, \cite{kumar2017secure}, \cite{dong2014efficient}, \cite{han2005pairing}, \cite{islam2015design}, \cite{elkamchouchi2008new}, \cite{james2017identity} and \cite{verma2017efficient} consume 126.73 ms, 80.5 ms, 42 ms, 82.77 ms, 53.6 ms, 62.7 ms, 62.7 ms, and 40 ms, respectively, of computation time. It is also shown graphically in Figure ((\ref{fig5.5})).  In the verification algorithm, our proposed IDBS-MR scheme needs $1T_{SM}+1T_A$ = 1*6.67+1*0.23 = 6.9 ms, while the schemes \cite{gao2012round}, \cite{kumar2017secure}, \cite{dong2014efficient}, \cite{han2005pairing}, \cite{islam2015design}, \cite{elkamchouchi2008new}, \cite{james2017identity} and \cite{verma2017efficient} take 80 ms, 46.77 ms, 26.68 ms, 64.86 ms, 46.7 ms, 44.85 ms, 46.7 ms, and 51.5 ms respectively, of computation time. Thus, the total computation cost for our proposed scheme is 34.25 ms while the schemes \cite{gao2012round}, \cite{kumar2017secure}, \cite{dong2014efficient}, \cite{han2005pairing}, \cite{islam2015design}, \cite{elkamchouchi2008new}, \cite{james2017identity} and \cite{verma2017efficient} need 206.77 ms, 147.2 ms, 66.7 ms, 147.2 ms, 100.3 ms, 107.6 ms, 106.4 ms, and 91.58 ms, respectively, for blind signature and verification algorithms, as shown in Figure (\ref{fig5.6}). Thus, our scheme needs 83\%, 76\%, 48\%, 77\%, 65\%, 68\%, 68\% and 62\% of computation cost of the schemes \cite{gao2012round}, \cite{kumar2017secure}, \cite{dong2014efficient}, \cite{han2005pairing}, \cite{islam2015design}, \cite{elkamchouchi2008new}, \cite{james2017identity} and \cite{verma2017efficient}, respectively. 
    
    \begin{table}
        \centering
        \caption{Security comparison of proposed IDBS-MR scheme with existing schemes}
        \label{tbl5.2}
        \begin{tabular}{|c|c|c|c|c|}
            \hline
            Schemes	& Cryptographic   &	Security   &	Security  &	Data   \\ 
             Schemes	&   primitives &	 Problem & Model &  Recovery \\

            \hline
            \hline
            \cite{gao2012round} &  IBS \& BS & 1-mBDHIP$^@$ & 	ROM &	No  \\
            \cite{kumar2017secure} &  Choon Cheon’s IBS  \& &	GDHP$^{\#\#}$  &	ROM &	No \\
             &   Boldyreva’s BS &	 &	 & \\
            \cite{dong2014efficient} & CLS \& BS &	ECDLP$^{*}$ &	ROM &	No \\
            \cite{han2005pairing} & IBS \& BS &	IMWPP$^{@@}$ &	ROM &	Yes \\
            
            \cite{islam2015design} &  CLS \& BS	& k-CAA3$^\#$ &	ROM &	No \\
            \cite{elkamchouchi2008new}  & IBS \& BS &	BDHP$^{\&}$ &	ROM &	Yes \\
            \cite{james2017identity} & IBS \& BS &	ECDLP &	ROM &	Yes \\
            \cite{verma2017efficient} &  IBS \& BS & k-CAA &	ROM &	Yes\\
            Our  & IBS \& BS &	ECDLP &	ROM &	Yes \\
            \hline
        \end{tabular}
        \footnotesize{\\$^@$One more bilinear Diffie-Hellman inversion problem, $^\#$3 Traitors collision attacks assumption, $^{\#\#}$ Gap Diffie-Hellman Problem, $^{@@}$ Inversion of Modified Weil Pairings Problem, $^{\&}$Bilinear Diffie-Hellman Problem, $^{*}$Discrete logarithm problem on elliptic curve }
\end{table}

	\item \textit{\textbf{Bandwidth cost}}. Here, we compare the signature size of our proposed scheme with related schemes \cite{gao2012round}, \cite{kumar2017secure}, \cite{dong2014efficient}, \cite{han2005pairing}, \cite{islam2015design}, \cite{elkamchouchi2008new}, \cite{james2017identity} and \cite{verma2017efficient} to evaluate the bandwidth consumption. Our proposed IDBS-MR scheme gives $|\mathbb{G}_1|+2|\mathbb{Z}_q| \approx 98$ bytes of signature, whereas the schemes \cite{gao2012round}, \cite{kumar2017secure}, \cite{dong2014efficient}, \cite{han2005pairing}, \cite{islam2015design}, \cite{elkamchouchi2008new}, \cite{james2017identity} and \cite{verma2017efficient} output 122 bytes, 88 bytes, 154 bytes, 256 bytes, 156 bytes,  148 bytes, 168 bytes and 66 bytes, respectively, of signature, as shown in Table \ref{tbl5.1} and Figure (\ref{fig5.7}). It can be observed from Table \ref{tbl5.1} that our scheme saves around 20\%, 36\%, 62\%, 37\%, 34\% and 42\% of bandwidth of the schemes \cite{gao2012round}, \cite{kumar2017secure}, \cite{dong2014efficient}, \cite{han2005pairing}, \cite{islam2015design}, \cite{elkamchouchi2008new}, \cite{james2017identity} and \cite{verma2017efficient}, respectively.
	
\begin{figure}
  \centering
  \includegraphics[width=0.8\linewidth]{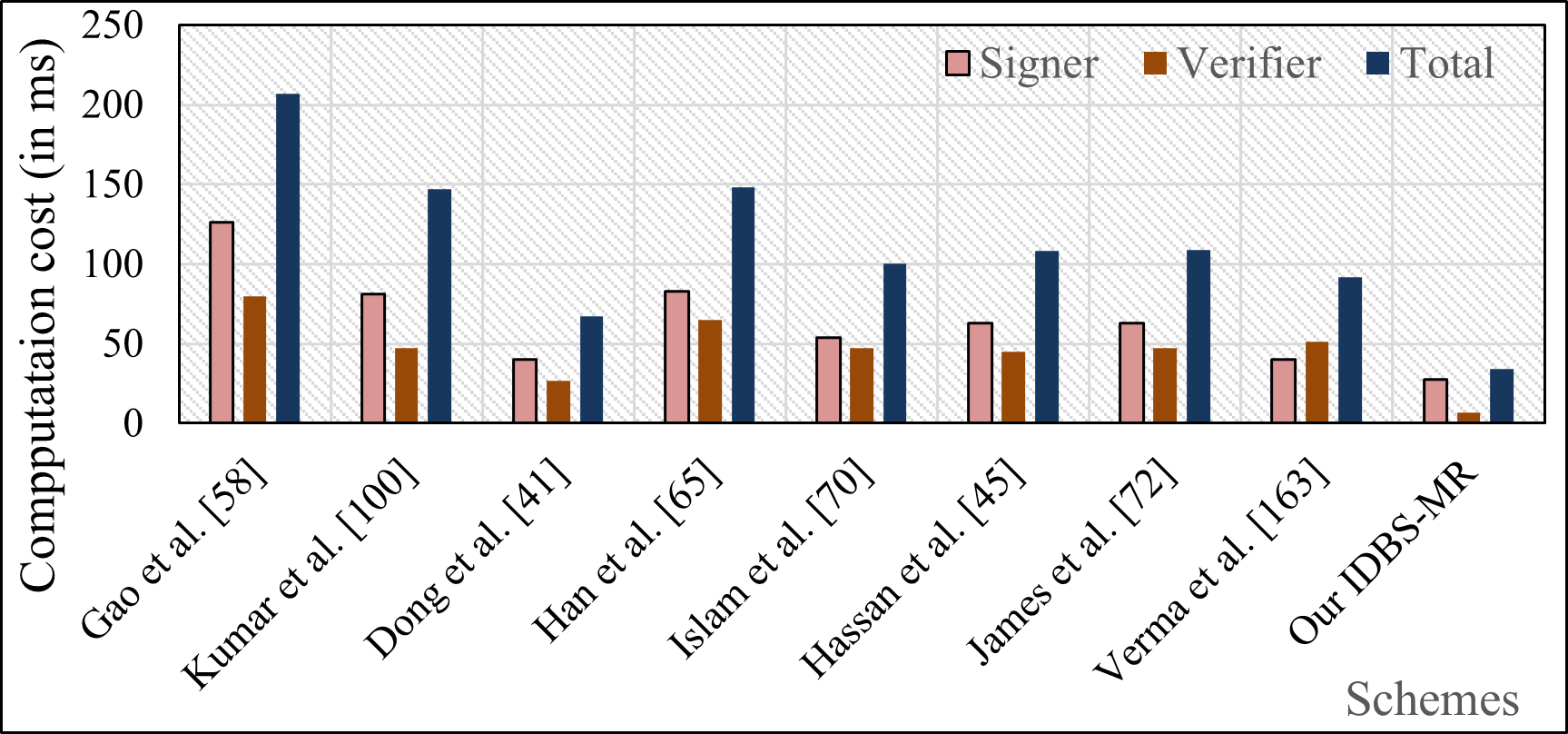}
  \caption{Computation cost of the proposed scheme with other related schemes}
\label{fig5.6}
\end{figure}

\begin{figure}
  \centering
  \includegraphics[width=0.8\linewidth]{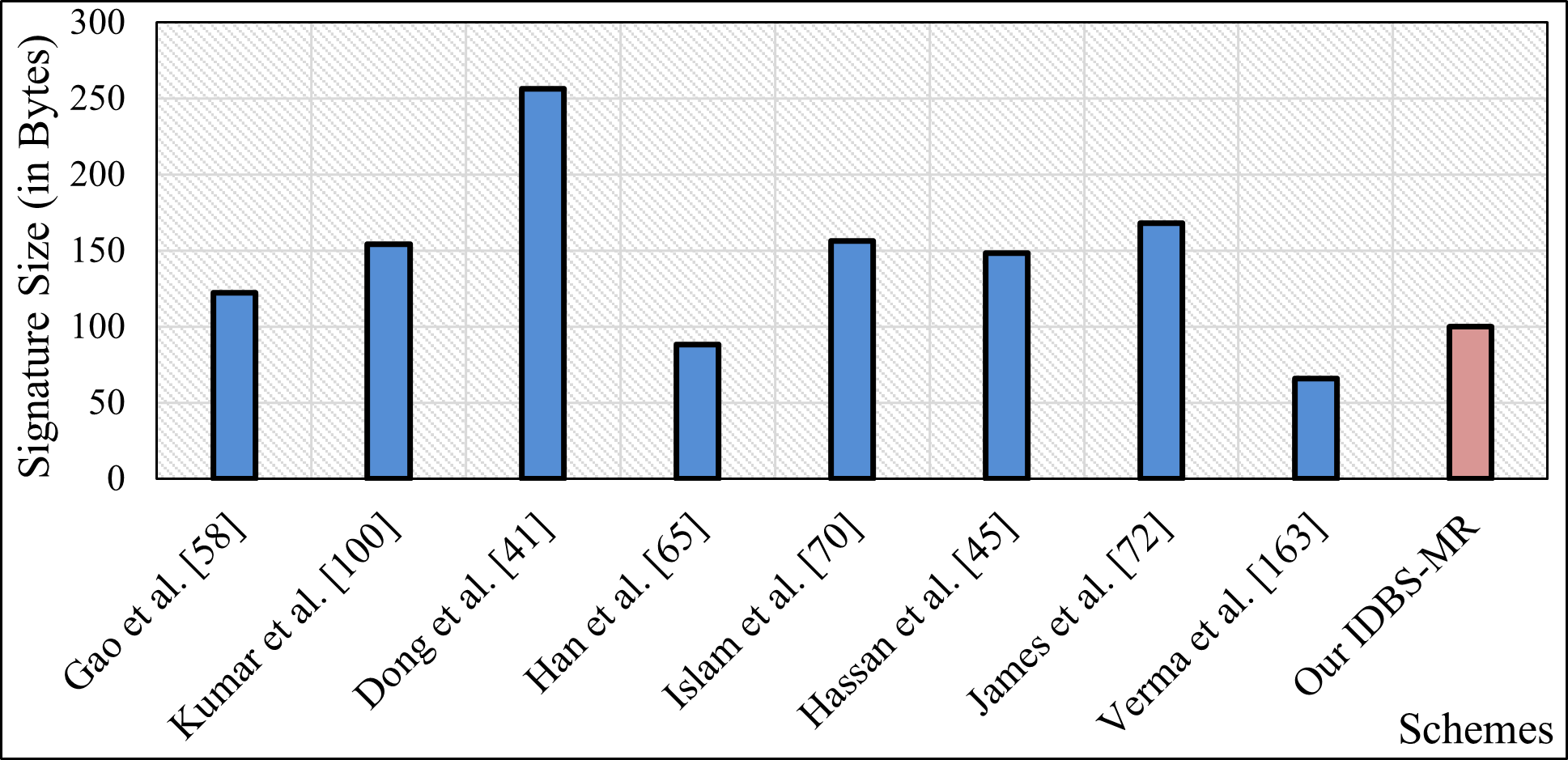}
  \caption{Signature size of the proposed scheme with related schemes}
\label{fig5.7}
\end{figure}

    \item \textit{\textbf{Security features}}. The proposed scheme is constructed on the Random Oracle Model (ROM), and its security is equivalent to solving the elliptic curve discrete logarithm problem (ECDLP). Table \ref{tbl5.2} summarizes the security features of our scheme along with the schemes \cite{gao2012round}, \cite{kumar2017secure}, \cite{dong2014efficient}, \cite{han2005pairing}, \cite{islam2015design}, \cite{elkamchouchi2008new}, \cite{james2017identity} and \cite{verma2017efficient}, in which we found that our scheme achieved all security features. 
\end{itemize}

\section{Summary} 
This chapter proposes a privacy-preserving data integrity auditing protocol for resource-limited devices. The protocol aims to delegate the signing and auditing process of data owners to a third-party auditor (TPA), thereby reducing the computational burden on resource-constrained devices. To achieve this goal, we introduce the IDBS-MR scheme, which eliminates costly operations, such as the map-to-hash function, modular exponentiation, and pairing during integrity checking. We prove that the proposed scheme is EUF-ID-CMA secured under ECDLP in ROM. Moreover, we demonstrate that the TPA can verify more data blocks in the least amount of time, and the user can verify multiple signatures simultaneously. Our evaluation of the proposed scheme shows that it exhibits high efficiency and adaptability, making it suitable for users with limited computational resources. Furthermore, we compare our scheme with other related schemes and demonstrate its superiority in terms of efficiency and adaptability.

\end{doublespace} \label{chapter5}
\begin{savequote}[75mm] 
However difficult life may seem, there is always something you can do and succeed at.
\qauthor{Stephen Hawking - British Theoretical Physicist} 
\end{savequote}

\chapter{ID-Based (Aggregated) Signcryption for Secure Video Streaming and Healthcare System}
\justify
\begin{doublespace}

Signcryption is a cryptographic primitive that performs encryption and signature simultaneously with the least computational overhead. The traditional RSA-based signcryption schemes are inefficient due to the problems of public key and certificate management which will later be overcome by introducing an ID-Based Signcryption (IDSC) scheme. Since IDSC is based on an ID-based setting, it suffers from an inherent issue, known as the key escrow problem. Certificate-less based signcryption scheme addresses the key escrow problem, but it loses the ID-based property. In this chapter, we present two variations of IDSC schemes: \textit{Escrow-Free Identity-Based Signcryption} (EF-IDSC) and \textit{Escrow-Free Identity-Based Aggregated Signcryption} (EF-IDASC) schemes. Both schemes are secured against confidentiality and enforceability attacks. Based on the proposed two schemes, we implement two secure emerging network systems. First, we implement a \textit{Secure Peer-to-Peer Video-on-Demand Streaming System}, whose security is based on the proposed EF-IDSC scheme. Second, we implement \textit{Secure Cloud-Assisted Health Care System} that achieves public verifiability based on IDASC scheme. Now, we discuss the related work.

Zheng \textit{et al.} \cite{zheng1997digital} were the first to propose the signcryption scheme, which achieves confidentiality, integrity, and authenticity of the data in a single step. Malone-Lee \cite{malone2002identity} proposed signcryption in an identity-based environment that addresses the certificate maintenance problem associated with scheme \cite{zheng1997digital}. In order to further optimize the computation cost, Sun \textit{et al.} \cite{sun2008generic, sun2008identity} proposed an Online/Offline Identity-Based Signcryption (OOIDSC) scheme for resource constraints environment by aiding the online/offline feature, in which the offline step executes the costly operations while online step includes the lightweight operations. 

Selvi \textit{et al.} \cite{selvi2010identityonline} pointed out that the scheme \cite{sun2008generic} has some security weaknesses, and the scheme \cite{sun2008identity} is easily forgeable, and proposed a new OOIDSC scheme. Lai \textit{et al.} \cite{lai2017efficient} presented an OOIDSC scheme with the short ciphertext. For a low-power device, Li \textit{et al.} \cite{li2012identity} gave another OOIDSC scheme whose offline phase is independent of message and recipient identity. Liu \textit{et al.} \cite{liu2010online} proposed an OOIDSC scheme that needs only the message and recipient’s identity during online signcryption. However, the scheme \cite{liu2010online} is not secure against the sender's anonymity. Barbosa \textit{et al.} \cite{barbosa2008certificateless} addressed the key escrow problem and discussed the certificate-less signcryption scheme, which is secure against insider attacks.  Li \textit{et al.} \cite{li2015certificateless} presented an Online/Offline Certificateless Signcryption (OOCLSC) scheme and prove its security under the well-known assumption of Computational Diffie-Hellman (CDH), q-Modified Bilinear Diffie-Hellman Inversion (q-mBDHI) and q-traitors collision attack algorithm (q-CAA) problem in the ROM. In \cite{luo2014security}, a new OOCLSC is implemented, however, Shi \textit{et al.} \cite{shi2015security} prove that it is not secure as an adversary can obtain the user’s private key. Li \textit{et al.} \cite{li2017certificateless} discuss an OOCLSC scheme for resource-constrained devices. Recently, Saeed \textit{et al.} \cite{saeed2018hoosc} present a heterogeneous online/offline signcryption for the Internet of things. 

Omala \textit{et al.} \cite{omala2016provably} presented a lightweight Certificateless Signcryption (CLSC) scheme for secure data transmission in the WBAN system. Yin \textit{et al.} \cite{yin2015certificateless} gave an efficient hybrid signcryption scheme in a certificateless setting for secure communication for WSNs. Unlike OOIDSC, schemes \cite{omala2016provably} and \cite{yin2015certificateless} are resistant to key escrow attacks.  Zhang \textit{et al.} \cite{zhang2016light} discuss a data communication scheme for e-health systems using Certificateless Generalized Signcryption (CLGSC) scheme. Caixue Zhou \cite{zhou2018comments} points out that Zhang \textit{et al.}’s scheme \cite{zhang2016light} is susceptible to an insider attack. Recently, Zhou \cite{zhou2019improved} has presented an improvised CLGSC scheme for the mobile healthcare system. 

In order to reduce the transmission and verification overhead, Selvi \textit{et al.} \cite{selvi2009identityagg}  discuss three aggregated signcryption schemes in the identity-based setting, which achieve public verifiability. Wang \textit{et al.} \cite{wang2016identity} propose a first identity-based aggregated signcryption scheme using multi-linear mapping in the standard model. Kar \cite{kar2013provably} propose a new identity-based aggregated signcryption scheme for low-processor devices. However, schemes \cite{selvi2009identityagg,wang2016identity, kar2013provably} are susceptible to the key escrow problem, which is addressed by Eslami \textit{et al.} \cite{eslami2014certificateless}, followed by presenting an aggregate-signcryption scheme in the certificateless setting. Niu \textit{et al.} \cite{niu2017privacy} propose a secure transmission scheme for heterogeneous devices, which transmits k messages from k senders in certificateless settings to m recipient in the IBC sets. Ullah \textit{et al.} \cite{ullah2019secure} give a secure multi-sender/multi-receiver data transfer scheme for a smart camera using an identity-based aggregated signcryption scheme. 

\section{Systems Models}
Here, we discuss the two architectures: Secure Peer-to-Peer Video-on-Demand Streaming (P2P-VoD), and secure cloud-assisted IoMT-enabled healthcare system. We also define the escrow-free identity-based signcryption scheme and escrow-free identity-based aggregated signcryption scheme. Besides, we discuss security threats. 

\begin{figure}
  \centering 
  \includegraphics[width=1\linewidth]{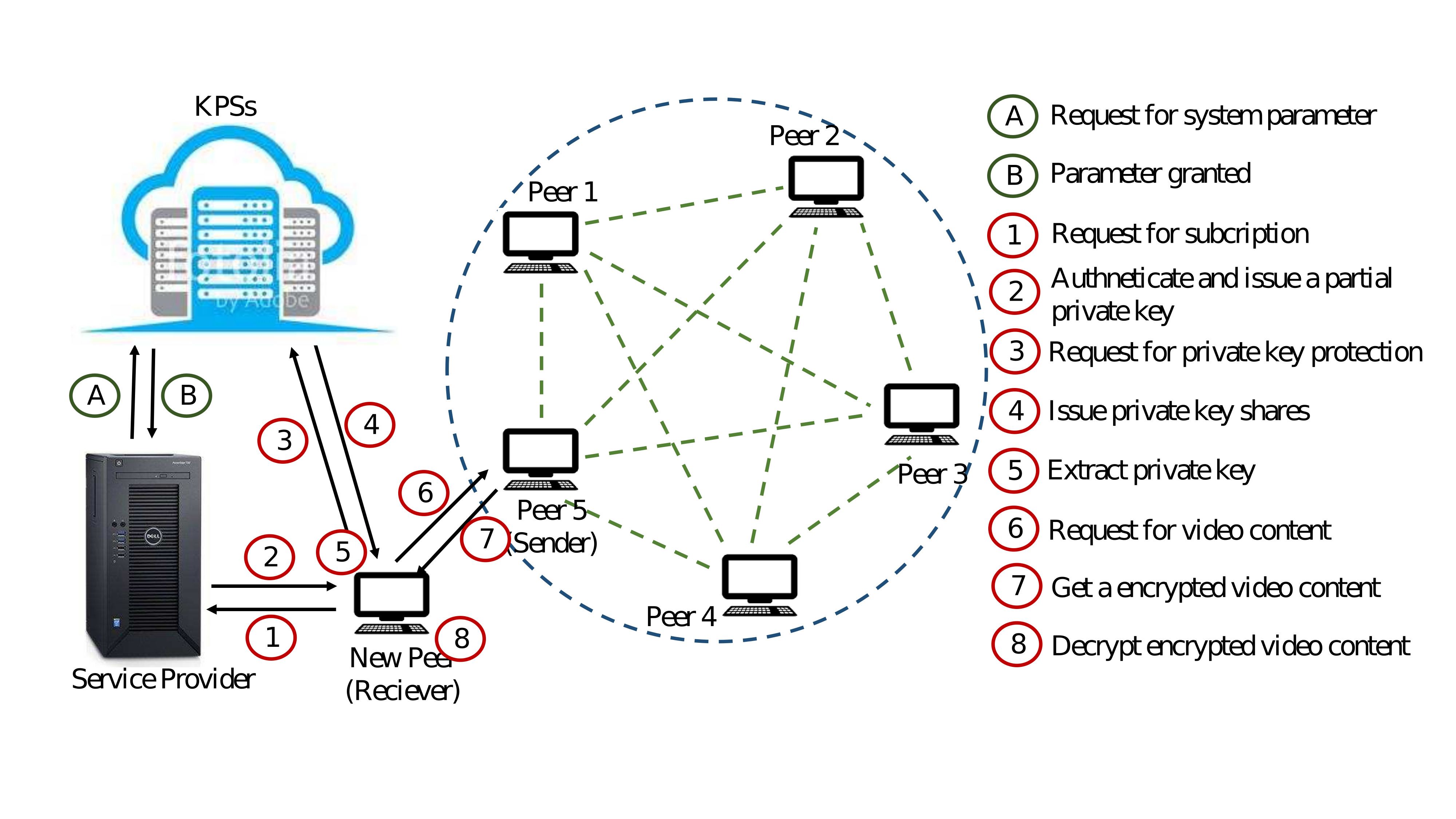}
  \caption{A mechanism of joining a new peer to the proposed P2P-VoD scheme}
\label{Fig6.1}
\end{figure}

\subsection{System Architecture of Video Streaming }

Here, we design a video-on-demand streaming architecture in the mesh-based P2P environment (P2P-VoD), which consists of three entities: service provider (SP), multiple KPSs, and peer (subscriber). 

\begin{itemize}
    \item \textit{Service Provider}: It is a semi-trusted authority that follows the correct way to compute the private key share for a peer but is still curious to know the original private key. It initializes the system with the help of multiple KPSs. It authenticates each peer in the network and provides a partial private key. 
    \item \textit{Key Privacy Servers}: KPSs protect the peer’s private key and issue the protected private key shares to a peer without knowing it. Besides, they offload the expensive computational on the cloud server. 
    \item \textit{Peer (Subscriber)}: Peer acts as a server as well as a client that can upload and download video content simultaneously. Each peer communicates with another peer using the session key in the network.
\end{itemize}

The proposed P2P-VoD system involves three phases: system initialization, connection establishment, and secure video streaming. The detail description of the proposed P2P-VoD system is demonstrated in Figure \ref{Fig6.1} 

\begin{itemize}
    \item \textit{System Initialization}: It involves two algorithms: system setup and peer registration. In the system setup algorithm, SP and KPSs set up the system, where they compute their secret keys and public parameters, given in step A and step B of Figure \ref{Fig6.1}. In the peer’s registration algorithm, a new peer requests the service provider for the subscription, shown in step 1. The SP authenticates the new peer and provides a partial private key (partial credential) against his identity, shown in step 2 of Figure \ref{Fig6.1} Now, the peer request to the multiple KPSs for key protection, given in step 3 followed by multiple Key Protection Servers (KPSs) protect the partial private key (partial credential) and issues him the protected private key (protected credential) shares, given in step 4. The new peer integrates the protected private key (protected credential) shares in order to extract his private key (credential), shown in step 5. Therefore, unlike generic IBC, the service provider, in our scheme does not have any information about the peer’s private key. 
    \item \textit{Connection establishment}: It allows two peers to establish a shared session key between them using their private keys and IDs.
    \item \textit{Secure video streaming}: It involves two algorithms: signcryption and unsigncryption. In the signcryption algorithm, a (receiver) peer requests another (sender) peer in the network, as shown in step 6. The (sender) peer divides the video content into chunks and signcrypts each chunk by its session key and private key, shown in Step 7. The signcrypted chunk is then transmitted to the (receiver) peer. Each peer has a buffer space which includes the information of desirable and available chunks. In order to upload and download video chunks, peers exchange the cache space of their choice. In unsigncryption algorithm, the (receiver) peer unsigncrypts the video chunk and can access the video content if it has a session key, as given in step 8. 
\end{itemize}

Since the proposed P2P-VoD system is dynamic and distributed in nature, it is susceptible to pollution attacks, wherein a polluter peer injects the polluted chunks into the P2P network, and so degrading the streaming quality at the recipient side. Besides, the system is also vulnerable to the untrusted service provider (SP), where SP can intentionally forge or alter the peer’s data chunks or allows any malicious peer in the network to access the data chunks. In order to overcome the pollution attack with the untrusted service provider in the P2P-VoD system, we discuss the Escrow-Free Identity-based Signcyrption (EF-IDSC) scheme. 

\subsection{Definition of Proposed EF-IDSC Scheme}

The proposed EF-IDSC scheme includes four algorithms: Setup, Keygen, Signcrypt and unsigncrypt, defined as follows.

\begin{itemize}
    \item \textit{\textbf{Setup}}: NM and KPSs compute the system parameter pp and secret keys $\langle s_0, s_1, s_2, \ldots, s_n \rangle$, where the master key $s_0$ and $\langle s_1, s_2, \ldots, s_n \rangle$ are kept secret, and $pp$ is published.
    \item \textit{\textbf{KeyGen}}: On given entity's identity $ID_E$, NM and KPSs compute the private key using their master key $s_0$ and secret keys $\langle s_1, s_2, \ldots, s_n \rangle$.
    \begin{itemize}
        \item \textit{ParPriKey}: Entity obtains a partial private key $D_0$ from NM.
        \item \textit{SecPriKey}: Entity obtains a protected private key shares $D_i$ from KPSs.
        \item \textit{KeyExt}: Entity computes its private key $d_E$.
    \end{itemize}
	\item \textit{\textbf{Signcrypt}}: On a given pp and message $M$, signer signcrypts $M$ using his private key $d_S$ and recipient  $ID_R$ and outputs the signcryptext $CT$.
	\item \textit{\textbf{Unsigncrypt}}: On given $pp$, and aggregated signcryptext $CT$, recipient decrypts it using his private key $d_R$ to obtain the messages and sender’s $ID_S$, and validates the messages $M$.
\end{itemize}

The ID-based signcryption scheme is secure if it is indistinguishably secured under chosen ciphertext and ID attack (IND-ID-CCA) and is existentially unforgeable under chosen message attack (EUF-CMA). 

\begin{definition}
Definition 6.1 (IND-CCA). The proposed EF-IDSC scheme is secured against IND-CCA, if no polynomially bounded adversaries $Adv$ has non-negligible advantage $\epsilon$ in-game, defined as:
    
\textbf{Setup}. On given security parameter $k$, challenger$Ch$ outputs public parameters $pp$, master key and sends $pp$ to $Adv$.
	
\textbf{Phase1}: $Adv$ runs the following queries. 
\begin{itemize}
    \item \textit{ParKeyIss query}. Given $ID$,$Ch$ responds the partial private key to $Adv$. 
    \item \textit{PriKeySec query}. Given a partial private key,$Ch$ responds to private key share to $Adv$.
    \item \textit{KeyExt query}. Given a private key share,$Ch$ responds to the full private key to $Adv$.
    \item \textit{Signcrypt query}. Given the identities $ID_S$, $ID_R$, KeyExt’s output, and a data content $chk$,$Ch$ returns the signcryption key and signcryptext $CT$ to $Adv$. 
    \item \textit{Unsigncrypt query}. Given ciphertext $CT$, challenger$Ch$ extracts the full private key $d_R$, outputs $chk$ to $Adv$.
\end{itemize}

\textbf{Challenger phase}. $Adv$ gives $<pkt_0,pkt_1,ID_S^*,ID_R^*>$, where  $|chk_0 |=|chk_1|$,$Ch$ picks a random $bit b\in \{0,1\}$ and computes $p2p_{sk}^*$, $CT^*$ and returns $CT^*$ to $Adv$, where keyExt query $ID_R$ has not been executed previously. 

\textbf{Phase 2}. Similar to phase 1, $Adv$ adaptively runs more queries such that no KeyExt query is allowed on $ID_R^*$  and unsigncrypt query on $<CT^*,ID_S^*,ID_R^*>$ respectively. 

\textbf{Guess}. At the end, $Adv$ guesses $b' \in {0,1}$. If $b = b'$, adversary $Adv$ wins the game with advantage.

\begin{equation} \label{eq6.1}
   A(Adv) := \lvert Pr[b=b'] - \frac{1}{2} \rvert \ge \epsilon
\end{equation}
\end{definition}

\begin{definition}
\textbf{(EUF-CMA)}. The proposed EF-IDSC scheme is secured against EUF-CMA, if no polynomial bounded forger of $F$ has a non-negligible advantage $\epsilon$ in the above game.
\begin{itemize} 
    \item 	\textit{Setup}. $F$ runs the query identical to the IND-CCA game. 
    \item \textit{Phase1}. $F$ runs the query identical to the IND-CCA game. 
	Forgery. $F$ responses the tuple $<chk^{ch},CT^{ch},ID_S^{ch},ID_R^{ch}>$, over which $F$ is not allowed to extract the ParKeyIss query on $ID_S$ but extract the KeyExt query on $ID_R$. $F$ wins the game if the output of unsigncrypt is not $\perp$.  Otherwise, $CT^{ch}$ is valid ciphertext under $ID_S^{ch}$ and $ID_R^{ch}$ (challenged identities of peers A and B).
\end{itemize}
 
\end{definition}

\subsection{Security Requirements for Video-On-Demand}
	
We adopt the security requirements, given in \cite{wang2018content, fiandrotti2015simple}for the P2P-VoD system.
\begin{itemize}
    \item \textit{Peer’s authentication}: It ensures that the claiming (receiver) peer is the legitimate peer in the network for accessing the online requested services.
    \item \textit{Access control}: It ensures that only authorized peers can access the online requested services. 
    \item \textit{Data integrity/privacy}: The video data content should not be altered or modified during transmission. That is, the content sent from the services provider is the same as the content received from the peer.
    \item \textit{Data confidentiality}: The video data content remains confidential other than the intended peer. 
    \item \textit{Replay attack}: It is to ensure that the valid video content transmission is not maliciously repeated or delayed. 
    \item \textit{Pollution attack}: An attacker may intentionally modify the received video chunks or forge video chunks and injects the polluted chunks into the video content which degrades the quality of video content on the recipient side.
    \item \textit{Flooding attack}:  It is a kind of denial of service attack. It occurs when the burdens of the Service providers, KPSs, and peers are bottlenecked under the condition that a flood of users may request key distribution and video encryption at the same time.
\end{itemize}

\subsection{System Architecture for Secure IoCAHC System}

Here, we discuss the network architecture of the proposed IoMT-based cloud-assisted healthcare (IoCAHC) system, illustrated in Figure (\ref{fig6.2}). 

\begin{definition}
(\textbf{IoMT-based healthcare system}). The proposed cloud-assisted IoMT-based healthcare system consists of six entities.  
\begin{itemize}
    \item \textit{Network manager (NM)}. The NM is a semi-trusted authority that computes its master and public key. It authenticates an entity and issues a partial private key to it. 
    \item \textit{Key protection servers (KPSs)}. KPSs protect the private key of an entity using their secret keys and issue a protected private key share to it. They perform computations on the cloud to mitigate the computation overhead.   
    \item \textit{Bio-medical sensor (BMS)}. It is a tiny sensor that has limited storage space, battery life, and computation power. It is installed either on/outside the patient’s body (wearable sensors) or deployed in the patient’s tissues (implanted sensors). 
    \item \textit{Personal assisted device (PAD)}. It is a data sink that has sufficient computation and storage but is not trustworthy as it is effortless for an attacker to retrieve the patient’s sensitive data by physically stealing the phone or statistically attack on it. 
    \item  \textit{Medical cloud server (MCS)}. It is a semi-trusted cloud server, which stores the patient’s PHI. It also validates the PHI and provides the accessibility of PHI to SD.
    \item \textit{Server Device (SD)}. A device on the medical institution side, can access the patient’s PHI on MCS and diagnoses the patient’s diseases based on their resulting PHI.
\end{itemize}
	 
\end{definition}

The proposed system achieves secure communication in the following phases: securing sensor devices within BAN, securing communication within BAN, and outside the BAN, and using the EF-IDASC scheme for secure data transmission between BMS and the server device.

\begin{figure}
  \centering
  \includegraphics[width=1\linewidth]{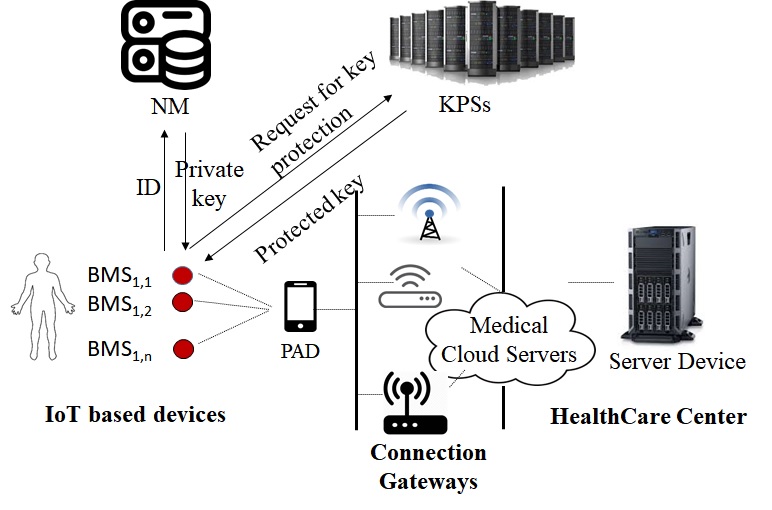}
  \caption{Detailed architecture of proposed secure aggregated signcryption system for cloud-centric IoMT-based healthcare system.}
\label{fig6.2}
\end{figure}

\subsection{Definition of Proposed EF-IDASC Scheme}

\begin{definition}
(\textbf{EF-IDASC}). The proposed EF-IDASC scheme includes the following five algorithms.
\begin{itemize}
    \item \textit{Setup}: Identical to setup phase of P2P-VoD system.
    \item \textit{KeyGen}: Identical to the KeyGen phase of the P2P-VoD system, including three sub-algorithms: ParPriKey, SecPriKey and KeyExt. 
    \item \textit{Signcrypt}: Identical to signcrypt phase of P2P-VoD system and outputs the signcryptext $CT$.
    \item \textit{Aggregation}: On given signcryptext $CT$ from different senders, it outputs the aggregated signcryptext $CT_{agg}$. 
    \item \textit{Unsigncrypt}: On given $pp$, and aggregated signcryptext $CT_{agg}$, recipient decrypts it using his private key $d_R$ to obtain the messages and sender’s $ID_S$, and validates the messages $M$.
\end{itemize}

The proposed IDASC is secure if it is IND-ID-CCA and EUF-CMA secure. The definitions are identical to the definitions of IND-ID-CCA and EUF-CMA, given in section 6.1.3.
\end{definition}

\subsection{Security Goals for IoCAHC System}
According to security goals defined in \cite{pantazis2012energy}, we consider the following security attacks for the proposed systems. 
\begin{itemize}
    \item \textit{Eavesdropping}. An unauthorized party could not listen to sensitive data in the WBAN communication.
    \item \textit{Data modification}. It ensures that data cannot be altered or modified by any adversary.
    \item \textit{Non-traceability}: It ensures that an adversary cannot trace the patient’s identity and any action on given encrypted data. 
    \item \textit{Mutual authentication}. The BMS and SD are authenticated to each other, which ensures that the PHI is coming from the intended BMS and arriving at the SD. 
    \item \textit{Contextual privacy}. Any entity in the system cannot link the PHI's source and destination if they do not collude. 
	Resilient to key escrow. The NM in the system could not compute the BMS’s signcrypted key. 
	\item \textit{Public auditing}. Anyone in the network can check the integrity of data stored on the cloud without downloading and knowing the actual content of the data.
\end{itemize}

\section{Peer-to-Peer Secure VoD Streaming System}

Here, we discuss implementing a secure P2P-VoD system, which consists of the following three phases: system initialization, connection establishment and secure video streaming. 

\subsection{System Initialization}

SP and KPSs are accountable for system initialization, which includes system setup and peer registration.  

\begin{itemize}
    \item \textit{System Setup}: Given a security parameter $k$, let two group $\mathbb{G}_1$, and $\mathbb{G}_2$ of same order $q$, where $q$ is large prime of $k$-bit, and $P$ be the generator of $\mathbb{G}_1$. Suppose $e$ be a pairing functions  $e: \mathbb{G}_1  \times \mathbb{G}_1 \rightarrow \mathbb{G}_2$ and three one-way cryptographic hash functions are,  $H_1:\{0,1\}^*  \rightarrow \mathbb{G}_1$, $H_2: \mathbb{G}_2 \rightarrow \{0,1\}^m \times \mathbb{G}_1 \times \{0,1\}^t$, and $H_3:\mathbb{G}_1  \times \{0,1\}^m \times \{0,1\}^t  \rightarrow \mathbb{Z}_q^*$. SP picks a random number $s_0 \in \mathbb{Z}_q^*$ known as a master secret key, sets the public key as $P_0=s_0P$, and sends $P_0$ to all KPSs. On a given $P_0$, each $KPS^i$ chooses an integer $s_i  \in \mathbb{Z}_q^*$, known as its secret key, sets public key $P_i=s_iP_0$ and sends $P_i$ to the SP and keep safe $s_i$. The SP combines to form the global public key $Y$, $Y=\sum_{i=1}^nP_i =\sum_{i=1}^ns_i P_0 = s_0 (s_1+s_2+ \dots +s_n )P$. The SP published the public parameter, $pp=<q,e,P,P_0,\mathbb{G}_1,\mathbb{G}_2,H_1,H_2,H_3,Y,P_0,..,P_n>$, and keeps $s_0$ secret.
	\item \textit{New Peer’s Registration}: When a new peer wants to join the network, it requests the SP and KPSs for its private key (credential), defined by the following steps:
	\begin{itemize}
	    \item \textit{ParKeyIss}: The peer (sender) chooses an element $x_S \in \mathbb{Z}_q^*$, sets  $X_S=x_S P$, and  $D_S=x_S Q_S$, where $Q_S=H_1(ID_S)$. The peer sends parameters $<X_S,ID_S,D_S>$ and asks to the SP for its partial private key. The SP validates the peer identification $ID_S$ and received parameters using $e(Q_S,X_S) \overset{?}{=}e(D_S,P)$. If the above equation satisfied, SP sets $D_{S0}=s_0 D_S$, $X_{S0}=s_0 X_S$ and responds to the peer with $< D_{S0},X_{S0}>$. The peer accepts the given parameter $< D_{S0},X_{S0}>$, if the equation $e(D_S,P_0) \overset{?}{=} e(D_S0,P)$ holds.
	    \item \textit{ParKeySec}: The peer requests to $KPS^i$ for protection by sending the parameters $< D_{S0},X_{S0}>$. $KPS^i$ accepts $<D_{S0},X_{S0}>$ if $e(Q_S,X_{S0}) \overset{?}{=}e(D_{S0},P)$ holds. Computes the protected private key share $D_{Si}=s_iD_{S0}$, and responses $D_{Si}$ to the peer.
	    \item \textit{KeyExt}: The peer accepts private key share $D_{Si}$, if  $e(D_S,P_i) \overset{?}{=} e(D_{Si},P)$ holds. Peer then extracts its private key $d_S$, $d_S=x_S^{-1} \sum_{i=1}^nD_{Si} =s_0(s_1+s_1+ \dots +s_n )Q_S$. Similarly (receiver) peer computes its private key $d_R$.
	\end{itemize}
\end{itemize}

\subsection{Connection Establishment}

This phase establishes a shared secret key between two peers in the network. This algorithm makes a connection between a (sender) peer and (receiver) peer and establishes a shared session key $p2p_{sk}$ between them as follows:
\begin{itemize}
    \item From his private key $d_S$, identity $ID_S$, receiver’s identity $ID_R$ and pp, the (sender) peer chooses a random element as $a \in \mathbb{Z}_q^*$, and computes $A=aQ_S$, where $Q_S=H_1 (ID_S)$.
    \item Compute session key as $p2p_{sk}=e(d_S,Q_R)^a$, where $Q_R=H_1 (ID_R)$.
    \item Send parameter $A$ to (receiver) peer.
    \item On given parameter $A$, and its private key $d_R$, (receiver) peer computes a shared session key as $p2p_{sk}= e(d_R,A)$.
\end{itemize}
	
It is noted that the P2PKeyExt algorithm is similar to the offline signcryption of a generic offline/online identity-based signcryption scheme with some additional properties. In related existing schemes, the offline signcryption algorithm computes expensive operations independent of the video chunks. It stores the pre-computed parameters for online signcryption whenever any peer requests the video chunk in the network. On the other hand, the P2PKeyExt algorithm in our proposed scheme also computes the expensive operations and stores the computed parameters for signcrypting the chunks. Besides, it establishes a shared session key between the network's two peers that will help signcrypting video chunks during signcrption algorithm. Thus, the P2PKeyExt algorithm performs the offline signcryption and establishes a shared secret key between two peers. Because of this reason, the offline computation cost is high. 

\subsection{Secure Video-On Demanding Streaming}

The proposed P2P-VoD system consists of three sub-algorithms: P2PKeyExt, signcrypt, and unsigncrypt, in which a (sender) peer divides the video content into chunks, signcrypts each chunk with the shared key $p2p_{sk}$ and transmits the signcrypted chunks to (receiver) peer where (receiver) peer unsigncryptes the signcrypted chunk with its shared key $p2p_{sk}$.
\begin{itemize}
    \item \textit{Signcrypt}: The (sender) peer signcrypts the video chunks $chk \in \{0,1\}^m$ using its private key $d_S$ and session key $p2p_sk$. It set $h=H_3(chk,A,T)$,  $C=(a+h)d_S$ and $D=H_2 (p2p_sk) \oplus chk||T||C$, where $T \in \{0,1\}^t$ denotes timestamp. It sends $CT=<D,A>$  to the (receiver) peer. 
	Unsigncrypt: The (receiver) peer extract $chk||T||C =D  \oplus H_2 (p2p_{sk})$, set $h=H_3(chk,A,T)$ and accepts $chk$, if Equation (\ref{eq6.2}) holds.
	\begin{equation} \label{eq6.2}
	    e(C,P)=e(Y,A+hQ_S )
	\end{equation}
\end{itemize}

This completes the full implementation of our proposed secure P2P-VoD system.

\section{System Analysis}

\subsection{Security Proof}

\begin{theorem} \label{thm6.1}
(\textbf{consistency}). The proposed P2P-VoD scheme is consistent.
\end{theorem}

\begin{proof}
Proof. In order to define the consistency of the P2P-VoD system, we discuss the consistency of session key $p2p_{sk}$ in  P2PKeyExt and unsigncryption algorithm.  

The consistency of session key $p2p_{sk}$ is verified as 

\vspace{-12mm}

\begin{align*}
    p2p_{sk}&=e(d_S,Q_R )^a=e(ad_S,Q_R ) \\
    &=e(as_0 (s_1+s_2+ \dots +s_n ) Q_S,Q_R ) \\
    &=e(aQ_S,s_0 (s_1+s_2+ \dots +s_n ) Q_R )
    &=e(A,d_R )=p2p_{sk}
\end{align*}
\vspace{-5mm}

The consistency of unsigncryption is verified as
\vspace{-12mm}

\begin{align*}
    e(C,P)=& e((a+h)d_S,P)\\
    & =e((a+h)s_0 (s_1+s_2+ \dots +s_n ) Q_S,P) \\
    & =e((a+h)Q_S,s_0 (s_1+s_2+ \dots +s_n )P) \\
    & =e(A+hQ_S,Y)
\end{align*}
\vspace{-5mm}

This proves the consistency of our proposed scheme.
\end{proof}

\begin{theorem} \label{thm6.2}
The proposed EF-IDSC scheme is $(q_1,q_2,q_3,q_{pp},q_p,q_{sc},q_{us},t,\epsilon)$-IND-CCA secure in the ROM, assuming the $(\epsilon',t')$-BDHP assumption, where $t'=t+t_B+q_1+q_2+q_3+ q_{pp}+ q_p+ q_{sc}+ q_{us}$, $\epsilon'=\epsilon/((q_1+q_{pp}q_P^n)q_3)(1-q_{sc}(q_2+2q_{sc})/2^k)$, and $q_1$, $q_2$, $q_3$, $q_{pp}$, $q_p$, $q_{sc}$, $q_{us}$ are the numbers of $H_1$, $H_2$, $H_3$, partial private key, private key, signcrypt queries, and unsigncrypt queries, respectively.
\end{theorem}

\begin{proof}
Suppose the generator of $\mathbb{G}_1$ be $P$, and the tuple $<P,aP,bP,cP>$ be a random instance of the BDH problem given to $B$. To break the IND-CCA security of the proposed scheme, $B$ helps the adversary $Adv$ solve the BDH problem, which is equivalent to computing $e(P,P)^{abc}$. Thus, $B$ responds to the following queries from $Adv$ in the IND-CCA game.

\textbf{Setup}. Suppose three random oracle hash functions are $H_1$, $H_2$, and $H_3$. Let $B$ maintain the following lists: $L_i$ for $H_i$ query, $L_{pp}$ for partial private key query, $L_p$ for private key extraction query, $L_{sc}$ for signcrypt query, and $L_{us}$ for unsigncrypt query. $B$ sets $Y=cP$ and sends it to $Adv$. $B$ randomly chooses a distinct number $j \in {1,2,\ldots,q_1}$.

\begin{itemize}
\item \textit{$H_1$ query}: $Adv$ asks an $H_1$ query on $ID_i$. On a given $ID_i$, $B$ sets $q_i=xP$ and $D_{0i}=xY$. If $ID_i \notin L_1$, $B$ adds the tuple $<ID_i,q_i,D_{0i},x>$ to $L_1$. Otherwise, $B$ outputs $q_Si=H_1(ID_{Si})=a'P$ and $q_Ri=H_1(ID_Ui) = b'P$ to $Adv$, where $a',b' \in \mathbb{Z}_q$.
\item \textit{$H_2$ query}: $Adv$ asks an $H_2$ query on $<A_i,chk_i,T_i>$. $B$ computes $H_2(A_i,chk_i,T_i)=h_i$. If the tuple $<A_i, chk_i,T_i>$ is not in $L_2$, $B$ adds the tuple $<A_i,chk_i,T_i,h_i>$ to $L_2$. Otherwise, $B$ outputs $h_i$ to $Adv$.
\item \textit{$H_3$ query}: $Adv$ asks an $H_3$ query on $<C_i>$. $B$ chooses $\lambda_i \in \mathbb{Z}_q$ and sets $H_3(C_i)=\lambda_i$. If the tuple $<C_i>$ is in $L_3$, $B$ adds the tuple $<C_i,\lambda_i>$ to $L_3$. Otherwise, $B$ outputs $\lambda_i$ to $Adv$.
\end{itemize}

\textbf{Phase 1}. The following queries are asked by $Adv$:

\begin{itemize}
\item \textit{ParKeyIss query}: Given $q_i$, $B$ aborts the process if $q_i$ has been previously asked; otherwise, $B$ outputs $D_{0i}$ to $Adv$.
\item \textit{PriKeySec query}: Given $D_{0i}$, $B$ aborts the process if $D_{0i}$ has been previously asked; otherwise, $B$ picks $s'k \in \mathbb{Z}q$ and sets $D{ki}=s'{ki}D_0$, where $1 \leq k \leq n$ denotes the $k$-th share of the private key. $B$ responds with $D_{ki}$ to $Adv$.
\item \textit{KeyExt query}: Given $D_{ki}$, $B$ aborts the process if $D_{ki}$ has been previously asked. Otherwise, $B$ computes $d_A=aP$ and $d_B=bP$ and responds with $d_B$ to $Adv$.
\item \textit{Signcrypt query}: For given data chunks $chk$, $ID_S$, and $ID_R$, $Adv$ asks the signcrypt query as follows:
\begin{itemize}
\item If $ID_S$ and $ID_R$ have not been previously queried, $B$ performs the following steps:
\begin{itemize}
\item Choose $r \in \mathbb{Z}_q^*$ and set $A= rQ_S$, $h=H_2(chk,B)$, $C=(r+h)d_S$, and $D=H_2(B) \oplus (chk||C||ID_S)$.
\item Respond with $<A,D>$ to $Adv$.
\end{itemize}
\item If $ID_S$ is queried but $ID_R$ has not been queried previously, $B$ already knows the receiver's private key $d_R$ and responds as follows:
\begin{itemize}
\item Pick numbers $r,h \in \mathbb{Z}_q^*$, set $A= rP-hQ_S$, $C= rQ_S$. Add $<A,m,h>$ to list $L_2$.
\item Find the entry $<ID_R,q_R,D{0R},x>$ in $L_1$, set $B=e(A,d_R)$, $D=H_2(B) \oplus chk||C||T$.
\item Respond with $<A,D>$ to $Adv$.
\end{itemize}
\item If $ID_S$ and $ID_R$ have been queried previously, $B$ knows the receiver's private key $d_R$ and responds as follows:
\begin{itemize}
\item Select the elements $r,h \in \mathbb{Z}q^*$ and set $A= rP-hQ_S$, $C= rY$. Add $<A,pkt,h>$ to list $L_2$.
\item Choose $\lambda \in {0,1}^m \times \mathbb{G}1 \times{0,1}^n$, set $D=\lambda \oplus chk||C||T$, and add $<ID_S,ID_R,A,B,C,r,chk,h,\lambda>$ to list $L{sc}$.
\item Respond with $<A,D>$ to $Adv$.
\end{itemize}
\end{itemize}
\item \textit{Unsigncrypt query}: For the given ciphertext $CT={A,D}$, $Adv$ asks the unsigncrypt query and performs the following steps:
\begin{itemize}
\item If $ID_R$ has not been previously queried, $B$ already knows $d_R$. $B$ may execute the unsigncrypt query to obtain $pkt$ and verify the signature. If it does not satisfy the verification, the process is aborted; otherwise, $B$ sends $chk$ to $Adv$.
\item If $ID_S$ and $ID_R$ have not been queried previously, $B$ already knows $d_S$ and responds as follows:
\begin{itemize}
\item Find the entry $<ID_R,q_R,D{0R},x> \in L_1$, set $B=e(A,d_R)$.
\item If $B \in L_3$, $ID_S \in L_1$, and $<A,chk> \in L_2$, compute $C||T||chk \leftarrow y \oplus H_3(B)$, $q_1 \leftarrow H_1 (ID_S)$, and $h \leftarrow H_2(A,chk)$.
\item Check if $e(C,P)=e(Y,A+hq_1)$, and respond with $chk$, $(A, C)$, and $ID_S$.
\item Otherwise, find the entry $<ID_S,ID_R,A,B,C,r,h,chk,\lambda> \in L_{sc}$, and compute $C||T||chk \leftarrow y \oplus h$, $q_1 \leftarrow H_1(ID_S)$, and $h \leftarrow H_2 (A,chk)$. Check if $e(C,P)=e(Y,A+hq_1)$, and respond with $chk$, $(A, C)$, and $ID_S$.
\end{itemize}
\end{itemize}
\end{itemize}

\textbf{Challenge}. $Adv$ responds with two chunks $chk_0$, $chk_1$, and identities $ID_S^{ch}$, $ID_R^{ch}$ to be challenged.

\begin{itemize}
\item If $ID_R^{ch}$ has not been previously queried, $B$ fails.
\item Otherwise, $B$ selects $b \in {0,1}$ to signcrypt $chk$ and responds to $Adv$.
\end{itemize}

\textbf{Phase 2}. Similar to Phase 1, $Adv$ extracts queries, but $Adv$ is not allowed to extract the unsigncrypt query for $CT^{ch}={chk_b,A^{ch},C^{ch}}$ under identities $ID_S^{ch}$ and $ID_R^{ch}$.

\textbf{Guess}. Finally, $Adv$ guesses a bit $b' \in {0,1}$. $Adv$ wins the game if $b=b'$. To solve the BDH problem, $B$ can retrieve the $L_{sc}$ for the tuple $<ID_S^{ch},ID_R^{ch},A^{ch},B^{ch},C^{ch},h,chk_b,\lambda>$, and compute $Z=(B^{ch}/(e(C^{ch},Y)))^{-1/h}$, where

\begin{align*}
Z=(B^{ch}/e(C^{ch},Y))^{-1/h}&=(e(rP-hQ_S^{ch},d_R^{ch})/e(C^{ch},Y) )^{-1/h} \
&=((e(rP,bcP)e(-hQ_S^{ch},d_R^{ch}))/e(rbP,cP))^{-1/h} \ &=e(-hQ_S^{ch},d_R^{ch})^{-1/h} \
&=e(aP,bcP)=e(P,P)^{abc}
\end{align*}

Thus, $B$ responds with $Z=e(P,P)^{abc}$ as the solution to the given BDH problem.

\textbf{Analysis}. The probability that $B$ does not abort the game is defined by four events:

\begin{itemize}
\item \textbf{$E_1$}: $Adv$ has asked for the partial and private key on $ID_i$ with probability $Pr[E_1]=(1/(q_1+q_{pp}q_P^n))$.
\item \textbf{$E_2$}: $Adv$ does not select $ID_i$ as the receiver during the challenge phase, and from $E_1$ we can deduce $\neg E_2$.
\item \textbf{$E_3$}: $B$ stops responding to $Adv$'s signcrypt query when there is a collision in $H_2$ or $H_3$ with probability $Pr[E_3]=1/q_3$.
\item \textbf{$E_4$}: $B$ aborts the valid ciphertext with probability $Pr[E_4]=(q_2+2q_{sc})/2^k$.
\end{itemize}

The probability that $B$ solves the BDH problem is given by:

\begin{equation}
A(Adv) \geq \frac{\epsilon}{(q_1+q_{pp}q_P^n)q_3(1-q_{sc}(q_2+2q_{sc})/2^k)}
\end{equation}
\end{proof}

\begin{theorem} \label{thm6.3}
The proposed EF-IDSC scheme is $(q_1,q_2,q_3,q_{pp},q_p,q_{sc},q_{us},t,\epsilon)$-EUF-CMA secure in the ROM, assuming the $(\epsilon',t')$-BDHP assumption, where $t'=t+t_B+q_1+q_2+q_3+q_{pp}+q_p+q_{sc}+q_{us}$, $\epsilon'=\frac{\epsilon}{(4(q_1+q_{pp}q_P^n)^2(q_2+2q_{sc}))(1-q_{sc}(q_2+2q_{sc})/2^k)^2}$, and $q_1$, $q_2$, $q_3$, $q_{pp}$, $q_p$, $q_{sc}$, $q_{us}$ are the numbers of $H_1$, $H_2$, $H_3$, partial private key, private key, signcrypt queries, and unsigncrypt queries, respectively.
\end{theorem}

\begin{proof}
Consider a forger $F$ against our proposed system. In order to forge a signature in the proposed system, we construct a simulator B that helps $F$ for solving the BDH problem with the probability of the least $\epsilon$. Identical to Theorem 6.1, $F$ runs $H_1$, $H_2$, $H_3$, partial private key, private key, signcrypt and unsigncrypt queries. $F$ asks the simulator $B$ in the same way as in Theorem 6.1. So, $B$ outputs the following queries for $F$ in the IND-CCA game.

\textbf{Setup}. Suppose $H_1$, $H_2$ and $H_3$ are three random oracle models. Let $B$ maintain the following lists: $L_i$ for $H_i$ query, $L_{pp}$ for partial private key query, $L_p$ for private key extraction query, $L_{sc}$ for Signcrypt query and list $L_{us}$ for unsigncrypt query. $B$ set $Y=cP$ and sends it to $F$. $B$ randomly chooses a distinct number $j \in \{1,2,..,q_1\}$. 
\begin{itemize}
    \item \textit{$H_1$ query}. $F$ asks $H_1$ query on $ID_i$. On given $ID_i$, $B$ set $q_i=xP$, and $D_{0i}=xY$, if $ID_i \notin L_1$, then adds the tuples $<ID_i,q_i,D_{0i},x>$ in $L_1$. Else, $B$ outputs $q_Si=H_1(ID_{Si})=a'P$ and $q_Ri=H_1 (ID_{Si})=b'P$ to $F$, where $a',b' \in \mathbb{Z}_q$.
    \item \textit{$H_2$ query}. $F$ asks $H_2$ query on $<A_i, chk_i,T_i>$. $B$ computes $H_2(A_i,chk_i,T_i)=h_i$,  if the tuple $<A_i,chk_i,T_i> \notin L_2$ then adds the tuples $<A_i,chk_i,T_i,h_i>$  in list $L_2$. Else, $B$ outputs $h_i$ to $F$.
    \item \textit{$H_3$ query}. $F$ asks $H_3$ query on $<C_i>$. $B$ chooses $\lambda_i\in \mathbb{Z}_q$ and set $H_3(C_i) = \lambda_i$,  if the tuple $<C_i> \in L_3$ and adds the tuples $<C_i,\lambda_i>$ in list $L_3$. Otherwise, $B$ outputs $\lambda_i$ to $F$.
\end{itemize}

\textbf{Phase 1}. The following queries are asked by the $F$.
\begin{itemize}
    \item \textit{ParKeyIss qyery}: On given $q_i$, $B$ abort the process, if $q_i$ is previously asked, otherwise, outputs $D_{0i}$ to $F$.
	\item \textit{PriKeySec query}. On given $D_{0i}$, $B$ abort the process, if $D_{0i}$ is previously asked. Otherwise, $B$ picks $s'_k \in \mathbb{Z}_q$ and set $D_{ki}=s'_{ki}D_0$, where $1 \le k \le n$,  denotes the $k$-th share of private key. $B$ responds $D_{ki}$  to the $F$.
	\item \textit{KeyExt query}. On given $D_{ki}$, $B$ aborts the process, if  $D_{ki}$ is previously asked. Otherwise, $B$ computes $d_S=aP$ and $d_R=bP$, and responds $d_R$ to the $F$. 
	
	\item \textit{Signcrypt query}.  For given data chunks $chk$, $ID_S$ and $ID_R$, $F$ asks the signcrypt query as follows.
	\begin{itemize}
	    \item If $ID_S$ and $ID_R$ are not previously queried, $B$ performs the steps. 
	    \begin{itemize}
	        \item Choose $r \in \mathbb{Z}_q^*$, set $A= rQ_S$, $h=H_2 (chk,B)$, $C=(r+h)d_S$ and $D=H_2(B) \oplus chk||C||ID_S$. 
	        \item Responds $<A,D>$ to the $F$. 
	    \end{itemize}
	    \item If $ID_S$ is queried but $ID_R$ is not queried previously, $B$ already knows the receiver’s private key $d_R$ and responds as follows. 
	    \begin{itemize}
	        \item Selects random numbers $r,h \in \mathbb{Z}_q^*$, set  $A= rP-hQ_S$,  $C= rQ_S$. Add $<A,m,h>$ in list $L_2$.
	        \item find the entry $<ID_R,q_R,D_{0R},x>$ in $L_1$,  set $B=e(A,d_R)$, set $D=H_2(B) \oplus  chk||C||T$. 
	        \item Responds $<A,D>$ to $F$. 
	    \end{itemize}
        \item If  $ID_S$ and $ID_R$ are queried previously, $B$ knows the receiver’s private key $d_R$ and responds as follows.
        \begin{itemize}
            \item 	Select the elements $r,h \in \mathbb{Z}_q^*$, and set  $A= rP-hQ_S$,  $C= rY$. Add $<A,pkt,h>$ in the list $L_2$. 
            \item Chooses $\lambda \in \{0,1\}^m \times \mathbb{G}_1 \times \{0,1\}^n$,  set $D=\lambda \oplus chk||C||T$,  add $<ID_S,ID_R,\\ A, B,C,r,chk,h,\lambda>$ to the list $L_{Sc}$. 
            \item Responds $<A,D>$ to $Adv$. 
        \end{itemize}
	\end{itemize}
	\item \textit{Unsigncrypt query}. For given ciphertext $CT=<A,D>$, $F$ asks the unsigncrypt query and performs the following steps. 
	\begin{itemize}
	    \item If $ID_R$ is not previously queried, $B$ already knows $d_R$. $B$ may execute the unsigncrypt query to obtain the $pkt$ and verify the signature. If it does not satisfied the verification is aborted, otherwise, sends $chk$ to $F$. 
	    \item If  $ID_S$ and $ID_R$ is not queried previously, $B$ already knows the $d_S$ and responds as follows:
	    \begin{itemize}
	        \item Find the entry $<ID_R,q_R,D_{0R},x> \in L_1$, set $B=e(A,d_R)$. 
	        \item If $B \in L_3$, $ID_S \in L_1$ and $<A,chk> \in L_2$, compute $C||T||chk \leftarrow y \oplus H_3(B)$,  $q_1 \leftarrow H_1 (ID_S)$, and  $h \leftarrow H_2 (A,chk)$.
	        \item check if $e(C,P)=e(Y,A+hq_1)$, and responds $m$, $<A, C>$ and $ID_S$,
	        \item Else, find the entry $<ID_S,ID_R,A,B,C,r,h,chk,\lambda> \in L_{sc}$, and compute $C||T||chk \leftarrow y \oplus h$, $q_1 \leftarrow H_1(ID_S)$ and $h \leftarrow H_2 (A,chk)$.
            \item check if $e(C,P)=e(Y,A+hq_1)$, and responds $chk$, $<A, C>$ and $ID_S$,
	    \end{itemize}
	\end{itemize}
\end{itemize}

\textbf{Forgery}. Forger $F$ responds the ciphertext  $CT^{ch}=<A^{ch},C^{ch}$ as valid signcryption under identities $ID_S^{ch}$  and $ID_R^{ch}$. If $ID_S^{ch}$ is not previously queried, the simulation is aborted. Otherwise, $B$ set $A= rP-hQ_S$ and $C= rQ_S$ by running the $H_1$ and $H_2$ query. 

Then, we have $Z=(B^{ch}/(e(C^{ch},Y)))^{-1/h}=e(P,P)^{abc}$ as the solution of the given BDH problem. It can be noticed that if a valid tuple does not found in $L_2$, then $F$ will not have any advantage.

Thus, the probability that $B$ solve the BDH problem is:

\begin{equation}
    A(F) \ge \frac{\epsilon}{(4(q_1+q_{pp} q_P^n )^2 (q_2+2q_{sc}))(1-q_{sc}(q_2+2q_sc)/2^k)^2}
\end{equation}

\end{proof}

\subsection{Security Requirements }

This section discusses the security requirements that our proposed P2P-VoD scheme has achieved.

\begin{itemize}
    \item \textit{Peer authentication}: In our proposed scheme, a peer requests to the SP for the subscription, in which SP authenticates the peer against his $ID$ and gives a partial private key to him. Followed by multiple KPSs providing the protected private keys to him. Thus, SP and multiple KPSs authenticate/register the peer. 
    \item \textit{Access control}: A (sender) peer, in the proposed P2P-VoD system, signcrypts the video chunk using a pair of the session key $p2p_{sk}$, private key $d_S$ and recipient’ peer $ID_R$. It broadcasts the signcrypted chunk in the network which could be unsigncrypted by that peer in the network who has pre-established session key $p2p_{sk}$. The intended (receiver) peer who has previously established a session key $p2p_{sk}$ can access the chunks by unsigncrypting the unsigncrypted chunks. Suppose an adversary or a malicious peer in the network wish to access the video content. He must first guess the private key $d_R$ of the (receiver) peer to compute the session key $p2p_sk$, which is equivalent to solving the BDH problem, as given in Theorem 7.10. Thus, the proposed P2P-VoD system provides access control only to the intended peer in the network.
    \item \textit{Data integrity}: On the (sender) peer side, the video content chunk is encrypted with the session key $p2p_{sk}$, computed from his private key $d_S$ and (receiver) peer $ID_R$. The encrypted content is decrypted by the intended peer who has the session key $p2p_{sk}$ corresponding private key $d_R$. Suppose an adversary alters or modifies the encrypted video content over the internet which will be caught by the (receiver) peer during unsigncryption process. Thus, the proposed P2P-VoD system provides the integrity of video content. 
	\item \textit{Data confidentiality}: The encrypted video data chunk is only intended for the peer whose private key $d_R$ is corresponding to the encrypted key $ID_R$. Suppose an adversary who wishes to alter or modify the video data chunk, he requires to guess the peer’s decryption key, which is equivalent to solving the BDH problem, as given in Theorem 7.10. 
	\item \textit{Replay attack}: The signcryption process of the P2P-VoD scheme uses the timestamp $T$, which ensures that video data chunk transmission is not maliciously repeated or delayed. Thus, the proposed P2P-VoD system is resilient to replay attacks.
	\item \textit{Pollution attacks}. Theorems 7.10 and 7.11 prove that under the assumption of the BDH problem adversary can alter and forge the video chunks respectively with negligible advantage. Therefore, the proposed P2P-VoD system defends polluted data chunks available to other peers. 
    \item \textit{Flood attack}. In the proposed EF-IDSC scheme, the system initialization and secure key mechanism are two independent algorithms. When a new peer wishes to request a video from a neighbour peer, it must be pre-registered with the service provider and KPSs during the system initialization phase and gets a private key (credential) against his identity ID. After obtaining a private key from the service provider, the new peers now can push or pull the video content using the proposed EF-IDSC scheme.
\end{itemize}

\subsection{Performance Evaluation }

This section simulates an experiment on the proposed IF-IDSC scheme of the P2P-VoD system, which includes the storage, communication, and computation cost. The experiment is performed on the same machine as discussed in Chapter \ref{chapter3}. Besides, we consider the cryptographic operations and their computation costs, as mentioned in Table \ref{tbl3.1}. In order to analyze the bandwidth and storage cost, we consider $|\mathbb{G}_1|$ as the size of multiplication group $\mathbb{G}_1$, $|\mathbb{G}_2|$ as the size of multiplication group $\mathbb{G}_2$, $|\mathbb{Z}_q|$ as the size of $\mathbb{Z}_q$, $|id|$ as the size of identity, and $|chk|$ as the size of the video data chunk. 

\begin{table}
        \centering
        \caption{Computational, storage and communication cost comparison of our scheme with other related schemes}
        \label{tbl6.1}
        \begin{tabular}{|c|c|c|c|c|c|}
            \hline
            \multicolumn{1}{|c|}{Schemes} & \multicolumn{3}{|c|}{Computation cost (in ms)} & \multicolumn{2}{|c|}{Size (in bytes)} \\
            \cline{2-6}
             &	Offline &	Online  &	Unsigncrpt &	Offline  &	Signcryptext  \\
                        \hline
            \hline
            \cite{lai2017efficient} &33.35	&0.23	&60.03	&198	&108\\
            \cite{li2015certificateless} &38.18	&6.67	&104.88	&298	&252\\
            \cite{li2017certificateless} &31.51	&0.23	&64.84	&166	&110\\
            \cite{saeed2018hoosc} &31.51	&0.23	&51.52	&234	&178 \\
            \cite{omala2016provably} &---	&26.63	&26.68	&---	&86\\
            Our	    &45.08	&6.67	&46.67	&100	&76\\
            \hline
        \end{tabular}
\end{table}

\begin{table}
        \centering
        \caption{General comparison of our scheme with related schemes}
        \label{tbl6.2}
        \begin{tabular}{|c|c|c|c|c|c|}
            \hline
            Schemes	& Cryptographic   &	Security  &	Analysis  &	Preserve  &	Resilient to \\
            	&   Primitives &	 Assumption. &	 model &	ID-based &	Key Escrow \\
            		&    &	  &	  &	 Feature &	 problem \\
            \hline
            \hline
            \cite{lai2017efficient} & OOIDSC &	k-CAA1, q-BDHI$^@$, q-SDH$^\#$ &	ROM	& Yes &	No\\
           \cite{li2015certificateless} &OOCLSC &	UF-CMA-I, q-CAA	& ROM &	No &	yes\\
            \cite{li2017certificateless} &OOCLSC &	q-BDHI, q-SDH, mBDHIP{\#\#}	& ROM &	No &	Yes\\
            \cite{saeed2018hoosc} &HOOCLSC &	q-BDHIP, q-SDHP, ECDLP &	ROM &	No	& Yes \\
            \cite{omala2016provably} &CLSC	& CDH, GDHP	&ROM	& No &	Yes\\
            Our	    &EF-IDSC &	BDH &	ROM &	Yes &	Yes\\
            \hline
        \end{tabular}
        {\\$^@$q-Bilinear-Diffie Inverse Problem, $^\#$q-strong Diffie-Hellman problem, $^{\#\#}$ modified Bilinear Diffie-Hellman inverse problem,}
\end{table}

We assume three metrics for the performance evaluation to examine the experimental result: computational cost (in msec), storage cost (in Bytes) and communication cost (in Bytes). For the performance analysis, we compare the proposed IDSC scheme with Lai \textit{et al.} \cite{lai2017efficient}, Li \textit{et al.} \cite{li2015certificateless}, Li \textit{et al.} \cite{li2017certificateless}, Saeed \textit{et al.} \cite{saeed2018hoosc} and Omala \textit{et al.} \cite{omala2016provably} schemes.

\begin{figure}
  \centering
  \includegraphics[width=0.8\linewidth]{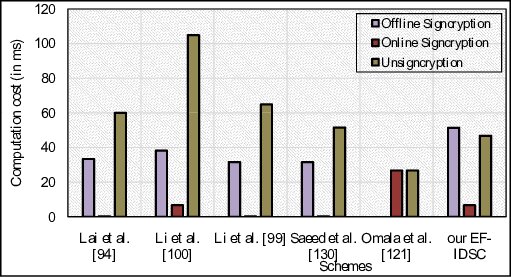}
  \caption{Offline, online and unsigncryption (in ms) comparison of proposed EF-IDSC scheme with other related schemes}
\label{fig6.3}
\end{figure}

\begin{figure}
  \centering
  \includegraphics[width=0.8\linewidth]{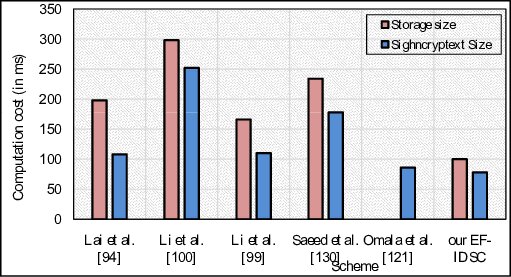}
  \caption{Storage and signcryptext size (in Bytes) comparison of proposed EF-IDSC scheme with other related schemes}
\label{fig6.4}
\end{figure}

\textbf{Computation cost}. The proposed scheme takes 2*20.01 + 1*0.23 + 1*4.83 = 45.08 ms for offline signcryption, while Lai \textit{et al.} \cite{lai2017efficient}, Li \textit{et al.} \cite{li2015certificateless}, Li \textit{et al.} \cite{li2017certificateless} and Saeed \textit{et al.} \cite{saeed2018hoosc} take 33.35 ms, 38.18 ms, 31.51 ms and 31.51 ms, respectively. In online signcryption, the proposed scheme consumes 1*6.67 = 6.67 ms while Lai \textit{et al.} \cite{lai2017efficient}, Li \textit{et al.} \cite{li2015certificateless}, Li \textit{et al.} \cite{li2017certificateless}, Saeed \textit{et al.} \cite{saeed2018hoosc} and Omala \textit{et al.} \cite{omala2016provably} schemes consume 0.23 ms, 0.23 ms, 0.23 ms, 26.68 ms and 6.67 ms, respectively. Similarly, the computational cost of our scheme with Lai \textit{et al.} \cite{lai2017efficient}, Li \textit{et al.} \cite{li2015certificateless}, Li \textit{et al.} \cite{li2017certificateless}, Saeed \textit{et al.} \cite{saeed2018hoosc} and Omala \textit{et al.} \cite{omala2016provably} schemes are 2*20.01 + 1*6.67 = 46.7 ms,  60.03 ms, 104.88 ms, 64.84 ms, 51.52 ms, and, 26.68 ms, respectively, illustrated in Figure (\ref{fig6.3}).

For evaluating storage and communication cost, we assume $|id|$ = 80 bits. For super-singular curve over the binary field  $\mathbb{F}_(2^{271})$ with the order of $\mathbb{G}_1$ is 252 bit prime and using compression technique [35], we consider $|\mathbb{G}_1|$ = 34 bytes, $|\mathbb{G}_2|$ = 34 bytes and $|\mathbb{Z}_q|$ = 32 bytes.

\textbf{Storage size (in Bytes)}. Figure (\ref{fig6.4}) demonstrates the storage cost (in Bytes) and signcryptext size (in Bytes) of our proposed scheme with other schemes. For offline storage,  Lai \textit{et al.} \cite{lai2017efficient} needs 198 bytes, Li \textit{et al.} \cite{li2015certificateless} needs 298 bytes, Li \textit{et al.} \cite{li2017certificateless} needs 166 bytes and Saeed \textit{et al.} \cite{saeed2018hoosc} needs 234 bytes while our scheme needs 100 bytes, summarize in Table \ref{tbl6.1}. 

\textbf{Signcryptext size (in Bytes)}. Figure (\ref{fig6.4}) illustrates the signcryptext size of the proposed scheme with other related schemes. It requires transmitting 108 bytes of messages in Lai \textit{et al.} \cite{lai2017efficient}, 252 bytes of messages in Li \textit{et al.} \cite{li2015certificateless}, 110 bytes of messages in Li \textit{et al.} \cite{li2017certificateless}, 178 bytes of messages in Saeed \textit{et al.} \cite{saeed2018hoosc} and 86 bytes Omala \textit{et al.} \cite{omala2016provably} scheme, while 76 bytes messages in our scheme, summarize in Table 6\ref{tbl6.1}. For transmitting a video chunk of 10 bytes, our proposed scheme saves 28\%, 69\%, 29\%, 56\% and 10\% of communication cost as compared to Lai \textit{et al.} \cite{lai2017efficient}, Li \textit{et al.} \cite{li2015certificateless}, Li \textit{et al.} \cite{li2017certificateless}, Saeed \textit{et al.} \cite{saeed2018hoosc} and Omala \textit{et al.} \cite{omala2016provably} schemes.

From Table \ref{tbl6.2}, it can be seen that the proposed EF-IDSC scheme is the only scheme that achieves identity-based features and is resilient to the key escrow attacks whose security is based on the well-known BDH problem under the random oracle model. The proposed scheme is more suitable for secure video streaming in the P2P-VoD system.

\section{Cloud-Centric Secure IoMT-Enabled Smart Healthcare System}
Here, we discuss implementing the proposed smart healthcare system, which consists of the following four phases. 

\subsection{System Initializations}

This section discusses the system setup and the entity’s registration. In Algorithm 6.1, NM and KPSs set up the proposed system, in which they generate the secret keys and system parameters. The secret keys are secret to them, and public parameters are broadcasted in the network. Algorithm 6.2 registers the new entity that wishes to add to the network. The NM authenticates the entity against its identity ID and issues a partial private key to it. The multiple KPSs protect the partial private key and forward the protected private key shares to the entity. Entity combines the shares and extracts its escrow-free private key. Since BMS has limited computation power, so the expensive computation can be outsourced on a powerful system. 

\begin{algorithm} 
	\caption{System setup} 
	\begin{algorithmic}[1]
	    \State Given a security parameter $k$, the NM picks an element $q$, a $k$-bit large prime number. Let an additive group $\mathbb{G}_1$, multiplicative group $\mathbb{G}_2$ of order of $q$,  $P$ be the generator of $\mathbb{G}_1$ and a pairing function  $e: \mathbb{G}_1  \times \mathbb{G}_1 \rightarrow \mathbb{G}_2$.
	    \State Suppose five one-way cryptographic hash functions are as follows. $H_1:\{0,1\}^l \rightarrow \mathbb{G}_1$, $H_2: \{0,1\}^* \rightarrow  \{0,1\}^{m+l+t} \times \mathbb{G}_1$, $H_3:\mathbb{G}_2  \times \{0,1\}^n \times \mathbb{Z}_q^*$, $H_4:\mathbb{G}_1^m \rightarrow \mathbb{Z}_q^*$, and $H_5:\mathbb{G}_1^n \rightarrow \mathbb{Z}_q^*$, where $m$, $l$ and $t$ denote the size message, identity and timestamp (in bits).
	    \State NM chooses an element $s_0 \in \mathbb{Z}_q^*$ and sets pubic key $P_0=s_0P$, and sends $P_0$ to $KPS^i$. 
	    \State $KPS^i$ chooses an element $s_i \in \mathbb{Z}_q^*$ and sets $P_i=s_i P_0$ and responses $P_i$ back to the NM.
	    \State NM combines all received parameters and computes the system public key $Y=\sum_{i=1}^n P_i = s_0 (s_1+s_2+\ldots+s_n)P$
	    \State NM keeps $s_0$ secret and published the public parameter $pp=<q,e,P,P_0,\mathbb{G}_1,\mathbb{G}_2,H_1,H_2,H_3,H_4,H_5,Y,P_1,P_2,..P_n>$.
    \end{algorithmic} 
    \label{alg6.1}
\end{algorithm}

\begin{algorithm} 
	\caption{Entity’s Authentication and registration} 
	\begin{algorithmic}[1]
	    \State Entity $E \in \{BMS, PAD, SD\}$ with identity $ID_E$,  picks an element $x_E \in \mathbb{Z}_q^*$, set  $X_E=x_E P$,  $D_E=x_E Q_E$, where $Q_E=H_1 (ID_E)$ and sends $<X_E,ID_E,D_E>$ to NM.
	    \State 	The NM validates the parameters using $e(Q_E,X_E) \overset{?}{=}e(D_E,P)$ compute the partial private key as  $D_{E0}=s_{D_E}$, $X_{E0}=s_0 X_E$, and responds $D_{E0}$ back to $E$
	    \State E aborts the process if $e(D_E,P_0) \overset{?}{=} e(D_{E0},P)$ does not hold.
	    \State Otherwise, $E$ accepts it and requests to $KPS^i$ for key protection.
	    \State $KPS^i$ aborts the process if above equation $e(Q_E,X_{E0}) \overset{?}{=}e(D_{E0},P)$  does not hold 
	    \State Otherwise, sends $D_{Ei}=s_iD_{E0}$  to $E$.
	    \State Entity E checks $D_{Ei}$ if equation $e(D_E,P_i)=e(D_{Ei},P)$ holds.
	    \State $E$ unblinds it and retrieves the private key $d_E=x_E^{-1} \sum_{i=1}^nD_{Ei} =s_0  (s_1+s_2+ \dots s_n ) Q_E$. 
	    \State Otherwise, abort the process. 
    \end{algorithmic} 
    \label{alg6.2}
\end{algorithm}

\subsection{Secure Data Communication from BMS to PAD}

This subsection discusses secure data communication from the BMS to PAD. Suppose a WBAN architecture consists of n bio-medical sensors (BMS) connected with a PAD. Let $M_{i,j}$ denotes the personal health information (PHI) collected by $j^{th}$ BMS at time $T_{i,j}$, where $1 \le i \le m$ and $1 \le j \le n$. The secure data transmission mechanism from $j^{th}$ BMS to the PAD is defined by Algorithm 6.3. In Algorithm 6.3, $j^{th}$ BMS collects the PHI $M_{i,j}$ at time $T_{i,j}$, signcrypts it using its private key $d_{BMS}^j$,  and sends the PHI $C_{i,j}$ in encrypted form to the PAD. In order to further improve the efficiency, $j^{th}$ BMS combines the signcryptext $C_{i,j}$ collected on different timestamps, signs it with its private key $d_{BMS}^j$ and stores it in the PAD. On receiving the PHI in encrypted form from $j^{th}$ BMS, the PAD first verifies the authenticity of data, using Eq. (7.5) defined in Algorithm 7.4 without accessing the original PHI and accepts the parameters. Otherwise, reject it.  

\begin{algorithm} 
	\caption{PHI Aggregate Signcryption } 
	\begin{algorithmic}[1]
	\State $j^{th}$ BMS chooses an element $a_j \in \mathbb{Z}_q^*$ and computes $A_j=a_j Q_{BMS}^j$, and $B_j=a_jP$ where $Q_{BMS}^j=H_1 (ID_{BMS}^j)$.
	\State $j^{th}$ BMS sets $K_j= e(a_j d_{BMS}^j,Q_{SD})$, where  $Q_{SD}=H_1(ID_{SD})$ and computes signcryption key as $S_j^k=H_2 (ID_{SD},A_j,K_j,S_j^{k-1})$. Here, $S_j^{k-1}$  is the previous key. 
	\State 	On given time-stamp $T_{i,j} \in \{0,1\}^t$ and PHI $M_{i,j} \in \{0,1\}^m$, $j^{th}$ BMS computes $h_{i,j}=H_3 (M_{i,j},A_j,T_{i,j})$ and $C_{i,j}=(a_j+h_{i,j})d_{BMS}^j$
	\State Sets $D_{i,j}=M_{i,j} ||C_{i,j} ||ID_{BMS}^j ||T_{i,j} \oplus S_j^k$, where $1 \le i \le m$ and $1 \le j \le n$. 
	\State 	Aggregates the signcypted PHI as $C_{aggr,j}=H_4 (C_{1,j},C_{2,j},..C_{m,j} )$ and $E_j=C_{aggr,j} d_{BMS}^j$
	\State Stores $CT_j=<A_j,B_j,C_{aggr,j},D_{i,j},E_j>$ in PAD.
	\State Secure data communication from PAD to MCS
    \end{algorithmic} 
    \label{alg6.3}
\end{algorithm}

This section discusses the secure data communication from the PAD to SD/MCS and the integrity of data on the MCS. PAD combines the aggregated signcryptexts received from multiple BMS into a single compact re-aggregated signcryptext, defined in Algorithm 6.4. In comparison with BMS, PAD, in our proposed model, has sufficient memory and computational power, so we could not bother about computation on PAD. It is to note that PAD is only allowed to verify the aggregated signcryptext and re-aggregated signcryptext without knowing the actual PHI data. Now, PAD offloads the re-aggregated data to the MCS. Algorithm 6.5 defines public verifiability; wherein anyone can verify the integrity of PHI without downloading the actual data from the MCS. Whenever the doctor needs the patient’s current health status, he fetches the encryption data from MCS. In Algorithm 6.6, SD decrypts the signcrypted data using an aggregate unsigncryption scheme to access the original PHI data. It can access the original PHI if the received parameters are successfully verified.

\begin{algorithm} 
	\caption{PHI Re-aggregation} 
	\begin{algorithmic}[1]
	\State Now, PAD collects signcrypted data $CT_j=<A_j,C_{aggr,j},D_{i,j},E_j>$ from $j^{th}$ BMS with $ID_{BMS}^j$, where $1 \le j \le m$. PAD re-aggregates them as
    \State $C_{PAD}=H_5 (C_{aggr,1},C_{aggr,2},..,C_{aggr,n})$
    \State 	Computes the $F=C_{PAD}d_{PAD}$.
    \State 	Re-aggregated signcryptext is $CT_{PAD}=<A_j,B_j,C_{PAD},C_{aggr,j},D_{i,j},E_j,F>$ and stores it on the MCS.
    \end{algorithmic} 
    \label{alg6.4}
\end{algorithm}

\begin{algorithm} 
	\caption{Public Verifiability} 
	\begin{algorithmic}[1]
    \State On stored encrypted data $CT_{PAD}=<A_j,B_j,C_{PAD},C_{aggr,j},D_{i,j},E_j,F>$ on the MCS, user can verify the integrity of data as follows. 
    \State	Checks the equality
    \State{
    \begin{equation} \label{eq6.4}
         e(F,P)=e( Q_{PAD},Y)^{C_{PAD}} 
    \end{equation}
    }
    \State	If yes, checks the equality
    \State{
    \begin{equation} \label{eq6.5}
         e(E_j,B_j)=e(A_j,Y)^{C_{aggr,j}}
    \end{equation}
    }
    \State	Accepts it. 
	\State Otherwise, Abort it. 
	    \end{algorithmic} 
	    \label{alg6.5}
\end{algorithm}

\begin{algorithm} 
	\caption{Aggregate Unsigncryption} 
	\begin{algorithmic}[1]
    	\State On given signcryptext $CT_{PAD}=<A_j,B_j,C_{PAD},C_{aggr,j},D_{i,j},E_j,F>$ from MCS, SD performs 
	    \State Sets $K'_j=e(d_{SD},A_j)$ and $S_j^k=H_2 (ID_{SD},A_j,K'_j,S_j^{k-1})$.
	    \State Decrypt PHI as $M_{i,j}||C_{i,j}||ID_{BMS}^j ||T_{i,j}=D_{i,j} \oplus S_j^k$.
	    \State 	Compute $h_{i,j}= H_3 (M_{i,j},A_j,T_{i,j} )$
	    \State 	Accept the PHI $M_{i,j}$ if the following conditions holds, 
	    \State{
	    \begin{equation}
	        e(\sum_{i=1}^m \sum_{j=1}^n C_{i,j},P)=e(\sum_{i=1}^m \sum_{j=1}^n(A_j+h_{i,j} Q_{BMS}^j) ,Y)
	    \end{equation}
	    }
	  \end{algorithmic} 
	  \label{alg6.6}
\end{algorithm}

The consistency of Equation (6.7) is verified as:
\vspace{-12mm}

\begin{align*}
    e(\sum_[i=1]^m \sum_{j=1}^n C_{i,j},P) &=e(\sum_{i=1}^m \sum_{j=1}^n(a_j+h_{i,j})d_{BMS}^j,P) \\
    & =e(\sum_{i=1}^m \sum_{j=1}^n(a_j+h_{i,j})Q_{BMS}^j,s_0  (s_1+s_1+\dots s_n )P)\\
    &=e(\sum_{i=1}^m \sum_{j=1}^n(A_j+h_{i,j}Q_{BMS}^j),Y)
\end{align*}
\vspace{-5mm}

Also,
\vspace{-12mm}

\begin{align*}
    K'&=e(d_{SD},A)=e(s_0 (s_1+s_1+ \dots s_n ) Q_{SD},aQ_{BMS})\\
    &=e(Q_{SD},as_0 (s_1+s_1+ \dots s_n ) Q_{BMS}) \\
    &=e(Q_{SD},ad_{BMS} )=K
\end{align*}

\subsection{Secure Data Communication from SD to BMS}
After accessing the patient's PHI from MCS, SD diagnoses the patient by sending the prescription P to the BMS in a secure way. Due to the resource-constrained behaviour of a BMS, we will leverage the PAD again, where SD signcrypts the prescription P using Algorithm 6.3 and sends it to MCS. On a particular time T, PAD fetches it and sends the signcrypted prescription to the target BMS. BMS runs Algorithm 6.6 to unsigncrypts the signcrypted prescription to obtain prescription P.

\section{System Analysis}

\subsection{Security Proof}

\begin{theorem}
\textbf{(IND-CCA)}. Under the assumption of random oracle model, assume a PPT adversary $A(q_i,q_p,q_{sc},q_{us},t, \epsilon)$ asking at most $q_i$ queries to oracles $H_i(i =1, 2, 3)$, $q_p$ private key extraction queries, $q_{sc}$ signcrypt queries, $q_{us}$ unsigncrypt queries and wins IND-CCA-I games with non-negligible advantage $\epsilon$, then there is an algorithm $B$ that can solve the BDH problem in polynomial time $t'$ with advantage: 

\centering
$Adv(A) \ge  \frac{\epsilon}{((q_1+q_P^n)q_3 )(1-q_{sc}(q_2+2q_{sc})/2^k   )}$

and,  
Time $t'=t+t_B O(q_1+q_2+q_3+ q_p+ q_{sc}+ q_{us})$

Where, $t_B$ is the running time for algorithm $B$. 
\end{theorem}

\begin{proof}
The proof is similar to Theorem 2 of \cite{Kumar2020IoMT}.
\end{proof}

\begin{theorem}
\textbf{(EUF-CMA).} Suppose a forger $F(q_1,q_2,q_3,q_p,q_{sc},q_{us},t,\epsilon)$ who can run no more than time $t$, $H_1$ queries no more than $q_1$, $H_2$ queries no more than $q_2$, $H_3$ queries no more than $q_3$, private key extraction queries no more than $q_P$, signcryption queries no more than $q_{sc}$, and unsigncryption queries no more than $q_{us}$, with an advantage more than equal to $\epsilon$. Under the adaptive chosen message and ID attacks, the proposed EF-IDASC scheme is secured against EUF-CMA-I attack, if no adversary $F(q_1,q_2,q_3,q_p,q_{sc},q_{us},t,k^{-n})$-breaks the scheme.

\centering
$Adv(F) \le \frac{\epsilon}{(4(q_1+q_P^n)^2(q_2+2q_{sc}))(1-q_{sc}(q_2+2q_{sc})/2^k)^2}$
\end{theorem}

\begin{proof}
The proof is similar to Theorem 3 of \cite{Kumar2020IoMT}.
\end{proof}

\subsection{Performance Evaluation}
This section evaluates the performance of the proposed smart healthcare system in terms of computational cost and communication cost. Our primary focus is to compute the energy consumption employed on the BMS side during data communication and computation since BMS is a resource-constrained device as compared to the machine on medical institutions. 

\textbf{Experimental result}: We perform the experiment on \textit{Acer E5-573-5108} laptop with \textit{Intel(R) Core(TM) i5-5200U CPU@2.20 GHz} and\textit{ 8 GB} RAM on Windows 10. With an 80-bit security level, all computations are estimated on a super-singular curve $y^2+y=x^3+x$ having embedding degree 4 and used the eta pairing, $\eta:E(\mathbb{F}_(2^{271})) \times E(\mathbb{F}_(2^{271})) \rightarrow E(\mathbb{F}_(2^{4.271}))$. In order to scrutinize the performance, we assume the following operations: elliptic curve (EC) based point scalar multiplication, elliptic curve-based point addition, the pairing of two points on an elliptic curve, exponentiation on pairing, map-to-point hash function, modular inversion and modular-multiplication operations. The execution time of such cryptographic operations is obtained by taking the average of 5 succeeding runs with different inputs using the PBC library. It is also observed in \cite{kumar2017secure} that $1T_P \approx 3T_{SM} \approx 87T_M$, $T_E \approx 21T_M$, $T_H \approx 23T_M$, $T_I \approx 11.6T_M$, and $T_A \approx 0.12T_M$. From \cite{shim2013eibas, cao2008imbas}, we have seen that the electrical requirement for the MICA2 sensor is as follows: the required voltage in active mode is 3.0V, the current drawn in active mode, transmitting and receiving mode are 8.0, 27.0, and 10.0 mA respectively and data rate is 12.4 Kbps. As per the calculation, is given in \cite{shim2014s2drp}, we consider that the energy consumption for operations is calculated as $W=V*I*T$. Table \ref{tbl6.3} summarizes the notations, computation cost (in msec), and energy consumption (in mJ) of various cryptographic operations. As per \cite{shim2014s2drp}, a sensor needs 0.0522 mJ and 0.0193 mJ energy for transmitting and receiving a 1 bytes message.

\begin{table} 
        \centering
        \caption{Computation cost of cryptographic operations}
        \label{tbl6.3}
        \begin{tabular}{|c|c|c|c|}
            \hline
            Operations	& Notation &	Execution time (ms) &	Energy consume. (mJ) \\
            \hline
            \hline
            
Modular multiplication &	$T_M$	&0.027	&0.65 \\
Scalar multiplication (ECC)&	$T_{SM}$	&0.304	&7.30 \\
Point addition (ECC)	& $T_A$	& 0.001 &	0.024 \\
Exponentiation & 	$T_E$	& 0.297 &	7.1 \\
Inversion & 	$T_I$	& 0.008 &	0.192 \\
Map-to-point hash &	$T_H$	& 0.319 &	7.7 \\
Bilinear pairing &	$T_P$	& 2.373	& 56.95\\
        \hline
    \end{tabular}
    
\end{table}

\textbf{Message size}. Here, we study the aggregated-message size overhead in one transmission and m transmission. For communication overhead, we consider $|T|$ = $|M|$ = 4 bytes and $|ID|$ =1 byte. For 1024-bit RSA level of security, $|\mathbb{G}_1|$ must be 64-bytes prime if $\mathbb{G}_2$ is a q-order of a subgroup of the multiplicative group of the finite field $\mathbb{F}_{p^2}$. According to [38], we assume $|\mathbb{G}_1|$ = 42.5 bytes for the finite field $\mathbb{F}_{p^3}$ and $|\mathbb{G}_1|$ =20 bytes for the finite field $\mathbb{F}_{p^6}$. The required aggregated-message size for storing m data collected from a BMS in single transmission and m transmission are $4|\mathbb{G}_1|+m*(|\mathbb{G}_1|+9)$ and $m*(2|\mathbb{G}_1|+9)$  bytes respectively. 

We also compute the storage cost of the re-aggregated message size of n*m PHI data collected from $n$ BMS to MCS via PAD in one transmission and $n*m$ in $n*m$ transmission, respectively. Thus, the required message size for storing $n*m$ data collected from n BMS in one transmission in single transmission and n*m transmission are $(n+3)|\mathbb{G}_1|+n*m*(|\mathbb{G}_1 |+9)$ and $m*n*(2|\mathbb{G}_1|+9)$ bytes respectively. It is obvious that PAD is a storage-rich device so it will not bother storing the large message. The required message size for storing a prescription at SD is $|\mathbb{G}_1|+9$ bytes.

\textbf{Communication overhead}. We consider the size of signcrypted PHI, in order to evaluate the communication overhead. The proposed system stores the signcrypted PHI in the PAD and transmits it to the SD via MCS. For communicating m messages to PAD, BMS requires $4|\mathbb{G}_1|+m*(|\mathbb{G}_1 |+9)$ and $m*(2|\mathbb{G}_1 |+9)$ bytes of overhead in one transmission and $m$ transmissions respectively. For communicating $m*n$ messages to MCS, PAD needs $(n+3)|\mathbb{G}_1|+n*m*(|\mathbb{G}_1 |+9)$ and $m*n*(2|\mathbb{G}_1 |+9)$ bytes of overhead a transmissions, and min $m*n$ transmissions respectively. And, SD requires ($|\mathbb{G}_1 |+9$) bytes of overhead for transmitting a prescription to BMS, as summarized in Table \ref{tbl6.4}. 

\textbf{Energy consumption on communication}.  In our proposed system, BMS consumes $(0.0715(4+m)|\mathbb{G}_1|+0.643m)$  mJ and $m*(0.143|\mathbb{G}_1 |+0.643)$ mJ  of energy communicating m aggregated messages collected from a BMS to PAD in one transmission and m transmission respectively. Figure (\ref{fig6.5}) shows the energy consumption for communication between BMS and PAD on various security parameters. From PAD to SD/MCS, the proposed scheme consumes $((0.214+0.0715*n*(m+1))|\mathbb{G}_1|+m*n*0.643)$ mJ and $m*n*(0.143|\mathbb{G}_1 |+0.643)$ mJ of energy for transmitting $m*n$ re-aggregated messages in one transmission and $m$ transmission respectively. Figure (\ref{fig6.6}) shows the energy consumption for communication between SD and PAD on various security parameters. Similarly, for transmitting a prescription from SD to BMS, the proposed systems consume $(0.0715|\mathbb{G}_1 |+0.643)$  mJ of energy.

\subsection{Performance Comparison} 
Here, we compare our proposed EF-IDASC scheme with Selvi \textit{et al.}’ scheme \cite{selvi2009identityagg}-I, scheme \cite{selvi2009identityagg}-II,  scheme \cite{selvi2009identityagg}-III, Eslami \textit{et al.} scheme \cite{eslami2014certificateless}, Kar \textit{et al.} scheme \cite{kar2013provably} and Niu \textit{et al.} scheme \cite{niu2017privacy}, in terms of computation (in ms), communication (in Bytes), energy consumption (in mJ) and security attack.

\begin{figure}
  \centering
  \includegraphics[width=0.8\linewidth]{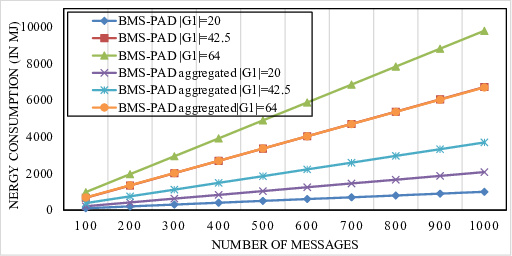}
  \caption{Total energy consumption between BMS and PAD of proposed health care system}
\label{fig6.5}
\end{figure}

\begin{figure}
  \centering
  \includegraphics[width=0.8\linewidth]{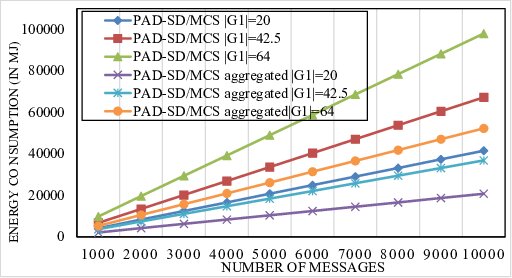}
  \caption{Total energy consumption between PAD and SD/MCS of proposed health care system}
\label{fig6.6}
\end{figure}
 	 
\begin{table} 
        \centering
        \caption{Communication overhead comparison of EF-IDASC scheme with related schemes}
        \label{tbl6.4}
        \begin{tabular}{|c|c|c|c|c|}
            \hline
            \multicolumn{1}{|c|}{Schemes} & \multicolumn{3}{|c|}{Storage cost (in bytes)} & \multicolumn{1}{|c|}{Computation cost (in mJ)} \\
            \cline{2-5}
           & Signcrypt. 	& Private key & 	Public key & 	Energy consumption \\
            \hline
            \hline
            \cite{selvi2009identityagg}-I & 5.5n+2.4 &	$|G1|$	& $|ID|$ & $n(2|G1|+|ID|+|T|+|M|)+1|G1|$	\\
            \cite{selvi2009identityagg}-II & 	5.5n+4.9 &	$|G1|$ &	$|ID|$ & $n(2|G1|+|ID|+|T|+|M|)+2|G1|$ \\
            \cite{selvi2009identityagg}-III & 5.5n+2.4 &	$|G1|$ & $|ID|$ & $n(2|G1|+|ID|+|T|+|M|)+1|G1|$ \\
           \cite{eslami2014certificateless} & 	4.9n+2.4 &	$|G1|+|Zq|$	& $|G1|+|ID|$ & $(2n+1)|\mathbb{G}_1|$\\
            \cite{kar2013provably} & 	7.6n+2.4 &	$|G1|$ &	$|ID|$ & $n(3|\mathbb{G}_1|+|M|)+1|\mathbb{G}_1|$ \\
            \cite{niu2017privacy} & 	5.1n+2.4 &	$|\mathbb{G}_1|+|\mathbb{Z}_q|$ &	$|\mathbb{G}_1|+|ID|$ & $n(2|\mathbb{G}_1|+|M|)+1|\mathbb{G}_1|$ \\
            Our &	3.1n+9.7 &	$|\mathbb{G}_1|$ &	$|ID|$ & $n(|\mathbb{G}_1|+|ID|+|T|+|M|)+4|\mathbb{G}_1|$ \\
        \hline
    \end{tabular}
\end{table}

\begin{table} 
        \centering
        \caption{Computation overhead comparison of EF-IDASC scheme with related schemes}
        \label{tbl6.5}
        \begin{tabular}{|c|c|c|c|c|}
            \hline
            \multicolumn{1}{|c|}{Schemes} & \multicolumn{4}{|c|}{computation cost (in mJ)}  \\
            \cline{2-5}
            Schemes	& Signcrypt  &	Unsigncrypt  &	Energy  &	Total Energy  \\
            	&   (A) & (B) &	 consumption   &	consumption  \\
            \hline
            \hline
            \cite{selvi2009identityagg}-I & $n(3T_{SM}+1T_P)$ &	$(2n+2)T_P$	& 78.85n	& 84.35n+2.4	\\
            \cite{selvi2009identityagg}-II & 	$n(4T_{SM}+1T_P)$	& $(2n+4)T_P+1T_{SM}$	& 86.15n	& 91.65n+4.9 \\
            \cite{selvi2009identityagg}-III & $n(3T_{SM}+1T_P)$	& $2nT_{SM} +4T_P$	& 78.85n &	84.35n+2.4 \\
            \cite{eslami2014certificateless} & $n(4T_{SM}+1T_P)$ & $nT_P$	& 86.15n &	91.05n+2.4\\
            \cite{kar2013provably} & $n(3T_{SM}+1T_P)$ &	$3nT_{SM} +nT_P$ &	78.85n	& 86.45n+2.4 \\
            \cite{niu2017privacy} & $n(4T_{SM}+1T_P)$ &	$nT_{SM} +3T_P$ &	86.15n &	91.25n+2.4 \\
            Our &	$(n+4)T_{SM}+1T_P$	& $nT_{SM} +3T_P$ &	7.3n+86.1	& 10.4n+95.8 \\
        \hline
    \end{tabular}
\end{table}

\textbf{Computation cost}. Due to the resource-constraint behaviour of BMS, we only consider the computation cost on the patient side. In our proposed EF-IDASC scheme, aggregate-signcryption needs $(n+3) T_{SM}+1T_P$ =0.304n+3.28 ms to signcrypt and combine n messages, while aggregate-unsigncryption needs $nT_{SM}+3T_P$ =0.304n+7.12 ms to unsigncrypts and verifies the aggregated signcryptext, which is the least cost as compared to  schemes \cite{selvi2009identityagg}-I, scheme \cite{selvi2009identityagg}-II,  scheme \cite{selvi2009identityagg}-III, Eslami \textit{et al.} scheme \cite{eslami2014certificateless}, Kar \textit{et al.} scheme \cite{kar2013provably} and Niu \textit{et al.} scheme \cite{niu2017privacy}, shown in Table \ref{tbl6.5}. 

\begin{table} 
        \centering
        \caption{Security comparison of EF-IDASC scheme with related schemes}
        \label{tbl6.6}
        \begin{tabular}{|c|c|c|c|c|c|c|c|c|c|c|c|}
            \hline
            \multicolumn{1}{|c|}{Schemes} & \multicolumn{11}{|c|}{computation cost (in mJ)}  \\
            \cline{2-12}
            & $S_1$ &	$S_2$ &	$S_3$&	$S_4$&	$S_5$&	$S_6$&	$S_7$&	$S_8$&	$S_9$&	$S_{10}$&	$S_{11}$  \\
            \hline
            \hline
            \cite{selvi2009identityagg}-I & $\checkmark$	&$\checkmark$	&$\checkmark$	&$\checkmark$	&$\times$	&$\times$	&$\times$	&$\checkmark$	&$\times$	&$\checkmark$	&$\times$	\\
            \cite{selvi2009identityagg}-II & 	$\checkmark$	&$\checkmark$	&$\checkmark$	&$\checkmark$	&$\times$	&$\times$	&$\times$	&$\checkmark$	&$\times$	&$\checkmark$	& $\checkmark$ \\
           \cite{selvi2009identityagg}-III & $\checkmark$	&$\checkmark$	&$\checkmark$	&$\checkmark$	&$\times$	&$\times$	&$\times$	&$\checkmark$	&$\times$	&$\checkmark$	&$\checkmark$ \\
            \cite{eslami2014certificateless}] & $\checkmark$	&$\checkmark$	&$\checkmark$	&$\times$	&$\times$	&$\checkmark$	&$\times$	&$\times$	&$\times$	&$\checkmark$	&$\checkmark$\\
            \cite{kar2013provably} & $\checkmark$	&$\checkmark$	&$\checkmark$	&$\checkmark$	&$\times$	&$\times$	&$\times$	&$\checkmark$	&$\times$	&$\checkmark$	&$\checkmark$ \\
            \cite{niu2017privacy} & $\checkmark$	&$\checkmark$	&$\times$	&$\times$	&$\times$	&$\times$	&$\times$	&$\times$	&$\times$	&$\checkmark$	&$\times$ \\
            Our &	$\checkmark$	&$\checkmark$	&$\checkmark$	&$\checkmark$	&$\checkmark$	&$\checkmark$	&$\checkmark$	&$\checkmark$	&$\checkmark$	&$\checkmark$	&$\checkmark$ \\
        \hline
    \end{tabular}
\end{table}

\textbf{Communication Cost}. To evaluate the communication cost, we consider $|M|$ = 4 bytes, $|tt|$ =4 bytes and $|ID|$ =1 bytes. For super-singular curve over the binary field  $\mathbb{F}_(2^{271})$ with the order of $\mathbb{G}_1$ is 252-bit prime and using compression technique [35], we assume $|\mathbb{G}_1|$ =34 bytes. The proposed EF-IDASC scheme outputs $4|\mathbb{G}_1|+n*(|\mathbb{G}_1 |+9)$ = (136+43n) bytes of aggregated signcryptext of n messages, which has the least communication cost as compared to \cite{selvi2009identityagg}-I, scheme \cite{selvi2009identityagg}-II,  scheme \cite{selvi2009identityagg}-III, Eslami \textit{\textit{et al.}} scheme \cite{eslami2014certificateless}, Kar \textit{et al.} scheme \cite{kar2013provably} and Niu \textit{et al.} scheme \cite{niu2017privacy} schemes, shown in Table \ref{tbl6.5}

\textbf{Total Energy consumption}. This section evaluates the total energy consumption for aggregate-signcrypting ‘$n$’ messages and communicating (transmission and receiving) them from the sender (BMS) to the receiver (SD). The proposed system consumes (7.3n+86.1) mJ to aggregate-signcrypt and (3.07n+9.7)  mJ to communicate n messages from sender (BMS) to receiver (SD). Thus, it consumes (10.4n+95.7) mJ of total energy, which has the least energy consumption as compared to related schemes, such as \cite{selvi2009identityagg}-I, scheme \cite{selvi2009identityagg}-II,  scheme \cite{selvi2009identityagg}-III, Eslami \textit{et al.} scheme \cite{eslami2014certificateless}, Kar \textit{et al.} scheme \cite{kar2013provably} and Niu \textit{et al.} scheme \cite{niu2017privacy} for signcrypting and communicating n message. Table \ref{tbl6.5} compares the total energy consumption of our scheme with other related schemes. Figure (\ref{fig6.7}) shows the total energy consumption comparison of our scheme with another scheme.
\begin{figure}
  \centering
  \includegraphics[width=0.8\linewidth]{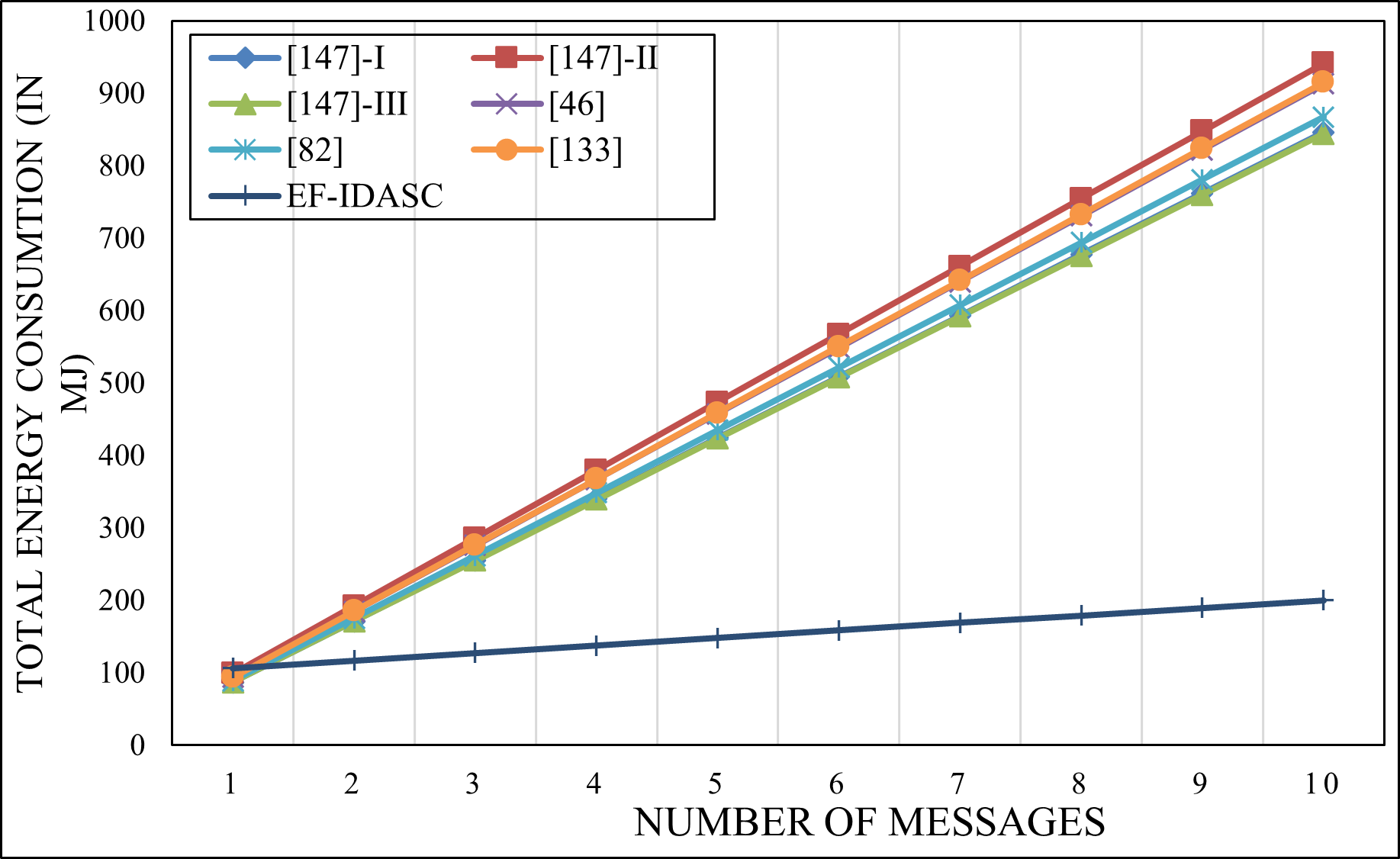}
  \caption{Total energy consumption comparison of the proposed healthcare system with other systems}
\label{fig6.7}
\end{figure}

\textbf{Security comparison}. Here, we use the following notations, $S_1$: mutual authenticity, $S_2$: data confidentiality, $S_3$: data integrity, $S_4$: Unlinkability $S_5$: contextual privacy, $S_6$: resilient to key escrow, $S_7$: entity revocation, $S_8$: patient anonymity, $S_9$: non-repudiation, $S_{10}$: forward secrecy and $S_{11}$: Public auditing. Table \ref{tbl6.6} illustrates that the proposed EF-IDASC scheme has better security on comparing with related schemes.

\section{Summary }
This chapter presents two escrow-free identity-based signcryption (EF-IDSC) schemes, namely, escrow-free identity-based signcryption (EF-IDSC) and escrow-free identity-based aggregated signcryption (EF-IDASC). Both schemes are designed to be secure against the IND-CCA and EUF-CMA attacks. We provide a performance analysis of the proposed schemes, demonstrating that they have the least communication, computation, and storage costs among similar schemes. Moreover, we implement a secure peer-to-peer video-on-demand (P2P-VoD) streaming system based on the proposed EF-IDSC scheme and evaluate its performance. Furthermore, we implement a secure cloud-assisted healthcare system that achieves public verifiability based on the EF-IDASC scheme. The evaluation of the system shows that it can provide secure and efficient healthcare services to patients while preserving their privacy.

\end{doublespace} \label{chapter6}
\begin{savequote}[75mm] 
Science without religion is lame, religion without science is blind.
\qauthor{Albert Einstein} 
\end{savequote}

\chapter{Identity-Based (Anonymous) Authenticated Key Agreement Protocol for Resource-Constrained Devices}
\justify
\begin{doublespace}

Users have benefited greatly from the recent advancements in ubiquitous technologies, particularly sensors, which have allowed for the creation of smart homes, wearable devices, and healthcare monitoring systems, among other applications and services. However, these technologies present significant security challenges, particularly in environments with limited resources like sensor networks. One of the main security concerns with sensor networks is their wireless communication, which is susceptible to eavesdropping, interception, and unauthorized access, allowing attackers to intercept and alter transmitted data and potentially leading to security breaches like data theft, tampering, and unauthorized access. Another major issue is the limited security mechanisms available due to resource constraints like processing power, memory, and battery life, leaving the sensors vulnerable to various types of attacks, including denial-of-service attacks, data tampering, and unauthorized access. Additionally, the resource-constrained nature of sensor networks makes them more vulnerable to node compromise, where an attacker can gain control of a sensor node and use it to launch attacks on the network, leading to data theft, manipulation, and service disruption.

This chapter presents two kinds of key agreement protocols for energy-constrained applications. First, we propose a one-round three-party authenticated ID-based key agreement protocol (OR-3PID-KAP) whose security is based on solving ECDLP and BDHP. The proposed system has the least computational cost, bandwidth cost and message exchange as compared to the related schemes. Second, we proposed an identity-based anonymous authentication and key agreement (IBAAKA) protocol for WBAN in the cloud-assisted environment, which achieves mutual authentication and user anonymity. Recently, there have been discussed several identity-based key agreement authentication protocols for lightweight devices. In 2009, Yang \textit{et al.} \cite{yang2009id} presented an identity-based authentication and key agreement scheme using elliptic curve cryptography. However, Yoon \textit{et al.} \cite{yoon2009robust} showed that the scheme \cite{yang2009id} was prone to impersonation attacks and did not provide perfect forward secrecy. Similar to Yang \textit{et al.} scheme \cite{yang2009id}, Cao \textit{et al.} \cite{cao2010pairing} discussed an identity-based authenticated key agreement protocol with the least message exchange. However, they did not achieve user anonymity and unlinkability.

In 2015, Tsai \textit{et al.} \cite{tsai2015privacy} proposed a new identity-based authentication protocol for mobile cloud computing services where the mobile user and service provider need to register with a trusted third party that provides the long-term secret key for them. This scheme was computationally inefficient as it is based on bilinear pairing. Jiang \textit{et al.} \cite{jiang2016security} proved that the scheme \cite{tsai2015privacy} is not secure against the impersonation attack and did not provide MA. In 2016, Yang \textit{et al.} \cite{yang2016efficient} presented a lightweight anonymous authentication scheme with untraceability for mobile cloud computing. The scheme \cite{yang2016efficient} provides a set of pseudo-IDs and corresponding secret keys to each user. Although managing the set of pseudo-IDs and secret keys is a challenge in a lightweight environment. In 2018, Kumar \textit{et al.} proposed a two-party identity-based authenticated key agreement protocol \cite{kumar2019PF} for resource constraint devices. However, it does not achieve user anonymity. Recently, He \textit{et al.} \cite{he2011efficient} present a new provably secure authentication scheme using an elliptic curve. Wang \textit{et al.} \cite{wang2013cryptanalysis} demonstrate that \cite{he2011efficient} does not provide MA and could not resist a reflection attack. Recently, Islam \textit{et al.} \cite{islam2011more} improved the Yoon \textit{et al.} scheme \cite{yoon2009robust} and presented a new authentication scheme, but it is not resilient to denial of service attacks. Recently, Lui \textit{et al.} \cite{liu2013certificateless} designed a new certificate-less authentication scheme for the WBAN system, which uses the pairing on elliptic curves. He \textit{et al.} \cite{he2016anonymous} proved that the scheme \cite{liu2013certificateless} is not suitable for a secure e-health care system as it is susceptible to impersonation attacks. Besides, they presented an AA scheme for the WBAN system and show that their scheme is provably secure.

Zhao \cite{zhao2014efficient} found that the scheme \cite{liu2013certificateless} is insecure against the public verifier table and does not achieve user anonymity followed by proposed an AA protocol for the WBAN system using ECC. Wang \textit{et al.} \cite{wang2015new} proved that the scheme \cite{zhao2014efficient} is incapable of achieving user unlinkability and put forward an improved authentication protocol using the pairing for WBANs. Jiang \textit{et al.} \cite{jiang2016bilinear} show that protocol \cite{wang2015new} is susceptible to user impersonation attacks. Recently, Jia \textit{et al.} \cite{jia2019provably} proposed an identity-based AA and key agreement scheme for mobile edge computing settings, which satisfies MA as well as user anonymity and untraceability. Li \textit{et al.} \cite{li2017enhanced} proposed an efficient end-to-end authenticated scheme for the healthcare system using ECC. Recently, in 2020, Sowjanya \textit{et al.} \cite{sowjanya2020elliptic} proved that Li \textit{et al.} \cite{li2017enhanced} scheme has some security loopholes followed by they address such issues and presented a new end-to-end authenticated scheme for healthcare system using ECC. In the same year, Nikooghadam \textit{et al.} \cite{nikooghadam2020secure} presented a secure two-factor authenticated key agreement protocol using ECC. However, the scheme satisfied security; it did not preserve user anonymity.

In this direction, we observed that none of the AA schemes is suitable for the e-health care system. Emerging technology like cloud computing, together with the WBAN system, could be advantageous for the next-generation e-health care system. Today, it remains a challenge to design an efficient cloud-enabled anonymous authenticated key agreement protocol for the WBAN system.

\begin{figure}
  \centering
  \includegraphics[width=1\linewidth]{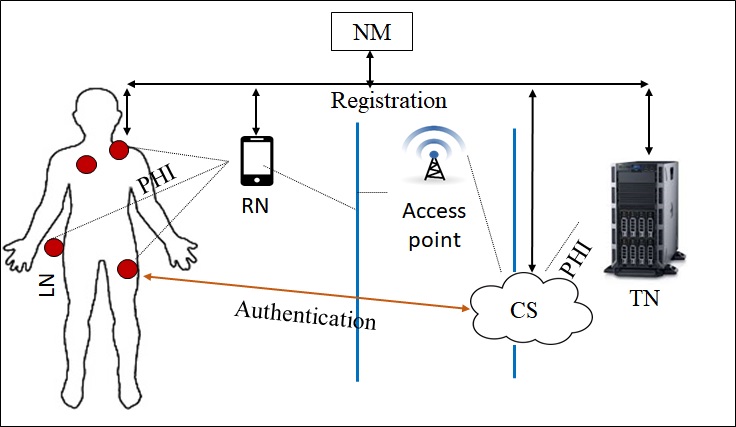}
  \caption{Proposed cloud-enabled WBAN network model}
\label{fig7.1}
\end{figure}

\section{System Model}

This section gives a brief discussion of the network and security model of the proposed WBAN cloud-enabled architecture. 

\subsection{Network Architecture }

The proposed model consists of five entities: Network manager (NM), cloud servers (CS), leaf node (LN), root node (RN) and target node (TN), as illustrated in Figure (\ref{fig7.1}).

\begin{itemize}
    \item \textit{NM}: It is a trusted third party that registers the leaf node, root node, target node, and cloud server and provides them with the secret keys corresponding to their identities.
    \item \textit{LN}: LN is a resource-constraint wearable/implanted sensor that has limited storage space, battery life, and computation power. It monitors and collects the patient’s PHI and forwards it to the TN via RN.
    \item \textit{TN}: TN is a resource-rich remote system that provides medical services to the patient. 
    \item \textit{RN}: It is an intermediate node between the leaf node and the target node, for example, a smartphone. Since LN has limited storage and power, the RN helps LN with long-range data transmission and data storage. It collects the patient’s PHI data from the multiple LNs without knowing any information about it and forwards it to the TN.
    \item \textit{CS}: It is a computationally rich storage device that stores PHI data. Due to the enormous amount of medical data, each target node/leaf node needs to offload its data into the cloud server. Before providing any service to the requested user, CS must be registered with the NM and preloaded with the system's public parameters.
\end{itemize}

\subsection{Security Threat }

Based on the security model of the key agreement protocol by Bellare \textit{et al.} \cite{bellare2000authenticated}, we discuss the security model of our proposed identity-based anonymous authentication and key agreement (IBAAKA) protocol. It can be defined as a game played between an adversary $A$ and challenger $C$. Suppose a scheme denotes $\Gamma$ and two participants denote as P, where P can be a sensor LN and cloud server CS. Let $O_P^i$ denote the random oracle of $i^{th}$ instant of participant $P$. An adversary adaptively communicates to the scheme $\Gamma$ by asking some queries to the random oracle and obtains information from responses. $A$ asks the following oracle queries: 

\begin{itemize}
    \item \textit{h(m)}: On given message m, $Adv$ asks a query to the hash oracle, and the oracle checks if m has been submitted previously. If not, randomly pick an element x and return it to A. Otherwise, return the same value. 
    \item \textit{Extract(ID)}: On given identity ID, $Adv$ execute this oracle, and the oracle gives the private key corresponding to the ID.
    \item \textit{Send($P^i$, msg)}: On given message msg, when $Adv$ submits a query to random oracle $O_P^i$, the oracle yields the real result output by the actual scheme. $Adv$ will execute the Send($P^i$,start) query to begin the protocol. 
    \item \textit{Reveal($P^i$)}: On a given query, A asks a query and obtains the session shared key of participant $\Gamma_P^i$.
    \item \textit{Execute($LN^i$,$CS^i$)}: When $Adv$ asks a query, this oracle gives all the message exchange between the instance $O_{LN}^i$ and $O_{CS}^j$ during the process.
    \item \textit{Test($O_P^i$)}: $Adv$ can ask this query only once during the whole process. The oracle randomly picks a bit $b \in \{0,1\}$ and outputs the original session key SK if $b =1$; otherwise, outputs a random value. 
\end{itemize}

\textit{Partnership}: Two instances, says $O_{LN}^i$ and $O_{CS}^j$ are said to be partner if: 
\begin{itemize}
    \item they exchange messages directly to each other.
    \item Share same session key.
    \item Except $O_{LN}^i$ and $O_{CS}^j$, no other instance accepts session key.
\end{itemize}

\textit{Freshness}: An instance is said to fresh if the session key SK has been accepted without the Reveal and Extract queries executed on $O_P^i$ and its partner. 

Now, we give the brief on the two-security definition that would be concluded from the above game playing between adversary $A$ and challenger $C$. 

On the given output coming from the Test oracle, $A$ guesses a bit $b'$. If $b' = b$, $A$ could break the security of the authentication protocol. The advantage of choosing $b$ can be defined as $Adv(A) = |Pr[b' = b] - \frac{1}{2}|$.

\begin{definition}
(\textbf{Provable security}). If any probabilistic polynomial time (PPT) bounded adversary $Adv$ has negligible advantage $Adv(Adv)$ in the above game, we say that the authenticated protocol is provably secure.
\end{definition}

We say that if an adversary $A$ forges the parameters exchange during transmission on behalf of the participant and forged parameters are received by the partner; then $Adv$ successfully breaks the mutual authentication (MA) protocol.

The advantage of breaking the MA is denoted as $Adv(A)=|Pr[E_{LN-CS}]+Pr[E_{CS-LN}]|$, where $E_{LN-CS}$ and $E_{CS-LN}$ represent the event that $A$ successfully impersonates the LN and produces an actual login parameter and the event that $Adv$ produces a valid response, respectively.

\begin{definition}
(\textbf{MA security}). If any PPT bounded adversary $Adv$ has negligible advantage Adv(A) in the above events. We say that the authenticated protocol is MA secure.
\end{definition}

\subsection{Security Requirement }

Here, we discuss the following security requirements for considering secure WBAN systems. 
\begin{itemize}
    \item \textit{User Anonymity}. It ensures that the identity of the legal user must be hidden from the adversary during the authentication phase. The sensor collected data refers to the user’s personal information, so he must be enjoyed the wireless medical facilities, and at the same time, his privacy will not be revealed to the unauthentic third person. Thus, a user’s identity should be hidden from everyone except the network manager and CS. 
    \item \textit{Perfect forward secrecy}. It ensures that the adversary cannot learn anything about the session key of the previous session, even if he gets the current private key of a user.
    \item \textit{Mutual authentication}. The MA ensures that LN and CS must be authenticated to each other which means the PHI is coming from the intended LN and arriving at the predetermined CS. Therefore, MA confirms the legitimacy of the user and CS’s identities in the network.
    \item \textit{Session key establishment}. In order to achieve the integrity, authenticity, and confidentiality of PHI data, it is necessary to establish a shared session key between the LN and CS. 
    \item \textit{Non-traceability}. In order to protect user privacy, only user anonymity is not sufficient, but the scheme must achieve non-traceability. The non-traceability ensures that the user’s action should not be traced by the network manager and cloud server.
    \item \textit{Key replacement attack}. It is ensured that the adversary could not replace the public key of CS/LN to forge a valid signature on any message.
\end{itemize}

\section{Proposed  Cloud-assisted ID-based AA Protocol}

Here, we discuss the implementation of OR-3PID-KAP and IBAAKA schemes. 

\subsection{One round 3-Party ID-based Key Agreements Protocol}

Suppose three leaf nodes, say A, B and C, are in the proposed network model, as disused in section 7.1, and want to share secret PHI amongst themselves in a secure way. Thus, they should be agreed on a shared secret key, $K_{ABC}$ using the OR-3PID-KAP scheme that consists of three algorithms: Setup, Key Generation and Key agreement.
\begin{itemize}
    \item \textit{Setup}: On given security parameter $k$, NM select a master key $s\in \mathbb{Z}_q$ and set the public key as $P_{Pub} = sP$, where $P$ is the generator of $\mathbb{G}_1$. Let two cryptographic hash functions are $H_1: \{0,1\}^n \times \mathbb{G}_1 \rightarrow \mathbb{Z}_q$ and $H_2: \{0,1\}^n \rightarrow \{0,1\}^n$, and a bilinear map $e: \mathbb{G}_1 \times \mathbb{G}_1 \rightarrow \mathbb{G}_2$. Now NM published the public parameters $pp=<P,P_{Pub},\mathbb{G}_1,\mathbb{G}_2,e,H_1,H_2>$ and keeps master key $s$ secret.

    \item \textit{Extract}: This phase takes system parameter pp, master key $s$, and node’s identities $ID_A$, $ID_B$ and $ID_C$ as input and gives the identity-based long term private key $S_A$, $S_B$ and $S_C$. NM performs the following operations:
    \begin{itemize}
        \item 	Pick three random number $r_A$, $r_B$  and $r_C  \in \mathbb{Z}_q$, and computes:
	\begin{align*}
	    R_A&= r_A P,Q_A= H_1 (ID_A ||R_A) \\
        R_B& = r_B P,Q_B= H_1 (ID_B ||R_B)\\
        R_C&= r_C P,Q_C= H_1 (ID_c ||R_C) \\
	\end{align*}
	    \vspace{-20mm}
    	\item NM computes private key as $S_A= r_AsQ_A$, $S_B= r_BsQ_B$, $S_C= r_CsQ_C$  and  sends $<S_A,R_A>$, $<S_B,R_B>$ and $<S_C,R_C>$ to nodes $A$, $B$ and $C$ respectively, via a secure channel or preloaded at the time of node installation.
	    \item  Each LN can validate its private key with the following equation. 
	    \begin{align*}
	        e(S_i P,P) \overset{?}{=}e(H_1(ID_i||R_i) P_{Pub},R_i)
	    \end{align*}
	\end{itemize}
    \vspace{-5mm}
    Correctness is verified as:
    \begin{align*}
        e(S_iP,P) &=e(r_isQ_iP,P)
        =e(r_i sH_1(ID_i||R_i)P,P) \\ 
        &=e(sH_1(ID_i||R_i)P,r_i P)
        =e(H_1(ID_i||R_i)P_{Pub},R_i)
    \end{align*}
    \item \textit{Key Agreement}: Shared session key $S_k$ is computed as follows:
    
    \vspace{-5mm}
    \begin{align*}
        A \rightarrow B, C &:  a \in \mathbb{Z}_q, T_A=aR_A \\
        B \rightarrow A, C &:  b \in \mathbb{Z}_q, T_B=bR_B \\
        C \rightarrow B, A &: c \in \mathbb{Z}_q, T_C=cR_C 
    \end{align*}
\vspace{-10mm}

A: Given $T_B$ and $T_C$, LN A compute $K_{A|BC}$:
\vspace{-5mm}
\begin{align*}
    K_{A|BC} &= e(Q_B T_B,Q_C T_C)^{aS_A} \\
S_k &=H_2 (ID_A,ID_B,ID_C,R_A,R_B,R_C,T_A,T_B,T_C,K_{A|BC})
\end{align*}

B: Given $T_A$ and $T_C$, LN $B$ compute $K_{B|AC}$:
\vspace{-5mm}
\begin{align*}
    K_{B|AC} &= e(Q_A T_A,Q_C T_C)^{bS_B}\\
S_k &=H_2(ID_A,ID_B,ID_C,R_A,R_B,R_C,T_A,T_B,T_C,K_{B|AC})
\end{align*}

C: Given $T_A$ and $T_B$, LN $C$ compute $K_{C|AB}$:
\begin{align*}
    K_{C|AB} &= e(Q_A T_A,Q_B T_B)^{cS_C} \\
S_k &=H_2(ID_A,ID_B,ID_C,R_A,R_B,R_C,T_A,T_B,T_C,K_{C|AB})
\end{align*}

\end{itemize}

The common secret key $K_{ABC}$ are agreed because: 
\vspace{-8mm}
\begin{align*}
    K_{A|BC} & = e(Q_B T_B,Q_C T_C)^{aS_A} \\
            &= e(Q_B br_B P,Q_C cr_C P)^{ar_A sQ_A}\\
            &= e(P,P)^{sabcQ_A r_A Q_B r_B Q_C r_C}
\end{align*}
\vspace{-10mm}

\begin{align*}
    K_{B|AC} &= e(Q_A T_A,Q_C T_C)^{bS_B}\\
    &= e(Q_A ar_A P,Q_C cr_C P)^{br_BsQ_B} \\
    &= e(P,P)^{sabcQ_A r_A Q_B r_B Q_C r_C } 
\end{align*}
\vspace{-10mm}

\begin{align*}
    K_{C|AB}&= e(Q_A T_A,Q_B T_B)^{cS_C}\\
    &= e(Q_Aar_A P,Q_Bbr_B P)^{cr_csQ_C}\\
    &= e(P,P)^{sabcQ_A r_A Q_B r_B Q_C r_C}
\end{align*}

The above three equations verify the correctness of shared session key $K_{ABC}=K_{A|BC}=K_{B|AC}=K_{C|AB}$.

\subsection{ID-based Anonymous Authenticated Key Agreement Protocol}

This section archives the anonymous authentication between the LN and the cloud. The proposed IBAAKA protocol consists of three algorithms: Setup, registration and authentication. 
\begin{itemize}
    \item \textit{Setup}: Given a security parameter $k$, the NM chooses a random $k$-bit large prime number $q$, an additive group $\mathbb{G}_1$ and multiplicative group $\mathbb{G}_2$ of order of $q$. Let $P$ be the generator of $\mathbb{G}_1$ and pairing function  $e: \mathbb{G}_1  \times \mathbb{G}_1 \rightarrow \mathbb{G}_2$. Suppose five cryptographic hash functions, $H_1:\{0,1\}^I \times \mathbb{G}_1 \rightarrow \mathbb{Z}_q^*$, $H_2:\{0,1\}^{2I} \times \mathbb{G}_1^2 \rightarrow \mathbb{Z}_q^*$, $H_3: \mathbb{G}_2  \rightarrow \{0,1\}^{I+t} \times \mathbb{G}_1^2 \times \mathbb{Z}_q^*$, $H_4:  {0,1}^(2I+t) \times \mathbb{G}_1^2 \rightarrow \mathbb{Z}_q^*$, and $H_5: \{0,1\}^{2I} \times \mathbb{G}_1^3 \rightarrow \mathbb{Z}_q^*$. 
    NM chooses an element $s_0 \in \mathbb{Z}_q^*$ (master key) and set $P_0=s_0P$ (pubic key).  NM keeps $s_0$ secret and published the public parameter $pp=<k,q,e,P,P_0,\mathbb{G}_1,\mathbb{G}_2,H_1,H_2,H_3,H_4,H_5>$.
    \item \textit{Registration}: This phase registers the leaf node and cloud server as follows: 
    \begin{itemize}
        \item 	\textit{Cloud server registration}: Cloud server picks identity $ID_{CS} \in \{0,1\}^*$ and requests to the NM for its private key. The NM checks $ID_{CS}$, chooses a random integer $r_{CS} \in \mathbb{Z}_q^*$, computes the private key $R_{CS}=r_{CS}P$, $s_{CS}=r_{CS}+s_0H_1(ID_{CS}||R_{CS})$ and  sends it to the CS. Now, CS’s private and public keys are denoted as $s_{CS}$ and $<ID_{CS},R_{CS}>$ respectively. 
        \item 	\textit{Leaf Node registration:} LN picks identity $ID_{LN} 
        \in \{0,1\}^*$ and requests to the NM for its private key. The NM checks $ID_{LN}$, chooses a random integer $r_{LN} 
        \in \mathbb{Z}_q^*$,  computes the private key $R_{LN}=r_{LN}P$, $s_{LN}=r_{LN}+s_0 H_1 (ID_{LN}||R_{LN})$ and  sends it to the LN. LN’s private key and public key are denoted as $s_{LN}$ and $<ID_{LN},R_{LN}>$ respectively. Additionally, NM pre-stores the parameter $R_{CS}$ in LN’s storage.
    \end{itemize}
    
\begin{figure}
  \centering
  \includegraphics[width=1\linewidth]{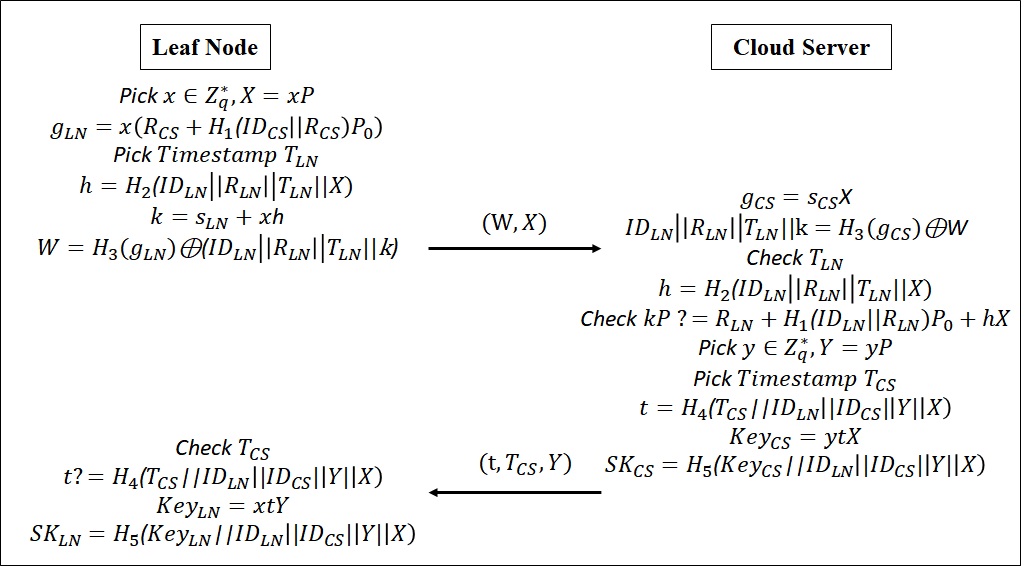}
  \caption{Authentication phase of proposed IBAAKA scheme}
\label{fig7.2}
\end{figure}

    \item \textit{Authentication}: After registration, this algorithm performs the MA between LN and CS in an intervening way, also illustrated in Figure (\ref{tbl7.2}).
    \begin{itemize}
        \item Leaf node (LN) chooses an element $x \in \mathbb{Z}q^*$, sets $X = xP$, and computes the key $\mathbb{G}{LN} = x(R_{CS} + H_1(ID_{CS} || R_{CS})P_0)$.
        \item LN chooses a current time stamp $T_{LN}$ and compute $h=H_2(ID_{LN} ||R_{LN}||T_{LN}||X)$, $k=s_{LN}+xh$ and encrypts  $W=H_3(\mathbb{G}_{LN})\oplus (ID_{LN}||R_{LN}|||T_{LN}||X||k)$ and sends $<W,X>$ to the CS. 
        \item On receiving the parameters, CS computes key $\mathbb{G}_{CS}=s_{CS}X$ and extracts $ID_{LN}$ on decrypting $W$ using key $\mathbb{G}_{CS}$ as $ID_{LN} R_{LN}||T_{LN}||X||k = H_3(\mathbb{G}_{CS}) \oplus W$.
        \item CS checks the timestamp $T_{LN}$ and checks whether the equation $kP \overset{?}{=} R_LN+H_1(ID_{LN}||R_{LN})P_0 +hX$ holds, where, $h=H_2(ID_{LN}||R_{LN}||T_{LN}||X)$. 
    \end{itemize}

    From the given equations, we have:

$\mathbb{G}{SN} = x(R{CS} + H_1(ID_{CS} || R_{CS})P_0)$

$kP \overset{?}{=} R_{LN} + H_1(ID_{LN} || R_{LN})P_0 + H_2(ID_{LN} || R_{LN} || T_{LN} || X)X$
    
    The consistency of the equations is verified as: 
    \vspace{-5mm}
    \begin{align*}
    \mathbb{G}{LN} &= x(R{CS} + H_1(ID_{CS} || R_{CS})P_0) \
    &= x(r_{CS}P + s_0H_1(ID_{CS} || R_{CS})P) \
    &= x(r_{CS} + s_0H_1(ID_{CS} || R_{CS}))P \
    &= s_{CS}X = \mathbb{G}{CS} \
    \
    kP &= (s{LN} + xh)P \
    &= (s_{LN}P + xhP) \
    &= (r_{LN} + s_0H_1(ID_{LN} || R_{LN}))P + xhP \
    &= (r_{LN}P + s_0H_1(ID_{LN} || R_{LN})P + xhP) \
    &= (R_{LN} + H_1(ID_{LN} || R_{LN})P_0 + hX)
    \end{align*}
    \vspace{-20mm}
    
    \item Now, CS chooses a current time stamp $T_{CS}$ and an element $y \in \mathbb{Z}_q^*$, sets  $Y=yP$ and $t=H_5(ID_{LN}||ID_{CS}||T_{CS}||X||Y)$.
    \item CS computes session key as $SK_{CS}=H_6(key_{CS}||ID_{LN}||ID_{CS} ||Y||X)$, where $key_{CS}=ytX$. Then sends $<t,T_{CS},Y>$ to the LN.
    \item LN checks the timestamp $T_{CS}$ and check $t \overset{?}{=}H_5(ID_{LN} ||ID_{CS}||T_{CS}||X||Y)$.
	\item LN computes the session key as $key_{LN}=xtY$ and $SK_{LN}=H_6(key_{LN}||ID_{LN}||ID_{CS}||Y||X)$. 
\end{itemize}

From the above equation, we have 
$key_{LN}=xtY=xtyP=tyX=key_{CS}$

\section{System Analysis}

Here, we discuss the security proof and performance evaluation of proposed OR-3PID-KAP and IBAAKA schemes.

\subsection{Security Analysis of OR-3PID-KAP Scheme}

Here, we discuss the security analysis of the OR-3PID-KAP scheme. 

\begin{theorem} \label{thm7.1}
The proposed OR-3PID-KAP scheme is provable and secure.
\end{theorem}

\begin{proof}
Let two random oracle models be H1 and H2 and assume ECDLP and BDH problems are difficult to compute, then our proposed one-round key agreement protocol is provable and secures against the adversary. The security proof of our proposal is proved similar to Theorem 2 in Wang \textit{et al.} \cite{wang2016identity}. 
\end{proof}

\begin{theorem}
 The proposed OR-3PID-KAP scheme is achieved the perfect forward secrecy.
\end{theorem}

\begin{proof}
The use of pre-image and uniformly distributed hash function H2 guarantees perfect forward secrecy. The output of hash function H2 gives the shared session key:
\vspace{-5mm}
\begin{align*}
    S_k &=H_2(ID_A,ID_B,ID_C,R_A,R_B,R_C,T_A,T_B,T_C,K_{A|B})\\
    &=H_2(ID_A,ID_B,ID_C,R_A,R_B,R_C,T_A,T_B,T_C,K_{B|AC}) \\
    &=H_2(ID_A,ID_B,ID_C,R_A,R_B,R_C,T_A,T_B,T_C,K_{C|BA})
\end{align*}

 If the adversary wishes to compromise the long-term private key, he must need to compute the user’s secret values a, b and c which is equivalent to solving the ECDLP and CDHP problem. Thus, assuming the ECDLP and CDHP are hard, our proposed scheme is resilient to perfect forward secrecy attacks. 

\end{proof}

\begin{theorem} \label{thm7.3}
The proposed OR-3PID-KAP scheme achieved Mutual Authentication.
 
\end{theorem}

\begin{proof}
In our proposed scheme, BS generates the user's long-term private key using master key s and a random variable $r_i$. Using the long-term private key, identity and random integer, nodes A, B and C in the sensor network compute a common session key $S_k$. Because session key $S_k$ includes the private key and identification of all nodes thus only registered/authenticated nodes can establish a session shared key. Therefore, the proposal provides mutual authentication $S_k$.
\end{proof}

\begin{theorem} \label{them7.4}
The proposed OR-3PID-KAP scheme is secured against Man-in-the-middle attacks.
\end{theorem}

\begin{proof}
Suppose an adversary and his assistant $X$ impersonate two nodes, say $A$ and $B$. Suppose the adversary attacks the public channel and obtains information $<R_A,R_B,R_C,T_A,T_B,T_C,ID_1,\dots,ID_n>$.

Let the attacker choose some random numbers $a^x$, $b^x$, $c^x$, $r_A^x$, $r_B^x$, and $r_C^x \in \mathbb{Z}_q$, and compute the parameters $R_A^x= r_A^xP$, $R_B^x= r_B^x P$, $R_C^x= r_C^x P$, $T_A^x= a^x R_A^x$,  $T_B^x= b^x R_B^x$, and $T_C^x= b^x R_C^x$, as defined below:

\begin{align*}
    A &\rightarrow B, C : \text{pick } a \in \mathbb{Z}_q, T_A=aR_A \\
    X_A &\rightarrow B, C : \text{pick } a_x  \in \mathbb{Z}_q, T_A^x= a^xR_A \\
    B &\rightarrow A, C : \text{pick } b \in \mathbb{Z}_q, T_B=bR_B \\
    X_B &\rightarrow A, C : \text{pick } b_x  \in \mathbb{Z}_q, T_B^x= b^xR_B \\
    C &\rightarrow B, A : \text{pick } c \in \mathbb{Z}_q, T_C=cR_C \\
    X_C &\rightarrow B, A : \text{pick } c_x  \in \mathbb{Z}_q, T_C^x= c^x R_C \\
\end{align*}

The nodes then compute as:
\begin{align*}
    A: K_{A|BC} &=e(Q_B T_B^x,Q_C T_C^x)^{aS_A} 
    = e(P,P)^{sa^xb^xc^x Q_A r_A Q_B r_B Q_C r_C } \\
    B: K_{B|AC}  &= e(Q_A T_A^x,Q_C T_C^x)^{bS_B} 
    = e(P,P)^{sa^x b^x c^x Q_A r_A Q_B r_B Q_C r_C} \\
    C: K_{C|AB} &=e(Q_A T_A^x,Q_B T_B^x)^{cS_C}
    = e(P,P)^{sa^x b^x c^x Q_A r_A Q_B r_B Q_C r_C}\\
\end{align*}

From the given parameters, the attacker $X$ computes the following session keys:

\begin{align*}
    X_A: K_{X|BC} &=e(Q_B T_B,Q_C T_C)^{a^x S_A^x} 
    = e(P,P)^{a^x S_A^x bcQ_B r_B Q_C r_C } \\
    X_B: K_{X|AC} &= e(Q_A T_A,Q_C T_C)^{b^x S_B^x} 
    = e(P,P)^{aS_B^x b^x cQ_A r_A Q_C r_C} \\
    X_C: K_{X|AB} &= e(Q_A T_A,Q_B T_B)^{cS_C} 
    = e(P,P)^{acS_C^x c^x Q_A r_A Q_B r_B} \\
\end{align*}

For a successful attack, the attacker must compute the value of the node’s long-term private key and random secret key, which is equivalent to solving the ECDLP problem. In this sense, we say that our proposed scheme is secure against man-in-the-middle attacks. 

\end{proof}

\begin{theorem} \label{thm7.5}
The proposed OR-3PID-KAP scheme is secured against no key control.
\end{theorem}

\begin{proof}
In our protocol, session Key $S_k$ is computed as:
\begin{align*}
    S_K &=H_2 (ID_A,ID_B,ID_C,R_A,R_B,R_C,T_A,T_B,T_C,K_{A|BC} )\\
&=H_2 (ID_A,ID_B,ID_C,R_A,R_B,R_C,T_A,T_B,T_C,K_{B|AC} )\\
&=H_2 (ID_A,ID_B,ID_C,R_A,R_B,R_C,T_A,T_B,T_C,K_{C|BA})
\end{align*}

where $K_{A|BC}$, $K_{B|AC}$,and $K_{C|AB}$  are defined in section 3, and $a$, $b$ and $c$ are node’ random chosen secret number. It can be seen that any node could not pre-compute the session key $S_k$ without knowing the information of others' secret keys. Therefore, it proves that our protocol does not have control over the key. 
\end{proof}

\begin{theorem} \label{thm7.6}
The proposed OR-3PID-KAP scheme is secured against known session key attacks.
\end{theorem}

\begin{proof}
Each node evaluates the hash of the participant’s identities, public parameters, and secret key of each node to compute the session key $S_k$. Since the hash function’s response is uniformly distributed in $\{0,1\}^k$, the probability that two session key has a relation is $1/2k$ which is negligible. Therefore, our protocol is resilient to the known session key attacks under the assumption of pre-image and uniformly distributed cryptographic hash function $H_2$.
\end{proof}

\subsection{Security Analysis of IBAAKA Scheme}
\begin{theorem} \label{thm7.7}
Suppose an adversary $Adv$ can compute the legal login parameters and break the LN-to-CS authentication of the proposed IBAAKA scheme with non-negligible probability. There could be a challenger$Ch$ that can solve the ECDL problem with probability $\epsilon_0 \ge \epsilon(1-1/q)/q_H1$, where $q_H1$ represents a query $H_1$ run by $A$. 
\end{theorem}

\begin{proof}
Consider an adversary $A$ that generates a valid login request tuple $<W,X>$, which can be justifiable by the CS with a non-negligible probability $\epsilon$. Here are the given equations:

\begin{align*}
X &= xP \
\mathbb{G}{LN} &= x(R{CS} + H_1(ID_{CS} || R_{CS})P_0) \
h &= H_2(ID_{LN} || R_{LN} || T_{LN} || X) \
k &= s_{LN} + xh \
W &= H_3(\mathbb{G}{LN}) \oplus (ID{LN} || R_{LN} || T_{LN} || X || k)
\end{align*}

Using the forking lemma, given in \cite{pointcheval2000security}, adversary $A$ will produce another tuple $<W',X>$ on a different choice of $H_2$ in the simulation, where $X = xP$, $\mathbb{G}{LN} = x(R{CS} + H_1(ID_{CS} || R_{CS})P_0)$, $h' = H'2(ID{LN} || R_{LN} || T_{LN} || X)$, $k' = s_{LN} + xh'$, and $W' = H_3(\mathbb{G}{LN}) \oplus (ID{LN} || R_{LN} || T_{LN} || X || k)$.

As $<W,X>$ and $<W',X>$ can provide satisfied verification on the CS side, we have the following two equations:

\begin{align*}
kP &= R_{LN} + H_1(ID_{LN} || R_{LN})P_0 + hP \
k'P &= R_{LN} + H_1(ID_{LN} || R_{LN})P_0 + h'P
\end{align*}

Suppose, $R_{LN} = aP_0 + bP$, $s_{LN} = b$, and $H_1(ID_{LN} || R_{LN}) = -a$. Then we have,

\begin{align*}
h'kP - hk'P &= h'(R_{LN} + H_1(ID_{LN} || R_{LN})P_0 + hP) - h(R_{LN} + H_1(ID_{LN} || R_{LN})P_0 + h'P) \
&= h'R_{LN} + h'H_1(ID_{LN} || R_{LN})P_0 + h'hP - hR_{LN} - hH_1(ID_{LN} || R_{LN})P_0 - hh'P \
&= (h' - h)R_{LN} + (h' - h)H_1(ID_{LN} || R_{LN})P_0 \
&= (h' - h)(aP_0 + bP) + (h' - h)(-a)P_0 \
&= (h' - h)bP
\end{align*}

and,

\begin{align*}
bP = (h' - h)^{-1}(h'k - hk')P
\end{align*}

Thus, for any given tuple $(P, bP)$ of the ECDL problem, adversary $Adv$ solves the problem with a non-negligible advantage $\epsilon_0 \ge (1 - 1/q)/q_{H1}$. Therefore, no adversary could break the LN-to-CS authentication. Thus, our proposed AA scheme achieves LN-to-CS authentication.

\end{proof}

\begin{theorem}
Suppose an adversary $Adv$ can compute legal login parameters and break the CS-to-LN authentication of the proposed IBAAKA scheme with non-negligible probability. There could be constructed a challenger$Ch$ that solves the ECDL problem with probability $\epsilon_0 \ge \epsilon(1-1/q)/q_H1$, where $q_H1$ represents a query $H_1$ made by A. 
\end{theorem}

\begin{proof}
Suppose $E^{LN-to-CS}$ represents an event that can break the LN-to-CS authentication, and $E^{CS-to-LN}$ represents an event that can break our proposed protocol's CS-to-LN authentication. From the proof in Lemma 1, we suppose that the event $E^{LN-to-CS}$ does not occur. The event $E^{CS-to-LN}$ denotes that the adversary can forge a response $<t,T_{CS},Y>$ upon intercepting the login tuple $<W,X>$. Now, one of three events will occur:

\begin{itemize}
\item The adversary may pick the correct value of $t$, so the probability that this event occurs is $1/q^k$, where $k$ is the size of the hash function.
\item The parameters $X$ and $k$ have been computed in another session. Let the adversary $Adv$ make $q_{LN}$ queries to the oracle of LN. The probability that this event occurs is $(q_{LN}/q) \times (q_{LN}/q) \times (q_{LN}-2) \le (q_{LN}^3)/q^2$.
\item On the given parameters $<key_{CS}, ID_{LN}, ID_{CS}, Y, X>$, the adversary $A$ runs queries on the $H_6$ oracle. Now we have $Pr[E^{CS-to-LN}|E^{CS-to-LN}] \le \epsilon_0 + 1/q^k + (q_{LN}^3)/q^2$.
\end{itemize}

Suppose $X = cP$, $R_{LN} = aP_0 + bP$, $s_{LN} = b$, and $H_1(ID_{LN}|| R_{LN}) = -a$ for some values $a, b, c \in \mathbb{Z}q^*$. For an instance $<X, R{LN} + H_1(ID_{LN}||R_{LN})P_0> = (cP, aP_0 + bP - aP_0) = (cP, bP)$ of the Computational Diffie-Hellman (CDH) problem, $Adv$ can compute $s_{CS} = bcP$. In this way, $Adv$ could solve the CDH problem with a non-negligible probability $\epsilon_0 \ge \epsilon - 1/q^k - (q_{LN}^3)/q^2$.

Therefore, our proposed AA scheme achieves the CS-to-LN authentication, and no adversary could break the CS-to-LN authentication.
\end{proof}

\begin{theorem} \label{thm7.9}
Suppose the ECDL problem is hard. Our proposed IBAAKA is MA secure.
\end{theorem}
\begin{proof}
From Lemma 8.7 and 8.8, we proved that no adversary could break our proposed AA protocol's LN-to-CS authentication or CS-to-LN authentication. Thus, it is proved that the proposed AA protocol is MA secure. 
\end{proof} 

\begin{theorem} \label{thm7.10}
Under the assumption of the random oracle model and hardness of computation Diffie–Hellman problem, our proposed IBAAKA scheme is AKA-secure. 
\end{theorem}

\begin{proof}
Let an adversary $Adv$ correctly guess the random bit $b$ value in a Test query with a non-negligible advantage $\epsilon$. Suppose $E^{H_{SK}}$ is the event that $Adv$ could guess the session key correctly, so we have $Pr[E^{H_{SK}}] \ge \epsilon/2$. Suppose the event that shows $ Adv$ runs the Test query to the oracle LN in the $i$-th instance is $E^{Test(LN)}$, and the event that shows $ Adv$ runs the Test query to the oracle CS in the $j$-th instance is $E^{Test(CS)}$. Also, the event that shows $Adv$ breaks the LN-to-CS authentication of our proposed system is $E^{LN-to-CS}$. We have:

\begin{align*}
Pr[E^{H_{SK}}] &= Pr[E^{H_{SK}} \wedge E^{Test(LN)}] \
&+ Pr[E^{H_{SK}}\wedge E^{Test(CS)}\wedge E^{(LN-to-CS)}] \
&+ Pr[E^{H_{SK}}\wedge E^{Test(CS)} \wedge \neg E^{(LN-to-CS)}] \ge \epsilon/2
\end{align*}

and

\begin{align*}
Pr[E^{H_{SK}} \wedge E^{Test(LN)}] &+ Pr[E^{H_{SK}} \wedge E^{Test(CS)} \wedge \neg E^{(LN-to-CS)}] \
&\ge \frac{\epsilon}{2} - Pr[E^{H_{SK}} \wedge E^{Test(CS)} \wedge E^{(LN-to-CS)}] \
&\ge \frac{\epsilon}{2} - Pr[E^{(LN-to-CS)}]
\end{align*}

As the events $E^{Test(CS)} \wedge \neg E^{(LN-to-CS)}$ and $E^{Test(LN)}$ are equivalent, we have:

\begin{align*}
Pr[E^{H_{SK}} \wedge E^{Test(LN)}] &\ge \frac{1}{2} \left(\frac{\epsilon}{2} - Pr[E^{(LN-to-CS)}]\right)
\end{align*}

Therefore, we conclude that:

\begin{align*}
Pr[t=H_5(ID_{LN}||ID_{CS}||T_{CS}||X||Y)|X,Y] \ge \epsilon
\end{align*}

Lemma 8.7 shows that $Pr[E^{LN-to-CS}]$ is negligible. Then, we obtain:

\begin{align*}
Pr[t=H_5(ID_{LN}||ID_{CS}||T_{CS}||X||Y)|X,Y] \ge \epsilon
\end{align*}

Suppose $X=xP$ and $Y=yP$, where $x,y\in \mathbb{Z}_q^*$. On the given CDH assumption tuple $<X,Y>=<xP,yP>$, $Adv$ must be able to compute $K=xY=yX=xyP$. $C$ uses $A$ to solve the CDH problem with a non-negligible advantage $\epsilon_0 \ge \frac{1}{2}\left(\frac{\epsilon}{2} - Pr[E^{LN-to-CS}]\right)$ to solve it. Therefore, we have proven that our proposed IBAAKA scheme is AKA-secure.
\end{proof}

\subsection{Security Discussion}
The proposed IBAAKA protocol achieves the following security properties:

\begin{itemize}
    \item \textit{Mutual authentication}. In our proposed scheme, LN computes the $\mathbb{G}_{SL}=x(R_{CS}+H_1( ID_{CS}||R_{CS})P_0)$, $h=H_2(ID_{LN}||R_{LN}||T_{LN}||X)$ and $k=s_{LN}+xh$ using its private key $s_{LN}$ and encrypts $W=H_3(\mathbb{G}_LN) \oplus(ID_{LN}||R_{LN}||T_{LN}||X||k)$.  No one other than LN could have knowledge of its private key. The CS then computes the decryption key $\mathbb{G}_{CS}=s_{CS}X$  and decrypts $W$ using key $\mathbb{G}_{CS}$ as $ID_{LN}||R_{LN}||T_{LN}||X||k=H_3(\mathbb{G}_{CS}) \oplus W$. Now, CS verifies the authenticity of LN by checking whether the equation $kP \overset{?}{=} R_{LN}+H_1(ID_{LN}||R_{LN})P_0+hX$ holds. On CS side, CS compute sets  $Y=yP$ and $t=H_5(ID_{LN}||ID_{CS}||T_{CS}||X||Y)$, $key_{CS}=ytX$ and $SK_{CS}=H_6(key_{CS}||ID_{LN} ||ID_{CS} ||Y||X)$. LN checks the timestamp $T_{CS}$ and check whether the equation $t \overset{?}{=}H_5(ID_{LN} ||ID_{CS} ||T_{CS}||X||Y)$ holds and computes the session key as $key_{LN}=xtY$ and $SK_{LN}=H_6(key_{LN}||ID_{LN} ||ID_{CS}||Y||X)$. In addition, Theorem 8.9 proves mutual authentication of the proposed IBAAKA scheme for WBAN. Therefore, the protocol provides mutual authentication.
	\item \textit{User anonymity}. User anonymity ensures that the user’s real identity remains hidden except for the CS and NM.  In our proposed scheme, LN’s identity $ID_{LN}$ is encrypted in login parameters $W$. To reveal the real identity of the LN, the adversary has to generate the CS’ private key $s_{CS}$, equivalent to solving the CDH problem. Thus, an adversary could not identify the real identity of LN. Thus, the protocol achieves sensor/LN anonymity. 
	\item \textit{Session key security}. In an authentication phase of the proposed anonymous authentication protocol, LN and CS agreed on the shared session key SK. The session key security depends on the ECDL problem and ensures that the session key will not be broken.
	\item \textit{Untraceability}. In the proposed protocol, the exchange parameters that is $<W,X>$ and $<t,T_{CS},Y>$ are based on the random numbers $x$ and $y$. Therefore, each new session entity chooses unique $x$ and $y$, so these parameters are different for each session. For two different session parameters, the adversary $A$ cannot find any relationship among these parameters nor trace the sensor activity. Therefore, the proposed scheme achieves the untraceability property. 
	\item \textit{Perfect forward secrecy}. Suppose adversary $A$ has the private key of LN and CS and also $A$ has the intercepted message  $<W,X>$ and $<t,T_{CS},Y>$. From these given parameters, an adversary could not compute the random numbers x and y to get the session key $SK=SK_{CS}=H_5(key_{CS}||ID_{LN}||ID_{CS} ||Y||X)=SK_{LN}=H_5(key_{LN}||ID_{LN}||ID_{CS}||Y||X)$, an adversary requires $key_{CS}=ytX$ or $key_{LN}=xtY$ which is hard to compute if he/she does not know the random element $x$ and $y$. Therefore, under the CDH assumption, the proposed scheme provides forward-perfect secrecy.
	\item \textit{Stolen verifier attack}. The CS has not any information related to any LN for mutual authentication. So, there is no information to be stolen. Thus, the scheme is secure against the stolen verifier attack.
	\item \textit{LN Impersonation attack}. To impersonate a sensor LN, the adversary must generate an actual login parameter $<W,X>$, such that W is encryption of $(ID_{LN}||R_{LN}||T_{LN}||X||k)$ and $k$ is signature on $ID_{LN}||R_{LN}||T_{LN}||X$ using LN’s private key $s_{LN}$. It is hard for an adversary to produce the same signature without LN’s private key; hence it cannot impersonate the leaf node LN. 
	\item \textit{CS impersonation attack}. To login the phase, LN encrypts $(ID_{LN}||R_{LN}||T_{LN} ||X||k)$ using $\mathbb{G}_{LN}$ which is derived from the random of LN and CS’s identity. Without knowing CS’s private key $s_{CS}$, the adversary cannot decrypt the $ID_{LN}$ and $X$ from the login parameters. Thus, he cannot generate the correct value of $t$. Thus, the adversary cannot impersonate the CS. 
	\item \textit{Man-in-the-middle attack}. In our proposed protocol, CS decrypts the login parameter W using its private key to authenticate the leaf node’s ID. By doing so, CS extracts the LN’s identity $ID_{LN}$ and $R_{LN}$. Without knowing the private key of CS, it is hard for an adversary to produce the message by itself. Therefore, the protocol is secure against the man-in-the-middle (MITM) attack. 
\end{itemize}
	
\textbf{Security comparison}. Table \ref{tbl7.1} compares the proposed IBAAKA scheme and other related schemes such as Liu \textit{et al.} scheme \cite{liu2013certificateless}, He \textit{et al.} scheme \cite{he2016anonymous},  Wang and Zhang’s \cite{wang2015new}, Jiang \textit{et al.} scheme \cite{jiang2016bilinear} and Jia \textit{et al.} scheme \cite{jia2019provably} in terms of security properties. We denote “Y” and “N” to represent whether the protocol fulfils the corresponding security property. Table \ref{tbl7.1} shows that the proposed IBAAKA scheme achieves the required security properties.

\begin{table} 
        \centering
        \caption{Security Comparison}
        \label{tbl7.1}
        \begin{tabular}{|c|c|c|c|c|c|c|}
            \hline
            Security requirements &	\cite{liu2013certificateless} & \cite{he2016anonymous} & \cite{wang2015new} & \cite{jiang2016bilinear} & \cite{jia2019provably} & IBAAKA\\
            \hline
            \hline
            MA  &	N	& N &	N &	Y &	Y &	Y\\
            User anonymity & 	N	& N	& Y &	Y	& Y& 	Y\\
            Session key security &	Y &	Y &	Y &	Y &	Y &	Y \\
            Untraceability & 	N &	Y &	Y &	Y &	Y &	Y\\
            Perfect forward & 	N &	Y &	Y &	Y &	Y &	Y\\
            Probable security  & N &Y &	Y &	Y &	Y &	Y\\
            MITM attack &	Y	& Y &	Y &	Y &	Y &	Y\\
            Impersonation attack & 	N &	Y &	N &	Y &	Y &	Y\\
            Key replacement & 	N &	N& 	N &	N &	Y &	Y\\
            User revocation &	N &	N &	N &	N &	N &	Y\\
        \hline
    \end{tabular}
\end{table}

\subsection{Performance Evaluation of OR-3PID-KAP}
We now compare our proposed protocol with other related schemes \cite{xiong2012provably} and  \cite{xiong2013new} in terms of computational cost (in operations), execution cost (in ms) and bandwidth cost (in Bytes). To achieve 1024-bit RSA level security for pairing-based cryptosystem, we found that the TinyPBC library, which has an efficient implementation over binary fields $F_(2^m)$ with 80-bit security \cite{xiong2010tinypairing} which is written in C language and running on TinyOS operating system and MICAs sensor user. MICA has only 4KB RAM, 128KB ROM, and an 8-bit/7.3828-MHz  ATmega128L microcontroller. The running cost of operations and their notations are shown in Table \ref{tbl7.2}.

\begin{table} 
        \centering
        \caption{Notation and Execution time of pairing-related operations on MICAs \cite{xiong2010tinypairing}}
        \label{tbl7.2}
        \begin{tabular}{|c|c|c|}
            \hline
            Operations &	Notation	&Execution cost (sec)  \\
            \hline
            \hline
            Hash-to-Point	&H	&0.89 \\
            Point Compression	& PC	& 0.38\\
            Point Decompression & PD	& 0.38 \\
            Scalar point multiplication & 	E	& 2.45\\
            Pairing 	& P	& 5.32 \\

        \hline
    \end{tabular}
\end{table}

We ignore the computation cost of the hash function operation as it takes less comparison cost than bilinear pairing operations. We focus on pairing operations, scalar multiplication on an elliptic curve, and scalar multiplication on pairing for approximation performance. We then examine the comparison of our scheme with related schemes in terms of computational cost (running in a sec),  

\begin{table} 
        \centering
        \caption{Computational and execution cost comparison of our proposed scheme with other related schemes}
        \label{tbl7.3}
        \begin{tabular}{|c|c|c|c|c|}
            \hline
            Schemes &	Computational Cost &	Execution cost (sec) &	Bandwidth (Bytes)	& \#Round  \\
            \hline
            \hline
            \cite{xiong2012provably} & $2P +7E +2Ex$ &	38.43 &	(2*80+16)/8 = 22 &	6 \\
            \cite{xiong2013new} & $P +5E + Ex$ &	22.89 &	(2*80+16)/8 = 22 &	6 \\
            Our Scheme &	$P + E$ &	7.77	& (2*80+16)/8 = 22	& 1 \\

        \hline
    \end{tabular}
\end{table}

In Table \ref{tbl7.3}, the reader can see that the proposed scheme takes only two multiplication operations on an elliptic curve executed in (5.32 +2.45) = 7.77 sec while [\cite{xiong2012provably} and  \cite{xiong2013new} take 38.43 sec and 22.89 sec respectively. That means our proposed scheme is 20.21\% and 33.94\% of \cite{xiong2012provably} and  \cite{xiong2013new} respectively, in terms of computational running time.

We assume the size of the user identity is 16-bit and the order of the elliptic curve group is 80-bit. As shown in Table \ref{tbl7.3}, we show that similar to \cite{xiong2012provably} and  \cite{xiong2013new}, our proposed scheme consumes (2*80+16)/8 = 22 bytes of channel bandwidth during the key establishment completion. Further, it can be seen that our protocol completes in one round while \cite{xiong2012provably} and  \cite{xiong2013new} take six rounds. 

\begin{table} 
        \centering
        \caption{Security  comparison of the proposed scheme with other schemes}
        \label{tbl7.4}
        \begin{tabular}{|c|c|c|c|c|c|c|c|c|}
            \hline
            Schemes &	PS$^@$ &	PFS$^\#$ &	MA$^{\#\#}$ &	MIMAR$^{@@}$ &	KCIR$^{\&}$ &	KCR$^{*}$ &	KSKS$^\&$ &	UKR$^\&\&$  \\
            \hline
            \hline
            \cite{xiong2012provably} & 	$\checkmark$	&$\checkmark$	&$\checkmark$	&$\checkmark$	& 	$\times$	&$\checkmark$& 	$\checkmark$	&$\checkmark$ \\
            \cite{xiong2013new} & 	$\checkmark$	&$\times$	&$\checkmark$	&$\checkmark$	& 	$\checkmark$	&$\checkmark$& 	$\times$	&$\checkmark$ \\
            Our Scheme &	$\checkmark$	&$\checkmark$	&$\checkmark$	&$\checkmark$	& 	$\checkmark$	&$\checkmark$& 	$\checkmark$	&$\checkmark$ \\

        \hline
    \end{tabular}
    \footnotesize{\\$^@$One Provable security, $^\#$perfect forward security, $^{\#\#}$ Mutual authentication, $^{@@}$ man-in-the-middle attack resilience, $^{\&}$ key compromised impersonation resilience, $^{*}$key control resilience. $^\&$known session key security. $^\&\&$unknown key share resilience }
\end{table}

Table \ref{tbl7.4} compares our scheme with \cite{xiong2012provably} and  \cite{xiong2013new}  in terms of the following security notions: Provable security, PFS: perfect forward security, mutual authentic-cation, man-in-middle attack resilience, key compromised impersonation resilience, key control resilience, the known session key security, unknown key share resilience. Similar \cite{xiong2012provably} and  \cite{xiong2013new}, our protocol also achieves all security notions.

\subsection{Performance Evaluation of IBAAKA }
This section illustrates the performance of the proposed IBAAKA scheme for WBAN systems in the cloud environment and compares it with other related systems such as Liu \textit{et al.} scheme \cite{liu2013certificateless}, He \textit{et al.} scheme \cite{he2016anonymous},  Wang and Zhang’s \cite{wang2015new}, Jiang \textit{et al.} scheme \cite{jiang2016bilinear} and Jia \textit{et al.} scheme \cite{jia2019provably}, in terms of computational cost and communication cost.
To simulate the proposed IBAAKA schemes, we experiment on \textit{Acer E5-573-5108 laptop} with \textit{Intel(R) Core(TM) i5-5200U CPU@2.20 GHz} and 8 GB RAM on Window 10. We consider the evaluation of operations on super-singular curve $y^2+y=x^3+x$ having embedding degree 4 and use the eta pairing $\eta_T:E(F_{2^{271}}) \times E(F_{2^{271}} ) \rightarrow E(F_{2^{4.271}})$. To analyse the performance of the proposed scheme, we take the following cryptographic operations: scalar multiplication on an elliptic curve, point addition on an elliptic curve, pairing operation on two points on an elliptic curve, pairing exponentiation, map-to-point hash function, modular inversion and modular-multiplication operations, which we represented by $T_{SM},T_A,T_P,T_E$,  $T_H$, $T_I$  and $T_M$ respectively. We run the various cryptographic operations over a bilinear pairing group and consider the average of 10 succeeding results using the PBC library \cite{lynn2007pbc}. We estimated that $1T_P=3T_{SM}=87T_M$, $T_E=21T_M$, $T_H=23T_M$, $T_I=11.6T_M$, and $T_A=0.12T_M$.  Table \ref{tbl7.5} abstracts the notations and computation cost (in msec) of required cryptographic operations.

\begin{table} 
        \centering
        \caption{Notations and computation cost (in msec) of different cryptographic operations}
        \label{tbl7.5}
        \begin{tabular}{|c|c|c|}
            \hline
            Operations &	Notation & 	Execution time (ms) \\
            \hline
            \hline
            Modular multiplication & $T_M$	& 0.027\\
            ECC-based multiplication & 	$T_SM$ &	0.304 \\
            ECC-based Point addition &	$T_A$	& 0.001\\
            Exponentiation &	$T_E$	& 0.297\\
            Inversion 	& $T_I$ &	0.008 \\
            Map-to-point hash & $T_H$ &	0.319 \\
            Bilinear pairing	& $T_P$ &	2.373\\
        \hline
    \end{tabular}
\end{table}

\begin{table} 
        \centering
        \caption{Computational cost (in msec): comparative summary
Schemes	Computational cost (msec)}
        \label{tbl7.6}
        \begin{tabular}{|c|c|c|}
            \hline
            Schemes & Sensor/Client/LN	 & Application provider/CS \\
            \hline
            \hline
            \cite{liu2013certificateless} & $4T_{SM}+1T_E$ (1.513) &	$1T_{SM}+1T_P+1T_E$ (2.974) \\
            \cite{he2016anonymous} & $4T_{SM}+1T_H$ (1.535)	& $4T_{SM}+2T_P+1T_H$ (6.281) \\
            \cite{wang2015new} & $3T_{SM}+1T_P+2T_H$ (3.923) &	$2T_{SM}+1T_P+2T_H$  (3.618) \\
            \cite{jia2019provably} & $4T_{SM}+1T_E$ (1.513) &	$4T_{SM}+2T_P+2T_E$ (6.556) \\
            \cite{sowjanya2020elliptic} & $3T_{SM}+1T_M$ (0.939) &	$6T_{SM}+1T_M$ (1.851) \\
            \cite{nikooghadam2020secure} & $3T_{SM}+2T_M$ (0.966) &	$4T_{SM}+2T_A$  (1.218)\\
            Our	& $3T_{SM}+2T_M$ (0.966) &	$5T_{SM}$ (1.52) \\
        \hline
    \end{tabular}
\end{table}

 Now, we evaluate the computational cost of our proposed IBAAKA protocol and other related schemes such as Liu \textit{et al.} scheme \cite{liu2013certificateless}, He \textit{et al.} scheme \cite{he2016anonymous},  Wang and Zhang’s \cite{wang2015new}, Jiang \textit{et al.} scheme \cite{jiang2016bilinear}, Jia \textit{et al.} scheme \cite{jia2019provably} and \cite{nikooghadam2020secure} schemes, as shown in Table \ref{tbl7.6}. The table shows that on the client/sensor side, Liu \textit{et al.} scheme \cite{liu2013certificateless} needs  1.513 ms, He \textit{et al.} scheme \cite{he2016anonymous} needs 1.535ms, Wang \textit{et al.} scheme \cite{wang2015new} needs 3.923 ms, Jia \textit{et al.} scheme \cite{jia2019provably} needs 1.513 ms, Sowjanya \textit{et al.} scheme \cite{sowjanya2020elliptic} needs 0.939 ms,  and Nikooghadam \textit{et al.} \cite{nikooghadam2020secure} needs 0.966 ms while the proposed IBAAKA protocol needs 0.966 ms. On the application provider/CS side, Liu \textit{et al.} scheme \cite{liu2013certificateless} needs 2.974 ms, He \textit{et al.} scheme \cite{he2016anonymous} needs 6.281ms, Wang \textit{et al.} scheme \cite{wang2015new} needs 3.618  ms, Jia \textit{et al.} scheme \cite{jia2019provably} needs 6.556 ms, Sowjanya \textit{et al.} scheme \cite{sowjanya2020elliptic} needs 1.218 ms,  and Nikooghadam \textit{et al.} \cite{nikooghadam2020secure} needs 1.52 ms while the proposed IBAAKA protocol needs 1.52 ms. The total computation costs of Liu \textit{et al.} \cite{liu2013certificateless}, He \textit{et al.} \cite{he2016anonymous},  Wang \textit{et al.} \cite{wang2015new},  Jia \textit{et al.} \cite{jia2019provably}, Sowjanya \textit{et al.} \cite{sowjanya2020elliptic} and Nikooghadam \textit{et al.} \cite{nikooghadam2020secure} and the proposed IBAAKA schemes are 4.487 ms, 7.816 ms, 7.451 ms, 8.069 ms, 2.79 ms, 2.184 ms and 2.486 ms, respectively. Thus, our proposed protocol has the least computation cost on the sensor side than other protocols, shown in Table \ref{tbl7.6} anf Figure (\ref{fig7.3}).

In order to achieve a trusted tight level of security, we consider the supersingular curve $E(\mathbb{F}_p)$ of order $q$ over a finite field $\mathbb{F}_p$, where $p$ and $q$ are two prime numbers, such that $|p|$ =512 and $|q|$=160 bits, respectively. The size of $\mathbb{G}_1$, $\mathbb{G}_2$ and $\mathbb{Z}_q$ are 1024, 512 and 160 bits respectively. For evaluating communication overhead, we consider the size of timestamp and identity as 2 bytes, i.e., $|T|$ = $|ID|$ =16 bits. Therefore, Liu \textit{et al.} scheme \cite{liu2013certificateless} requires transmitting 3408 bits, He \textit{et al.}’s scheme \cite{he2016anonymous} requires transmitting 4288 bits,  Wang \textit{et al.}’s scheme  \cite{wang2015new} requires transmitting 3296 bits, Jiang \textit{et al.}’s scheme \cite{jiang2016bilinear}, requires to transmit 2424 bits, Jia \textit{et al.}’s scheme \cite{jia2019provably} requires to transmit 5328 bits, and our scheme requires to transmit 3440 bits during authentication (see in Table \ref{tbl7.7}) and Figure (\ref{fig7.4}).

\begin{figure}
  \centering
  \includegraphics[width=0.8\linewidth]{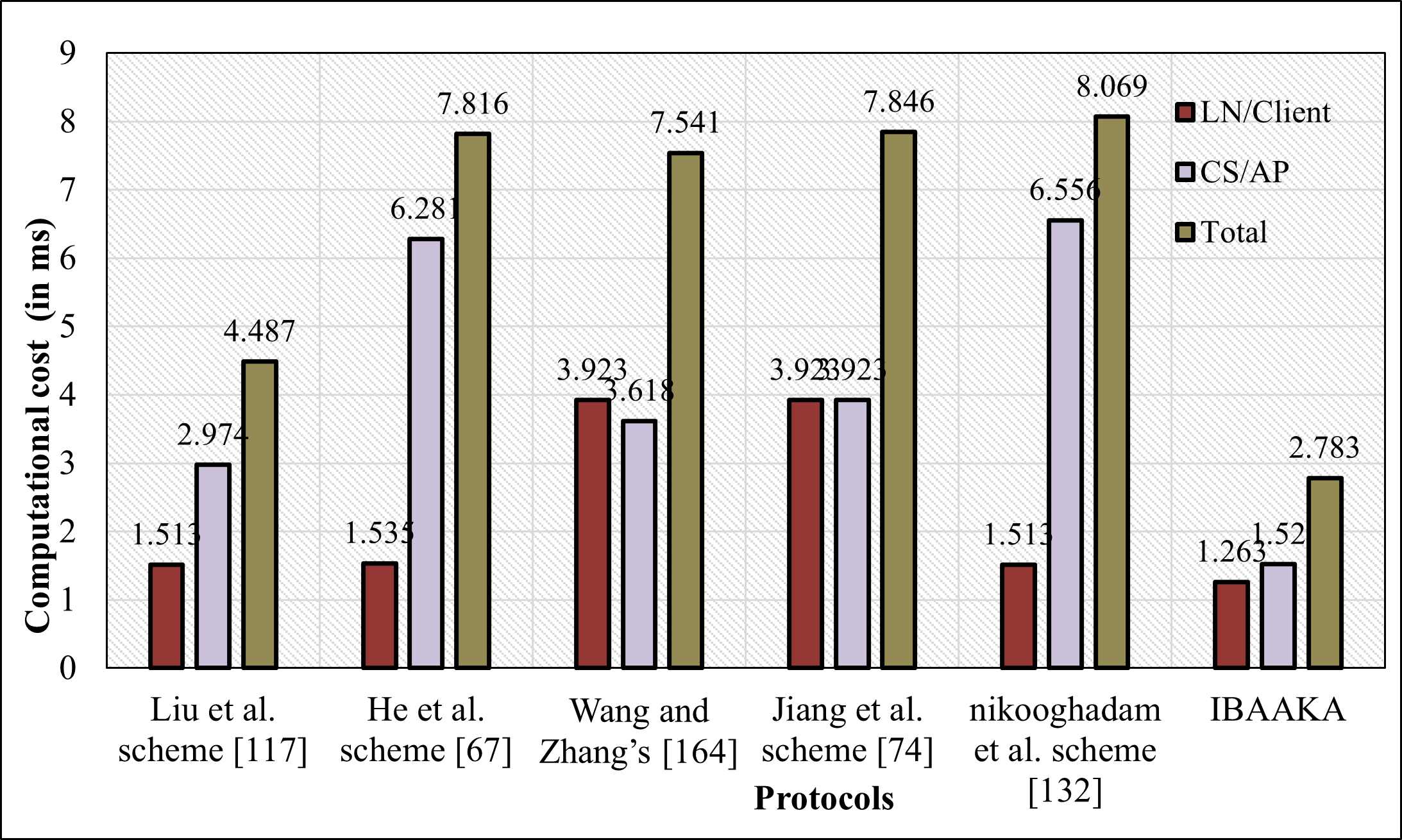}
  \caption{Computational cost comparison of proposed IBAAKA with other related schemes}
\label{fig7.3}
\end{figure}

\begin{figure}
  \centering
  \includegraphics[width=0.8\linewidth]{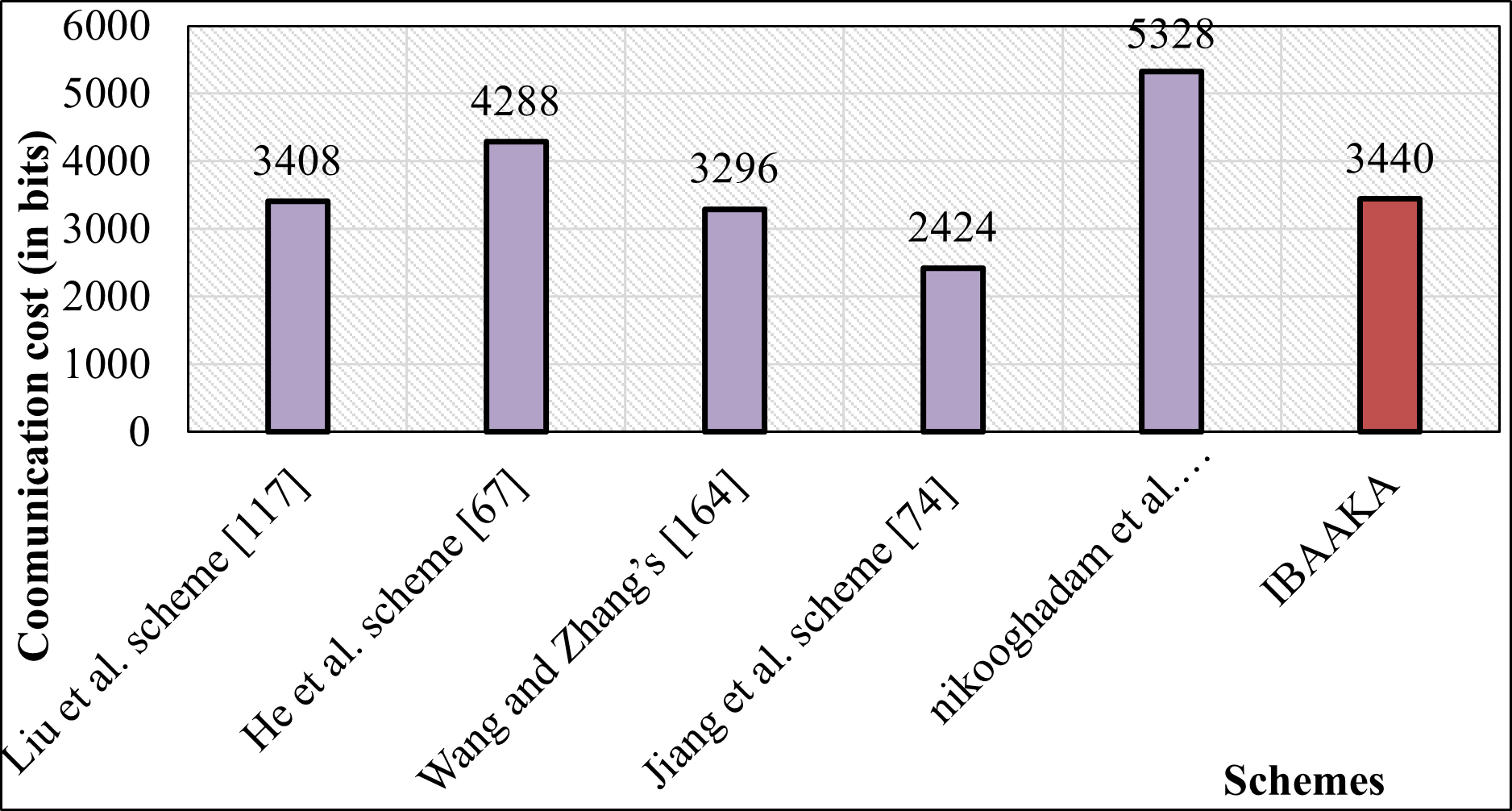}
  \caption{Communication cost comparison of proposed IBAAKA scheme with other schemes}
\label{fig7.4}
\end{figure}

\begin{table} 
        \centering
        \caption{Computational cost (in msec): comparative summary}
        \label{tbl7.7}
        \begin{tabular}{|c|c|c|}
            \hline
            Schemes & 	Computational cost (msec) \\
            \hline
            \hline
            \cite{liu2013certificateless} & $1|T|+3|\mathbb{G}_1 |+2|\mathbb{Z}_q |$	(3408) \\
            \cite{he2016anonymous} & $1|ID|+1|T|+4|\mathbb{G}_1 |+1|\mathbb{Z}_q |$	(4288) \\
            \cite{wang2015new} & $1|ID|+3|T|+3|\mathbb{G}_1 |+1|\mathbb{Z}_q |$ 	(3296) \\
            \cite{jia2019provably} & $1|ID|+2|T|+5|\mathbb{G}_1 |+1|\mathbb{Z}_q |$	(5328) \\
            \cite{sowjanya2020elliptic} & $3|T|+6|\mathbb{G}_1 |+3|\mathbb{Z}_q |$	(6672) \\
            \cite{nikooghadam2020secure} & $2|T|+3|\mathbb{G}_1 |+1|\mathbb{Z}_q |$	(3264)\\
            Our	& $1|ID|+2|T|+3|\mathbb{G}_1 |+2|\mathbb{Z}_q |$	(3440)\\
        \hline
    \end{tabular}
\end{table}

\section{Summary}
In this chapter, we have presented two key authenticated agreement protocols that are designed specifically for the WBAN system operating in a cloud-assisted environment. The first protocol we proposed is a three-party authenticated key agreement protocol that facilitates the computation of a secret shared key between three leaf nodes within the WBAN network. The second protocol we proposed is an efficient identity-based anonymous authentication and key agreement (IBAAKA) scheme that enables mutual authentication between the leaf node and cloud server and facilitates the generation of a session key between them. Our proposed IBAAKA scheme guarantees the anonymity of the user's identity except for the network manager in the registration phase. We have conducted a comprehensive security analysis of our protocols, and we have shown that they are secure under the random oracle model and CDH assumption. When compared to other existing protocols, our schemes have the least computation cost and comparable communication cost. Therefore, our protocols offer a highly efficient solution for authenticated key agreements in the WBAN system, which is of utmost importance in ensuring the security and privacy of the users' sensitive health data.

\end{doublespace} \label{chapter7}
\begin{savequote}[75mm] 
Cultivation of mind should be the ultimate aim of human existence.
\qauthor{Bhimrao Ramji Ambedkar} 
\end{savequote}

\chapter{Conclusion}
\justify
\begin{doublespace}

In this thesis, we survey the mathematics underlying elliptic curves and paring-friendly elliptic curves and their benefits in several emerging internet technology. In particular, we discussed that the most famous families of pairing-friendly elliptic curves, such as BN, BLS12, BLS24, KSS16, and KSS18, are vulnerable to the recent Tower Number Field Sieve (TNFS) and its variants attack. Therefore, we compared and estimated the practical security of several families of pairing-friendly curves. We found those families of curves better than BN, KSS, and BLS in terms of the required key size. 

After investigating the pairing-friendly curves, we discussed the key-escrow and its by-product problems associated with identity-based cryptosystems, one of the widely accepted applications of pairing-based cryptosystems. We gave an overview of solutions to these problems and proposed an efficient and secure key issuing scheme. Thereby, demonstrating that they are feasible to escrow-free identity-based encryption that is secured against confidentiality and escrow-free identity-based signature scheme that is forgeable secured.

We further discussed the security inconsistencies of the electronic voting system (EVS) and proposed two efficient identity-based blind signature (IDBS) schemes: IDBS-I and IDBS-II, secured against the existential forgery attack. We design a state-of-the-art End-to-End verifiable internet voting (E2E-VIV) system whose security is defined by the IDBS-II scheme. The performance analysis demonstrated that the E2E-VIV system is efficient and secured against the required security goals on comparing the related schemes. It also allows batch verifiability to verify multiple ballots and vote ballots simultaneously. 

In the subsequent work of the thesis, we discussed data outsourcing on cloud storage, wherein we found that it is challenging to audit the integrity of data on cloud storage for resource-constraint devices. Therefore, we presented an identity-based blind signature with a message recovery (IDBS-MR) scheme and formulated its security. The system is secured against an existential forgery attack under the adaptive chosen message and ID attacks (EF-ID-CMA) and ECDL problems. Further, we have presented a privacy-preserving data outsourcing mechanism that audits the integrity of stored data for resource-limited devices in cloud computing using the proposed IDBS-MR scheme.

Additionally, we discussed some security issues associated with two emerging areas: video-on-demand streaming with an untrusted service provider and the eHealth care system. We presented an escrow-free identity-based signcryption (EF-IDSC) to address the first issue and implement a secure peer-to-peer video-on-demand (P2P-VoD) streaming system, whose security is based on the proposed EF-IDSC scheme. To address the latter issue, we proposed escrow-free identity-based aggregated signcryption (EF-IDASC) schemes and implemented a secure cloud-assisted healthcare system that achieves public verifiability based on the IDASC scheme. Both systems are secured toward confidentiality and forgeability attacks. 

Finally, we proposed two identity-based authenticated key agreement protocols to obtain a secure session key in the health care system. In this work, we have first proposed a one-round, three-party authenticated ID-based key agreement protocol (OR-3PID-KAP) whose security is based on solving ECDLP and BDHP. Compared to the related schemes, the proposed system has the least computational cost, bandwidth cost, and message exchange. Second, we proposed an identity-based anonymous authentication and key agreement (IBAAKA) protocol for WBAN in the cloud-assisted environment, which achieves mutual authentication and user anonymity.

\section{Contribution in COVID-19}

An emerging infectious disease, coronavirus disease 2019 (COVID-19), has been reported in Wuhan, China, and subsequently spread worldwide within a couple of months. We have observed that the rate of infected patients is overwhelmed by the limited medical institutions, which can be addressed and solved using emerging digital technologies. Inspired by this problem, we have worked on COVID-19 and proposed a state-of-art privacy-preserving medical data-sharing system based on Hyperledger Fabric (MedHypChain) and regularised it to implement the patient-centred interoperability healthcare system. We have analyzed the performance of MedHypChain in three metrics (latency time, execution time, and throughput) for up to 20 permissioned nodes. At last, we compare MedHypChain with new blockchain-based healthcare systems in terms of time, communication, and high-level security features and found that MedHypChain achieves all security features. Here, we give a brief overview of the model.

\subsection{Network Model}
Our proposed MedHypChain system consists of the following participants: Network Administrator (NA), Client, Personal Digital Assistant (PDA), Permissioned Peer (PP), and Tracing Agent (TA).
\begin{itemize}
    \item 	\textbf{NA}: It is a trusted entity responsible for registering other participants in the network. Any Client, PDA, or PP must obtain credentials (key information) from NA before joining the proposed network. NA will play the same role as MSP in Hyperledger Fabric but does not manage peers' identities. 
	\item \textbf{Client}: A client is a user or organization who wishes to use the ledger service. Each client owns the credential obtained from NA for achieving secure, anonymous, and traceable data.
	\item \textbf{PDA}. It is a centralized device (e.g., smartphone) with adequate computation power and storage but is not trustworthy as it is easy for any adversary to retrieve the patient’s sensitive data by physically stealing the phone or statistical attack on it. It can collect the data from the client securely, verify it without knowing anything about the data, and aggregate it in its storage. Besides, the PDA installed a fabric software development kit (SDK) that allows it to interact with the Fabric blockchain and provides a simple API to query data from the ledger and submit the transition to a ledger.
	\item \textbf{PP}: The authorized entities responsible for managing, processing, and maintaining transactions. We adopt execute-order-validate architecture and classify it into Endorsing and ordering peer (E\&OP) and committing peer (CP). The E\&CP is accountable for endorsing the transaction using chaincode (smart contract) and updating the block in the ledger notifying the user, while OP is accountable for ordering the endorsed transactions PBFT consensus protocol. 
	\item \textbf{TA}: It is a trusted entity and is accountable for tracking users' malicious behaviour and revocation. It possesses the account's address and password $(add_t,pwd_t)$ and the private key $sk_t$ corresponding to its identity $ID_t$. Also, it has an extra tracing key to accomplish its responsibility.
\end{itemize}

\subsection{System Components}
The main idea of the proposed MedHypChain model is to design a platform for patient-centric health record sharing in the COVID-19 crisis among the participants anonymously, securely, accurately and efficiently. We employ a permissioned blockchain via Hyperledger that will preserve the privacy and confidentiality of data by adopting suitable encryption and access control schemes. 

The proposed MedHypChain model mainly consists of the following phases: Initialization, Registration, Transaction, Transition, Chaining and Tracing, as defined below.
\begin{itemize}
    \item \textbf{$(msk,pp)\leftarrow Initialization(1^k)$}. This initialization phase will respond to a public parameter $pp$ and master key msk according to the input security parameter $k$. 
	\item \textbf{($add_i,psw_i,sk_i) \leftarrow Registeration(msk,pp,ID_i)$}It allows the participants to registers their account address $add_i$ and account password $pwd_i$. Also, it generates a private key $sk_i$ corresponding to the participant’s $ID_i$. 
	\item \textbf{$(\sigma_p,Cx,add_p) \leftarrow TranPrsl(data,pws_s,sk_s,RS)$}. It signcrypts data using the sender’s private key $sk_s$ and set of recipient identities $RS={ID_1,ID_2,...,ID_n}$ and outputs the signcrypted data Cx. 
	\item \textbf{$\{0 \text{abort},1/(\sigma_e,Tx,add_e)\} \leftarrow Endorse(\sigma_p,Cx,add_p,pws_e)$}. This phase allows the PP to unsigncrypt the transaction $C_x$ using his private key $sk_{pp}$ corresponding to identity $ID_{pp}  \in RS$, validate the transaction $Tx$ using $add_p$, sign it with its private key $pwd_p$. 
	\item \textbf{$\{0,1\} \leftarrow Ordering(\sigma_e,Tx,add_e)$}. It verifies the signature $\sigma_e$ using $add_e$, validates the transaction using the PBFT consensus mechanism, and orders the valid transaction $Tx$. 
	\item \textbf{Commit\&Update(Tx)}. It validates each transaction in a block, updates the blocks in the ledger and notifies the sender. 
	\item \textbf{$ID \leftarrow Tracing(sk_{ta},Tx)$}. This allows TA to trace the actual identity of any malicious transaction. It takes TA private key $sk_{ta}$, malicious transaction $Tx$ and outputs the user’s identity $ID_i$. 
\end{itemize}
	
\subsection{Regulation of MedHypChain for Patient-Centred Interoperability}
Here, we will discuss how the proposed MedHypChain regularizes constructing a private patient-centred care blockchain network in which the data as a transaction is accessible to authorized MS. The proposed patient-centred interoperability healthcare system consists of four participants: NA, Patient, and medical server (MS) implemented under the framework of the proposed MedHypChain. The NA is the same in MedHypChain, and the patient and MS are the two organizations that interoperability manages the patient PHI, record in the ledger, and share over MedHypChain. A patient is a person who has been seen as a symptom of coronavirus or a quarantined person who has a chance of disease. Each patient is surrounded by WBAN, which comprises various tiny sensors with limited battery life and storage space, installed on/outside the patient's body (wearable sensors) or deployed in the patient's tissues (implanted sensors). It collects the patient's PHI data, and due to a limited broadcasting range, it stores PHI in a personal assisted device (PDA). For simplicity, we use smartphones as PDAs with the same functionality as the given PDA in MedHypChain, a network coordinator that helps WBAN communicate with the blockchain network. The MS is a device in the medical institution which can access the patient's PHI and diagnoses the patient's diseases based on their resulting PHI. 

Each patient manages two blockchains: patient-data blockchain $B_1$, and patient-prescription blockchain $B_2$. The first blockchain $B_1$ is created by the patient that contains the patient's health-related information, such as a patient's PHI, his identity, and the physician's name. The blockchain $B_2$ is created by the MS that contains the patient's diagnosis-related information, such as the identity of the physician assigned, PHI, and prescription details. Similarly, MS maintains each patient's blockchains $B_1$ and $B_2$. The MS will access the patient's PHI data if the patient allows MS to access blockchain $B_1$. It will process the data on its server and diagnose a patient remotely. On the flip side, patients via PDA will access the blockchain $B_2$ and read the prescription Pr if the patient is allowed to access blockchain $B_1$. The working of patient-centred health interoperability is defined as follows.

A patient is surrounded by a WBAN, which includes various wearable, implanted, or off-body sensors, each capable of sensing, processing, sampling, and communicating the medical signal to the recipient. For instance, an EEG sensor senses brain electrical activity and a breathing sensor senses respiration. Various sensors sense the patient's body and obtain the PHI $data_i$. Any sensor (leader) aggregates the PHI data as $data = \sum data_i$ and retrieves the current value of vrs from the world state that is maintained by its PDA and proposes a transaction $Tx_P$ as $Tx_P = (vrs+1)|(|data|)|RS|(|RW|)|CID||ESign_{pws_s}(data,vrs+1)||add_s$ via Transaction proposal. The sensor sends the transaction $Cx$ to the PDA, where the PDA verifies the data via EVer and keeps the health data encrypted so that PDA cannot access the data. PDA broadcasts $Cx$ to E\&CP for endorsing Endorsement. Now, E\&CP runs the Inspection and ordering algorithm to validate and order the transaction. If the transaction is valid, E\&CP forwards the transaction to PDA. The PDA passes the transaction to E\&CP for validating the transaction via Committing transaction. After commitment, E\&CP appends the block to the blockchain and notifies the PDA. In this way, patient PHI is stored in the channel blockchain.

MS will access the transaction from the patient-created permissioned blockchain if the patient permits. MS will access data and examine the PHI data based on their experience. MS will suggest a prescription Pr (which includes the instruction to the sensor to actuate as per their command) and create a transaction Cx using the Transaction algorithm. Like patients, MS will make their permissioned blockchain that blocks transactions. Each transaction $Cx=(vrs+1)|(|Pr|)|data|(|RS||RW|)|CID||ESign_(pws_m) (data,vrs+1)||add_m$ is created under MS's private key $sk_m$, similar to the patient side, which consists of prescription Pr and patient's PHI data. Accordingly, the transaction is endorsed by E\&CP via Endorsement, ordered by OP via Inspection and order, and updated by CP via Committing in the blockchain. During blockchain creation, if any participant doubts any illegal transaction in the network, he can request the TA, where TA will trace the participant's original identity for an unlawful transaction via the Tracing algorithm. 

Due to the lack of pages, we give only an overview of the model in this chapter. The complete implementation is given in the full paper.

\section{Future Directions}

    In the future, I am interested in working in post-quantum cryptography and secure computation techniques (for example, homomorphic encryption and secure multi-party computation).

\begin{enumerate}
   
    \item \textit{\textbf{Post-quantum cryptography}}: Several scientists and researchers are developing quantum-based computers that utilize quantum-mechanical phenomena like superposition and entanglement. These computers have the potential to perform complex computations incredibly quickly, such as solving number theory-based mathematical problems using Shor's algorithm. However, current quantum computers lack the processing power to attack any real cryptographic systems. Even so, researchers are continuously working on constructing and implementing new algorithms for quantum computers, making it crucial to create a secure cryptosystem against these attacks. This is where post-quantum cryptography comes in. Post-quantum cryptography is a public-key cryptosystem that is considered to be secure against attacks by quantum computers. Examples of post-cryptosystem include lattice-based cryptosystems and NTRU. I am very interested in working with the lattice-based cryptosystem after I complete my PhD.
    \item \textit{\textbf{Secure computation}}: I am also fascinated by the fields of Secure Multi-Party Computation (MPC) and Homomorphic Encryption in the realm of cryptosystems. MPC allows multiple parties to collaborate on computations without revealing their inputs, making it ideal for secure voting and contract signing. Meanwhile, Homomorphic Encryption enables computations on encrypted data, which helps audit data  on sensitive information. These areas interest me greatly, and I am excited to contribute to their progress and explore their potential. 
\end{enumerate}

\end{doublespace} \label{chapter8}

\singlespacing


\clearpage

\addcontentsline{toc}{chapter}{References}
\bibliographystyle{plainnat}
\bibliography{references}
\end{document}